\PassOptionsToPackage{prologue,dvipsnames}{xcolor}
\documentclass[12pt,oneside]{memoir}
\usepackage{titlesec}

\setsecnumdepth{subsection}% number subsections
 
\makeatletter
\renewcommand{\fnum@figure}{Figure \thefigure}
\makeatother
 
\usepackage{float}
\usepackage{tikz-cd}
\usepackage{subcaption}
\usepackage{quiver}

\newcommand{\bea}{\textsc{\textbf{Bean}}}
\newcommand{\LangS}{$\Lambda_S$}

\newcommand{\inl}{\textbf{inl }}
\newcommand{\inr}{\textbf{inr }}
\newcommand{\op}{\textbf{op}}
\newcommand{\slet}[3]{\mathbf{let} \ {#1} = {#2} \ \mathbf{in} \ {#3}}
\newcommand{\dlet}[3]{\mathbf{dlet} \ {#1} = {#2} \ \mathbf{in} \ {#3}}
\newcommand{\add}[2]{\mathbf{add} \ {#1} \ {#2}}
\newcommand{\sub}[2]{\mathbf{sub} \ {#1} \ {#2}}
\newcommand{\mul}[2]{\mathbf{mul} \ {#1} \ {#2}}
\newcommand{\nmul}[1]{\mathbf{mul}{#1}}
\newcommand{\dmul}[2]{\mathbf{dmul} \ {#1} \ {#2}}
\newcommand{\fdiv}[2]{\mathbf{div} \ {#1} \ {#2}}

\newcommand{\stexttt}[1]{\texttt{{#1}}}

\usepackage{amsmath}
\usepackage{amsthm}
\usepackage{mathpartir}
\usepackage{mathtools}
\usepackage{xspace}
\usepackage{textcomp}
\usepackage{amssymb}
\usepackage{bussproofs}
    \EnableBpAbbreviations
\usepackage{stmaryrd}
\newcommand{\R}{\mathbb{R}}

\newcommand{\NNR}{\mathbb{R}^{\geq 0}}

\newcommand{\F}{\mathbb{F}}
\newcommand{\tp}{$\tilde{\textbf{P}}$}
\newcommand{\p}{$\textbf{P}$}
\newcommand{\Bel}{\textbf{Bel}}
\newcommand{\id}{\text{id}}
\newcommand{\Met}{\mathbf{Met}}
\newcommand{\Set}{\textbf{Set}}
\newcommand{\eff}{\textcolor{black}}
\newcommand{\ceff}{\textcolor{black}}
\newcommand{\denot}[1]{\llbracket {#1} \rrbracket}
\newcommand{\pdenot}[1]{\llparenthesis {#1} \rrparenthesis}

\newcommand{\dom}{\text{dom}}

\newcommand{\pdenotid}[1]{{\llparenthesis {#1} \rrparenthesis}_{id}}
\newcommand{\pdenotap}[1]{{\llparenthesis {#1} \rrparenthesis}_{ap}}

\newcommand{\Del}{\textbf{Del}}

\newcommand{\stepid}{\Downarrow_{id}} 
\newcommand{\stepap}{\Downarrow_{ap}}

\newcommand{\Uid}{U_{id}}
\newcommand{\Uap}{U_{ap}}

\newcommand{\tensor}{\otimes}

\theoremstyle{definition}
\newtheorem{definition}{Definition}
\newtheorem{theorem}{Theorem}
\newtheorem{lemma}{Lemma}
\newtheorem{corollary}{Corollary}
\newtheorem*{remark}{Remark}

\newtheoremstyle{named}{}{}{\itshape}{}{\bfseries}{.}{.5em}{\thmnote{#3}#1}
\theoremstyle{named}
\newtheorem*{namedtheorem}{}

\newcommand{\oset}[1]{\xrightarrow{\raisebox{-0.45ex}[0ex][0ex]{
\text{\footnotesize{#1}}}}}

\usepackage{newtxtext}
\usepackage[frenchmath,libertine]{newtxmath}
\usepackage{microtype}
\usepackage{geometry}
\usepackage{import}
\usepackage{tablefootnote}
\usepackage{multirow}
\usepackage{colortbl}
\usepackage[inline]{enumitem}
\usepackage{pdflscape}

\usepackage{listings}
\lstdefinelanguage{fz}{
    mathescape=true,
    morekeywords=[2]{function,let,in,ret,rnd,fun,def,return},% 
    morekeywords=[3]{num,f64},%
    columns=fullflexible,
    sensitive,
    basicstyle =\linespread{2},
    keepspaces=true,
    showstringspaces=false,
    breaklines=true,
    breakatwhitespace=false,
    basicstyle=\sffamily,
    keywordstyle=[2]{\sffamily\color{Mahogany}},
    keywordstyle=[3]{\sffamily\color{Plum}},
    keywordstyle=[1]{\sffamily\color{ForestGreen}},
    morecomment=[l][\color{Gray}]{\/}
   }[keywords,comments,strings]%
\DeclareTextFontCommand{\linec}{\sffamily\selectfont}

\newcommand{\rnd}{{\textbf{rnd}}}
\newcommand{\rndd}{\textbf{rnd64}}
\newcommand{\rnds}{\textbf{rnd32}}

\newcommand{\ret}{\textbf{ret}}
\newcommand{\tlet}{\textbf{let}}
\newcommand{\M}{{\color{black}{M}}}
\newcommand{\tin}{\textbf{in}}

\newcommand{\letmx}{{\linec{{\textbf{let}$_\linec{\textbf{M}}$}}}}
\newcommand{\mlet}[3]{\linec{{let}~{$#1$} = {$#2$} 
  {in} {$#3$}}}
\newcommand{\letm}[3]{\linec{{\textbf{let}$_\linec{\textbf{M}}$}~{$#1$} = {$#2$} 
  {\textbf{in}} {$#3$}}}
\newcommand{\letc}[3]{\linec{\textbf{let}~{[$#1$]} = {$#2$} 
  {\textbf{in}} {$#3$}}}
\newcommand{\lett}[3]{\linec{{\textbf{let}}~{$#1$} = {$#2$} 
  {\textbf{in}} {$#3$}}}
\newcommand{\case}[3]{\linec{{\textbf{case}}~{$#1$}~\textbf{of}~(\inl{$#2$} {}$\mid${} \inr{$#3$})}}

\newcommand{\boxx}[1]{[#1]}

\newcommand{\Lang}{\textbf{\textsc{NumFuzz}}}
\newcommand{\bean}{\textsc{\textbf{Bean}}}
\usepackage{listings}
\usepackage[scaled]{beramono}
\usepackage{adjustbox}
\usepackage{verbatim}

\definecolor{verbgray}{gray}{0.9}
   {\par\noindent\adjustbox{margin=1ex,
    bgcolor=verbgray,margin=0ex \medskipamount}
    \bgroup\minipage\linewidth\verbatim}%
   {\endverbatim\endminipage\egroup}

\newcommand{\unit}{\text{unit}}
\newcommand{\num}{\text{num}}
\newcommand{\dnum}{\text{dnum}}

\newcommand{\tand}{\ \& \ }
\newcommand{\bang}[1]{{!_\linec{#1}}}

\newcommand{\monad}[1]{{{M}_{{#1}}}}
\newcommand{\rnderr}{\varepsilon}
\newcommand{\lin}{\multimap}
\newcommand{\err}{\textbf{err}}

\usepackage{hyperref}
\hypersetup{
    colorlinks,
    linktoc=all,
    citecolor=PineGreen,
    filecolor=black,
    linkcolor=RoyalBlue,
    urlcolor=BrickRed
}
\usepackage[nameinlink,capitalize]{cleveref}
\usepackage{epigraph}

\usepackage{natbib}

\title{Type-Based Approaches to Rounding Error Analysis}
\author{Ariel Eileen Kellison}
\date{}

\begin{document}

%%% TITLE
\maketitle
\thispagestyle{empty}
\clearpage

%%% ABSTRACT
\thispagestyle{empty}
\section*{Abstract}
{
This dissertation explores the design and implementation of programming
languages that represent rounding error analysis through typing.  A rounding
error analysis establishes an \emph{a priori} bound on the rounding errors
introduced by finite-precision arithmetic in a numerical program, providing a
measure of the quality of the program as an approximation to its ideal,
infinitely precise counterpart.  This can be achieved by measuring the distance
between the computed output and the ideal result, known as the forward error,
or by determining the smallest distance by which the input must be perturbed to
make the computed output match the ideal result, known as the backward error.
Due to the complex interaction between the rounding errors produced locally by
a program and those propagated from its inputs, performing a rounding error
analysis is known to be challenging.  Moreover, while most programs can be
shown to have a forward error bound, even if trivially large, many programs
cannot be shown to have a backward error bound. This makes the task of deriving
backward error bounds for complex programs even more challenging than the
corresponding task of deriving forward error bounds. 

One way to view programs that use finite-precision arithmetic is as
computations that \emph{produce} rounding error, while also \emph{consuming}
and \emph{amplifying} rounding errors passed to them as inputs from other
functions. This perspective aligns closely with the notions of \emph{effects}
(what programs produce aside from values) and \emph{coeffects} (how programs
use their inputs).  Both effectful and coeffectful program behaviors can be
analyzed using type-based approaches, where graded monadic types track effects,
and graded comonadic types track coeffects. While recent work has studied the
interaction of effects and coeffects in various domains, the work presented in
this dissertation is the first to investigate these behaviors in the context of
numerical analysis and the propagation of rounding errors in numerical
programs.

The first part of this dissertation demonstrates that it is possible to design
languages for forward error analysis, as illustrated with \Lang, a functional
programming language whose type system expresses quantitative bounds on
rounding error. This type system combines a sensitivity analysis, enforced
through a linear typing discipline, with a novel graded monad to track the
accumulation of rounding errors. We establish the soundness of the type system
by relating the denotational semantics of the language to both an exact and
floating-point operational semantics.  To illustrate the capabilities of our
system, we instantiate \Lang{} with error metrics from the numerical analysis
literature and incorporate rounding operations that accurately model key
aspects of the IEEE 754 floating-point standard. Furthermore, to demonstrate
the practical utility of \Lang{} as a tool for automated error analysis, we
have developed a prototype implementation capable of inferring error bounds.
This implementation produces bounds competitive with existing tools, while
often achieving significantly faster analysis times.

The second part of this dissertation explores a type-based approach to backward
error analysis with \bea{}, a first-order programming language with a linear
type system that can express quantitative bounds on backward error.  \bea{}'s
type system combines a graded coeffect system with strict linearity to soundly
track the flow of backward error through programs.  To illustrate \bea{}'s
potential as a practical tool for automated backward error analysis, we
implement a variety of standard algorithms from numerical linear algebra in
\bea, establishing fine-grained backward error bounds via typing in a
compositional style. We also develop a prototype implementation of \bea{} that
infers backward error bounds automatically. Our evaluation shows that these
inferred bounds match worst-case theoretical relative backward error bounds
from the literature, underscoring \bea{}'s utility in validating a key property
of numerical programs: \emph{numerical stability}.
}
\clearpage

%%% ACKS
\thispagestyle{empty}
\section*{Acknowledgments}
The work presented in this dissertation is the result of years of patient
guidance and constant encouragement from my academic mentors:  David Bindel,
Andrew Appel, and Justin Hsu. David  encouraged me to follow my curiosity in
using ideas from formal methods to solve problems in numerical analysis, while
Andrew and Justin patiently taught me the technical tools to explore and
develop those ideas further. I am truly grateful for the time each of them has
taken to teach me countless lessons about reading, writing, and computer
arithmetic.  

The core of this dissertation builds on papers co-authored with David Bindel,
Justin Hsu, and Laura Zielinski. Laura made significant contributions to the
implementation of \bea{} and is responsible for the work presented in
\Cref{app:algorithm}.

Beyond my primary mentors, many others have shaped the work presented in this
dissertation through their mentorship and collaboration. Bob Constable first
introduced me to formal methods, and his enthusiasm for type theory continues
to inspire me today. It is hard to say what my PhD research would have looked
like without the influence of Bob and the rest of the Nuprl research
group, including Mark Bickford, Richard Eaton, and Liron Cohen. 

At Sandia National Laboratories, members of the formal methods group provided
invaluable expertise. Geoffrey Hulette first directed me towards foundational
approaches to verifying floating-point programs, Sam Pollard provided friendly
expertise on the nuances of floating-point arithmetic, and Heidi Thornquist
grounded our speculative discussions with practical, applications-focused
insights. 

I am also grateful for collaborations with Jean-Baptiste Jeannin and Mohit
Tekriwal. Our work on formally verifying numerical algorithms in Coq led to
many rewarding QEDs and helped solidify my understanding of the IEEE Standard. 
\clearpage

%%% CONTENTS
{\hypersetup{linkcolor=black}
\tableofcontents* }
\thispagestyle{empty}
\clearpage

\aliaspagestyle{chapter}{empty}
\pagestyle{companion}
\setlength{\headwidth}{\textwidth}

\pagenumbering{arabic}

\chapter{Introduction}\label{topintro}
\begin{SingleSpace}
\epigraph{O dear Ophelia!\\ I am ill at these numbers:\\ I have not art to 
reckon my groans.}{Hamlet (Act II, Scene 2, Line 120)}
\end{SingleSpace}
Many fields of computing are concerned with designing algorithms for solving
problems in continuous mathematics. These algorithms are usually specified
using real numbers but are implemented in a finite-precision arithmetic, like
floating-point arithmetic. This approximation leads to rounding errors which
can degrade the accuracy of results. When accuracy guarantees are required, a
\emph{rounding error analysis} can provide an \emph{a priori} bound on the
rounding error in a floating-point result. There are many examples from
scientific and engineering domains where this type of analysis is essential:
\begin{itemize}
\item{\textbf{Solid modeling:}} 
In computer-aided design (CAD) and computer graphics, rounding error bounds are
used to guarantee the correctness of computer representations and manipulations
of geometric shapes  \citep{Sherman:2019:Solid,Hu:1996:Solid}. 
\item{\textbf{Secure multiparty computation:}} 
In encrypted signal processing, floating-point arithmetic is used to process
signals that might be real-valued; a rigorous rounding error analysis
guarantees that implementations will not leak sensitive information by
producing exceptional values
\citep{kamm:2015:secure,Franz:2011:secure,aliasgari:2012:secure}.
\item{\textbf{Numerical linear algebra:}}
In numerical linear algebra, specifications of basic operations like the dot
product and matrix multiplication are often accompanied by a rounding error
analysis describing their expected accuracy and stability
\citep{LAPACK,blackford:2002:blas,Li:2002:extended}. 
\item{\textbf{Zero-knowledge proofs for machine learning:}} In neural-network
inference, zero-knowledge proofs guarantee that sensitive data remain secret
\citep{Rosetta,weng:2021:mystique}; a rounding error analysis is central to the
development of efficient methods for constructing these proofs
\citep{Garg:22:secure}.
\end{itemize}

However, while accuracy guarantees are needed in many areas of computing,
there are very few tools available to help programmers construct them. In the
examples listed above, the rounding error analyses are done by hand. This
process requires specialized knowledge of how to reason about the subtle
details of floating-point arithmetic, and becomes increasingly impractical as
programs grow larger, rely on mixed precision computations, include
dependencies on multiple modules, or use foreign function interfaces. 

In this dissertation, we propose a novel approach to developing numerical
programs with accuracy guarantees: designing \emph{languages} that unify the
tasks of writing programs and reasoning about their accuracy. Our thesis is
that it is feasible to design and implement languages that represent rounding
error analysis through typing, and that these languages have the potential to 
make reasoning about numerical accuracy convenient for programmers.

%%%%%
\section{Rounding Error and Program Distances}
We begin by briefly introducing some of the fundamental concepts related to
designing languages for rounding error analysis. Suppose \p~is an algorithm
specified on real numbers and \tp~is an implementation of $\textbf{P}$ that
uses floating-point arithmetic, which might introduce rounding errors during
computation. We can think of \p~and \tp~as being non-equivalent but closely
related programs: every floating-point operation in $\tilde{\textbf{P}}$
approximates its real valued counterpart in $\textbf{P}$ with an accuracy that
depends on the number of bits in the floating-point format, as well as the
rounding strategy that is used.  The purpose of a rounding error analysis is to
quantitatively determine how close \p~and \tp~are by deriving an \emph{a
priori} bound on the effects of rounding errors. In numerical analysis, the
notion of the ``closeness'' of programs---or \emph{program distance}---can be
determined by performing either a \emph{forward} or \emph{backward} error
analysis \citep{Higham:2002:Accuracy}. 

\subsection*{Forward Error Analysis} 
Even if the programs \p~and \tp~are closely related, they might produce
different outputs when given the same input. This behavior leads to the
following natural question: for a given input $x$, how close is the output
$\tilde{\textbf{P}}(x)$ to the target output $\textbf{P}(x)$? The distance
between $\tilde{y} = \tilde{\textbf{P}}(x)$ and $y= \textbf{P}(x)$ is known as
the \emph{forward error}, and a bound on the forward error is obtained by
performing a \emph{forward error analysis}. The forward error is represented by
the quantity $\delta y$ in \Cref{fig:intro}.

\begin{figure}
\begin{center}
\begin{tikzpicture}[scale=1.2]   
    % Nodes
    \node (A) at (0, 1.75) {Domain $X$};
    \node (A) at (4, 1.75) {Codomain $Y$};
    \node (x) at (0, 0) {$x \in X$};
    \node (y) at (4, 1) {$y \in Y$};
    \node (b) at (0, -2) {$\tilde{x} \in X$};
    \node (c) at (4, -3) {$\tilde{y} \in Y$};
    % Arrows
    \draw[thick,->] (x) -- (y) node[midway, above] {\large ${P}$};
    \draw[dashed,thick] (x) -- (b) node[midway, left] {$\delta x$};
    \draw[dashed,thick] (y) -- (c) node[midway, left] {$\delta y$};
    \draw[thick, ->] (x) -- (c) node[midway, above] {\large$\tilde{P}$};
    \draw[thick, ->] (b) -- (c) node[midway, below] {\large${P}$};
\end{tikzpicture}
\end{center} 
\caption{Forward and backward error. \label{fig:intro}}
\end{figure}
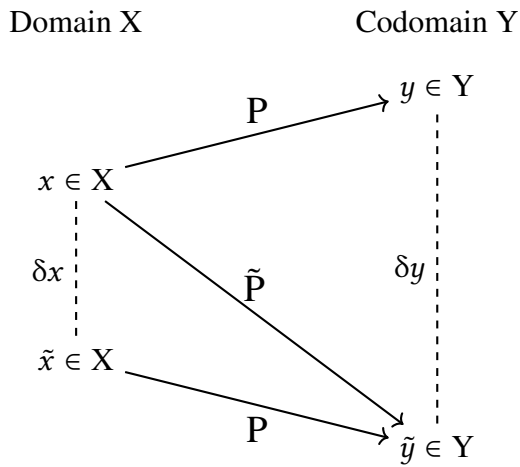

Accurate programs are typically characterized by having a small forward error.
Unfortunately, automated tools that statically compute sound \emph{a priori}
bounds on forward errors suffer from significant limitations. These tools tend
to overestimate rounding errors, are mostly restricted to analyzing
straight-line programs (those without loops or conditional branches), and
struggle to scale to larger programs with more than a few hundred
floating-point operations.  
Another significant limitation of many of these tools is their inability 
to account for how programs amplify and propagate rounding errors from their 
inputs---a critical factor for analyzing and scaling to realistic software
\footnote{Recent work has begun to address this limitation
\citep{PRECISA:2024,Abbasi:2023:compose}.}. 

\subsection*{Backward Error Analysis}
We can also compare $\textbf{P}$ and $\tilde{\textbf{P}}$ by answering the
following question: do $\textbf{P}$ and $\tilde{\textbf{P}}$ behave like
equivalent programs when given non-equivalent but closely related inputs? More
specifically, given a point $x$ in the domain of \tp, is there an input
$\tilde{x}$ close to $x$ such that
$\textbf{P}(\tilde{x})=\tilde{\mathbf{P}}(x)$? This question is answered
through a \emph{backward error analysis}, and the distance between the input
$x$ and the input $\tilde{x}$ is known as the \emph{backward error}. In
\Cref{fig:intro}, the backward error is represented by the quantity $\delta x$.  
While an accurate program is characterized by a small forward error, a backward
error analysis gives more insight into the quality of an approximating program:
a large backward error suggests that the programs $\tilde{\textbf{P}}$ and
$\textbf{P}$ are not closely related. Moreover, a small backward error implies
a small forward error so long as $\textbf{P}$ is robust to small perturbations
in its inputs. 

However, there are several challenges presented by backward error analysis.
First, if $\textbf{P}$ is not surjective, then a solution produced by
$\tilde{\textbf{P}}(x)$ might be outside of the codomain of $\textbf{P}$. In
such cases, the value $\tilde{x} \in X$ in \Cref{fig:intro} does not exist.
Another complication with backward error analysis is that backward error
guarantees are generally not compositional. This means that a bound on the
backward error of a program cannot be reliably derived from bounds on the
backward error of its individual components. As a result, statically analyzing
backward error remains an open challenge.

\subsection*{Program Distance}
The concept of program distance has recently emerged as formalism that offers a
more refined analysis of the behavior of programs compared to the standard
notion of program equivalence
\citep{gavazzo:2019:phd,dal:2022:effectful,Amorim:2017:metric,Gavazzo:2018:applicative}.
This is particularly true for programs with effectful properties---programs
whose behavior have an effect on the environment in which they are
evaluated---and coeffectful properties---programs whose behavior depends on the
environment in which they are evaluated.  While various formalisms for
reasoning about program distance have been used to study probabilistic
languages \citep{Amorim2019,Crubille:2015:metric} and languages for
differential privacy \citep{BunchedFuzz,Fuzz,chatzikokolakis:2014:generalized}, 
the characterization above of forward and backward error suggests that these
ideas can also  be applied to study languages for forward and backward error
analysis.   Intuitively, forward rounding error can be modeled using both
effects and coeffects: the total forward rounding error resulting from the
evaluation of a finite-precision program depends on the rounding error produced
locally by the function, as well as the errors propagated from the inputs.  On
the other hand, backward error can be described as a coeffectful property: when
backward error bounds exist, they characterize the relationship between a
finite-precision computation and its ideal counterpart in terms of
perturbations to the input space. 

%%%
\section{Type-Based Approaches to Rounding Error Analysis}
One approach to developing correct numerical software is to write programs and
the specifications of their numerical behavior together, using a typed
programming language. In carefully designed languages that soundly represent
rounding error analysis through typing, types can capture properties like
accuracy requirements and acceptable error bounds.  By representing these
properties through typing, accuracy requirements can be enforced via type
checking, ensuring that programs adhere to their intended numerical behavior.
For many applications, this approach potentially eliminates the need for manual
reasoning about numerical accuracy or reliance on external error analysis
tools, as the language itself provides guaranteed rounding error bounds.

To illustrate this approach in practice, consider the following example.
Elementary functions like \linec{sin}, \linec{exp}, and \linec{log} are often
implemented in math libraries using polynomial approximations, and the accuracy
of an implementation is highly dependent on the method used for polynomial
evaluation.  For instance, the following example is an IEEE binary64 polynomial
approximation to the function $2^x$ for $0 < x \le 1/512$:
\begin{align*}
P(x) = 1 + \frac{6243314768165347}{2^{53}}\cdot x &+ 
\frac{8655072058047061}{2^{55}} \cdot x^2 \\ &+ 
\frac{3999491778607567}{2^{56}}\cdot x^3 + 
\frac{5548134412280811}{2^{59}}\cdot x^4
\end{align*}
According to the analysis given by \citet{muller:2016:elementary}, the relative
forward error due to rounding of an implementation 
\lstinline{exp2 : f64 $\rightarrow$ f64} of $P(x)$ should ideally be several 
orders of magnitude smaller than $1.11\times10^{-16}$, which is the tightest 
bound that can be derived if a naive implementation of polynomial evaluation is used.  
In a language with an expressive  type system that explicitly tracks rounding error 
bounds at the level of types, type inference serves as a mechanism for automatically
computing sound rounding error bounds for every program.  In
this setting, as a prelude to the syntax of the languages proposed in this
dissertation, the type of the function \lstinline{exp2} that returns a value of
type \lstinline{f64} and produces at most $1.11\times10^{-16}$ relative
forward error due to rounding is:
\begin{center}
\begin{tabular}{c}
\begin{lstlisting}
exp2 : f64 mapsto M[$\text{\color{ForestGreen}{1.11e-16}}$]f64. 
\end{lstlisting}
\end{tabular}
\end{center}

Functions that require highly accurate implementations of the polynomial
approximation $P(x)$ can specify this fact in the type declarations of their
arguments. For instance, consider the function \lstinline{foo}, defined below:
\begin{lstlisting}
fun foo (accurate_exp2 : f64 $\rightarrow$ M[$\text{\color{ForestGreen}{2.17e-19}}$]f64, y : f64) {
    let x = accurate_exp2(y); // call to an accurate exp2  
    return x + 3.0;
}
\end{lstlisting}
This function is a \emph{higher-order function} that takes as an argument a 
function with the following type: 
\begin{center}
\begin{tabular}{c}
\begin{lstlisting}
accurate_exp2 : f64 $\rightarrow$ M[$\text{\color{ForestGreen}{2.17e-19}}$]f64 
\end{lstlisting}
\end{tabular}
\end{center}
This type specifies a function that returns a value of type \lstinline{f64} with 
at most $2.17\times 10^{-19}$ relative error due to rounding. If a function that 
produces more than  $2.17\times 10^{-19}$ relative rounding error is provided as 
an argument 
to \lstinline{foo}, the program will fail to type check. Specifically,
the type checker for the language enforces \lstinline{foo}'s accuracy 
requirement, rejecting cases where the argument does not meet this constraint. 
For example, passing the function 
\lstinline{exp2 : f64 mapsto M[${\text{\color{ForestGreen}{1.11e-16}}}$]f64} 
which allows a relative error of $1.11 \times 10^{-16}$ as an argument to 
\lstinline{foo} will fail to type check: 
\begin{lstlisting}
let z = foo(exp2); // ERROR mismatched types: expected $\color{Gray} \text{1.11e-16} \le \text{2.17e-19}$
\end{lstlisting}

Thus, by embedding types that track rounding error directly into the type 
system of a language, we can specify and enforce rigorous accuracy requirements
without requiring external tools or performing a manual error analysis. The 
ability to track and verify accuracy constraints at the level of types ensures that programs 
adhere to their intended numerical behavior and provides strong guarantees 
about the reliability of numerical computations. 

Moreover, higher-order functions like \lstinline{foo} highlight how these 
type systems facilitate modular reasoning about accuracy.  
Functions can explicitly require arguments to meet strict error constraints, 
and the type checker for the language enforces these requirements, ensuring 
that only implementations with sufficiently small rounding error can be used.
This guarantees correctness by construction, reducing the risk of subtle 
numerical errors propagating through larger programs. 

To ensure that a program's type accurately represents its 
rounding error bounds, the type system must satisfy a soundness property,
which we call
\emph{error soundness}.  This property establishes that the rounding error 
bound expressed in a program's type, such as the value 
{\lstinline{$\text{\color{ForestGreen}{2.17e-19}}$}} in the type
{\lstinline{M[$\text{\color{ForestGreen}{2.17e-19}}$]f64}},
is not
merely an annotation but is \emph{meaningful}. Specifically, 
if the type system assigns this type to a program, then the program is 
guaranteed adhere to this error bound during execution, thereby
ensuring the correctness of its numerical behavior.
Moreover,
\emph{error soundness} guarantees that these bounds are sound
overapproximations of the true rounding error, providing a reliable foundation
for reasoning about numerical accuracy. 
To establish error soundness for our languages, 
we construct denotational semantic models that assign precise meaning to 
types, capturing the notion of program distance as characterized 
by backward and forward error. We 

\section{Goals and Contributions}
The goal of this dissertation is to demonstrate linguistic features and type
systems that unify the tasks of reasoning about numerical accuracy and writing
numerical programs. It presents novel programming languages designed for both
forward error analysis and backward error analysis and details the development
of formal categorical denotational semantics for these languages, along with
implementations of type inference and type checking algorithms.  The methods
and languages presented in this dissertation serve as a blueprint for future
languages that make reasoning about numerical accuracy more accessible and
convenient for programmers.

Our contributions can be categorized into three  main areas: denotational
semantics, language design, and implementation.  The remainder of this chapter
is devoted to summarizing each of these contributions.  

\subsection*{Contribution 1: {Denotational Semantics}}
A desirable property for tools that automatically bound floating-point rounding
errors is \emph{soundness}: the bounds produced by the tool should
overapproximate the true error. The first contribution of this dissertation is
the introduction of categorical structures that soundly model forward and
backward error analysis. These semantic structures give insight into the
foundational concepts underlying programming languages that can soundly express
rounding error bounds through typing. 

For forward error analysis, \Cref{sec:nfuzz:semantics} defines the
\emph{neighborhood monad}, a novel graded monad on the category of metric
spaces. This monad uses grade information to model bounds on the distance
between two closely related computations. Forward relative error bounds are
obtained by instantiating the neighborhood monad with distance functions
proposed in the numerical analysis literature that approximate relative error.

For backward error analysis, \Cref{sec:bean:semantics} introduces a semantic
structure called \emph{backward error lenses}, which describes computations
suitable for backward error analysis. A backward error lens consists of a
triple of related transformations that collectively satisfy a backward error
guarantee. This structure supports the representation of standard numerical
primitives and their associated backward error bounds. The \emph{category of
backward error lenses} (\Bel{}) serves as a semantic framework for reasoning
about backward error analysis.

\subsection*{Contribution 2: {Language Design}}
Categorical frameworks for forward and backward error analysis provide an
abstract domain for specifying programming language constructs---such as
function definitions, pairs, and conditionals---that preserve the semantics of
an error analysis. To this end, the second contribution of this dissertation is
the design of two languages: \Lang{} for forward error analysis, and \bean{}
for backward error analysis, each equipped with a type system that guarantees
bounded error. 

The type system for \Lang{} is presented in \Cref{sec:nfuzz:language} and is
based on \emph{Fuzz} \citep{Fuzz}, a linear call-by-value $\lambda$-calculus
designed for differential privacy. \emph{Fuzz} uses a linear type system and a
graded comonadic type to statically perform a sensitivity analysis.  \Lang{}
extends \emph{Fuzz}'s type system with a graded monadic type for tracking local
rounding errors.  Additionally, \Lang{} introduces typing rules combining
sensitivity and local rounding error, enabling a compositional approach to
analyzing the rounding error of larger programs.  \Lang{}'s type system
guarantees that programs have bounded forward error, a property
we refer to as \emph{forward error soundness}. We establish this guarantee in
\Cref{sec:nfuzz:sound} by modeling the graded monadic type using the
neighborhood monad described above, and by relating the categorical
denotational semantics to an ideal and floating-point operational semantics. 

\bean{} is a first-order \emph{bidirectional} programming language based on
numerical primitives.  In bidirectional programming languages
\citep{Bohannon:2008:bidirectional,foster:2009:bidirectional,foster:2012:three},
expressions usually denote a related pair of transformations: a \emph{forward}
transformation mapping inputs to outputs, as well as a \emph{backward}
transformation mapping an updated output together with an original input to a
corresponding updated version of the input
\citep{Bohannon:2008:bidirectional}.  To capture backward error analysis in
\bean{}, each expression instead represents a \emph{triple} of transformations:
two forward transformations---one for the \emph{ideal} program and one for the
\emph{approximate} program---and a backward transformation that relates these
forward transformations under the constraints of a \emph{backward error lens}.
\Cref{sec:bean:language} introduces \bean{}'s type system, which  uses strict
linearity to ensure that when subexpressions are composed, the resulting larger
expression also satisfies the constraints of a backward error lens. A graded
coeffect system then computes static bounds on the backward error, based on 
bounds for the numerical primitives of
the language. This type and effect system guarantees that \bean{} programs have
bounded backward error, a property we refer to as
\emph{backward error soundness}.  The proof of backward error soundness for
\bean{} is given in \Cref{sec:bean:sound}.

\subsection*{Contribution 3: {Implementation}}
Many analysis tools have been developed to automatically bound the forward
rounding error of floating-point expressions; however, none of these tools have
employed a type-based approach, nor do they address backward error analysis.
The third contribution of this dissertation, described in
\Cref{sec:nfuzz:implementation} and \Cref{sec:bean:implementation}, is the
implementation of type checkers and grade inference algorithms for both \Lang{}
and \bean{}. Both implementations build on the sensitivity
inference algorithm introduced by \citet{Gaboardi:2013:dfuzz}.

In this type-based approach,
error bounds can be inferred and checked 
for simple numerical
programs: forward error bound in the implementation of \Lang{}, and 
backward error bounds in the implementation of \bean{}. 
This approach is attractive because it automatically provides a
formal proof that a given program has a certain error bound.  

We evaluate our implementation of \Lang{} using a variety of benchmarks from the
literature and demonstrate that it infers error bounds that are competitive
with those produced by other tools. We also show that our implementation is
capable of handling the largest benchmarks in the literature, such as a 128 x
128 matrix multiplication involving over 4 million floating-point operations.
Although our prototype implementation currently only supports a limited set of
primitive floating-point operations (addition, multiplication, division, square
root), our empirical evaluation shows that the time complexity of the inference
algorithm is linear in the size of the \Lang{} program. We therefore expect
that any further extensions to \Lang{} to support additional primitive
operations will not affect the complexity of the type checker and inference
algorithm. This distinguishes our type-based approach from other tools, where
the complexity depends on the specific floating-point operation being
performed. 

For \bean{}, we translate several large benchmarks into the language,
demonstrating that its implementation effectively infers useful error bounds
and scales to handle large numerical programs.  Since \bean{} is the first tool
to statically derive \emph{sound} backward error bounds, a direct comparison
with existing tools is limited.  We therefore evaluate our implementation of
\bean{} using three complementary methods.  First, we compare the backward error
bounds inferred by \bean{} to those from a dynamic analysis tool by
\citet{Fu:2015:BEA}, which provides the only automated quantitative backward
error bounds available, as well as to theoretical worst-case bounds from the
literature.  Additionally, we use forward error as a proxy by deriving forward
error bounds from \bean{}'s backward error bounds, leveraging known values of
the condition number. These forward error bounds are then
compared those produced by other tools, including \Lang{}. 

The artifact for our implementation of \Lang{} is available at
\url{https://zenodo.org/records/10967298}, and the working branch is available
at \url{https://github.com/ak-2485/NumFuzz}. The working branch of \bean{}
is available at \url{https://github.com/ak-2485/NumFuzz/tree/bean}.

\section{Dissertation Outline}
This dissertation connects several topics that have not been previously linked,
including floating-point arithmetic, type systems, category theory, and
bidirectional programming languages. To support this integration, we present
essential mathematical notations and preliminaries in
\Cref{chapter:background}.  

In \Cref{chapter:nfuzz}, we describe \Lang{}, covering its type system
(\Cref{sec:nfuzz:language}), denotational semantics
(\Cref{sec:nfuzz:semantics}), operational semantics
(\Cref{sec:nfuzz:opsemantics}), and soundness guarantee
(\Cref{sec:nfuzz:sound}). We provide examples of how to soundly instantiate the
language parameters in \Cref{sec:nfuzz:examples}, and describe a prototype
implementation along with its evaluation in \Cref{sec:nfuzz:implementation}.
Related work on static analysis techniques for sensitivity and forward rounding
error analysis is discussed in \Cref{sec:nfuzz:related}, and future directions
for \Lang{} are summarized in \Cref{sec:nfuzz:conclusion}. Omitted lemmas and
proofs from this chapter are provided in \Cref{app:nfuzz}.

In \Cref{chapter:bean}, we describe \bean{}, introducing its type system
(\Cref{sec:bean:language}) and denotational semantics
(\Cref{sec:bean:semantics}). The soundness guarantee, along with the necessary
operational constructions, is presented in \Cref{sec:bean:sound}. Examples in
\Cref{sec:bean:examples} demonstrate how \bean{} can be used to establish sound
backward error bounds for various numerical problems through typing. The
implementation and evaluation are described in \Cref{sec:bean:implementation}.
Related work on static analysis techniques for backward error analysis is
covered in \Cref{sec:bean:related}, and \Cref{sec:bean:conclusion} concludes
the chapter and describes future work. Omitted lemmas and proofs from this
chapter  are provided in \Cref{app:bean}.

%%%%%%%%%%%%%%%%%%%
\chapter{Background}\label{chapter:background}
  The languages described in this dissertation link several topics: floating-point
  arithmetic, type systems, category theory, and bidirectional
  programming languages. While the role of type systems as tools for statically
  reasoning about the behavior of programs is well-established in many areas of
  computing, connecting this core idea from the theory of programming languages
  to numerical analysis is a novel contribution of the work described in this
  dissertation. To support the integration of these concepts, this part of the
  dissertation provides definitions, notation, and background results on
  floating-point arithmetic, as well as on type systems and their categorical 
  semantics. 
\section{Floating-Point Arithmetic}\label{background:floatingpoint}
\begin{table}
\centering
\caption{Floating-point format parameters according to the 
	revised IEEE 754-2008 standard.}
\label{tab:formats}
\begin{tabular}{ l r r r r }
 \hline
	\textbf{Parameter}&\textbf{binary16} & 
		\textbf{binary32} & 
	\textbf{binary64} &\textbf{binary128} \\
 \hline
	p   & 11 & 24 & 53 & 113     \\
	emax & 31 & 127  & 1023 & 16383 \\
 \hline
\end{tabular}
\end{table}
A \emph{finite} floating-point number $x$ in a floating-point format 
$\F \subseteq \R$ has the form
\begin{equation}
x =  (-1)^s \cdot m \cdot \beta^{e-p+1} \label{eq:fpdef}
\end{equation}
where 
$\beta \in\{ b\in \mathbb{N} \mid b \ge 2\}$ 
is the \emph{base}, 
$p \in \{prec \in \mathbb{N} \mid prec \ge 2\}$ 
is the \emph{precision}, 
$m \in \mathbb{N} \cap [0,\beta^p)$ 
is the \emph{significand}, 
$e \in \mathbb{Z} \cap
[\text{emin}, \text{emax}]$ 
is the \emph{exponent}, and 
$s \in \{0,1\}$ 
is the \emph{sign} of $x$.
A complete definition of a floating-point format also specifies binary
encodings, and describes how to handle \emph{non-finite} special values, such
as infinities and NaNs.
For finite floating-point numbers, the parameter values for formats defined in
the IEEE 754 standard for floating-point arithmetic \citep{IEEE2019} are given
in \Cref{tab:formats}; in all cases, $\text{emin}=1$-$\text{emax}$. Although
single (binary32) and double (binary64) precision floating-point formats have
historically been the most widely used, modern architectures are increasingly
supporting half precision (binary16). 

Finite floating-point numbers are commonly categorized into two types: normal
and subnormal. A number of the form given in \Cref{eq:fpdef} is considered normal
if its significand $m$ and exponent $e$ satisfy the bounds $e \ge
\text{emin}$ and $\beta^{p-1} \le m < \beta^p$; otherwise, it is said to be
subnormal. If the magnitude of the result of a floating-point operation is 
subnormal, then the operation is said to have \emph{underflowed}. On the other 
hand, if the magnitude of the result exceeds the largest representable normal 
number in the format, the result is said to have \emph{overflowed}. While 
overflows result in a total loss of precision, underflows in floating-point 
formats that support 
subnormal numbers result in a gradual, rather than total, loss of precision.

Even at higher precisions, most real numbers cannot be represented exactly by
a finite floating-point number. Additionally, the result of most elementary
operations on floating-point numbers cannot be represented exactly and must be
\emph{rounded} to the nearest representable value, following a specific
rounding strategy. This process can introduce some \emph{rounding error} in the
result.

\subsection{Rounding Functions} 
According to the IEEE 754 standard, rounding is viewed as an operation that
maps real numbers to elements of the extended real numbers $\R \cup
\{-\infty,+\infty\}$, which allows for representing the results of 
operations that overflow.  
Four rounding modes are specified in the standard: round towards
$+\infty$, round towards -$\infty$, round towards $0$, and round towards
nearest (with defined tie-breaking schemes). The properties of these modes
are given in \Cref{tab:rnd_modes}. 

In practice, we are often concerned with 
analyzing the rounding error without considering overflow. Thus,  
it is usually sufficient to define rounding into a
given format as a function into the corresponding format with unbounded
exponents \citep{harrison:1997:floating-point,boldo:2017:formalproofs,
boldo:2023:floatingpoint}. Indeed, most textbook presentations of rounding
error analysis also assume that the range of the floating-point format
being rounded into 
is extended to arbitrarily small values; i.e., it is also assumed that 
there is no underflow. To formalize these assumptions, given a
real number $x$ and a floating-point format $\F$ with an unbounded exponent
range, we define a \emph{rounding function} $\rho: \R \rightarrow \F$ as a
function that takes $x$ and returns a nearby floating-point number $\rho(x) \in
\F$.

In general, well-defined rounding functions are monotone and act as the
identity function on the set of floating-point numbers \citep{MullerBook}:
\begin{itemize}
  \item $\forall x,y \in \R. ~ x \le y \rightarrow \rho(x) \le \rho(y)$
  \item $\forall x \in \F. ~ \rho(x) = x$
\end{itemize}

The following result establishes a bound on the magnitude of the rounding error
produced by rounding a real number,
where $u$ indicates the \emph{unit roundoff} for the given rounding function
and format.
\begin{theorem}\label{thm:main_round}
  Given a real number $x \in \R$, and assuming no underflow or overflow occurs, 
  the following equality holds \cite[Theorem 2.2]{Higham:2002:Accuracy}: 
  \[
  \rho(x) = x (1+\delta), \quad \text{where } |\delta| \le u.
  \] 
\end{theorem}
\begin{table}
\caption{Rounding Operations (Modes).}
\label{tab:rnd_modes}
\begin{tabular}{ c c c c }
 \hline
\textbf{Rounding mode} & \textbf{Behavior} & \textbf{Notation} & \textbf{Unit Roundoff} \\
\hline
 Round towards $+\infty$   & $\min\{y \in \F \mid y \ge x \}$  & $\rho_{RU}(x)$ & $\beta^{1-p}$   \\
 Round towards $-\infty$   & $\max\{y \in \F \mid y \le x \}$  &   $\rho_{RD}(x)$ & $\beta^{1-p}$ \\
 Round towards $0$   &   $\rho_{RU}(x)$ if $x < 0 $, otherwise  $\rho_{RD}(x)$  & $\rho_{RZ}(x)$
        & $\beta^{1-p} $ \\
 Round towards nearest\tablefootnote{For round towards nearest, there are
	several possible tie-breaking choices.}   & $\{y \in \F \mid \forall z
	\in \F,  |x - y| \le |x-z| \}
         $   & $\rho_{RN(x)}$ & $\frac{1}{2}\beta^{1-p}  $    \\
 \hline
\end{tabular}
\end{table}

\subsection{Models of Floating-Point Arithmetic}
The purpose of a rounding error analysis is to derive an a priori bound on the
floating-point rounding errors that are produced during the execution of a
program. Performing a rounding error analysis requires that we first establish
a model describing the accuracy of the basic arithmetic operations. Following
the IEEE 754 standard, each basic arithmetic operation ($+,-,*,\div,\sqrt{}$)
should behave as if it first computed a correct, infinitely precise result, and
then rounded this result using one of the functions in \Cref{tab:rnd_modes}.
Given the result in \Cref{thm:main_round}, this assumption on the computational
behavior of each basic arithmetic operation leads to the following commonly
used \emph{standard rounding error model}: 

\begin{definition}[The Standard Rounding Error Model]  \label{def:std_model}
  For any floating-point numbers $x, y \in \F$, and for some rounding function
  $\rho: \R \rightarrow \F$, the standard rounding error model for basic
  floating-point operations is given as follows \citep{Higham:2002:Accuracy}:
\begin{align}
  \rho(x~op~y)  
  = (x~op~y)(1+\delta), \quad |\delta| \le u, \quad op\in\{+,-,*,\div\}.
\end{align}
\end{definition}

An alternative to the standard rounding error model was proposed by
Olver \citep{Olver:1978:rp}, and was later applied to an early error
analysis of  Gaussian
elimination \citep{Olver:1982:Gaussian,Olver:1982:Arithmetic2} and also of
matrix computations \citep{Pryce:1984:Vectors,Pryce:1985:Matrix}:

\begin{definition}[Alternative Rounding Error Model] \label{def:olver_model}
  For any floating-point numbers $x, y \in \F$, and for some rounding function
  $\rho: \R \rightarrow \F$, an alternative rounding error model for the
  basic floating-point operations is given as follows:
  \begin{align}
     \rho(x~op~y)  
	  = (x~op~y)e^\delta, \quad |\delta| \le \frac{u}{1-u}, 
	  	\quad op\in\{+,-,*,\div\}.
  \end{align}
\end{definition}   

If the rounding function  is round towards $+ \infty$, then we can guarantee a
slightly tighter bound than the one given in \Cref{def:olver_model}:

\begin{lemma}  \label{lem:olver_model_pos}
  For any floating-point numbers $x, y \in \F$ we have:
  \begin{align}
     \rho_{RU}(x~op~y)
	  = (x~op~y)e^\delta, \quad |\delta| \le u,
	  	\quad op\in\{+,-,*,\div\}.
  \end{align}
\end{lemma}

\subsection{Measures of Accuracy}
While the most common measures of accuracy are \emph{relative error} and
\emph{absolute error}, the alternative model given in \Cref{def:olver_model}
naturally accommodates the notion of \emph{relative precision}. We define each
of these measures below. 

\paragraph{Relative and Absolute Error}
Absolute error $(er_{abs})$ and relative error $(er_{rel})$ are commonly used
to measure the error of approximating a value $x$ by a value $\tilde{x}$.
\begin{align}
  er_{abs}(x,\tilde{x}) &= |\tilde{x}-x| \\
        \quad 
  er_{rel}(x,\tilde{x}) &= |(\tilde{x}-x)/x| 
    \quad \text{if}~ x \neq 0 \label{def:rel_err}
\end{align}
According to the standard model for floating-point arithmetic
(\Cref{def:std_model}), the relative error of each of the basic floating-point
operations is bounded by the unit roundoff. The relative and absolute error do
not apply uniformly to all values: the absolute error is well-behaved for small
values, while the relative error is well-behaved for large values.

\paragraph{Relative Precision}
We use the following definition of relative precision, adapted from
\citet{Olver:1978:rp}:

\begin{definition}[Relative Precision (RP)]  \label{def:rp}
The \emph{relative precision} (RP) of $\tilde{x}$ as an approximation to $x$ is 
given by 
\\
\begin{equation}\label{eq:olver}
  RP(\tilde{x},x) = 
  \begin{cases}
    |\ln(\tilde{x}/x)| & \text{if $\text{sgn}({\tilde{x}}) = \text{sgn}(x)$ and $\tilde{x},x 
    \neq 0$} \\
    0 &  \text{if $\tilde{x} = x = 0$} \\
    \infty & \text{otherwise.}
  \end{cases}
\end{equation}
\end{definition}
According to the alternative model for floating-point arithmetic
(\Cref{def:olver_model}), the
{relative precision} of each of the basic floating-point operations is bounded
by slightly more than the unit roundoff. If the rounding function is fixed as
round towards $+ \infty$, then the bound on the relative precision can be shown
to be somewhat tighter using the standard model (\Cref{def:std_model}). In that
case, the error variable $\delta$ in \Cref{def:std_model} is non-negative, so
that $|\ln(1+\delta)| \le \delta$, and the relative precision of each of the
basic floating-point operations is bounded by the unit roundoff: 
\begin{align}{\label{eq:stdu}}
	\forall x,y \in \R. ~ RP(\rho(x~op~y),(x~op~y)) &\le 
		\left| \ln(1 + \delta)  \right|, \quad |\delta| \le u
\end{align}

The close relationship between relative precision and relative error can be
seen by rewriting \Cref{def:rp} and \Cref{def:rel_err} as follows:
\begin{align} 
  er_{rel}(x,\tilde{x}) &= |\gamma|; \quad \tilde{x} = (1+\gamma)x
  \label{eq:RPbnd1}\\ 
  RP(x,\tilde{x}) &= |\gamma|; \quad \tilde{x} = e^\gamma x
  \label{eq:RPbnd2} 
\end{align}
If we consider the Taylor expansion of the exponential, then from
\Cref{eq:RPbnd1} and \Cref{eq:RPbnd2} we can see that the relative precision is
a close approximation to the relative error so long as $\gamma \ll 1$.
Moreover, for some $\tilde{x}$ and $x$ with $RP(x,\tilde{x})=\alpha$ and
$\alpha < 1$, the following inequality holds:
\begin{equation} 
  er_{rel} = |e^\alpha -1| \le \alpha/(1-\alpha) \label{eq:conv} 
\end{equation}

A main advantage of relative precision in comparison to relative error is that
\Cref{def:rp} defines an extended pseudometric for all real numbers. 
Specifically, unlike relative error, relative precision satisfies the following
properties: 
\begin{enumerate}
  \item Reflexivity:  $\forall x \in \R.~RP(x,x) = 0$
  \item Symmetry: $\forall x,y \in \R.~RP(x,y) = RP(y,x)$
  \item Triangle Inequlity: $\forall x,y,z \in \R.~RP(x,z) \le RP(x,y) + RP(y,z)$
\end{enumerate}

%\subsection{Sensitivity and Conditioning}

%he backward error concept discussed in Section 2 is traditionally accompanied 
%%by a corresponding concept of conditioning (or well-posedness). A condition number measures 
%the worst-case sensitivity of the output to small changes in the input and, by construction, 
%the forward 
%error is approximately bounded by the product of a condition number and a backward error measure 
\section{Type Systems}
A type system is a principled system for assigning types to
programs.  This assignment is carried out using a set of rules,
which inductively define a set of valid typing judgments. The typing judgment,
which asserts that an expression $e$ can be given type $\tau$ relative to a
typing environment $\Gamma$ for the free variables of $e$, has the form $\Gamma
\vdash e : \tau$; the typing environment $\Gamma$ can be viewed as a partial
map from variables to types.  The premise judgments in each rule are written
above a horizontal inference line, and a single conclusion judgment is written
below the line, with the name of the rule appearing to the left of the line.
Given a typing environment $\Gamma$ and expression $e$, if there is some $\tau$
such that $\Gamma \vdash e : \tau$ , we say that $e$ is \emph{well-typed under
context $\Gamma$}; if $\Gamma$ is the empty context ($\emptyset$), we say $e$
is \emph{well-typed}, and write the judgment as $ \vdash e : \tau$.

Type checking an expression $e$ amounts to showing that the term is well-typed
by constructing a derivation of the judgment $\vdash e : \tau$ for some type
$\tau$.  For example, in the simply-typed lambda calculus extended with a let
expressions, the program \linec{\mlet{x}{3}{x+1}} is well-typed, with type
$\mathbf{int}$. The relevant syntax for types and terms includes integer
literals $n$: 
\begin{alignat*}{3} &\text{Types } &\tau &::=~
  \mathbf{int}
         \mid \sigma \rightarrow \tau 
         \\
         &\text{Expressions } &e &::=~ 
         x
         \mid n \in \mathbb{N}
         \mid \lambda x : \tau.~e
         \mid e ~ f 
         \mid e + f
         \mid \mlet{x}{e}{f}
  \end{alignat*}
To type check this program, only a few typing rules are needed:
\vspace{8pt}
\begin{center}
\AXC{$n \in \mathbb{N}$}
\LeftLabel{(Const)}
\UIC{$\Gamma \vdash n: \textbf{int}$}
\bottomAlignProof
\DisplayProof
\hskip20pt
\AXC{$\Gamma(x) = \tau$}
\LeftLabel{(Var)}
\UIC{$\Gamma \vdash x: \tau$}
\bottomAlignProof
\DisplayProof
\vskip1pt
\end{center}
\vspace{5pt}
\begin{center}
\AXC{$\Gamma \vdash e:\mathbf{int}$}
\AXC{$\Gamma \vdash f:\mathbf{int}$}
\LeftLabel{(Add)}
\BIC{$\Gamma \vdash e + f: \mathbf{int}$}
\bottomAlignProof
\DisplayProof
\hskip20pt
\AXC{$\Gamma \vdash e:\sigma$}
\AXC{$\Gamma,x:\sigma \vdash f:\tau$}
\LeftLabel{(Let)}
\BIC{$\Gamma \vdash \mlet{x}{e}{f} : \tau$}
\bottomAlignProof
\DisplayProof
\vskip1pt
\end{center}
The type derivation for the program proceeds as follows: 
\vspace{8pt}
\begin{center}
\AXC{}
\LeftLabel{(Const)}
\UIC{$ \vdash 3:\mathbf{int}$}
\AXC{}
\LeftLabel{(Var)}
\UIC{$x:\sigma \vdash x:\mathbf{int}$}
\AXC{}
\LeftLabel{(Const)}
\UIC{$x:\sigma \vdash 1:\mathbf{int}$}
\LeftLabel{(Add)}
\BIC{$x:\sigma \vdash x + 1:\mathbf{int}$}
\LeftLabel{(Let)}
\BIC{$ \vdash\linec{\mlet{x}{3}{x+1}} : \mathbf{int}$}
\bottomAlignProof
\DisplayProof
\end{center}

Typically, types describe only the basic structure of the data that programs
operate on, but it is possible for types to express more detailed information
about programs.  For instance, \emph{graded monadic types} refine type
information to describe and track \emph{effects}---any observable behaviors
that might arise during evaluation beyond the production of values. Examples of
computations with effects include partial functions, raising errors or
exceptions, performing input/output (IO), and, as a novel contribution of our
work, \emph{rounding}.  On the other hand, \emph{graded comonadic types}
represent a different approach to refining types by focusing on describing and
tracking \emph{coeffects}---how programs depend on their environment, rather
than the effects they have on it.  Coeffects have been used to describe a wide
range of program behaviors, including resource requirements, program
sensitivity \citep{Fuzz}, and, as part of this dissertation, \emph{backward
error}. 

\subsection{Effects and Graded Monadic Types}\label{background:effects}
Graded monadic types unify two approaches to describing and tracking
computational effects: effect systems and monads.  Introduced by Gifford and
Lucassen \citep{Gifford:1986:effects, lucassen:1987:dissertation,
Lucassen:1988:effects}, and later developed by Talpin and Jouvelot
\citep{Talpin1992,Talpin1993,Talpin1994} and others \citep{Nielson:1999:effects},
effect systems are a class of static analysis techniques that extend the types
and typing rules of an underlying type system with annotations describing the
effects that the primitive operations in the language might produce. For
instance, in effect system, function types have the form $\sigma \oset{\eff{s}}
\tau$ where the effect annotation $\eff{s}$ is traditionally taken from a join
semilattice
$(\mathcal{S},\sqcup,\bot,\sqsubseteq)$ \citep{Mycroft:2015:effects}.  This type
describes a computation that returns a value of type $\tau$ and may have an
effect described by $\eff{s}$ when applied to a value of type $\sigma$. 

The typing rules in effect systems track individual effects, and provide a
fine-grained description of how effects accumulate and interact.  An example
typing rule from a simple effect system is the (Let-E) rule for let expressions
below, which says that the overall effect is the (maximum) of the effect of the
binding expression $e$ and the effect of the body expression $f$:
\vspace{8pt}
\begin{center}
\AXC{$\Gamma \vdash e:\sigma, \eff{r}$}
\AXC{$\Gamma,x:\sigma \vdash f:\tau, \eff{s}$}
\LeftLabel{(Let-E)}
\BIC{$\Gamma \vdash \mlet{x}{e}{f} : \tau, \eff{r} \sqcup \eff{s}$}
\bottomAlignProof
\DisplayProof
\end{center}
A key advantage to this approach is that effect annotations can often be
automatically inferred using natural extensions of existing type inference
algorithms \citep{Nielson:1999:ProgramAnalysis}. 

In a separate line of research, \citet{Moggi:1989:monads,moggi:1991:notions}
demonstrated that effectful computations can be modeled using \emph{monads}
from category theory. Syntactically, monads are incorporated into the language
via a monadic type constructor,  written as $\M$, where a type $M{\tau}$
represents computations that yield a value of type $\tau$ and may produce
effects.  At the syntactic level, this monadic type provides a coarser view of
effects than the annotations in effect systems. For example, consider the
following (Let-M) monadic typing rule for let expressions, which  provides a
syntax-directed way to sequence effectful computations:
\vspace{8pt}
\begin{center}
\AXC{$\Gamma \vdash e:\M\sigma $}
\AXC{$\Gamma,x:\sigma \vdash f:\M\tau$}
\LeftLabel{(Let-M)}
\BIC{$\Gamma \vdash \mlet{x}{e}{f} : \M \tau$}
\bottomAlignProof
\DisplayProof
\end{center}
Clearly, the (Let-M) lacks the fine-grained detail provided by the (Let-E) rule
for effect systems above: the (Let-M) rule indicates that the computation
resulting from a monadic let-binding may have \emph{some} effect, but it
doesn't specify the exact nature or extent of the effect. In contrast, the
(Let-E) rule precisely characterizes the maximum effect the result may have;
i.e., $\eff{r} \sqcup \eff{s}$. Consequently, the utility of inference
algorithms in this setting is less apparent when compared to effect systems.

\citet{wadler:2003:marriage} were the first to propose that the monad structure
introduced by Moggi could be generalized into a family of monads, and used an
\emph{effect-annotated monadic type} to syntactically integrate monads and
effect systems.  Later, \citet{Katsumata:2014:effects} introduced a
denotational semantic framework for the effect-annotated monadic type based on
\emph{graded monads}, which we will soon introduce in \Cref{background:cat}.
Consequently, \emph{graded monadic types} embed graded monads into the syntax
of a type system. The graded monadic type constructor $M_{r}$ refines the
monadic type constructor $M$ into a family of type constructors indexed by the
grade $r$, where $r$ is an element of a \emph{preordered monoid}:
\begin{definition} \label{def:pmonoid}
  A preordered monoid is a tuple $\mathcal{E} = (E,\le,1,\odot)$ such that
  $(E,\le)$ is a preordered set and the binary operator $\odot$ is monotone
  with respect to $\le$ in each argument. 
\end{definition}

One advantage of using elements of a preordered monoid rather than a join
semilattice is that, although using the join operator $(\sqcup)$ to compute the
effect of let expressions is sound, it is not always precise 
\citep{Katsumata:2014:effects}.  The corresponding (Let-G) typing rule for 
sequencing computations of graded monadic type is an intuitive combination of the 
(Let-E) rule and (Let-M) rule:
%
% https://www.cs.kent.ac.uk/people/staff/dao7/publ/effects-revisited.pdf
% also, Parametric Effect Monads and Semantics of Effect Systems
\vspace{8pt} 
\begin{center}
\AXC{$\Gamma \vdash e:M_{r}\sigma $}
\AXC{$\Gamma,x:\sigma \vdash f:M_{s}\tau$}
\LeftLabel{(Let-G)}
\BIC{$\Gamma \vdash \mlet{x}{e}{f} : \M_{{r} \odot \eff{s}} \tau$}
\bottomAlignProof
\DisplayProof
\end{center}

\subsection{Coeffects and Graded Comonadic Types}\label{background:coeffects}
While graded monadic types provide an expressive mechanism for tracking how
programs affect their environment, they cannot track how programs \emph{depend}
on their environment. Type systems that precisely characterize these
\emph{contextual program properties}, known as coeffects, have been developed
by \citet{Brunel:2014:coeffects}, \citet{Ghica:2014:bounded},
\citet{Petricek:2014:coeffects,petricek:2013:coeffects}, and
\citet{petricek:2017:context}. The denotational semantics of these systems,
which use \emph{comonads}---the categorical dual of monads---is
well-established.  Here, we follow the presentation of
\citet{Gaboardi:2016:combining} and \citet{Brunel:2014:coeffects} where comonads
are embedded into the syntax using a \emph{graded comonadic type} constructor,
$!_{s}$.

The type constructor $!_{s}$ generalizes the exponential type constructor $!$
from linear type systems \citep{Wadler:1990:linear}.  Linear type systems are
derived from linear
logic \citep{Katsumata:2014:effects,Girard:1987:LinearLogic}, where each
assumption must be used exactly once. If an assumption $A$ can be discarded or
duplicated, it is marked as $!A$. Similarly, in linear type systems, the type
constructor $!$ differentiates between linear, single-use data (denoted by a
plain type $\tau$) and non-linear, reusable data (denoted by the exponential
type $!\tau$).  For instance, if the body of a function can access a free
variable $x$ of type $\tau$ an arbitrary number of times, then $x$ would be
assigned the type $!\tau$. 

Refining the type constructor  $!\tau$ into a family of type constructors $!_s$
allows for tracking more fine-grained program properties, such as
\emph{bounding} or \emph{limiting} the number of times a variable can be
accessed.  The first attempt to refine the exponential type in this way was
introduced in bounded linear logic \citep{Girard:1992:BoundedLinearLogic},
where the annotation $s$ on the constructor $!_s$  is a polynomial bounding the
computational complexity of a program. More generally, if an expression can access 
a free variable $x$ of type $\tau$ at most $s$ times, then $x$ would be assigned 
the type $!_{s}\tau$, where the grade $s$ is an element of a 
\emph{preordered semiring}:
\begin{definition}\label{def:semiring}
	A \emph{preordered semiring} is a tuple 
	$\mathcal{R} = ({R},\le,0,+,1,\cdot)$ 
  where $({R},\le)$ is a preordered set,
  $ ({R},0,+,1,\cdot)$ is a semiring and 
   both $+$ and $\cdot$ are monotone with  
  respect to $\le$ in both arguments. 
\end{definition}
In this framework, accurately tracking and limiting the usage of variables
requires a redefinition of typing environments as partial maps from variables
to both types \emph{and} grades. This allows environments to not only assign
types to variables but to also track the specific number of times each variable
can be accessed, ensuring that the variable usage requirements of programs
match the grade associated with each variable. For instance, if $\Gamma(x)=
(\tau,r)$, then we have the binding $x:_r\tau$ in $\Gamma$. The type judgment
$x:_r\sigma \vdash e : \tau$ then asserts that the expression $e$ requires
access to the variable $x$ of type $\sigma$ a total of $\ceff{r}$ times.
Environments defined in this way naturally support \emph{sum}, \emph{scaling},
and \emph{translation} operations:
\begin{definition} \label{def:env_sum}
  If two typing environments $\Gamma$ and $\Delta$ always map  
  shared variables to the same type, i.e., if 
  $\Gamma(x) = (\sigma,s)$ and $x \in dom(\Delta)$ imply $\Delta(x) = (\sigma,r)$
  for some grade $r$, then their \emph{sum} is defined as follows:
\\
\
  \[(\Gamma + \Delta)(x) \triangleq  
  \begin{cases}
      (\sigma,s+r) \quad &\text{ if }  \Gamma(x) =(\sigma,s) \text{ and }
                          \Delta(x) = (r,\sigma) \\
      \Gamma(x) \quad &\text{ if } x \notin dom(\Delta) \\
      \Delta(x) \quad &\text{ if } x \notin dom(\Gamma)
  \end{cases}
  \]
\
\end{definition}
\begin{definition} \label{def:env_scale}
The \emph{scaling} operation scales the grades in a typing environment by a 
    given grade: 
\\
  \[(s\cdot \Gamma)(x) \triangleq  
  \begin{cases}
      (\sigma,s\cdot r) \quad &\text{ if }  \Gamma(x) =(\sigma,r)  \\
      \bot \quad &\text{ if } x \notin dom(\Gamma) 
  \end{cases}
  \]
\end{definition}
\begin{definition} \label{def:env_trans}
The \emph{translation} operation translates grades in a typing environment
  by a given grade:  
\\
  \[(s + \Gamma)(x) \triangleq  
  \begin{cases}
      (\sigma,s+ r) \quad &\text{ if }  \Gamma(x) =(\sigma,r)  \\
      \bot \quad &\text{ if } x \notin dom(\Gamma) 
  \end{cases}
  \]
\end{definition}

Given these operations on typing environments, we can consider an example of a
typing rule for let expressions, which provides sequencing for coeffectful
computations:
% https://cs-people.bu.edu/gaboardi/publication/GaboardiEtAlIicfp16.pdf
\vspace{8pt}
\begin{center}
\AXC{$\Gamma  \vdash e: !_{s} \sigma$}
\AXC{$\Gamma, x:_{\ceff{r \cdot s}}\sigma \vdash f:\tau$}
\LeftLabel{(Let-C)}
\BIC{$r \cdot \Gamma \vdash \mlet{x}{e}{f}:\tau$}
\bottomAlignProof
\DisplayProof
\end{center}
This rule composes two computations: one specifying how many times an
expression is capable of being used, and one that has a usage requirement.  It
states that  the overall let expression uses the free variables in the
binding expression a scale factor of $r$ times when the binding expression with
capability $s$ is substituted for a variable used $r \cdot s$ times in the body
of the let expression. 

\subsection{Categorical Semantics}\label{background:cat}
The language guarantees presented in \Cref{sec:nfuzz:sound} and
\Cref{sec:bean:sound} of this dissertation are obtained using a
denotational-semantic framework.  In denotational semantics,  the meaning of a
type $\tau$ is represented by an object $\denot{\tau}$ in a mathematical
domain, defined inductively over the structure of $\tau$.  Similarly, the
meaning of an expression $e$ is interpreted as an element $\denot{e}$ of
$\denot{\tau}$.  More generally, the meaning of an expression depends on its
context and its type. A judgment ${\Gamma \vdash e : \tau}$ is therefore
typically interpreted as an element of the space $\denot{\Gamma} \rightarrow
\denot{\tau}$.  Contexts are traditionally interpreted as the product of the
underlying types:
$\denot{\Gamma = x_1: \tau_1 \dots, x_n : \tau_n} \triangleq 
  \denot{\tau_1} \times \dots \times \denot{\tau_n}$.

Our domains of interest are categories, both common (such as the category
$\Met$, which we will see in \Cref{sec:nfuzz:semantics}) and novel (such as
the category $\Bel$, which we will see in \Cref{sec:bean:semantics}), and
so we will also refer to our denotational semantics as categorical semantics.
In this setting, types are interpreted as objects in a category, and
typing judgments are interpreted as morphisms between these objects.    

Although this section should contain the definitions and notation 
necessary for the presentations of categorical semantics in \Cref{sec:nfuzz:semantics}
and \Cref{sec:bean:semantics}, more detailed explanations can
be found in the introductory textbooks by
\citet{AwodeyBook,Leinster_2014} and \citet{Abramsky:2010:cat}.

\begin{definition}\label{def:cat}
A \emph{category} $\mathbf{C}$ consists of:
\begin{itemize}
\item 
A collection $Ob_\mathbf{C}$ of objects.
\item 
A collection $Hom_\mathbf{C}(A,B)$ 
of morphisms for every pair of objects $A, B \in Ob_\mathbf{C}$. 
\item
The \emph{composition} morphism $g \circ f$ in $Hom_\mathbf{C}(A,C)$ for every pair 
of morphisms $f \in Hom_\mathbf{C}(A, B)$ and $g \in Hom_\mathbf{C}(B, C)$ such that, 
$h \circ (g \circ f) = (h \circ g) \circ f$ for any maps 
$f \in Hom_\mathbf{C}(A, B)$, $g \in Hom_\mathbf{C}(B, C)$, 
and  $h \in Hom_\mathbf{C}(C, D)$.
\item   
For each object $A$, an \emph{identity} morphism $id_A \in Hom_\mathbf{C}(A, A)$
corresponding to object $A$, which acts as the identity under
composition: $f \circ id = id \circ f = f$.
\end{itemize}
\end{definition}

Languages with graded monadic and graded comonadic types embedded in the syntax 
interpret these types using \emph{graded monads} and \emph{graded comonads}
on a category $\mathbf{C}$. 

\subsubsection{Graded Monads}
Graded monads generalize the definition of a monad,
and provide a mathematical structure for interpreting 
graded monadic types.  
Variations on the  
generalization have been proposed by  
\citet{Atkey:2009:effects}, \citet{Tate:2013:effects},
\citet{Katsumata:2014:effects}, and
\citet{Orchard:2014:graded,Orchard:2020:Unifying}.  
The following elementary presentation is due to 
\citet{Katsumata:2014:effects}:

\begin{definition}\label{def:graded_monad}
Let $\mathcal{E} = (E,\le,1,\odot)$ be a preordered monoid. 
A $\mathcal{E}$-graded monad on a category $C$ consists of 
the following data:
\begin{itemize}
\item An functor $T_q: C \rightarrow C$ for every $q \in \mathcal{E}$
\item A natural transformation $T(q \le q') : T_q \rightarrow T_q'$ 
for every $q \le q'$ satisfying
  \begin{align*}
  T(q \le q')  &= \id_{T_q} \\
  T(q' \le q'') \circ (T(q \le q')) &= T(q \le q'') 
  \end{align*}
\item The natural transformation $\eta: \text{Id}_C \rightarrow T_1$
called the \emph{unit map}.
\item The natural transformation 
$\mu_{q,q'}: T_q \circ T_{q'} \rightarrow T_{q \odot q'}$ called the 
\emph{graded multiplication map}.
\end{itemize}
\end{definition}
These data make the following diagrams commute: 
\begin{center}
\[
\begin{tikzcd}
	{T_q \circ T_{q'}} && {T_{q \odot q'}} \\
	\\
	{T_r \circ T_{r'}} && {T_{r \odot r'} }
	\arrow["{T(q \le r) \circ T(q' \le r')}"', from=1-1, to=3-1]
	\arrow["{\mu_{r,r'}}"',  from=3-1, to=3-3]
	\arrow["{\mu_{q, q'}}",  from=1-1, to=1-3]
	\arrow["{T(q \odot q' \le r \odot r')}", from=1-3, to=3-3]
\end{tikzcd}
\qquad\quad
\begin{tikzcd}
	{T_q} && {T_1 \circ T_q} \\
	\\
	{T_q \circ T_1} && {T_q}
	\arrow["{T_q \circ \eta}"', from=1-1, to=3-1]
	\arrow["{\eta \circ T_q}", from=1-1, to=1-3]
	\arrow["{\mu_{1, q}}", from=1-3, to=3-3]
	\arrow["{\mu_{q, 1}}"', from=3-1, to=3-3]
	\arrow[shift right=1, no head, from=1-1, to=3-3]
	\arrow[shift left=1, no head, from=1-1, to=3-3]
\end{tikzcd}
\]
\end{center}
\begin{center}
\[ 
\begin{tikzcd}
	{T_q \circ T_{q'} \circ T_{q''}} && {T_q\circ T_{q' \odot q''}} \\
	\\
	{T_{q \odot q'} \circ T_{q''}} && {T_{q \odot q' \odot q''} }
	\arrow["{\mu_{q, q' \circ T_{q''}}}"', from=1-1, to=3-1]
	\arrow["{\mu_{q \odot q', q''}}"', from=3-1, to=3-3]
	\arrow["{T_q  \circ \mu_{q', q''}}", from=1-1, to=1-3]
	\arrow["{\mu_{q, q' \odot q''}}", from=1-3, to=3-3]
\end{tikzcd}
\]
\end{center}

{Functors} and {natural transformations} are defined as follows:

\begin{definition}\label{def:functor} 
A \emph{functor} $F: \mathbf{C} \rightarrow \mathbf{D}$ between categories 
$\mathbf{C}$ and $\mathbf{D}$
consists of:
\begin{itemize}
\item An object-map $F : Ob_\mathbf{C} \to Ob_\mathbf{D}$, 
    assigning an object $FA$ of $\mathbf{D}$
    to every object  $A$ of $\mathbf{C}$.
\item A function on morphisms $F : Hom_\mathbf{C}(A, B) \to Hom_\mathbf{D}(A, B)$, 
    assigning a morphism $Ff: FA \rightarrow FB$
    of $Hom_\mathbf{D}(A, B)$ to every morphism $f : A \rightarrow B$ of 
    $Hom_\mathbf{C}(A, B)$, so 
    that composition and identities are preserved:
    \[ F(g \circ f) = Fg \circ Ff, \qquad F (\id_{FA}) = \id_{FA}\]
\end{itemize}
\end{definition} 

\begin{definition}\label{def:nat_trans}
  Let $F,G: \mathbf{C} \rightarrow \mathbf{D}$ be functors.  A \emph{natural
  transformation} $\alpha: F \rightarrow G$ consists of a family of morphisms
  $\alpha_A \in Hom_\mathbf{D}(F(A), G(A))$, one per object $A \in
  Ob_\mathbf{C}$, that commute with the functors $F$ and $G$ applied to any
  morphism: for every $f \in Hom_\mathbf{C}(A, B)$, we have $F(f);\alpha_B =
  \alpha_A;  G(f)$. Diagrammatically,
\begin{center}
\begin{tikzcd}
F(A) \arrow[r, "F(f)"] \arrow[d, "\alpha_A"'] & F(B) \arrow[d, "\alpha_B"] \\
G(A) \arrow[r, "G(f)"']                     & G(B)                    
\end{tikzcd}
\end{center}
\end{definition}

\subsubsection{Graded Comonads}\label{background:comonad}

Dual to graded monads and graded monadic types, 
\emph{graded comonads} generalize the definition 
of a comonad and serve as the mathematical structure for interpreting 
graded comonadic types. 

Here, we provide an abridged and elementary definition for graded comonads. 
Detailed definitions are given by 
\citet{Gaboardi:2016:combining}, \citet{Brunel:2014:coeffects}, and 
\citet{Katsumata:2018:comonads}.

\begin{definition}
Let $\mathcal{S} = (S,\le,0,+,1,\odot)$ be a preordered semiring.
A $\mathcal{S}$-graded \emph{comonad} on a \emph{symmetric monoidal category} 
$C$ with tensor unit \textbf{I} consists of the following data: 
\begin{enumerate}
\item For every $s \in {S}$, a functor $D_s: C \rightarrow C$.
\item For every $s \in {S}$, a natural transformation 
$m_{s,\textbf{I}}: \textbf{I}  \rightarrow  D_s\textbf{I}$, called \emph{0-monoidality}.
\item For every $s \in {S}$, a natural transformation 
  $m_{s,A,B} : D_sA \otimes D_sB \rightarrow D_s(A \otimes B)$,
  called \emph{2-monoidality}.
\item A natural transformation $\varepsilon_A : D_1A \rightarrow A$,
called \emph{dereliction}. 
\item A natural transformation $w_A: D_0 A \rightarrow \textbf{I}$, called 
\emph{weakening}.
\item For every $r,s \in {S}$, a natural transformation 
$c_{r,s,A}: D_{(r+s)} A \rightarrow D_r A \otimes D_s A$, called \emph{contraction}.
\item For every $r,s \in {S}$, a natural transformation 
$\delta_{r,s,A} : D_{r \odot s}A \rightarrow D_r(D_sA)$, called \emph{digging}. 
\end{enumerate}
These six natural transformations satisfy over 20 equational axioms, which we will not 
write here. 
\end{definition}

%%%%%%%%%%%%%%%%%%%
\chapter{A Language for Forward Error Analysis}\label{chapter:nfuzz}

This chapter presents \Lang{} (\textsc{Num}erical \textsc{Fuzz}), a typed 
higher-order functional programming language with a linear type system that 
can express quantitative bounds on forward error. 

\section{Introduction}\label{sec:nfuzz:introduction}
From a numerical perspective, the \Lang{} approach to rounding error analysis
follows a well-established method: a global, compositional rounding error
analysis is modeled by combining a sensitivity analysis with a local rounding
error analysis. A sensitivity analysis describes how small changes in input
values can affect the overall output, while a local rounding error analysis
focuses on how individual arithmetic operations, when rounded, contribute to
the overall error. By decomposing a large-scale analysis into these smaller,
analyzable parts, a global view of rounding error can be built up from local
analyses.

What distinguishes the \Lang{} approach is how it formalizes this process in a
type system. The key novelty lies in capturing the rounding error analysis at
the level of types, enabling the static computation of rounding error bounds
for floating-point programs. 
In this approach, \emph{graded comonadic types} (see
\Cref{background:coeffects}) describe function sensitivity and \emph{graded
monadic types} describe rounding error (see \Cref{background:effects}). 
The typing rules then provide a formal, logical method for reasoning about the
interaction between function sensitivity and local rounding error, and
for analyzing the rounding error of increasingly complex programs.  Thus, the
overall rounding error of a program can be derived from the known rounding
error of numerical primitives. 
This approach is attractive because valid derivations correspond to formal
proofs that a given program satisfies the error bound assigned to it by the
type system. This guarantee is rigorously established by connecting a metric
denotational semantics for \Lang{}, which specifies the mathematical meaning of
\Lang{} programs as non-expansive maps between metric spaces, to both an ideal
and floating-point operational semantics, which specify the computational
behavior of \Lang{} programs.

While sensitivity type systems have been previously proposed in the
differential privacy literature, the key innovation of \Lang{} is extending
this concept to explicitly account for the propagation of rounding errors in
numerical computations. By encoding the interplay between sensitivity and
rounding errors directly into the type system, \Lang{} ensures that the
rounding error of floating-point programs can be reasoned about
compositionally. From the perspective of language design, this
approach offers a robust framework for developing software that requires
precise control over numerical accuracy. 

\subsubsection{Sensitivity Type Systems} The core of \Lang{}'s type system is
based on \emph{Fuzz} \citep{Fuzz}, a family of languages for differential privacy that
use linear type systems to track function sensitivity. The fundamental idea in
linear type systems is that functions must use their arguments exactly once.
This contrasts with conventional type systems, where functions are unrestricted
and can use their arguments an arbitrary number of times. To distinguish
conventional functions from linear ones, the types of linear functions  are
written as $\sigma \multimap \tau$.  Unrestricted functions are encoded in
linear type systems as $! \sigma \multimap \tau$, where the $!$ constructor is
used to indicate that the argument to the function does not need to adhere to a
linear usage constraint; \Cref{background:coeffects} provides more details on
the $!$ constructor.  

Sensitivity type systems like Fuzz build on the concept of linearity to
represent \emph{c-sensitive} functions.  Intuitively, a function is
$c$-sensitive if it can amplify distances between inputs by a factor of at most
$c$.  Formally, \emph{c-sensitivity} is defined as follows:

\begin{definition}[C-Sensitivity] A function $f : X \rightarrow Y$ between
  metric spaces is said to be $c$-\emph{sensitive} (or C-Lipschitz)
  \emph{iff} $d_Y(f(x), f(y)) \le c \cdot d_X(x, y)$ for all $x,y \in X$.
  \label{def:csensitive}
\end{definition}
In sensitivity type systems, the function type $\sigma \multimap \tau$
describes functions that are 1-sensitive; these functions are also referred
\emph{non-expansive} functions because they do not amplify distances between
inputs. Interpreting the function type in this way 
requires that both the input and
output types of the function have an associated metric. This idea is quite
natural when types are viewed as metric spaces, such as  the real numbers $\R$
with the standard metric $d_{\R}(x, y) = |x - y|$. In this setting, functions
like $g(x) = x$  and $ h(x)=sin(x)$ are 1-sensitive and can be typed with the
signature $\R \multimap \R$. 
To express functions with varying degrees of sensitivity, the $!$ constructor
is refined into a family of \emph{graded comonadic type constructors},
$\bang{s}$, where the grade $s$  indicates a \emph{metric scaling} and 
is an element of the preordered semiring
(\Cref{def:semiring}) $\mathcal{R}$ of
extended non-negative real numbers $\mathbb{R}_{\ge 0} \cup \{\infty\}$.  
For example, the
type $\bang{r} \sigma$ scales the metric of the type $\sigma$ by a factor of
$r$. The type $ \bang{r} \sigma \multimap \tau$  then describes the type of a
function that is \emph{r-sensitive} with respect to its argument. 

To adapt these core ideas from sensitivity type systems to reason about
relative rounding error in \Lang{}, we rely on the \emph{relative precision}
(\Cref{def:rp}), a pseudometric on the real numbers proposed by
\citet{Olver:1978:rp}.  If we denote the relative precision of $\tilde{x} \in
\R$ as an approximation to $x \in \R$ as $RP(\tilde{x},x)$ then the function
$f(x) = x^2$ is $2$-sensitive under the RP metric:
\begin{align}
	RP(f(x),f(y)) &= \left| \ln\left(\frac{x^2}{y^2}\right) \right| \\
	&= 2 \cdot RP(x,y)
\end{align}
\noindent Spelling this out, if we have two inputs $v$ and $v \cdot e^\delta$,
which are at distance $\delta$ under the RP metric, then applying the function
$f(x) = x^2$ results in outputs $v^2$ and $(v \cdot e^\delta)^2 = v^2 \cdot e^{2
\cdot \delta}$. These outputs are at a distance of at most $2 \cdot \delta$
under the RP metric.

The function $f(x) = x^2$ can be implemented in \Lang{} as follows:
\begin{align*}
	&\linec{f} : ~\bang{2} \num \multimap \num \\
	&\linec{f} \triangleq \lambda x. ~\nmul{(x,x)}
\end{align*}
For now, we can think of the numeric type $\num$ as the real numbers $\R$
equipped with the RP metric. The type $\bang{2} \num \multimap \num$ then
indicates that the function is 2-sensitive under this metric.

\subsubsection{Rounding Error in \Lang{}}
So far, we have not considered rounding error: the function $\linec{f}
\triangleq \lambda x. ~\nmul{(x,x)}$ simply squares its argument without
performing any rounding. To better understand how rounding error is modeled in
\Lang{}, and how function sensitivity interacts with rounding error, consider
the function ${pow2} : \R\rightarrow \R$, which squares a real number and then
rounds the result using an arbitrary rounding function $\rho$: 
\[ 
  pow2(x) = \rho(x^2)
\]
Using the alternative model for floating-point arithmetic (\Cref{def:olver_model}),
the error analysis is simple: 
\begin{equation}\label{eq:ex_pow2}
  pow2(x) =  (x \cdot x)e^\delta 
\end{equation}
where $|\delta|$ is the relative precision and $|\delta| \le u/(1-u)$.  Our
insight is that  type system can be used to perform this analysis, by modeling
rounding as an \emph{error producing} effectful operation. To see how this
works, the function \linec{pow2} can be defined in \Lang{} as follows:
\begin{align*}
  &\linec{pow2} : ~\bang{2} \num \multimap \monad{\varepsilon} \num \\
  &\linec{pow2} \triangleq \lambda x.  ~\rnd (\nmul{(x,x)}) 
\end{align*}
Here, $\rnd$ is a primitive operation that produces values of graded monadic
type $\monad{\varepsilon} \num$, where $\varepsilon$ is a constant that
models the error due to a single rounding, measured as relative precision.  
Therefore, $\varepsilon \le u/(1-u)$. 

More generally, the type $\monad{r} \num$ describes computations that produce
numeric results and might also perform an arbitrary number of roundings. The
grade $r$ expresses an upper bound on the total rounding error produced by the
computation, measured as relative precision. Thus, the type for \linec{pow2}
captures the desired error bound from \Cref{eq:ex_pow2}: when applied to any
input, \linec{pow2} produces an output that approximates its ideal, infinitely
precise counterpart to within RP distance at most $u/(1-u)$.

To formalize this guarantee, our denotational semantics in
\Cref{sec:nfuzz:semantics} interprets values of graded monadic type
$\monad{q}\tau$ as pairs of values whose components are separated by a distance
no greater than $q$. Next, our operational semantics in
\Cref{sec:nfuzz:opsemantics} specify two ways to execute programs of graded
monadic type: under an ideal operational semantics, where rounding
operations act as the identity function, and under a floating-point operational
semantics, where rounding operations round their arguments following some
prescribed rounding strategy.  Then, our main soundness theorem in
\Cref{sec:nfuzz:sound} connects our denotational and operational
semantics, so that the first component of the interpretation of a value of type
$\monad{q}\num$ is the result under the ideal operational semantics and the
second component is the result under the floating-point operational semantics.
This theorem guarantees that programs of monadic type $\monad{q}\num$ represent
computations that produce values with at most $q$ rounding error.

\subsubsection{Composing Error Bounds}
The type of $\linec{pow2} : \bang{2} \num \lin \monad{u} \num$ actually
guarantees a bit more than just a bound on the roundoff: it also guarantees
that the function is $2$-sensitive under an \emph{ideal} semantics. This additional 
piece of information is crucial for analyzing how functions that produce rounding 
error compose.

To see why, consider the function $h(x) = x^4$.  We can implement this function
using $\linec{pow2}$:
\begin{align*}
	&\linec{pow4}: \num \multimap  \monad{3\varepsilon} \num \\
	&\linec{pow4}\triangleq \lambda x.~ \letm{y}{\linec{pow2}~ x}{
        \linec{pow2}~y}
\end{align*}

The $\letm{-}{-}{-}$ construct sequentially composes two monadic, effectful
computations. To keep the example readable, some \Lang{} syntax is elided.
Thus, $\linec{pow4}$ first squares its argument, rounds the result, then
squares again, rounding a second time.

The bound $3\varepsilon$ on the total roundoff error deserves some explanation. 
In the typing rules for \Lang{} given in \Cref{sec:nfuzz:language}, we will 
see that this grade on the monadic type is computed as the sum 
$2\varepsilon + \varepsilon$, where the first term $2\varepsilon$ is the error 
$\varepsilon$ from the first rounding operation \emph{amplified by $2$ since this 
error is fed into the second call of $\linec{pow2}$, a $2$-sensitive function}, and the 
second term $\varepsilon$ is the roundoff error from the second rounding operation. 

Let $\mapsto_{id}$ denote the evaluation of \linec{pow4} under an ideal
execution model, where $\rnd$ in the body of \linec{pow2} behaves like the
identity function, and let $\mapsto_{fp}$ denote the evaluation of
\linec{pow4} under a floating-point execution model, where $\rnd$ in the body
of \linec{pow2} behaves like a specified rounding function.  We can then
visualize \linec{pow4} applied to a numeric value $a$ as the following
composition: 

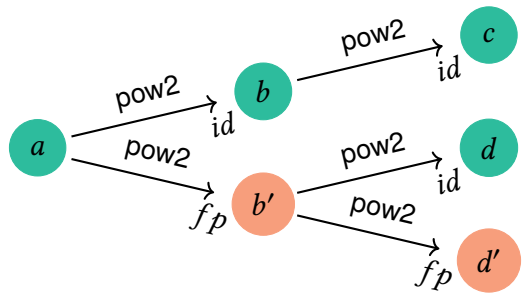
\begin{figure}[H] 
  \centering
\begin{tikzpicture}[scale=1.5]

  \node[draw, circle, fill=SeaGreen, fill opacity=1,
        draw opacity =0.35, text opacity = 1,
        draw=SeaGreen,minimum size=.75cm] (A) at (0,0) { $a$};

  \node[draw, circle, fill=SeaGreen, fill opacity=1,
        draw opacity =0.35, text opacity = 1,
        draw=SeaGreen,minimum size=.75cm] (Y1) at (2, 0.5) { $b$};

  \node[draw, circle, fill=Melon, fill opacity=1,
        draw opacity =0.35, text opacity = 1,
        draw=Melon,minimum size=.75cm] (Y2) at (2, -0.5) { $b'$};

  \node[draw, circle, fill=Melon, fill opacity=1,
        draw opacity =0.35, text opacity = 1,
        draw=Melon,minimum size=.75cm] (Z3) at (4, -1) {$d'$};

  \node[draw, circle,  fill=SeaGreen, fill opacity=1,
        draw opacity =0.35, text opacity = 1,
         draw=SeaGreen,minimum size=.75cm] (Z2) at (4, 1) { $c$};

  \node[draw, circle,  fill=SeaGreen, fill opacity=1,
        draw opacity =0.35, text opacity = 1,
        draw=SeaGreen,minimum size=.75cm] (Z1) at (4, 0) {$d$ };

  \draw (0.3,0.1) edge[thick,black, ->,shorten >=0.25cm] node[above,sloped] 
	{$\small \linec{pow2}$} node[below, xshift=25pt,sloped] {$ id$}  (Y1);
  \draw (0.3,-0.1) edge[thick,black, ->,shorten >=0.25cm] node[above,sloped] 
	{$\small \linec{pow2}$} node[below, xshift=25pt,sloped] {$fp$}  (Y2);  
  \draw (2.3,-0.4) edge[thick,black, ->,shorten >=0.25cm] node[above,sloped] 
	{$\small \linec{pow2}$} node[below, xshift=25pt,sloped] {$id$}  (Z1);    
  \draw (2.3,0.6) edge[thick,black, ->,shorten >=0.25cm] node[above,sloped] 
	{$\small \linec{pow2}$} node[below, xshift=25pt,sloped] {$id$}  (Z2);  
  \draw (2.3,-0.6) edge[thick,black, ->,shorten >=0.25cm] node[above,sloped] 
	{$\small \linec{pow2}$} node[below, xshift=25pt,sloped] {$fp$}  (Z3);

\end{tikzpicture}
\end{figure}

From left-to-right, the ideal and approximate results of $\linec{pow2}(a)$ are
$b$ and ${b'}$, respectively; error soundness for \Lang{}, which we will see in
\Cref{sec:nfuzz:sound} guarantees that the grade $\varepsilon$ on the
monadic return type of $\linec{pow2}$ is an upper bound on the distance between
these values. The ideal result of $\linec{pow4}(a)$ is $c$, while the
approximate result of $\linec{pow4}(a)$ is ${d'}$.  (The value $d$ arises from
mixing ideal and approximate computations, and does not fully correspond to
either the ideal or approximate semantics.) The $2$-sensitivity guarantee of
$\linec{pow2}$ ensures that the distance between $c$ and $d$ is at most twice
the distance between $b$ and ${b'}$---leading to the $2\varepsilon$ term in the
error---while the distance between $d$ and ${d'}$ is at most $\varepsilon$.
Applying the triangle inequality yields an overall error bound of at most
$2\varepsilon + \varepsilon = 3\varepsilon$.

From a numerical perspective, the meaning of the two terms in the total error
$2\varepsilon + \varepsilon$ of $\linec{pow4}$ is clear: the first reflects how
the function $\linec{pow2}$ magnifies errors in the inputs---the sensitivity of
the function, and the second reflects the \emph{local} rounding error of
$\linec{pow2}$---how much error due to rounding is produced locally in the
body of a function.

\section{Type System}\label{sec:nfuzz:language}
This section describes the syntax of \Lang{}, which was briefly introduced in 
the previous section. \Lang{} is based on \emph{Fuzz} \citep{Fuzz}, a linear 
call-by-value $\lambda$-calculus, extended with explicit constructs for monadic 
types to model rounding. For simplicity we do not treat recursive types, and \Lang{} 
does not have general recursion. 

\subsection{Types}\label{sec:nfuzz:types}
\begin{figure}[tbp]
  \begin{alignat*}{3}
         &\sigma, \tau &::=~ &\unit
         \mid \num
         \mid \sigma ~\&~ \tau 
         \mid \sigma \otimes \tau
         \mid \sigma + \tau 
         \mid \sigma \multimap \tau
         \mid {\bang{s} \sigma}
    	   \mid {\monad{q} \tau}
         \tag*{(Types)} \\
         &v, w \ &::=~ &x
         \mid ()
         \mid k \in R
         \mid \langle v,w \rangle 
         \mid  (v, w)
         \mid \inl \ v
         \mid \inr \ v
         \tag*{(Values)}
         \\
         & & & \mid \lambda x.~e 
         \mid \boxx{v}
         \mid \rnd ~ v
         \mid \ret ~ v 
         \mid \letm{x}{\rnd~ v}{f} 
         \\
         &e, f &::=~ & v
	       \mid v~w
         \mid {\pi}_i~ v
         \mid \lett{(x,y)}{v}{e}
         \mid \case{v}{x.e}{y.f}
         \tag*{(Terms)}
         \\
         & & & 
        \mid \letc{x}{v}{e}
         \mid \letm{x}{v}{f} 
	  \mid \lett{x}{e}{f}
         \mid \mathbf{op}(v) \quad \mathbf{op} \in \mathcal{O}
  \end{alignat*}
  \caption{\Lang{} types, values, and terms; $s,q \in \NNR \cup \{\infty\}$.}
  \label{fig:syntax}
\end{figure}

The syntax of \Lang{} types is given in \Cref{fig:syntax}.  The linear function
type $\sigma \multimap \tau$, the graded comonadic (metric scaling) type
$\bang{s}\sigma$, and the graded monadic type $\monad{q}\tau$ have already been
introduced in \Cref{sec:nfuzz:introduction}. Additional background on these types
is given in \Cref{background:coeffects} and \Cref{background:effects}.

The base types in the language are a $\unit$ type and a base numeric type
$\num$. 
The $\unit$ type, combined with the binary sum type constructor + is used to
encode Boolean types, i.e.,  $\mathbb{B} \triangleq \unit + \unit$. The sum
type constructor itself represents a choice between two values, and is used to
encode conditional statements. 
%
%https://dl.acm.org/doi/pdf/10.1145/3589207
Like \emph{Fuzz} and other linear type systems, \Lang{} supports two product types: a
\emph{multiplicative} product $\tensor$ and an \emph{additive} product \&. 

%%%%%
\subsection{Values and Terms} 
Aside from the monadic and comonadic constructs, most values in \Lang{}
correspond to those found in a linear call-by-value typed $\lambda$-calculus
without recursive types.  These include variables, a unit value (), additive
$\langle$-,-$\rangle$ products, multiplicative (-,-) products, sum constructors
$\inl$ and $\inr$, and lambda abstractions.  

Languages with monadic types embedded in their syntax typically separate values
and terms into two disjoint classes, and use \linec{\textbf{return}} and
\linec{\textbf{let}} constructs to sequence monadic computations
\citep{DalLago:2022:relational,torczon:2023:effects-coeffects,Levy2003}.  In \Lang{}, although
values and terms are not disjoint, all computations are explicitly sequenced
using let expressions, ${\lett{x}{v}{e}}$, and term constructors and
eliminators are restricted to values. To sequence monadic and comonadic types,
\Lang{} provides the eliminators $\letm{x}{v}{e}$ and $\letc{x}{v}{e}$,
respectively.  The constructs $\rnd~ v$ and $\ret~ v$ lift values of plain type
to monadic type, while the comonadic construct $\boxx{v}$ indicates scaling the
metric of the underlying type by a constant.

\Lang{} is parameterized by a set $R$ of numeric constants with type $\num$,
a fixed constant $\varepsilon$ representing the rounding error produced by
the evaluation of a rounding function,
and a signature $\Sigma$ defining the primitive operations in the language: a
set of operation symbols $\mathbf{op} \in \mathcal{O}$, each with a type
$\sigma \lin \tau$, and a function $op : CV(\sigma) \rightarrow CV(\tau)$
mapping closed values of type $\sigma$ to closed values of type $\tau$. We
write $\{ \mathbf{op} : \sigma \multimap \tau\}$ in place of the tuple $(\sigma
\multimap \tau , op : CV(\sigma) \rightarrow CV(\tau), \mathbf{op})$.   For
now, we make no assumptions on the functions $op$; we will see in
\Cref{sec:nfuzz:semantics} that, for soundness, we need to ensure that each
function is non-expansive with respect to its type signature.  In
\Cref{sec:nfuzz:examples}, we instantiate $R$, interpret $\num$ as a
concrete set of numbers with a particular metric, and provide a concrete
signature $\Sigma$. 

\subsection{Typing Relation}
\begin{figure}
\begin{center}
%% ROW1	
%% unit
\AXC{}
\LeftLabel{(Unit)}
\UIC{$\Gamma \vdash (): \unit$}
\bottomAlignProof
\DisplayProof
\hskip 2em
%% const
\AXC{$k \in R$}
\LeftLabel{(Const)}
\UIC{$\Gamma \vdash k : \num$}
\bottomAlignProof
\DisplayProof
\hskip 1em
%% var
\AXC{$x \notin dom(\Gamma)$}
\AXC{$s \ge 1$}
\LeftLabel{(Var)}
\BIC{$\Gamma, x:_s \sigma, \Delta \vdash x : \sigma$}
\bottomAlignProof
\DisplayProof
\vskip 2em

%% ROW2
%% abs intro
\AXC{$\Gamma, x:_1 \sigma \vdash e : \tau$}
\LeftLabel{($\multimap$ I)}
\UnaryInfC{$\Gamma \vdash \lambda x. e : \sigma \multimap \tau $}
\bottomAlignProof
\DisplayProof
\hskip 2em
%% abs elim
\AXC{$\Gamma \vdash v : \sigma \multimap \tau$}
\AXC{$\Theta \vdash w : \sigma $}
\LeftLabel{($\multimap$ E)}
\BinaryInfC{$\Gamma + \Theta \vdash vw : \tau $}
\bottomAlignProof
\DisplayProof
\vskip 2em
%%

%% ROW3
%% prod intro
\AXC{$\Gamma \vdash v : \sigma$}
\AXC{$\Gamma \vdash w : \tau$}
\LeftLabel{($\&$ I)}
\BinaryInfC{$\Gamma \vdash \langle v, w \rangle: \sigma ~\&~ \tau $}
\bottomAlignProof
\DisplayProof
\hskip 2em
%% prod elim
\AXC{$\Gamma \vdash v : \tau_1 ~\& ~ \tau_2$}
\LeftLabel{($\&$ E)}
\UIC{$\Gamma \vdash {\pi}_i \ v : \tau_i$}
\bottomAlignProof
\DisplayProof
\vskip 2em
%%

%% ROW4
%% tens intro
\AXC{$\Gamma \vdash v : \sigma $}
\AXC{$\Theta \vdash w : \tau$}
\LeftLabel{($\otimes$ I)}
\BIC{$\Gamma + \Theta \vdash (v, w) : \sigma~ \otimes ~\tau$}
\bottomAlignProof
\DisplayProof
\hskip 2em
%% tens elim
\AXC{$\Gamma \vdash v : \sigma ~\otimes ~\tau$ }
\AXC{$\Theta,x:_s \sigma,y:_s\tau \vdash e: \rho $}
\LeftLabel{($\otimes$ E)}
\BIC{$s \cdot \Gamma + \Theta \vdash \tlet~ (x,y) \ = \ v \ \tin~ e : \rho $}
\bottomAlignProof
\DisplayProof
\vskip 2em
%%

%% ROW5
%% sum introl
\AXC{$\Gamma \vdash v : \sigma$ }
\LeftLabel{($+$ $\text{I}_L$)}
\UIC{$\Gamma \vdash \inl \ v : \sigma + \tau$}
\bottomAlignProof
\DisplayProof
\hskip 2em
%% sum intror
\AXC{$\Gamma \vdash v : \tau$ }
\LeftLabel{($+$ $\text{I}_R$)}
\UIC{$\Gamma \vdash \inr \ v : \sigma + \tau$}
\bottomAlignProof
\DisplayProof
\vskip 2em

%% ROW 5
% sum elim
\AXC{$\Gamma \vdash v : \sigma+\tau$}
\AXC{$\Theta, x:_s \sigma \vdash e : \rho$ \qquad
$\Theta, y:_s \tau \vdash f: \rho$}
\AXC{$s > 0$}
\LeftLabel{($+$ E)}
\TIC{$s \cdot \Gamma + \Theta \vdash \mathbf{case} \ v \ 
	\mathbf{of} \ (\inl ~x.e \ | \ \inr~ y.f) : \rho$}
\bottomAlignProof
\DisplayProof
\vskip 2em

%% ROW 6
% box intro
\AXC{$\Gamma \vdash v : \sigma$ }
\LeftLabel{($!$ I)}
\UIC{$s \cdot \Gamma \vdash \boxx{v} : {!_s \sigma}$}
\bottomAlignProof
\DisplayProof
\hskip 2em
% box elim
\AXC{$\Gamma \vdash v : {!_s \sigma}$}
\AXC{$\Theta, x:_{t\cdot s} \sigma \vdash e : \tau$}
\LeftLabel{($!$ E)}
\BIC{$t \cdot \Gamma + \Theta \vdash \letc{x}{v}{e} : \tau$}
\bottomAlignProof
\DisplayProof
\vskip 2em
%%

%%% ROW 7
% let 
\AXC{$\Gamma \vdash e :  \tau$}
\AXC{$\Theta, x:_{s} \tau \vdash f : \sigma$}
\AXC{$s > 0$}
\LeftLabel{(Let)}
\TIC{$s \cdot \Gamma + \Theta \vdash \tlet~ x = e \ \tin~ f : \sigma$}
\bottomAlignProof
\DisplayProof
\vskip 2em

%%% ROW 8
%% RET
\AXC{$\Gamma \vdash v : \tau$}
\LeftLabel{(Ret)}
\UIC{$\Gamma \vdash \ret~ v : M_0 \tau$}
\bottomAlignProof
\DisplayProof
\hskip 2em
%% RND
\AXC{$\Gamma \vdash v : \num$}
\LeftLabel{(Rnd)}
	\UIC{$\Gamma \vdash \rnd \ v : M_{\varepsilon} \num$}
\bottomAlignProof
\DisplayProof
\hskip 2em
\AXC{$\Gamma \vdash e :  M_q \tau$}
\AXC{$r \ge q$}
\LeftLabel{(MSub)}
\BIC{$\Gamma \vdash e :  M_{r} \tau$}
\bottomAlignProof
\DisplayProof
\vskip 2em

%%% ROW 9
%% LETM
\AXC{$\Gamma \vdash v : M_r \sigma$}
\AXC{$\Theta, x:_{s} \sigma \vdash f : M_{q} \tau$}
\LeftLabel{(MLet)}
\BIC{$s \cdot \Gamma + \Theta \vdash \letm{x}{v}{f} : M_{s \cdot r+q} \tau$}
\bottomAlignProof
\DisplayProof
\hskip 2em
%% OPS
\AXC{$\Gamma \vdash v : \sigma$}
\AXC{$\{ \mathbf{op} :\sigma \lin \num \} \in {\Sigma}$}
\LeftLabel{(Op)}
\BIC{$\Gamma \vdash \mathbf{op}(v) : \num$}
\bottomAlignProof
\DisplayProof

\end{center}
    \caption{Typing rules for \Lang, with $s,t,q,r \in \NNR \cup \{\infty\}$. 
                \Lang{} is parametric in $R$ (Const), ${\Sigma}$ (Op), and $\rnderr$ (Rnd). }
    \label{fig:typing_rules}
\end{figure}

The typing relation of \Lang{} is presented in \Cref{fig:typing_rules}.  Before
stepping through each rule defining the relation, we provide some background on
typing judgments and typing environments.

Typing environments in \Lang{} are defined as follows:
\[
	\Gamma, \Delta ::= \emptyset \mid \Gamma, x:_r \sigma
\] 
where grade annotations $r$ are elements of the preordered semiring
(\Cref{def:semiring}) of extended positive real numbers, i.e.  $r \in
([0,\infty],\le,0,+,1,\cdot)$.  We extend the definition of multiplication as
follows:
\\
\
\begin{align}
	r \cdot \infty = \infty \cdot r =
	\begin{cases}
		0 \quad &\text{if } r=0\\
		\infty \quad &\text{otherwise}
	\end{cases}
\end{align}
\

A well-typed expression \[x:_r \sigma \vdash e : \tau\] represents a
computation that is $r$-sensitive to perturbations in the variable $x$. Zero
sensitivity (r=0) indicates that $e$ is independent of $x$, while infinite
sensitivity (r=$\infty$) means that any perturbation in $x$ can result in
arbitrarily large changes in $e$. 

Following  \Cref{background:coeffects}, a typing environment $\Gamma$ can also
be viewed as a partial map from variables to types and sensitivities, where
$(\sigma, r) = \Gamma(x)$ when ${x:_r\sigma \in \Gamma}$. The \emph{sum}
$\Gamma + \Delta$ of two typing environments (\Cref{def:env_sum}) and the
\emph{scaling} $r\cdot \Gamma$ of a type environment by a grade
(\Cref{def:env_scale}) are defined as in \Cref{background:coeffects}. 

With the structure of typing environments established, we now describe the
rules in \Cref{fig:typing_rules}. 
The ({Const}) rule allows any numeric constants with type $\num$ to be used
under any environment; for now, these constants can be though of as real
numbers, but their exact meaning will be fixed in example instantiations of the
language given in \Cref{sec:nfuzz:examples}. 
The ({Var}) rule allows a variable from the environment to be used so long as
its sensitivity is at least 1. This rule also embeds a form of weakening into
the system: the treatment of typing environments in the rule allowing variables
to be declared but not used, and also allows a variables to be declared with a
higher sensitivity than is actually required. Intuitively, this captures the
fact that $r$-sensitive functions are also $s$-sensitive for $r \le s$.

The introduction and elimination rules for multiplicative products $\tensor$
and additive product $\&$ are identical to those used in \emph{Fuzz}.  To understand
the difference between these two products, consider the treatment of typing
environments in their respective introduction rules: 
\vspace{8pt}
\begin{center}
\AXC{$\Gamma \vdash v : \sigma $}
\AXC{$\Theta \vdash w : \tau$}
\LeftLabel{($\otimes$ I)}
\BIC{$\Gamma + \Theta \vdash (v, w) : \sigma~ \otimes ~\tau$}
\bottomAlignProof
\DisplayProof
\hskip 2em
\AXC{$\Gamma \vdash v : \sigma$}
\AXC{$\Gamma \vdash w : \tau$}
\LeftLabel{($\&$ I)}
\BinaryInfC{$\Gamma \vdash \langle v, w \rangle: \sigma ~\&~ \tau $}
\bottomAlignProof
\DisplayProof
\end{center}
In the multiplicative product ($\tensor$ I), the components of the pair have
free variables in summable environments, and the (variablewise) sensitivity of
the resulting pair is determined by the sum of the environments.  
In the additive product ($\&$ I), the components of the pair share a typing
environment, and the (variablewise) sensitivity of the pair is determined by
this shared environment. 
Although the descriptors \emph{multiplicative} and \emph{additive} are
inherited from linear logic \citep{wood:2022:substructural}, they conflict with
the context operations in sensitivity type systems, where  multiplicative rules
\emph{add} their contexts, and additive rules \emph{share} their contexts.

The typing rules for sequencing (Let) and case analysis ($+$ E) both require
that the sensitivity $s$ scaling the environment in the conclusion of the rule
be strictly positive. While this restriction in the (Let) rule for let
expressions is really only required for soundness in the presence of
non-termination \citep{Gavazzo:2018:applicative} and can be omitted for a terminating
calculus like \Lang{}, it is essential for soundness in the ($+$ E) rule, as
described by \citet{Amorim:2017:metric}. 

The remaining interesting rules are those for metric scaling and monadic types.
In the ($!$ I) rule, the box constructor $\boxx{-}$ indicates scalar
multiplication of an environment.  The ($!$ E) rule is similar to ($\tensor$
{E}), but includes the scaling on the variable in the scope of the elimination.

The rules (MSub), ({Ret}), (Rnd), and (MLet) are the core rules for rounding
error analysis in \Lang. Intuitively, the monadic type $\monad{\rnderr}\num$
describes computations that produce numeric results while performing rounding,
and incur at most $\rnderr$ in rounding error.  The subsumption rule states
that rounding error bounds can be loosened.  The ({Ret}) rule states that we
can lift terms of plain type to monadic type without introducing rounding
error. The (Rnd) rule types the primitive rounding operation, which introduces
roundoff errors.  Here, $\rnderr$ is a fixed numeric constant describing the
roundoff error incurred by a rounding operation. The precise value of this
constant depends on the precision of the format and the specified rounding
function; we leave $\rnderr$ unspecified for now. In
\Cref{sec:nfuzz:examples}, we will illustrate how to instantiate our
language to different settings.

The monadic elimination rule (MLet) allows sequencing two rounded computations
together. This rule formalizes the interaction between sensitivities and
rounding, as illustrated by example in \Cref{sec:nfuzz:introduction}: the
rounding error of the overall let expression $\letm{x}{v}{f}$ is upper bounded by 
the sum of the roundoff error of the value $v$ scaled by the sensitivity of $f$ to 
$x$, and the roundoff error of $f$.

Before introducing our operational semantics, we note that the static aspects
of our system introduced so far satisfy the properties of weakening and
substitution, and define the notion of a subenvironment. We write $e[v/x]$ for
the capture-avoiding substitution of the value $v$ for all free occurrences of
the variable $x$ in the expression $e$.  We will use the notation
$e[\vec{v}/dom(\Gamma)]$ to indicate the (simultaneous) substitution of all
values $v_i$ corresponding to variables $x_i$ in the domain of the typing
environment $\Gamma$ into the expression $e$.

\begin{definition}[Subenvironment] \label{def:subenv}
	The environment $\Delta$ is a \emph{subenvironment} of $\Gamma$, written
	$\Delta \sqsubseteq \Gamma$, if whenever $\Gamma(x) = (\sigma,s)$ for some
	sensitivity $s$ and type $\sigma$,  then there exists a sensitivity $s'$ such
	that $s' \ge s$ and $\Delta(x) = (\sigma,s')$.
\end{definition}

\begin{lemma}[Weakening] \label{thm:weakening} 
	Let $\Gamma \vdash e : \tau$ be a well-typed term. Then for any typing
	environment $\Delta \sqsubseteq \Gamma$, there is a derivation of
	$\Delta \vdash e : \tau$.
\end{lemma}

\begin{proof} 
By induction on the typing derivation of $\Gamma \vdash e: \tau$.
\end{proof}

\begin{lemma}[Substitution] 	\label{thm:substitution}
	Let $\Gamma, \Delta \vdash e : \tau$ be a well-typed term, and let
	$\vec{v} \vDash \Delta$ be a well-typed substitution of closed values, i.e.,
	we have derivations $\vdash v_i : \Delta(x_i)$ for every 
	$x_i \in dom({\Gamma})$. Then there is a
	derivation of \[ \Gamma \vdash e[\vec{v}/dom(\Delta)] : \tau. \]
\end{lemma}

\begin{proof} 
	The base cases (Unit), (Const), and (Var) are direct, and the
	remaining of the cases follow by applying the induction hypothesis to
	every premise of the relevant typing rule. 
\end{proof}

\section{Operational Semantics}\label{sec:nfuzz:opsemantics}
\begin{figure}
\begin{center}

\begin{align*}
	\mathbf{op}(v) &\mapsto op(v) \\
	(\lambda x.e) \ v &\mapsto e[v/x] \\
	\pi_i\langle v_1,v_2 \rangle &\mapsto v_i \\
	\letc{x}{\boxx{v}}{e} &\mapsto e[v/x] \\
	\letm{x}{\ret~ v}{e} &\mapsto e[v/x] \\ 
	\lett{x \tensor y}{(v, w)}{e} &\mapsto e[v/x][w/y] \\
	\case{\inl v}{x.e}{y.f} &\mapsto e[v/x] \\
	\case{\inr v}{x.e}{y.f} &\mapsto f[v/x]
\end{align*}
\vskip -1.75em
\begin{align*}
	\letm{y}{(\letm{x}{\rnd ~v}{f})}{g} &\mapsto 
	\letm{x}{\rnd ~v}{(\letm{y}{f}{g})} \quad x\notin FV(g)
\end{align*}
\vskip -0.25em

	\AXC{$e \mapsto e'$}
	\UIC{$\lett{x}{e}{f} \mapsto \lett{x}{e'}{f}$}
	\DisplayProof

\end{center}
    \caption{Evaluation rules for \Lang.}
    \label{fig:eval_rules}
\end{figure}

The operational semantics described in this section specify the computational
behavior of \Lang{} by defining an evaluation strategy for terms. To capture a
forward rounding error analysis, we ultimately define two operational
semantics: one that describes how terms evaluate under an ideal semantics, and
one that describes how terms evaluate under a floating-point semantics.  Our
denotational semantics given in \Cref{sec:nfuzz:semantics} then describe
the distance between the values that terms reduce to under these two different
semantics; this connection is made precise in \Cref{cor:err-sound}.

We start by defining a general small-step operational semantics, based on the
operational semantics of \emph{Fuzz}~\citep{Fuzz}, and then refine these general
semantics into an ideal operational semantics and a floating-point operational
semantics. The complete set of evaluation rules is given in Figure
\ref{fig:eval_rules}, where the judgment $e \mapsto e'$ indicates that the
expression $e$ takes a single step, resulting in the expression $e'$. 

Although our language does not have recursive types, the $\letmx$ construct
makes it somewhat less obvious that the calculus is terminating: the evaluation
rules for $\letmx$ rearrange the term but do not reduce its size.  Even so, a
standard logical relations argument can be used to show that well-typed
programs are terminating. If we denote the set of closed values of type $\tau$
by $CV(\tau)$ and the set of closed terms of type $\tau$ by $CT(\tau)$, so that
$CV(\tau) \subseteq CT(\tau)$, and define $\mapsto^*$ as the reflexive
transitive closure of the single step judgment $\mapsto$, then we can state our
termination theorem as follows.

\begin{theorem}[Termination] \label{thm:SN}
If $\emptyset \vdash e : \tau$ then there exists $v \in CV(\tau)$ such that $e
\mapsto^* v$.  
\end{theorem}

The proof of \Cref{thm:SN} follows by a standard logical relations argument.
Below, we define the logical relation as the reducibility predicate
$\mathcal{R}_\tau$ and state the key auxiliary lemmas used in the proof of
\Cref{thm:SN}.  A detailed proof is given in \Cref{app:nfuzz:termination}.

\begin{definition} \label{def:redpred}
We define the reducibility predicate $\mathcal{R}_\tau$ inductively on types in
Figure \ref{fig:redpred}.
\end{definition} 
\begin{figure} 
\begin{align*}
\mathcal{R}_\tau &\triangleq \{e \mid e \in CT(\tau) ~\wedge~\exists v \in
CV(\tau). \ e \mapsto^* v ~\wedge~v \in \mathcal{VR}_\tau \}   \\
\mathcal{VR}_{\unit} &\triangleq \{ \langle \rangle \} \\
\mathcal{VR}_{\num} &\triangleq R \\
\mathcal{VR}_{\bang{s} \tau} & \triangleq \{ [v] \mid v \in \mathcal{R}_\tau \}
\\ 
\mathcal{VR}_{\sigma \tand \tau} &\triangleq \{ \langle v,w \rangle \mid v \in
\mathcal{R}_\sigma ~\wedge~w \in \mathcal{R}_\tau \} \\
\mathcal{VR}_{\sigma \tensor \tau} &\triangleq \{ (v, w) \mid v \in
\mathcal{R}_\sigma ~\wedge~w \in \mathcal{R}_\tau \} \\ 
\mathcal{VR}_{\sigma + \tau} &\triangleq \{ v \mid \inl v ~\wedge~ v \in
\mathcal{R}_\sigma \text{ or } \inr v ~\wedge~v \in \mathcal{R}_\tau \} \\
\mathcal{VR}_{\sigma \multimap \tau} &\triangleq 
\{ v \mid \forall w \in \mathcal{VR}_\sigma.\ vw \in \mathcal{R}_\tau \} \\ 
\mathcal{VR}_{\monad{u} \tau} & \triangleq \bigcup_{n \in \mathbb{N}}
\mathcal{VR}^n_{\monad{u} \tau} \\ 
\mathcal{VR}^0_{\monad{u} \tau} & \triangleq \{ v \mid v \equiv \ret~ w
~\wedge~ w \in \mathcal{R}_\tau \text{ or } v \equiv \rnd \ k ~\wedge~k \in
\mathcal{R}_{\num}\} \\
\mathcal{VR}^{n+1}_{\monad{u} \tau} &\triangleq \mathcal{VR}^n_{\monad{u} \tau}
\cup \Bigg\{ \letm{x}{v}{f} \mid \exists \ \sigma, u_1, u_2, j. \ u \ge u_1 +
u_2 ~\wedge~ n > j \\  &\qquad \qquad \qquad \qquad ~\wedge~v \in
\mathcal{VR}^j_{\monad{u_1} \sigma} \Bigg.  \Bigg. ~\wedge~ \left(\forall w\in
\mathcal{VR}_\sigma, f[w/x] \in \mathcal{VR}^{n-j}_{\monad{u_2} \tau}\right)
\Bigg\} 
\end{align*}
\caption{Reducibility Predicate.} \label{fig:redpred} \end{figure}

The proof of \Cref{thm:SN} relies on two key lemmas: \Cref{lem:SN_aux1}
and \Cref{thm:subsumption}.  The definition of the reducibility predicate
ensures that the proofs of these lemmas follow without difficulty.
\Cref{lem:SN_aux1} is fairly standard and follows by induction on the
typing derivation $\emptyset \vdash e : \tau$.  The proof of
\Cref{thm:subsumption} proceeds by induction on the depth of the predicate
$\mathcal{VR}$.

\begin{lemma}\label{lem:SN_aux1}
The predicate $\mathcal{R}_\tau$ is preserved by backward and forward
reductions.  Specifically, if $\emptyset \vdash e : \tau$ and $e \mapsto e'$
then $e \in \mathcal{R}_\tau \iff e' \in \mathcal{R}_\tau$.
\end{lemma}

\begin{lemma}[Subsumption]\label{thm:subsumption}
For any $m \in \mathbb{N}$, and for any monadic grades $q,r$ such that  $r \ge
q$, if $e \in \mathcal{VR}^m_{\monad{q} \tau}$, then $e \in
\mathcal{VR}^m_{M_{r} \tau}$. 
\end{lemma}

\subsubsection{Ideal and Floating-Point Operational Semantics}
Thus far, terms that include the primitive monadic rounding operation $\rnd$
have been treated as values, both in our presentation of the syntax of \Lang{}
in \Cref{sec:nfuzz:language} and in our definition of the evaluation rules,
given in \Cref{fig:eval_rules}.  To define an ideal and floating-point
operational semantics, we refine our syntax and semantics, so that the rounding
operation is now an expression that steps to a number.  The syntax of \Lang{}
is updated as follows. 
  \begin{alignat*}{3}
         &v, w \ &::=~ &x
         \mid ()
         \mid k \in R
         \mid \langle v,w \rangle 
         \mid  (v, w)
         \mid \inl \ v
         \mid \inr \ v
         \tag*{(Values)}
         \mid \lambda x.~e 
         \mid \boxx{v}
         \mid \ret ~ v 
         \\
         &e, f &::=~ & \dots \mid \rnd ~ v
         \tag*{(Terms)}
  \end{alignat*}
The evaluation rules are refined by defining two distinct step relations that
capture the behavior of ideal and floating-point computations.  Under the ideal
semantics, the rounding operation behaves like the identity function, and under
the floating-point semantics, the rounding operation behaves like a rounding
function, $\rho$. For now,
we make no assumptions on the function $\rho$, but further assumptions 
will be needed in our denotational semantics. It is sufficient to think of
$\rho$ as a well-defined rounding function as described in
\Cref{background:floatingpoint}.
\begin{definition} \label{def:id-fp-steps}
We define two step relations $e \mapsto_{id} e'$ and $e \mapsto_{fp} e'$ by
augmenting the operational semantics in \Cref{fig:eval_rules} with the
following rules:
  \begin{align*}
    \rnd~ k \mapsto_{id} \ret~ k
    \qquad \text{and} \qquad
    \rnd~ k \mapsto_{fp} \ret~ \rho(k)
  \end{align*}
\end{definition}

\section{Denotational Semantics}\label{sec:nfuzz:semantics}
Our type system is designed to bound the distance between the outputs of two 
closely related computations: an ideal
computation and its floating-point counterpart. In the previous section, we
demonstrated how programs of graded monadic type can be executed under two
different operational semantics, producing both an ideal, 
infinitely precise value and a floating-point value that may have incurred
rounding error during execution. Formally, the operational semantics say
nothing about the distance between the these two results.  In this section, we
will show that programs of type $\monad{\varepsilon}\tau$ can be interpreted as
pairs of computations that produce values separated by a distance of at most
$\varepsilon$ under a metric on $\mathbb{R}$.  In the next section, we
will connect the operational results from the previous section with the
denotational results presented here to   establish our main result,
demonstrating that well-typed programs of type $\monad{\varepsilon}\tau$
produce at most $\varepsilon$ rounding error.

\subsection{The Category of Metric Spaces}
To formally capture the notion of the distance between program outputs, the
denotational semantics for \Lang{} are based on the categorical semantics for
Fuzz introduced by \citet{Amorim:2017:metric}, where types are interpreted
as \emph{extended pseudo-metric spaces} and programs are interpreted and
non-expansive maps in the category ($\Met$) of these metric spaces. 

\begin{definition}[Extended Pseudo-Metric Space] \label{def:psmet}
  An \emph{extended pseudo-metric space} $(X, d_X)$ consists of a \emph{carrier}
  set $X$, denoted by $|X|$, and a \emph{distance function} 
  $d_X : X \times X \to \NNR \cup \{ \infty \}$
  satisfying the following properties for all $a, b, c, \in X$:
\begin{itemize}
\item reflexivity: $d(a, a) = 0$
\item symmetry: $d(a, b) = d(b, a)$
\item triangle inequality: $d(a, c) \leq d(a, b) + d(b, c)$
\end{itemize}
\end{definition}

Extended pseudo-metric spaces differ from standard metric spaces in two ways.
First, their distance functions can assign infinite distances (\emph{extended}
real numbers). Second, their distance functions are only \emph{pseudo}-metrics
because they can assign distance zero to pairs of distinct points. Since we
will only be concerned with extended pseudo-metric spaces, we will refer to
them as metric spaces.

Now, before we define $\Met$, we define non-expansive maps: 
\begin{definition}\label{def:non_expansive}
  A \emph{non-expansive map} $f : (X, d_x) \to (Y, d_Y)$ between extended
  pseudo-metric spaces consists of a set-map $f : X \to Y$ such that $d_Y(f(x),
  f(x')) \leq d_X(x, x')$. 
\end{definition}
\begin{definition}[The Category of Metric Spaces ($\Met$)]\label{def:met-cat}
  The category $\Met{}$ of \emph{extended pseudo-metric spaces} is the category
  with the following data.
\begin{itemize}
  \item The obejcts are extended pseudo-metric spaces.
  \item The morphisms from $X$ to $Y$ are non-expansive maps from $X$ to $Y$.
\end{itemize}
  The identity function is a non-expansive map, and
  non-expansive maps are closed under composition. Therefore, extended
  pseudo-metric spaces and non-expansive maps form a category $\Met$.
\end{definition}

The category $\Met$ supports several constructions that are useful for
interpreting linear type systems:
\begin{itemize}
\item The Cartesian product $(A, d_A) \tand (B, d_B)$ with carrier $A \times
B$ and metric $d_{A \tand B} ((a, b), (a', b')) =
\text{max}(d_A(a, a'), d_B(b, b'))$. 
\item The tensor product $(A, d_A) \otimes (B,
d_B)$ with carrier $A \times B$ and metric $d_{A
\otimes B} ((a, b), (a', b')) = d_A(a, a') + d_B(b, b')$.  
\item Coproducts $(A, d_A) + (B, d_B)$, where the carrier is the disjoint union
$A \uplus B$ and the metric $d_{A + B}$ assigns distance $\infty$ to pairs of
elements in different injections, and distance $d_A$ or $d_B$ to pairs of
elements in $A$ or $B$, respectively.
\item  Non-expansive functions
$(A, d_A) \lin (B, d_B)$, where the carrier set is $\{ f : A \to B \mid
f~\text{non-expansive} \}$ and the metric is given by the supremum norm: 
\[ 
d_{A\lin B}(f, g) = \text{sup}_{a \in A} d_B(f(a), g(a)).
\] 
\item Terminal objects are the singleton metric space $I = (\{\star\}, d_I)$ 
with a single element and a constant distance function $d_I({\star,\star})=0$. 
Specifically, for every object $X \in \Met$, there is a morphism 
$e_X : X \to I$ given by $e_X := x \mapsto \star$.
\end{itemize}

\begin{theorem} \label{thm:smcc}
The category $(\Met, I, \otimes, \lin)$ is a symmetric monoidal 
closed category (SMCC), where the unit object $I$ is the metric space with a single 
element.
\end{theorem}

\Cref{thm:smcc} follows by observation of the following:
the functor $(- \otimes B)$ is left-adjoint to the functor $(B \lin -)$, so
maps $f : A \otimes B \to C$ can be curried to $\lambda(f) : A \to (B \lin
C)$, and uncurried. 

\subsection{A Graded Comonad on $\Met$}
As discussed in \Cref{background:coeffects}, graded comonadic types can be 
modeled by a categorical structure called a $\mathcal{S}$-\emph{graded exponential
comonad} \citep{Brunel:2014:coeffects,Gaboardi:2016:combining,Katsumata:2018:comonads}. 
Given any metric space $(A, d_A)$ and non-negative number $r$, there is an evident
operation that scales the metric by $r$: $(A, r \cdot d_A)$. This operation can
be extended to a graded comonad:

\begin{definition} \label{def:scale-comonad}
  Let the pre-ordered semiring $\mathcal{S}$ be the extended non-negative real
  numbers $\NNR \cup \{ \infty \}$ with the usual order, addition, and
  multiplication; $0 \cdot \infty$ and $\infty \cdot 0$ are defined to be $0$.
  We define functors $\{ D_s : \Met \to \Met \mid s \in \mathcal{S} \}$ such
  that $D_s : \Met \to \Met$ takes metric spaces $(A, d_A)$ to metric spaces
  $(A, s \cdot d_A)$, and non-expansive maps $f : A \to B$ to $D_s f : D_s A \to
  D_s B$, with the same underlying map.

  We also define the following associated natural transformations:
  \begin{itemize}
    \item For $s, t \in \mathcal{S}$ and $s \leq t$, the map $(s \leq t)_A : D_t
      A \to D_s A$ is the identity; note the direction.
    \item The map $m_{s, I} : I \to D_s I$ is the identity map on the singleton
      metric space.
    \item The map $m_{s, A, B} : D_s A \otimes D_s B \to D_s (A \otimes B)$ is the
      identity map on the underlying set.
    \item The map $w_A : D_0 A \to I$ maps all elements to the singleton.
    \item The map $c_{s, t, A} : D_{s + t} A \to D_s A \otimes D_t A$ is the
      diagonal map taking $a$ to $(a, a)$.
    \item The map $\epsilon_A : D_1 A \to A$ is the identity.
    \item The map $\delta_{s, t, A} : D_{s \cdot t} A \to D_s (D_t A)$ is the identity.
  \end{itemize}
\end{definition}

These maps are all non-expansive and it can be shown that they satisfy the
diagrams \citep{Gaboardi:2016:combining} defining a $\mathcal{S}$-graded
exponential comonad.

\subsection{A Graded Monad on $\Met$}\label{sec:nfuzz:monad}
Our type system is designed to bound the distance between various kinds of
program outputs. Intuitively, types should be interpreted as \emph{metric
spaces}, which are sets equipped with a distance function satisfying several
standard axioms. \citet{Amorim:2017:metric} identified the following slight
generalization of metric spaces as a suitable category to interpret \emph{Fuzz}.

The categorical structures we have seen so far are enough to interpret the
non-monadic fragment of our language, which is essentially the core of the \emph{Fuzz}
language~\citep{Amorim:2017:metric}. As proposed
by \citet{Gaboardi:2016:combining}, this core language can model
effectful computations using a graded monadic type, which can be modeled
categorically by (i) a \emph{graded strong monad}, and (ii) a
\emph{distributive law} modeling the interaction of the graded comonad and the
graded monad.

\subsubsection{The Neighborhood Monad}
Recall the intuition behind our system: closed programs $e$ of type
$M_{\rnderr} \num$ are computations producing outputs in $\num$ that may
perform rounding operations. The index $\rnderr$ should bound the distance
between the output under the \emph{ideal} semantics, where rounding is the
identity, and the \emph{floating-point (FP)} semantics, where rounding maps a
real number to a representable floating-point number following a prescribed
rounding procedure. Accordingly, the interpretation of the graded monad should
track \emph{pairs} of values---the ideal value, and the FP value.

This perspective points towards the following graded monad on $\Met$, which we
call the \emph{neighborhood monad}. While the definition appears quite natural
mathematically, we are not aware of this graded monad appearing in prior work.

\begin{definition} \label{def:nhd-monad}
  Let the pre-ordered monoid $\mathcal{R}$ be the extended non-negative real
  numbers $\NNR \cup \{ \infty \}$ with the usual order and addition. The
  \emph{neighborhood monad} is defined by the functors $\{ T_r : \Met \to \Met
  \mid r \in \mathcal{R} \}$ and associated natural transformations as follows:
  \begin{itemize}
    \item The functor $T_r : \Met \to \Met$ takes a metric space
      $M$ to a metric space with underlying set
      \[
        |T_r M| \triangleq \{ (x, y) \in M \mid d_M(x, y) \leq r \}
      \]
     and metric 
      \[
        d_{T_r M} ( (x, y), (x', y') ) \triangleq d_M (x, x'). 
      \]
    \item The functor $T_r$ takes a non-expansive function $f : A \to B$ to
      $T_r f : T_r A \to T_r B$ with
      \[
        (T_r f)( (x, y) ) \triangleq (f(x), f(y))
      \]
    \item For $r, q \in \mathcal{R}$ and $q \leq r$, the map $(q \leq r)_A : T_q
      A \to T_r A$ is the identity.
    \item The unit map $\eta_A : A \to T_0 A$ is defined via:
      $
        \eta_A(x) \triangleq (x, x).
      $
    \item The graded multiplication map $\mu_{q, r, A} : T_q (T_r A) \to T_{r +
      q} A$ is defined via:
      \[
        \mu_{q, r, A} ( (x, y), (x', y') ) \triangleq (x, y').
      \]
  \end{itemize}
\end{definition}

The definitions of $T_r$ are evidently functors. The associated maps are
natural transformations, and define a graded
monad~\citep{Katsumata:2014:effects,Fujii:2016:graded}. %

\begin{lemma}\label{lem:monad-nat}
  Let $q, r \in \mathcal{R}$. For any metric space $A$, the maps $(q \leq r)_A$,
  $\eta_A$, and $\mu_{q, r, A}$ are non-expansive maps and natural in $A$.
\end{lemma}
\noindent The proof of \Cref{lem:monad-nat} is provided in 
\Cref{app:nfuzz:monad}.

\begin{lemma} \label{lem:monad}
  The functors $T_r$ and its associated maps form a $\mathcal{R}$-graded monad on
  $\Met$.
\end{lemma}
\begin{proof}
Establishing this fact requires checking that the diagrams in 
\Cref{def:graded_monad} commute, which follows directly by unfolding 
definitions.
\end{proof}

As we will soon see, the monad structure defined so far is
insufficient for interpreting the graded monadic sequencing rule (MLet).
Simply put, the issue is that 
sequencing requires the input and output 
types of programs to match, but the natural transformations we 
have defined for the neighborhood monad lack a 
mechanism to ensure this in our semantics. 
This issue also arises for the standard (not graded) monadic sequencing rule 
described in \Cref{background:effects}, for which
\citet{moggi:1991:notions} originally proposed the use 
of \emph{strong monads} to address the issue. The basic idea is that 
a strong monad is a monad along with an additional natural 
transformation $\sigma_{A,B} : A \otimes TB \rightarrow T(A \otimes B)$
known as the
\emph{strength map}.  Generalizations of the notion of  
strong monads have been presented by 
\citet{Atkey:2009:effects} and \citet{Katsumata:2014:effects}. 
We follows the presentation of \citet{Gaboardi:2016:combining}:

\begin{lemma}\label{lemma:monad_strong}
  The neighborhood monad (\Cref{def:nhd-monad}) together with 
  the \emph{tensorial strength maps} 
  $st_{r, A, B} : A \otimes T_r B \to T_r (A \otimes B)$ 
  defined as
  \begin{align*}
    st_{r, A, B}(a, (b, b')) &\triangleq ((a, b), (a, b'))
  \end{align*}
  for every $r \in \mathcal{R}$ form a 
  $\mathcal{R}$-strong graded monad on $\Met$.
\end{lemma}
\noindent
We check the non-expansiveness (\Cref{def:non_expansive}) and 
naturality (\Cref{def:nat_trans}) of the 
tensorial strength map in \Cref{app:nfuzz:monad}.

\subsubsection{A Graded Distributive Law}
\citet{Gaboardi:2016:combining} showed that languages supporting graded
coeffects and graded effects can be modeled with a graded comonad, a graded
monad, and a graded distributive law. In our setting, we have the following
family of maps defining the interaction between the neighborhood monad 
$(T_r)_{r \in \mathcal{R}}$ and the graded comonad $(D_s)_{s \in \mathcal{S}}$. 

\begin{lemma} \label{lem:distr}
  Let $s \in \mathcal{S}$ and $r \in \mathcal{R}$ be grades, and let $A$ be a
  metric space. Then identity map on the carrier set $|A| \times |A|$ is a
  non-expansive map
  \[
    \lambda_{s, r, A} : D_s (T_r A) \to T_{s \cdot r} (D_s A)
  \]
  Moreover, these maps are natural in $A$.
\end{lemma}
\noindent 
It is straightforward to verify the non-expansiveness (\Cref{def:non_expansive}) and 
naturality (\Cref{def:nat_trans}) of the 
distributive map. Details are provided in \Cref{app:nfuzz:monad}.

Similarly, it is straightforward to show that the maps $\lambda_{s, r, A}$ form
a graded distributive law in the sense
of \citet{Gaboardi:2016:combining}: for $s \leq s'$ and $r \leq r'$ the
identity map $T_{s \cdot r}(D_s A) \to T_{s' \cdot r'}(D_{s'} A)$ is also
natural in $A$, and the four diagrams required for a graded distributive law
all commute~\citep[Fig.  8]{Gaboardi:2016:combining}, but since we do not
rely on these properties we will omit these details.

\subsection{Interpreting \Lang{}}\label{subsec:interp}
We are now ready to define an interpretation of \Lang{} in the category $\Met$.

\subsubsection{Interpreting  Types}
We interpret each type $\tau$ as a metric space $\denot{\tau}$, using
the constructions described in the previous sections: the basic constrictions in 
$\Met$ along with the graded comonad and graded monad. 

\begin{definition}\label{def:interp-ty}
We define the interpretation of types by induction on the type syntax:
\begin{equation*}
\begin{aligned}[c]
  \denot{\unit} &\triangleq I = (\{ \star \}, 0) \\
  \denot{A \otimes B} &\triangleq \denot{A} \otimes \denot{B}\\
  \denot{A + B} &\triangleq \denot{A} + \denot{B} \\
  \denot{\bang{s} A} &\triangleq D_s \denot{A}
\end{aligned}
\qquad \qquad
\begin{aligned}[c]
  \denot{\num} &\triangleq (R, d_R) \\
  \denot{A \tand B} &\triangleq \denot{A} \tand \denot{B}\\
  \denot{A \lin B} &\triangleq \denot{A} \lin \denot{B} \\
  \denot{M_r A} &\triangleq T_r \denot{A}
\end{aligned}
\end{equation*}
\end{definition}
It is not yet necessary to fix the interpretation of the base type $\num$. For now,
$(R, d_R)$ can be any metric space.

\subsubsection{Interpreting Judgments}
We interpret a typing judgment of the form $\Gamma \vdash e : \tau$ as a morphism in 
$\Met$ from the metric space $\denot{\Gamma}$ to the metric space $\denot{\tau}$. 
Since all morphisms in $\Met$ are non-expansive, this interpretation ensures 
a version of \emph{metric preservation} for \Lang{}, which is the central 
language guarantee for \emph{Fuzz} \citep{Fuzz,Amorim:2017:metric}. Intuitively, 
this property guarantees that well-typed 
programs respect the sensitivity bound assigned by the type system. 

To interpret typing judgments, we first require an interpretation of typing 
contexts, as well as few auxiliary maps. 
The interpretation of typing contexts is defined inductively as follows:
\begin{align*}
  \denot{\cdot} &\triangleq I = (\{ \star \}, 0) \\
  \denot{\Gamma, x :_s \tau} &\triangleq \denot{\Gamma} \otimes D_s \denot{\tau}
\end{align*}

The interpretation of typing judgments is defined inductively over the 
typing derivation (\Cref{fig:typing_rules}). Since many of the typing rules 
rely on the scaling ($s \cdot \Gamma$) and summing ($\Gamma + \Delta$) of contexts, 
we need to give a clear semantics to these context operations. 

First, given any binding $x :_r \tau \in \Gamma$, there is a non-expansive map 
from $\denot{\Gamma}$ to $\denot{\tau}$ projecting out the $x$-th position. 
Formally, projections are defined via the weakening maps $(0 \leq s)_A ; w_A :
D_s A \to I$ and the unitors, but we will use notation that treats an element 
$\gamma \in \denot{\Gamma}$ as a function, so that $\gamma(x) \in \denot{\tau}$. 
We use this notation to state the following lemma about the sum of two contexts.
\begin{lemma} \label{lem:ctx-contr}
  Let $\Gamma$ and $\Delta$ such that $\Gamma + \Delta$ is defined. Then there
  is a non-expansive map $c_{\Gamma, \Delta} : \denot{\Gamma + \Delta} \to
  \denot{\Gamma} \otimes \denot{\Delta}$ given by:
  \[
    c_{\Gamma, \Delta}(\gamma) \triangleq (\gamma_\Gamma, \gamma_\Delta)
  \]
  where $\gamma_\Gamma$ and $\gamma_\Delta$ project out the positions in
  $dom(\Gamma)$ and $dom(\Delta)$, respectively.
\end{lemma}

Finally, we use the graded comonad to interpret the scaling of a context: 

\begin{lemma} \label{lem:ctx-scale}
  Let $\Gamma$ be a context and $s \in \mathcal{S}$ be a sensitivity. Then the
  identity function is a non-expansive map from $\denot{s \cdot \Gamma} \to D_s
  \denot{\Gamma}$.
\end{lemma}

We are now ready to define our interpretation of typing judgments. Our
definition is parametric in the interpretation of three things: the numeric type 
$\denot{\num} = (R, d_R)$, the rounding operation $\rnd$, and the operations in 
the signature $\Sigma$.

\begin{definition}(Interpretation of \Lang{} Terms.) \label{def:interp-prog}
  Fix $\rho : R \to R$ to be a (set) function such that for every $r \in R$ we
  have \[d_R(r, \rho(r)) \leq \rnderr.\]
  Furthermore, for every operation $\{ \mathbf{op} :
  \sigma \lin \tau \} \in \Sigma$ fix an interpretation
  $\denot{\op} : \denot{\sigma} \to \denot{\tau}$ such that, for every closed
  value $\emptyset \vdash v : \sigma$, we have $\denot{\op}(\denot{v}) =
  \denot{op(v)}$.
  Given these assumptions, 
  we can interpret each well-typed program
  $\Gamma \vdash e : \tau$ as a non-expansive map $\denot{\Gamma \vdash e :
  \tau} : \denot{\Gamma} \to \denot{\tau}$ by induction on the typing
  derivation, via case analysis on the last rule.
\end{definition}

We demonstrate the construction for several cases here, 
including all cases involving terms of monadic type.
The remaining cases can be found in \Cref{sec:app_interp_nfuzz}. 
To reduce
notation, we elide the  
the unitors $\lambda_A : I \otimes A \to A$
and $\rho_A: A \otimes I \to I$; the associators $\alpha_{A, B, C}: (A \otimes
B) \otimes C \to A \otimes (B \otimes C)$; and the symmetries 
$\sigma_{A, B}: A \otimes B \to B \otimes A$. 
\begin{description}
% CONST
\item[(\textbf{Const}).] Define $\denot{\Gamma \vdash k : \num} :
\denot{\Gamma} \to \denot{\num}$ to be the constant function returning
$k \in R$.
% OP
\item [(\textbf{Op}).] 
By assumption, we have an interpretation $\denot{\op}$ for every operation in
the signature $\Sigma$. We can then define:
\[
  \denot{\Gamma \vdash \op(v) : \tau} \triangleq \denot{\Gamma \vdash v :
  \sigma} ; \denot{\op}
\]
% RET
\item[(\textbf{Ret}).] Let $f = \denot{\Gamma \vdash v : \tau}$. Define 
\[\denot{\Gamma \vdash \ret ~v : M_0 \tau} \triangleq f;\eta_{\denot{\tau}}.\]
% SUBSUM
\item[(\textbf{MSub}).] Let $f = \denot{\Gamma \vdash e : M_r
\tau}$. Define \[\denot{\Gamma \vdash e : M_{r'} \tau}  \triangleq f ;
(r \leq r')_{\denot{\tau}}.\]
% RND
\item[(\textbf{Rnd}).] Let $f = \denot{\Gamma \vdash k : \num}$. Define
\[
\denot{\Gamma \vdash \rnd ~k : M_{\rnderr} \num}
\triangleq f ; \langle id, \rho \rangle.
\]
Explicitly, the second map takes $r \in R$ to the pair $(r, \rho(r))$. The
output is in $\denot{M_\rnderr \num}$ by our assumption of the rounding
function $\rho$, and the function is non-expansive by the definition of the
metric on $\denot{M_{\rnderr} \num}$.
% MLET
\item[(\textbf{MLet}).] Let $f = \denot{\Gamma \vdash e : M_r \sigma}$
and $g = \denot{\Theta, x:_s \sigma \vdash e' : M_q \tau}$. 
Define  
\[
\denot{s \cdot \Gamma + \Theta \vdash \letm{x}{e}{f} : \monad{s\cdot r +q} \tau}
\triangleq c_{s \cdot \denot{\Gamma}, \denot{\Theta}}; h_1;h_2.
\]
The maps $h_1$ and $h_2$ are defined as follows. First, apply the comonad to
$f$ and then compose with the distributive law to get:
\[
D_s f ; \lambda_{s, r, \denot{\sigma}} : D_s \denot{\Gamma} \to T_{s \cdot r}
D_s \denot{\sigma} .
\]
Repeatedly pre-composing with the map 
$m_{s,A,B}: D_sA \otimes D_sB \rightarrow D_s(A \otimes B)$ produces a map 
$\denot{s \cdot \Gamma} \to T_{s\cdot r}D_s\denot{\sigma}$.
Then, composing in parallel with 
$id_{\denot{\Theta}}$ and then post-composing with the strength
map $st_{s \cdot r, \denot{\Theta}, \denot{\sigma}}$ yields:
\[
h_1 = ((m;D_s f ; \lambda_{s, r, \sigma}) \otimes id_{\denot{\Theta}}) ;  \otimes st_{s
\cdot r, \denot{\Theta}, \denot{\sigma}} : \denot{\Theta} \otimes \denot{s
\cdot \Gamma} \to T_{s \cdot r} (\denot{\Theta} \otimes D_s \denot{\sigma})
\]
Next, applying the functor $T_{s \cdot r}$ to $g$ and then post-composing with
the multiplication $\mu_{s \cdot r, q, \denot{\tau}}$, we have:
\[
h_2 = T_{s \cdot r} g ; \mu_{s \cdot r, q, \denot{\sigma}} : T_{s \cdot r}
(\denot{\Theta} \otimes D_s \denot{\sigma}) \to T_{s \cdot r + q} \denot{\tau}.
\]
The composition $h_1;h_2$ yields a map from $\denot{\Theta} \otimes \denot{s
\cdot \Gamma}$ to $T_{s \cdot r + q} \denot {\tau} = \denot{M_{s \cdot r + q}
\tau}$, as required. 
\end{description}

\subsubsection{Ideal and Floating-Point Denotational Semantics}
The metric semantics we have just defined
interprets each \Lang{} program as a non-expansive map. 
Ultimately, we aim to
show that values of monadic type $\monad{r} \sigma$ are interpreted as pairs of
values: the first being the result under the ideal operational semantics
and the second being the result under an approximate, floating-point
operational semantics. 
These operational semantics were defined in \Cref{sec:nfuzz:opsemantics}.

In this section, we define two \emph{denotational
semantics} that capture the ideal and floating-point 
behaviors of our programs, respectively. We then relate
the metric semantics from the previous section to the ideal and 
floating-point denotational semantics in a \emph{pairing lemma}
(\Cref{lem:pairing}).
 
We develop both the ideal and floating-point 
semantics in $\Set$, where maps are not required to be non-expansive.
Intuitively, while we can
define an ideal semantics of well-typed programs as non-expansive maps in
$\Met$, programs under the floating-point  semantics are not guaranteed to be
non-expansive---a tiny change in the input to a rounding operation could lead to
a relatively large change in the rounded output. 
We therefore develop both semantics in $\Set$, where maps are not required to be 
non-expansive.

\begin{definition} \label{def:id-fp-sem}
Let $\Gamma \vdash e : \tau$ be a well-typed program. We can define two
semantics in $\Set$:
\begin{align*}
\pdenot{\Gamma \vdash e : \tau}_{id} &: \pdenot{\Gamma}_{id} \to \pdenot{\tau}_{id} \\
\pdenot{\Gamma \vdash e : \tau}_{fp} &: \pdenot{\Gamma}_{fp} \to \pdenot{\tau}_{fp}
\end{align*}
We take the graded comonad $D_s$ and the graded monad $T_r$ to both be the
identity functor on $\Set$:
\begin{align*}
\pdenot{\monad{q}\tau}_{id} &= \pdenot{\bang{s}~\tau}_{id} \triangleq \pdenot{\tau}_{id}\\
\pdenot{\monad{q}\tau}_{fp} &= \pdenot{\bang{s}~\tau}_{fp} \triangleq \pdenot{\tau}_{fp}
\end{align*}

The ideal and floating point interpretations of well-typed programs are both
straightforward, by induction on the derivation of the typing judgment. The
only interesting case is for the rule (Rnd) for the rounding operation:
\begin{align*}
\pdenot{\Gamma \vdash \rnd~ k : \monad{\rnderr} \num}_{id}
&\triangleq \pdenot{\Gamma \vdash k : \num}_{id} \\
\pdenot{\Gamma \vdash \rnd~ k : \monad{\rnderr} \num}_{fp}
&\triangleq \pdenot{\Gamma \vdash k : \num}_{fp} ; \rho
\end{align*}
where $\rho : R \to R$ is a rounding function.
\end{definition}

We now relate the metric semantics with the ideal and floating-point semantics
we have just defined. Let $U : \Met \to \Set$ be the forgetful functor mapping 
each metric space to its
underlying set, and each morphism of metric spaces to its underlying function on
sets. We have:

\begin{lemma}[Pairing] \label{lem:nfuzz:pairing}
  Let $\emptyset \vdash e : \monad{r} \num$. Then we have:
  \[
    U \denot{e} = \langle \pdenot{e}_{id}, \pdenot{e}_{fp} \rangle
  \]
  in $\Set$. The first projection of $U \denot{e}$ is $\pdenot{e}_{id}$, and the
  second projection is $\pdenot{e}_{fp}$.
\end{lemma}

\begin{proof}
By the logical relation for termination, the judgment $\emptyset \vdash e : M_r
\num$ implies that $e$ is in $\mathcal{R}_{\monad{r} \num}^n$ for some $n \in
\mathbb{N}$. We proceed by induction on $n$.

\begin{description}
\item[Case.]
For the base case $n = 0$, we know that $e$ reduces to either $\ret ~v$
or $\rnd ~v$. We can conclude since by inversion $v$ must be a real
constant, and $U\denot{v} = \pdenot{v}_{id} = \pdenot{v}_{fp}$.
\item[Case.] 
For the inductive case $n = m + 1$, we know that $e$ reduces to
$\letm{x}{\rnd ~k}{f}$. By the logical relation, we have
$f[v/x] \in R_{\monad{r} \num}^m$ for all values $v$ such that
$\emptyset \vdash v : \num$. By induction we then have:
\begin{align*}
\pdenot{f[k/x]}_{id}
&\triangleq \pdenot{\letm{x}{\rnd~ k}{f}}_{id}
= U \denot{f[k/x]} ; \pi_1 \\
\pdenot{f[\rho(k)/x]}_{fp}
&\triangleq \pdenot{\letm{x}{\rnd~ k}{f}}_{fp}
= U \denot{f[\rho(k)/x]} ; \pi_2
\end{align*}
Thus we just need to show:
\[
U\denot{\letm{x}{\rnd~ k}{f}}
= \langle U\denot{f[k/x]}; \pi_1, U\denot{f[\rho(k)/x]}; \pi_2 \rangle
\]
where we have judgments $\emptyset \vdash \rnd ~k : \monad{r} \num$ and $x :_s
\num \vdash f : \monad{q} \num$. 
We can conclude by applying the substitution lemma (\Cref{lem:subst-sem}) and 
unfolding the definition of $\denot{\letm{x}{\rnd ~k}{f}}$.
\qedhere
\end{description}
\end{proof}

\section{Forward Error Soundness}\label{sec:nfuzz:sound}
The primary guarantee for \Lang{} is forward error soundness, which  
ensures that well-typed programs of graded monadic type satisfy the error
bound indicated by their type. In this section, we prove this 
guarantee by demonstrating that the ideal and floating-pint 
operational semantics, as described in \Cref{sec:nfuzz:opsemantics},
are computationally sound: stepping a well-typed \Lang{} term
does not change its semantics.
Forward error soundness then follows as a 
corollary to 
computational soundness and 
the pairing lemma (\Cref{lem:nfuzz:pairing}), which 
relates the metric semantics to both the ideal and 
floating-point semantics. 

\begin{lemma}[Computational Soundness] \label{lem:pres-id-fp}
Let $\emptyset \vdash e : \tau$ be a well-typed closed term, and suppose $e
\mapsto_{id} e'$. Then there is a derivation of $\emptyset \vdash e' : \tau$ and
the semantics of both derivations are equal: $\pdenot{\vdash e : \tau}_{id} =
\pdenot{\vdash e' : \tau}_{id}$. The same holds for the floating-point denotational and
operational semantics.
\end{lemma}

\begin{proof}
By case analysis on the step relation. We detail the cases where $e = \rnd ~ k$.
\begin{description}
\item[Case: $\rnd ~k \mapsto_{id} \ret ~k$.] Suppose that $\emptyset \vdash \rnd~k
: \monad{q} \num$, where $\rnderr \leq q$. Then
the rules (Ret) and (MSub) can be used to derive the judgment 
$\cdot \vdash \ret~k : \monad{q} \num$, and
\[
\pdenot{\emptyset \vdash \rnd~k : \monad{q} \num}_{id}
= \pdenot{\emptyset \vdash k : \num}_{id}
= \pdenot{\emptyset \vdash \ret~k : \monad{q} \num}_{id} .
\]
\item[Case: $\rnd~k \mapsto_{fp} \ret ~ \rho(k)$.] Suppose that $\emptyset \vdash
\rnd~k : \monad{q} \num$, where $\rnderr \leq q$. Then
the rules (Ret) and (MSub) can be used to derive the judgment 
$\emptyset \vdash \ret~k : \monad{q} \num$, and
\[
\pdenot{\emptyset \vdash \rnd~k : \monad{q} \num}_{fp}
= \pdenot{\emptyset \vdash \rho(k) : \num}_{fp}
= \pdenot{\emptyset \vdash \ret \rho(k) : \monad{q} \num}_{fp} .
\]
\qedhere
\end{description}
\end{proof}

As a corollary, we have soundness of the error bound for programs with monadic type.

\begin{corollary}[Forward Error Soundness]\label{cor:err-sound}
  Let $\emptyset \vdash e : M_r \num$ be a well-typed program. Then $e
  \mapsto_{id}^* \ret ~ v_{id}$ and $e \mapsto_{fp}^* \ret ~ v_{fp}$ such that
  $d_{\denot{\num}}(\pdenotid{v_{id}}, \pdenot{v_{fp}}_{fp}) \leq r$.
\end{corollary}

\begin{proof}
Under the ideal and floating point semantics, the only values of monadic type
are of the form $\ret ~v$. Since these operational semantics are
type-preserving and normalizing, we must have 
\[
e \mapsto_{id}^* \ret~ v_{id} \quad \text{and} \quad e \mapsto_{fp}^* \ret~ v_{fp}.
\] 
By computational soundness~(\Cref{lem:pres-id-fp}), we have
\[
\pdenot{e}_{id} = \pdenot{\ret~v_{id}}_{id} \quad \text{and} \quad 
\pdenot{e}_{fp} = \pdenot{\ret~ v_{fp}}_{fp}.
\] 
Now, by pairing (\Cref{lem:nfuzz:pairing}), we have
\[
U\denot{e} = \langle \pdenot{\ret~ v_{id}}_{id}, \pdenot{\ret~ v_{fp}}_{fp} \rangle.
\]
Since the forgetful functor is the identity on
morphisms, we have $U\denot{e} = \denot{e}: I \to T_r \denot{\num}$ in $\Met$. By
definition, $\pdenot{\ret~ v_{id}}_{id} = \pdenot{v_{id}}_{id}$ and
$\pdenot{\ret~ v_{fp}}_{fp} = \pdenot{v_{fp}}_{fp}$. Thus,
$\langle \pdenot{ v_{id}}_{id}, \pdenot{ v_{fp}}_{fp} \rangle$ is 
an element of $T_r \denot{num}$ and we conclude by
the definition of the monad $T_r$: 
\[
d_{\denot{\num}}(\pdenotid{v_{id}}, \pdenot{v_{fp}}_{fp}) \leq r.
\]
\end{proof}

\section{Examples}\label{sec:nfuzz:examples}
This section illustrates how to instantiate \Lang{} in practice. Recall that
\Lang{}  is parameterized by a set $R$ of numeric constants of type $\num$, a
fixed constant $\varepsilon$ representing an upper bound on the rounding error
produced by evaluating a rounding function, and a signature $\Sigma$ defining
the primitive operations in the language.  Our language guarantee of forward
error soundness (\Cref{cor:err-sound}) holds under the following
assumptions on these parameters. 

First, the interpretation $\denot{\num} = (R,d_R)$ of the numeric type must be
a metric space. Second, for every operation $\{ \mathbf{op} : \sigma \lin \tau
\} \in \Sigma$, we must fix an interpretation $\denot{\op} : \denot{\sigma} \to
\denot{\tau}$ such that, for every closed value $\emptyset \vdash v : \sigma$,
we have $\denot{\op}(\denot{v}) = \denot{op(v)}$. Simply stated, this means any
primitive operation included in an instantiation of \Lang{} must be a
non-expansive map.  Third, the interpretation $\denot{\rnd}$ of the primitive
rounding operation must use a well-defined rounding function $\rho : R \to R$
such that, for every $r \in R$, we have: \[d_R(r, \rho(r)) \leq \rnderr.\] This
choice of the rounding function fixes the language constant $\rnderr$. 

Now, we will explore how instances of these parameters can be soundly 
chosen in practice.  In the next section, we evaluate
how an implementation of this instance compares to existing sound tools
that automatically bound the relative rounding error of floating-point
programs.

\paragraph{Interpreting $\num$.}
If we interpret our numeric type $\num$ as the set of strictly positive real
numbers $\R_{>0}$ with the relative precision (RP) metric (\Cref{def:rp}), then
we can use \Lang{} to perform a relative error analysis as described by
\citet{Olver:1978:rp}.  

\paragraph{Defining primitive operations.}
Using the $RP$ metric, we can extend the language with four primitive
arithmetic operations, typed as follows: 

\begin{alignat*}{2}
&\mathbf{add} &&: (\num ~ \& ~ \num) \multimap \num  \\
&\mathbf{mul} &&: (\num \otimes \num) \multimap \num \\
&\mathbf{div} &&: (\num \otimes \num) \multimap \num    \\
&\mathbf{sqrt} &&:~ ![0.5]\num \multimap \num
\end{alignat*}

Above, we use the syntax $![s]$ in place of $\bang{s}$ and the syntax $M[q]$ in
place of $\monad{q}$ for readability. This is also the syntax of our
implementation of \Lang{}, which we will introduce in the next section.

In order to soundly add these operations as primitives to the language, we must
verify that their interpretations are each non-expansive maps.  If we fix the
interpretation of these operations as their natural mathematical counterparts
over the positive real numbers, then these interpretations are non-expansive
functions.

As an example, consider the operation $\mathbf{add}: (\num ~ \& ~ \num)
\multimap \num$. Showing non-expansiveness amounts to verifying, for every
$a,b,a',b' \in \R_{>0}$, that 
\[
  RP((a+b),(a'+b')) \le \max(RP(a,a'),RP(b,b')).
\]
If we let $\alpha = RP(a,a')$ and $\beta = RP(b,b')$, then we have that
\begin{align*}
RP((a+b),(a'+b')) &\le RP((a+b) \cdot \max(e^\alpha,e^\beta),(a+b))\\
&= \max(RP(a,a'),RP(b,b'))
\end{align*}
as required. 
Similarly, for the operation
$\mathbf{mul}: (\num ~ \otimes ~ \num) \multimap \num$. 
Showing non-expansiveness
amounts to verifying, for every $a,b,a',b' \in \R_{>0}$, that 
\[
RP((a \cdot b),(a' \cdot b')) \le RP(a,a') + RP(b,b').
\]
Recall that when we first introduced the multiplicative ($\tensor$) product 
and
additive ($\&$) product in \Cref{sec:nfuzz:language}, the descriptors
\emph{multiplicative} and \emph{additive} inherited from linear logic
conflicted with the context operations in sensitivity type systems. Here,
however, they are fitting: multiplication is most naturally typed with the
multiplicative product, while addition is most naturally typed with the
additive product. While we could also type $\mathbf{add}$ with the
multiplicative product, the type is coarser than necessary and would ultimately
result in looser rounding error bounds. (Note that while we can soundly type
$\mathbf{add}$ using both the additive and multiplicative product, we can not
soundly type $\mathbf{mul}$ with the additive product).  To fully understand
why, we first need to interpret our rounding operation $\rnd$. For now, suppose
we added addition typed with the multiplicative product, $\mathbf{add}': (\num
\otimes \num) \multimap \num$, to the primitive operations. Then, consider the
following typing derivations:
\begin{center}
\AXC{$x:_1 R \vdash x : R$}
\AXC{$x:_1 R \vdash x : R$}
\LeftLabel{($\&$ I)}
\BIC{$x:_1 R \vdash \langle x,x \rangle : R\tand R$}
\AXC{$ \{ \mathbf{add} : R\tand R\multimap R\in \Sigma \}$}
\LeftLabel{(Op)}
\BIC{$ x:_{1} R\vdash 
  \mathbf{add}~ \langle x,x \rangle : \R$}
\bottomAlignProof
\DisplayProof
\end{center}
\vspace{10pt}
\begin{center}
\AXC{$x:_1 R\vdash x : R$}
\AXC{$x:_1 R\vdash x : \R$}
\LeftLabel{($\otimes$ I)}
\BIC{$x:_2 R\vdash (x,x) : R\tand \R$}
\AXC{$ \{ \mathbf{add}' : R\otimes R\multimap R\in \Sigma \}$}
\LeftLabel{(Op)}
\BIC{$ x:_{2} R\vdash 
  \mathbf{add}'~ (x,x): \R$}
\bottomAlignProof
\DisplayProof
\end{center}

In the first derivation, the expression $\mathbf{add}~ \langle x,x \rangle$ is
$1$-sensitive in $x$. This is because, in the left branch where the pair
introduction rule ($\&$ I) is used, the pair $\langle x,x \rangle$ is typed in
the same environment as the components, resulting in the pair being
$1$-sensitive in $x$.  In contrast, in the second derivation,  the pair
introduction rule ($\otimes$ I) is used, and the environments used to type each
component are summed as \[x :_2 R = x :_1 R+ x :_1 R,\] making the pair
$2$-sensitive in $x$.  The overall expression $\mathbf{add}'~  (x,x)$ is
therefore $2$-sensitive in $x$.  We will soon see how this difference would
propagate through an error analysis.

\paragraph{Choosing the rounding function.}
Given our interpretation $\denot{\num} = (\R_{>0},RP)$ for the numeric type, we
require the rounding function $\rho$ to be a function such that for every $x
\in \R_{>0}$, we have $RP(x, \rho(x)) \le \epsilon$; that is, the rounding
function must satisfy an accuracy guarantee with respect to the metric $RP$ on
$\R_{>0}$. If we choose $\rho_{RU}: \R \rightarrow \R$ to be
round towards $+\infty$, then by \cref{lem:olver_model_pos} we have that
$RP(x,\rho(x)) \le u$, where $u$ is the unit roundoff. 

Now, using the $\rnd$ operation, we can write the floating-point counterparts
of the primitive operations $\mathbf{add}$, $\mathbf{mul}$, $\mathbf{div}$, and
$\mathbf{sqrt}$ defined above in \Lang{}. Consider the example for
$\mathbf{add}$:
\begin{align*}
  &\mathbf{addfp}: (\num ~\&~ \num) \multimap M[u]\num  \\
  &\mathbf{addfp}~ z \triangleq 
  \lett{y}{\mathbf{add} ~ z }{\rnd~ y} 
\end{align*}
Observe that the type $M[u]$ reflects our interpretation of $\rnd$, which
produces at most unit roundoff ($u$) rounding error when evaluated.  Error
soundness (\Cref{cor:err-sound}) guarantees that 
the function $\mathbf{addfp}$
approximates an infinitely precise computation with relative precision $u$.
Similarly,
for the function  $\mathbf{add}$, we could lift the result to monadic type 
using the return construct ($\ret$), and error soundness would 
guarantee that the function is an exact approximation to an infinitely precise 
computation.
The
functions $\mathbf{mul}$, $\mathbf{div}$, and $\mathbf{sqrt}$ can be similarly
defined, with the following type signatures: 
\begin{alignat*}{2}
  &\mathbf{mulfp} &&: (\num \otimes \num) \multimap M[u]\num\\
  &\mathbf{divfp} &&: (\num \otimes \num) \multimap M[u]\num \\   
  &\mathbf{sqrtfp} &&: ~![0.5]\num \multimap M[u]\num
\end{alignat*}

We can now see how the choice of type signature for our primitive
operations impacts the result of a rounding error analysis.  First, let
$\mathbf{addfp}': (\num \otimes \num) \multimap M[u]\num$ be the program that
rounds the result of the primitive operation $\mathbf{add}' : \num \otimes \num
\multimap \num$, described above.  Then, consider the following \Lang{}
programs\footnote{Recall that the 
$\lett{[y']}{y}{-}$ deconstructor is used to 
sequence values with comonadic type.}:
\begin{align*}
  &\mathbf{fun1}: \num \multimap M[2u]\num  \\
  &\mathbf{fun1} \triangleq  \lambda y.~
  \letm{x}{\rnd ~y}{\mathbf{addfp}~\langle x,x\rangle} \\
    \\  
  &\mathbf{fun2}: ~![2]\num \multimap M[3u]\num  \\
  &\mathbf{fun2} \triangleq \lambda y.~
  \lett{[y']}{y}{\letm{x}{\rnd ~y'}{\mathbf{addfp}'~(x,x)}} 
\end{align*}
The programs $\mathbf{fun1}$ and $\mathbf{fun2}$ both take a value of numeric
type as an input, round this value, and sum the result with itself using
floating-point addition. However, $\mathbf{fun2}$ is $2$-sensitive in its 
argument, while $\mathbf{fun1}$ is only 
$1$-sensitive in its argument. We will see that this is because 
$\mathbf{fun1}$ uses $\mathbf{addfp}: (\num
~\&~ \num) \multimap M[u]\num$ for floating-point addition, while
$\mathbf{fun2}$ uses $\mathbf{addfp}': (\num \otimes \num) \multimap M[u]\num$
for floating-point addition. The consequence of this choice is seen in the 
return types of the function
signatures: the type of $\mathbf{fun1}$ guarantees at most $2u$ roundoff error,
while  the type of $\mathbf{fun2}$ guarantees a looser bound, of at most $3u$
roundoff error. 

To see how this works, consider the corresponding derivations of the typing
judgments. For $\mathbf{fun1}$, we have the following valid derivation: 
\begin{center}
\AXC{}
\LeftLabel{(Var)}
\UIC{$y:_1 R\vdash y: R$}
\LeftLabel{(Rnd)}
\UIC{$y:_1 R\vdash \rnd~ y: M_u R$}
\AXC{[A1]}
\noLine
\UIC{$\vdots$}
\noLine
\UIC{$ x:_{1} R\vdash  \mathbf{add}~ \langle x,x\rangle: R$}
\AXC{[B1]}
\noLine
\UIC{$\vdots$}
\noLine
\UIC{$ y:_1 R\vdash \rnd ~y: R$}
\LeftLabel{(Let)}
\BIC{$ x:_{1} R\vdash  \mathbf{addfp}~ \langle x,x\rangle
  : M_{u} \R$}
\LeftLabel{(MLet)}
\BIC{$1 \cdot y:_1 R \vdash \letm{x}{\rnd ~y}{\mathbf{addfp}
      ~ \langle x,x\rangle}: M_{1 \cdot u+u}\R$}
\LeftLabel{(Abs)}
\UIC{$\emptyset \vdash \mathbf{fun1}: M_{2u} \R$}
\bottomAlignProof
\DisplayProof
\end{center}
In the right branch of the derivation above, we have already seen 
the derivation $[A1]$, showing that the expression 
$\mathbf{add}~ \langle x,x\rangle$ is 
$1$-sensitive in $x$.

For $\mathbf{fun2}$, the derivation is:
\begin{center}
\AXC{}
\LeftLabel{(Var)}
\UIC{$y:_1 R\vdash y: \R$}
\LeftLabel{(Rnd)}
\UIC{$y:_1 R\vdash \rnd~ y: M_u \R$}
\AXC{[A2]}
\noLine
\UIC{$\vdots$}
\noLine
\UIC{$ x:_{2} R\vdash  \mathbf{add}'~ (x,x): R$}
\AXC{[B2]}
\noLine
\UIC{$\vdots$}
\noLine
\UIC{$ y:_1 R\vdash \rnd ~y: \R$}
\LeftLabel{(Let)}
\BIC{$ x:_{2} R\vdash  \mathbf{addfp}'~ (x,x)
  : M_{u} \R$}
\LeftLabel{(MLet)}
\BIC{$2 \cdot y:_1 \num \vdash \letm{x}{\rnd ~y}{\mathbf{addfp}'
 ~ (x,x)}: M_{2 \cdot u+u} \num$}
\LeftLabel{(Abs)}
\UIC{$ \emptyset \vdash \mathbf{fun2}: M_{3u} \num$}
\bottomAlignProof
\DisplayProof
\end{center}

In the right branch of the derivation above, we have already seen 
the derivation $[A2]$ showing that the expression 
$\mathbf{add}'~ (x,x)$ is 
$2$-sensitive in $x$. These derivations illustrate
the importance of using the best possible type 
for primitive operations. 

Now that we have described an instantiation of \Lang{} and 
justified our choice of the types for primitive operations 
in this instance, we can proceed to consider some more detailed examples. 

\subsubsection{Examples}
The examples presented in this section use the actual syntax of an
implementation of \Lang, introduced in \Cref{sec:nfuzz:implementation}. The
implementation closely follow the language syntax presented in
\Cref{fig:typing_rules}, with some additional syntactic sugar, defined below:

\begin{align*}
\text{\lstinline{(x = e; f)}} &\equiv \lett{x}{e}{f} 
  \tag{pure sequencing} \\
\text{\lstinline{(let x = v; f)}} &\equiv \letm{x}{e}{f}
  \tag{monadic sequencing}\\
\text{\lstinline{(let [x] = v; f)}} &\equiv \lett{\boxx{x}}{v}{f} 
  \tag{comonadic sequencing}
\end{align*}

For top-level
programs, we write \lstinline{(function ID args} \{\lstinline{v}\} \lstinline{e)} to
denote the let-binding $\lett{\text{ID}}{v}{e}$, where $v$ is a lambda term with
arguments \lstinline{args}. We write additive and 
multiplicative pairs as \lstinline{$\langle$-,-$\rangle$} and \lstinline{(-,-)}, 
respectively.
Finally, for types, we write \lstinline{M[u]num} to represent monadic types
with a numeric grade \lstinline{u} and we write \lstinline{![s]} to represent
comonadic types with a numeric grade \lstinline{s}.

\paragraph{Example: Fused Multiply-Add}
We warm up with a simple example of a \emph{multiply-add} (MA) operation: given
$x, y, z$, we want to compute $x \cdot y + z$.  The \Lang{} implementation of
\lstinline{MA} is: 

\begin{lstlisting}
// Multiply$\color{Gray}\text{-}$add operation
function MA (x: num, y: num, z: num) { 
  let a = mulfp (x,y); // monadic sequencing of multiplication with rounding
  addfp $\langle$a,z$\rangle$
}
\end{lstlisting}

We can soundly type the program \lstinline{MA} in \Lang{} as 
\begin{center}
{\lstinline{MA : num $\multimap$ num $\multimap$ num $\multimap$ M[2u]num}},
\end{center}
where the index \lstinline{2u} on the
return type indicates that the roundoff error is at most twice the unit
roundoff,  due to the two separate rounding operations in \lstinline{mulfp} and
\lstinline{addfp}. The monadic sequencing \lstinline{let a = mulfp (x,y)} 
allows us to use the result of \lstinline{mulfp (x,y)}---which has a monadic type---
as an argument to \lstinline{addfp}, which accepts pure (non-monadic) numeric
arguments. 

Multiply-add is extremely common in numerical code, and modern architectures
typically support a \emph{fused} multiply-add (FMA) operation.  This operation
performs a multiplication followed by an addition, $x\cdot y+z$, 
as though it were a
single floating-point operation. Consequently, the FMA operation incurs a 
single rounding error instead of two.
The \Lang{} implementation of the FMA operation is:

\begin{lstlisting}
// Fused multiply$\color{Gray}\text{-}$add operation
function FMA (x: num, y: num, z: num) { 
  a = mul (x,y); // multiplication without rounding
  b = add $\langle$a,z$\rangle$; // addition without rounding
  rnd b 
}
\end{lstlisting}

We can soundly type the program \lstinline{MA} in \Lang{} as 
\begin{center}
{\lstinline{FMA : num $\multimap$ num $\multimap$ num $\multimap$ M[u]num}}.
\end{center} 
The index \lstinline{u} on the return type of \lstinline{FMA} is
reflects the  
reduced rounding error when compared to \lstinline{MA}.
 
\paragraph{Example: Polynomial Evaluation}
A standard method for evaluating a polynomial is Horner's scheme, which
rewrites an $n$th-degree polynomial $p(x) = a_0+a_1x+\cdots a_n x^n$ as 
\[ p(x) = a_0 + x(a_1 + x(a_2  + \cdots x(a_{n-1} + ax_n)\cdots)),\]
and computes the result using only $n$ multiplications and $n$ additions.
Using \Lang, we can perform an error analysis on a version of Horner's scheme
that uses the FMA operation to evaluate second-order polynomials of the form
$p(\vec{a},x) = a_2x^2 + a_1x + a_0$ where $x$ and all $a_i$s are non-zero
positive constants.  The implementation \lstinline{Horner2} in \Lang{}: 

\begin{lstlisting}
// Horner's scheme for a second order polynomial
function Horner2  
  (a0: num, a1: num, a2: num, x: ![2]num) {
  let [x1] = x; // comonadic sequencing
  s1 = FMA a2 x1 a1; 
  let z = s1; // monadic sequencing
  FMA z x1 a0 
}
\end{lstlisting}

We can soundly type \lstinline{Horner2} in \Lang{} with the following signature:
\begin{center}
{\lstinline{Horner2 : num $\multimap$ num $\multimap$ num $\multimap$ ![2]num $\multimap$ M[2u]num}}
\end{center}
The type of \lstinline{Horner2} guarantees that at most \lstinline{2u} rounding error is produced by 
the function, measured as relative precision.

As a consequence of the metric interpretation of programs
(\Cref{subsec:interp}), the type of \lstinline{Horner2} 
also ensures bounded sensitivity of the ideal semantics, which corresponds to the polynomial
\[p(\vec{a},x) = a_2x^2
+ a_1x + a_0.\] For any $\vec{a},\vec{u} \in \R^3_{>0}$, and
for any $x, x' \in \R_{>0}$, we can measure the
sensitivity of \lstinline{Horner2} to rounding errors introduced by the inputs:
if $RP(x,x') = q$ and  $RP(a_i,u_i) = r$  for each $i \in {0,1,2}$, then 
\begin{align} 
RP(p(\vec{a},x), p(\vec{u},x')) &\le RP(\vec{a}, \vec{u}) + 2 \cdot RP(x, x') \nonumber \\ 
&= \sum_{i=0}^{2} RP(a_i,u_i) + 2 \cdot RP(x, x') \nonumber \\
&= 3r + 2q \label{eq:horner_sens} 
\end{align} 
The term $2q$ reflects that \lstinline{Horner2} is $2$-sensitive in the
variable $x$. The fact that we take the sum of the RP distances over
the $a_i$'s follows from the metric on the function type
(\Cref{subsec:interp}). Since \Lang{} supports currying (see \Cref{thm:smcc}),
the metric is the same as for the multiplicative (tensor) product. 

The interaction between the sensitivity of the function under its ideal
semantics and the local rounding error incurred in the body of the 
function \lstinline{Horner2} over exact
inputs is illustrated by the function \lstinline{Horner2_with_error},
which takes arguments that have rounding error:

\begin{lstlisting}
// Horner's scheme for a second order polynomial with input error
function Horner2_with_error   
  (a0: M[u]num, a1: M[u]num, a2: M[u]num, x: ![2](M[u]num)) { 
  let [x1] = x; // comonadic sequencing
  let x' = x1; // monadic sequencing needed for FMA
  let a0' = a0; 
  let a1' = a1;
  let a2' = a2; 
  Horner2 a0' a1' a2' x'
}
\end{lstlisting}

We can soundly type \lstinline{Horner2_with_error} in \Lang{} with the following signature:
\begin{center}
{\lstinline{Horner2 : M[u]num $\multimap$ M[u]num $\multimap$ M[u]num $\multimap$ ![2](M[u]num) $\multimap$ M[7u]num}}
\end{center}
From the type, we see that \Lang{} guarantees that 
\lstinline{Horner2_with_error} produces at most  \lstinline{7u} rounding error: from \cref{eq:horner_sens}
it follows that the sensitivity of the function contributes \lstinline{5u},
and rounding error incurred by evaluating \lstinline{Horner2} over exact inputs
contributes the remaining \lstinline{2u}.

\Lang{} is a higher-order language,
and it is possible to implement Horner's scheme for polynomials of fixed
degree $n$ by first writing a $n$-ary monadic fold function---a higher-order
function---and then applying it to a product of $n$ coefficients $a_i$. The type
system of \Lang{} is capable of expressing the fold function along with its
roundoff error. For example, \lstinline{fold3} is a $2$-ary monadic fold function:

\begin{lstlisting}
// A fold$\color{Gray}\text{-}$like function over lists of length 3
function fold3   
  (a : num otimes num otimes num, g : ![2](num mapsto num mapsto M[u]num)) {
  let [g'] = g; 
  let (a0, b) = a;
  let (a1, a2)= b;
  s = g' a2 a1;
  let z = s;
  g' z a0
}
\end{lstlisting}

The type of \lstinline{fold3} is:
\begin{center}
{\lstinline{fold3 :  num otimes num otimes num mapsto M[u]num mapsto ![2](num mapsto num mapsto M[u]num) mapsto M[2u]num}}.
\end{center}
We can use this fold-like function to implement \lstinline{Horner2} like so: 

\begin{lstlisting}
// Horner's scheme for a second order polynomial using a fold$\color{Gray}\text{-}$like function
function Horner2  
  (a: num otimes num otimes num, x: ![2]num) {
  let [x1] = x; // comonadic sequencing
  g = fun (a: num) {fun (b: num) {FMA a x1 b}};
  fold3 a [g$\color{Cerulean}{\{2\}}$]
}
\end{lstlisting}

We will see why the annotation $\color{Cerulean}{\{2\}}$ is required in 
\Cref{sec:nfuzz:implementation}. 

\subsection{Floating-Point Conditionals} 
In the presence of rounding error, conditional branches present a particular
challenge: while the ideal execution may follow one branch, the floating-point
execution may follow another.  In \Lang, we can perform rounding error analysis
on programs with conditional expressions  (case analysis) when executions
\emph{take the same branch}, for instance, when the data in the conditional is
a boolean expression that does not  have floating-point error because it is
some kind of parameter to the system, or some exactly-represented value that is
computed only from other exactly-represented values.  This is a restriction of
the \emph{Fuzz}-style type system of \Lang, which is not able to compare the
difference between two different branches since the main metatheoretic
guarantee only serves as a sensitivity analysis describing how a single program
behaves on two different inputs.  In \Lang, the rounding error of a program
with a case analysis is then a measure of the maximum rounding error that
occurs in any single branch.

As an example of performing rounding error analysis in \Lang~ on functions with
conditionals, we first add the primitive operation $\mathbf{is\_pos}: !_\infty
\R \multimap \mathbb{B}$, which tests if a real number is greater than zero.
The sensitivity on the argument to $\mathbf{is\_pos}$ is necessarily  infinity,
since an arbitrarily small  change in the argument to could lead to an
infinitely large change in the boolean output. Using $\mathbf{is\_pos}$ we
define the function $\mathbf{case1}$, which computes the square of a negative
number, or returns the value 0 (lifted to monadic type):
\begin{align*}
&\mathbf{case1} : ~ !_\infty \R \multimap M_u\R \\
&\mathbf{case1}~ x \triangleq ~\lett{\boxx{c}}{\mathbf{is\_pos}}{x} \\
& \qquad \qquad \ \mathbf{if } ~c ~\mathbf{then}~ 
\mathbf{mulfp} (x,x) ~\mathbf{ else } ~ \ret ~0.
\end{align*} 
From the signature of $\mathbf{case1}$, we see that the relative precision (RP)
is unit roundoff, due to the single rounding in $\mathbf{mulfp} ~(x,x) $. 

%\paragraph*{Underflow and overflow.}
%In the following examples, \emph{we assume that the results of computations do
%not overflow or underflow}. Recall from \Cref{background:float} that the standard
%model for floating-point arithmetic given in \cref{eq:op_model} is only valid
%under this assumption.
\section{Implementation}\label{sec:nfuzz:implementation}

We have developed a prototype  cker for \Lang{} in OCaml, based on the
sensitivity-inference algorithm due to \citet{Amorim:2014:typecheck}
developed for a dependently-typed extension of \emph{Fuzz} \citep{Gaboardi:2013:dfuzz}.
Given an environment $\Gamma$, a term $e$, and a type $\sigma$, the goal of
type checking is to determine if a derivation $\Gamma \vdash e : \sigma$
exists.  For sensitivity type systems, type checking and type inference can be
achieved by solving the sensitivity inference problem. The sensitivity
inference problem is defined using \emph{context skeletons} $\Gamma^{\bullet}$
which are partial maps from variables to \Lang{} types. If we denote by
$\overline{\Gamma}$ the context $\Gamma$ with all sensitivity assignments
removed, then the sensitivity inference problem is defined 
\citep[Definition 5]{Amorim:2014:typecheck} as follows.

\begin{definition}[Sensitivity Inference]
Given a skeleton $\Gamma^{\bullet}$ and a term $e$, the \emph{sensitivity
inference problem} computes an environment $\Gamma$ and a type $\sigma$ with a
derivation $\Gamma \vdash e : \sigma$ such that $\Gamma^{\bullet} =
\overline{\Gamma}$. 
\end{definition}

\begin{figure}[htbp]
%% ROW1
\begin{center}
%% var
\AXC{}
\LeftLabel{(Var)}
\UIC{$\Gamma^{\bullet}, x: \sigma; x \Rightarrow \Gamma^0, x :_1\sigma;\sigma$}
\bottomAlignProof
\DisplayProof
\hskip 2em
%% const
\AXC{}
\LeftLabel{(Const)}
\UIC{$\Gamma^{\bullet};k \in R \Rightarrow \Gamma^0; \num$}
\bottomAlignProof
\DisplayProof
\vskip 2em

%% unit
\AXC{}
\LeftLabel{(Unit)}
\UIC{$\Gamma^{\bullet};e \in \unit \Rightarrow \Gamma^0; \unit$}
\bottomAlignProof
\DisplayProof
\hskip 2em
%%%ROW3
% let 
\AXC{$\Gamma^{\bullet}; e \Rightarrow \Gamma;  \tau$}
\noLine \UIC{$\Gamma^{\bullet},x ; f \Rightarrow \Theta, x:_{s} \tau; \sigma$}
\AXC{}
\noLine \UIC{$s > 0$}
\LeftLabel{(Let)}
\BIC{$\Gamma^{\bullet}; \lett{x}{e}{f}  \Rightarrow s * \Gamma + \Theta; \sigma$}
\bottomAlignProof
\DisplayProof
\vskip 2em

%% fun
\AXC{$\Gamma^{\bullet}, x : \sigma; e \Rightarrow \Gamma, x:_{s}\sigma; \tau$}
\noLine \UIC{$s \ge 1$}
\LeftLabel{($\multimap$ I)}
\UnaryInfC{$\Gamma^{\bullet}; \lambda (x : \textcolor{blue}{\sigma}).e \Rightarrow \Gamma; \sigma \multimap \tau $}
\bottomAlignProof
\DisplayProof
\hskip 2em
%% app
\AXC{$\Gamma^{\bullet}; v \Rightarrow \Gamma; \sigma \multimap \tau$}
\noLine \UIC{$\Gamma^{\bullet}; w \Rightarrow \Delta;\sigma' $}
\noLine \UIC{ $\sigma' \sqsubseteq \sigma$}
\LeftLabel{($\multimap$ E)}
\UnaryInfC{$\Gamma^{\bullet}; v w \Rightarrow \Gamma + \Delta; \tau $}
\bottomAlignProof
\DisplayProof
\vskip 2em
%%

%% dep prod intro
\AXC{$\Gamma^{\bullet}; v \Rightarrow \Gamma_1;  \sigma$}
\AXC{}
\noLine \UIC{$\Gamma^{\bullet}; w \Rightarrow \Gamma_2; \tau$}
\LeftLabel{($\tand$ I)}
\BIC{$\Gamma^{\bullet}; \langle v, w \rangle \Rightarrow \max({\Gamma_1, \Gamma_2}); \sigma \tand \tau $}
\bottomAlignProof
\DisplayProof
\hskip 2em
%% dep prod elim
\AXC{$\Gamma^{\bullet};  v \Rightarrow \Gamma; \tau_1 \tand \tau_2$}
\LeftLabel{($\tand$ E)}
\UIC{$\Gamma^{\bullet}; {\pi}_i \ v \Rightarrow \Gamma; \tau_i$}
\bottomAlignProof
\DisplayProof
\vskip 2em

%% tens intro
\AXC{}
\noLine \UIC{$\Gamma^{\bullet};v  \Rightarrow \Gamma_1; \sigma $}
\noLine \UIC{$\Gamma^{\bullet}; w \Rightarrow \Gamma_2; \tau$}
\LeftLabel{($\tensor$ I)}
\UIC{$\Gamma;(v,w) \Rightarrow \Gamma_1 + \Gamma_2; \sigma \tensor \tau$}
\bottomAlignProof
\DisplayProof
\hskip 2em
%% tens elim
\AXC{$\Gamma^{\bullet}; v  \Rightarrow \Delta; \sigma \tensor \tau$ }
\noLine \UIC{$\Gamma^{\bullet},x: \sigma,y:\tau; e \Rightarrow \Gamma, x:_{s_1}\sigma,  y:_{s_2} \tau; \rho $}
\LeftLabel{($\tensor$ E)}
\UIC{$ \Gamma^{\bullet}; \lett{(x,y)}{v}{e}  \Rightarrow \Gamma + \max(s_1,s_2)* \Delta; \rho$}
\bottomAlignProof
\DisplayProof
\vskip 2em

%% sum intro
\AXC{$\Gamma^{\bullet}; v \Rightarrow \Gamma; \sigma$ }
\LeftLabel{($+$ I)}
\UIC{$\Gamma^{\bullet}; \inl \ v \Rightarrow \Gamma; \sigma + \tau$}
\bottomAlignProof
\DisplayProof
\hskip 2em
% sum elim
\AXC{$\Gamma^{\bullet}, x; e \Rightarrow \Theta,  x:_s \sigma; \rho_1$}
\noLine \UIC{$\Gamma^{\bullet},y ; f \Rightarrow \Theta,  y:_s \tau; \rho_2$}
\AXC{}
\noLine \UIC{$\Gamma^{\bullet}; v \Rightarrow \Gamma;  \sigma+\tau$}
\LeftLabel{($+$ E)}
\BIC{$\Gamma^{\bullet}; \mathbf{case} \ v \ \mathbf{of} \ (\inl x.e \ | \ \inr y.f) 
\Rightarrow \bar{s} * \Gamma + \Theta; \max(\rho_1,\rho_2)$}
\bottomAlignProof
\DisplayProof
\vskip 2em

% box intro
\AXC{$\Gamma^{\bullet}; v \Rightarrow \Gamma; \sigma$ }
\LeftLabel{($!$ I)}
\UIC{$\Gamma^{\bullet}; [v\textcolor{blue}{\{s\}}] \Rightarrow s * \Gamma; {!_s \sigma}$}
\bottomAlignProof
\DisplayProof
\hskip 0.5em
% box elim
\AXC{$\Gamma^{\bullet}; v \Rightarrow \Gamma; {!_s \sigma}$}
\noLine \UIC{$\Gamma^{\bullet}; x \Rightarrow \Theta, x:_{t*s} \sigma ; e : \tau$}
\LeftLabel{($!$ E)}
\UIC{$\Gamma^{\bullet} ; \lett{\boxx{x}}{v}{e}\Rightarrow t*\Gamma + \Theta; \tau$}
\bottomAlignProof
\DisplayProof
\vskip 2em

%% return
\AXC{$\Gamma^{\bullet}; v \Rightarrow \Gamma; \tau$}
\LeftLabel{(Ret)}
\UIC{$\Gamma^{\bullet}; \ret v \Rightarrow \Gamma; M_0 \tau$}
\bottomAlignProof
\DisplayProof
\hskip 2em
%% RND
\AXC{$\Gamma^{\bullet}; v \Rightarrow \Gamma; \num$}
\LeftLabel{(Rnd)}
\UIC{$\Gamma^{\bullet}; \rnd ~v \Rightarrow \Gamma; M_q \num$}
\bottomAlignProof
\DisplayProof
\vskip 2em

% OPS
\AXC{$\Gamma^{\bullet} ; v \Rightarrow \Gamma; \sigma$}
\noLine \UIC{$\{ \mathbf{op} :\sigma \lin \num \} \in {\Sigma}$}
\LeftLabel{(Op)}
\UIC{$\Gamma^{\bullet}; \mathbf{op}(v) \Rightarrow \Gamma; \num$}
\bottomAlignProof
\DisplayProof
\hskip 2em
% let-bind
\AXC{$\Gamma^{\bullet}; v  \Rightarrow \Gamma; M_r \sigma$}
\noLine \UIC{$\Gamma^{\bullet}, x; f \Rightarrow \Theta, x:_{s} \sigma;  M_{q} \tau$}
\LeftLabel{($M_u$ E)}
\UIC{$\Gamma^{\bullet}; \letm{x}{v}{f} \Rightarrow s * \Gamma + \Theta;  M_{s*r+q} \tau$}
\bottomAlignProof
\DisplayProof

\begin{equation*}
\bar{s} \triangleq
\begin{cases}
    s &  s > 0 \\ \epsilon &  \text{otherwise} 
\end{cases}
\end{equation*}

\end{center}
    \caption{Algorithmic rules for \Lang.  $\Gamma^0$ denotes  
    an environment where all variables have sensitivity zero. 
    The supertype ($\max$) and subtype ($\sqsubseteq$) relations on 
    \Lang{} types are given in \Cref{fig:max_ty} and \Cref{fig:sub_ty}, 
    respectively.}
    \label{fig:alg_rules}
\end{figure}

\begin{figure}[htbp]

\begin{align*}
%% CART PROD
\max(\sigma_1 \tand \tau_1, \sigma_2 \tand \tau_2) &\triangleq
\max(\sigma_1, \sigma_2) \tand \max(\tau_1,\tau_2) \\ 
%% TEN PROD
\max(\sigma_1 \otimes \tau_1, \sigma_2 \otimes \tau_2) &\triangleq
\max(\sigma_1, \sigma_2) \otimes \max(\tau_1,\tau_2) \\ 
%% SUM
\max(\sigma_1 + \tau_1, \sigma_2 + \tau_2) &\triangleq
\max(\sigma_1, \sigma_2) + \max(\tau_1,\tau_2) \\
%% MONAD
\max(\text{M}_{s_1}\tau_1,\text{M}_{s_2}\tau_2) &\triangleq
\text{M}_{\max(s_1,s_2)}\max(\tau_1,\tau_2) \\
%% BANG
\max(!_{s_1}\tau_1,!_{s_2}\tau_2) &\triangleq
\ !_{\min(s_1,s_2)}\max(\tau_1,\tau_2) \\
%%%%%
%% CART PROD
\min(\sigma_1 \tand \tau_1, \sigma_2 \tand \tau_2) &\triangleq
\min(\sigma_1, \sigma_2) \tand \min(\tau_1,\tau_2) \\ 
%% TEN PROD
\min(\sigma_1 \otimes \tau_1, \sigma_2 \otimes \tau_2) &\triangleq
\min(\sigma_1, \sigma_2) \otimes \min(\tau_1,\tau_2) \\ 
%% SUM
\min(\sigma_1 + \tau_1, \sigma_2 + \tau_2) &\triangleq
\min(\sigma_1, \sigma_2) + \min(\tau_1,\tau_2) \\
%% MONAD
\min(\text{M}_{s_1}\tau_1,\text{M}_{s_2}\tau_2) &\triangleq
\text{M}_{\min(s_1,s_2)}\min(\tau_1,\tau_2) \\
%% BANG
\min(!_{s_1}\tau_1,!_{s_2}\tau_2) &\triangleq
\ !_{\max(s_1,s_2)}\min(\tau_1,\tau_2) \\
%%%%%
%% FUNMAX
\max(\sigma_1 \multimap \tau_1, \sigma_2 \multimap \tau_2) &\triangleq
\min(\sigma_1, \sigma_2) \multimap \max(\tau_1,\tau_2) \\
%% FUNMIN
\min(\sigma_1 \multimap \tau_1, \sigma_2 \multimap \tau_2) &\triangleq
\max(\sigma_1, \sigma_2) \multimap \min(\tau_1,\tau_2) 
\end{align*}
    \caption{The $\max$ (supertype) relation on \Lang{} types, with $s_1,s_2 \in \NNR \cup \{\infty\}$. }
    \label{fig:max_ty}
\end{figure}

\begin{figure}[htbp]
%% ROW1
\begin{center}
%% CART PROD
\AXC{$\sigma \sqsubseteq \sigma'$}
\AXC{$\tau \sqsubseteq \tau'$}
\RightLabel{$\sqsubseteq .\tand$}
\BIC{$ \sigma \tand \tau \sqsubseteq \sigma' \tand \tau' $}
\DisplayProof
\hskip 2.5em
%% TEN PROD
\AXC{$\sigma \sqsubseteq \sigma'$}
\AXC{$\tau \sqsubseteq \tau'$}
\RightLabel{$\sqsubseteq .\tensor$}
\BIC{$ \sigma \tensor \tau \sqsubseteq \sigma' \tensor \tau' $}
\DisplayProof
\hskip 2.5em
%% SUM
\AXC{$\sigma \sqsubseteq \sigma'$}
\AXC{$\tau \sqsubseteq \tau'$}
\RightLabel{$\sqsubseteq .+$}
\BIC{$ \sigma + \tau \sqsubseteq \sigma' + \tau' $}
\DisplayProof
\vskip 1em

%% FUN
\AXC{$\sigma' \sqsubseteq \sigma$}
\AXC{$\tau \sqsubseteq \tau' $}
\RightLabel{$\sqsubseteq .\multimap$}
\BIC{$ \sigma \multimap \tau \sqsubseteq \sigma' \multimap \tau' $}
\DisplayProof
\hskip 2.5em
%% MONAD
\AXC{$\sigma \sqsubseteq \sigma'$}
\AXC{$u \le u'$}
\RightLabel{$\sqsubseteq .$M}
\BIC{$ \text{M}_u \sigma \sqsubseteq \text{M}_{u'} \sigma' $}
\DisplayProof
\hskip 2.5em
%% BANG
\AXC{$\sigma \sqsubseteq \sigma'$}
\AXC{$s \le s'$}
\RightLabel{$\sqsubseteq .!$}
\BIC{$ \bang{s'} \sigma \sqsubseteq !_{s} \sigma' $}
\DisplayProof
\vskip 1.5em

%% UNIT
\AXC{}
\RightLabel{$\sqsubseteq$-refl}
\UIC{$\sigma \ \sqsubseteq \sigma$ }
\DisplayProof

\end{center}
    \caption{\Lang{} subtyping relation, with $s,s',u,u' \in \NNR \cup \{\infty\}$. }
    \label{fig:sub_ty}
\end{figure}

We solve the sensitivity inference problem using the algorithm given in
\Cref{fig:alg_rules}. 

Given a term $e$ and a skeleton environment $\Gamma^{\bullet}$, the algorithm
produces an environment $\Gamma^{\bullet}$ with sensitivity information and a
type $\sigma$. Calls to the algorithm are written as $\Gamma^{\bullet}; e
\Rightarrow \Delta; \sigma$.  

Every step of the algorithm corresponds to a derivation in \Lang.  The syntax
of the algorithmic rules differs from the syntax of \Lang{}
(\Cref{fig:typing_rules}) in two places: the argument of lambda terms require
type annotations $(x : \textcolor{blue}{\sigma})$, and the box constructor
requires a sensitivity annotation $([v\textcolor{blue}{\{s\}}])$. The
algorithmic rules for these constructs are as follows:
\begin{center}
% box intro
\AXC{$\Gamma^{\bullet}; v \Rightarrow \Gamma; \sigma$ }
\RightLabel{($!$ I)}
\UIC{$\Gamma^{\bullet}; [v\textcolor{blue}{\{s\}}] \Rightarrow s * \Gamma; {!_s \sigma}$}
\bottomAlignProof
\DisplayProof
\qquad\qquad
%% fun
\AXC{$\Gamma^{\bullet}, x : \sigma; e \Rightarrow \Gamma, x:_{s}\sigma; \tau$}
\AXC{$s \ge 1$}
\RightLabel{($\multimap$ I)}
\BIC{$\Gamma^{\bullet}; \lambda (x : \textcolor{blue}{\sigma}).e \Rightarrow \Gamma; \sigma \multimap \tau$}
\bottomAlignProof
\DisplayProof
\end{center}

Following \citet{Amorim:2014:typecheck} the algorithm uses a \emph{bottom-up}
rather than a \emph{top-down} approach.  In the \emph{top-down} approach, given
a term $e$, type $\sigma$, and environment $\Gamma$, the environment is split
and used recursively to type the subterms of the expression $e$.  The
\emph{bottom-up} approach avoids splitting the environment $\Gamma$ by
calculating the minimal sensitivities and roundoff errors required to type each
subexpression.  The sensitivities and errors of each subexpression are then
combined and compared to $\Gamma$ and $\sigma$ using subtyping. The  subtyping
relation in \Lang{} is defined in \Cref{fig:sub_ty} and captures the fact that a
{$k$-sensitive} function is also {$k'$-sensitive} for $k \le k'$. Importantly,
subtyping is admissible in \Lang.
\begin{theorem}\label{thm:sub_type}
The typing judgment $\Gamma \vdash e : \tau'$ is derivable given a derivation
$\Gamma \vdash e : \tau$ and a type $\tau'$ such that $\tau \sqsubseteq \tau'$.
\end{theorem}
\begin{proof}
The proof follows by induction on the derivation $\Gamma \vdash e : \tau$. Most
cases are immediate, but some require weakening (\Cref{thm:weakening}). We
detail two examples here.  
\begin{description}
\item[Case (Var).] We are required to show $\Gamma,x:_s \tau,\Delta \vdash x :
\tau'$ given $\tau \sqsubseteq \tau'$. We can conclude by weakening
(\Cref{thm:weakening}) and the fact that the subenvironment relation is
preserved by subtyping; i.e., $x:_s \tau' \sqsubseteq x:_s \tau$.  %
\item[Case ($!$ I).] We are required to show 
$s * \Gamma \vdash [v] : \bang{s'}\sigma'$ for some $s'\le s$ and $\sigma'
\sqsubseteq \sigma$. By the induction hypothesis we have $\Gamma \vdash v :
\sigma'$, and by the box introduction rule ($!$ I) we have $s' * \Gamma \vdash
[v] : \bang{s'} \sigma'$. Because $s * \Gamma \sqsubseteq s' * \Gamma$ for $s'
\le s$ we can conclude by weakening.  \end{description}
\end{proof}

The algorithmic rules presented in \Cref{fig:alg_rules} define a \emph{sound}
type checking algorithm for \Lang: %
\begin{theorem}[Algorithmic Soundness]\label{thm:algsound}
If $ \ \Gamma^{\bullet};e \Rightarrow \Gamma; \sigma$ then 
there exists a derivation $\Gamma \vdash e : \sigma$. 
\end{theorem}

\begin{proof}
By induction on the algorithmic derivations, we see that every step of the
algorithm corresponds to a derivation in \Lang. Many cases are immediate, but
those that use subtyping, supertyping, or subenvironments are not; we detail
those here. 

\begin{description}
\item[Case ($\multimap$ E).] Applying subtyping (\Cref{thm:sub_type}) to the
induction hypothesis we have $\Delta \vdash w : \sigma$. We conclude by the
($\multimap$ E) rule.  % 
\item[Case ($\tand$ I).] We define the $\max$
relation on any two subenvironments $\Gamma_1$ and $\Gamma_2$ so that
$\max(\Gamma_1,\Gamma_2) \sqsubseteq \Gamma_1$ and $\max(\Gamma_1,\Gamma_2)
\sqsubseteq \Gamma_2$. Let us denote $\max(\Gamma_1,\Gamma_2)$ by the
environment $\Delta$. By the induction hypothesis and weakening
(\Cref{thm:weakening}) we have  $\Delta \vdash v : \sigma $ and 
$\Delta \vdash w : \tau$. We conclude by the ($\tand$ I) rule.
\item[Case ({$\tensor$ E}).] 
Let us denote by $\Theta$ the environment $\Gamma + {\max}(s_1,s_2) * \Delta$
and by $s$ the sensitivity ${\max}(s_1,s_2)$. By the induction hypothesis and
weakening (\Cref{thm:weakening}), we have \newline ${\Gamma, x:_{s}\sigma,
y:_{s}\tau \vdash e : \rho}$ and we can conclude by the ($\otimes$ E) rule.
\item[Case (+ E).] 
The proof relies on the fact that, given a $\rho$ such that $\rho =
\max(\rho_1,\rho_2)$, both $\rho_1$ and $\rho_2$ are subtypes of $\rho$. Using
this fact, and by subtyping (\Cref{thm:sub_type}) and the induction hypothesis,
we have $\Theta, x:_s \sigma \vdash e : \rho$  and $\Theta, y:_s \tau \vdash f
: \rho$. If $s >0$, we can can conclude directly by the ${\text{(+ E)}}$ rule.
Otherwise, we first apply weakening (\Cref{thm:weakening}) and then conclude by
the ${\text{(+ E)}}$ rule.  \end{description}
\end{proof}

\subsection{Evaluation}
In order to serve as a practical tool, our type checker must infer useful error
bounds within a reasonable amount of time.  Our empirical evaluation therefore
focuses on measuring two key properties: tightness of the inferred error bounds
and performance.  To this end, our evaluation includes a comparison in terms of
relative error and performance to two popular tools that {soundly} and
automatically bound relative error: {FPTaylor} \citep{FPTaylor} and
{Gappa} \citep{GAPPA}.  Although Daisy \citep{DAISY} and Rosa \citep{Rosa2} also
compute relative error bounds, they do not compute error bounds for the
directed rounding modes, and our instantiation of \Lang{} requires round towards
$+\infty$ (see \Cref{sec:nfuzz:examples}).  For our comparison to Gappa and FPTaylor,
we use benchmarks from FPBench \citep{FPBench}, which is the standard set of
benchmarks used in the domain; we also include the Horner scheme discussed in
\Cref{sec:nfuzz:examples}. There are limitations, summarized below, to the arithmetic
operations that the instantiation of \Lang{} used in our type checker can handle,
so we are only able to evaluate a subset of the FPBench benchmarks.  Even so,
larger examples with more than 50 floating-point operations are intractable for
most tools \citep{SATIRE}, including FPTaylor and Gappa, and are not part of
FPBench. Our evaluation therefore includes larger examples with well-known
relative error bounds that we compare against.  Finally, we used our
type checker to analyze the rounding error of four floating-point conditionals.

Our experiments were performed on a MacBook with a 1.4 GHz processor and 8 GB
of memory. Relative error bounds are derived from the relative precision
computed by \Lang{} using \Cref{eq:conv}.

\subsubsection{Limitations of \Lang} \label{sec:limit}
Soundness of the error bounds inferred by our type checker 
is guaranteed by \Cref{cor:err-sound} and the instantiation of
\Lang{} described in \Cref{sec:nfuzz:examples}. 
This instantiation imposes the following limitations on the benchmarks we can
consider in our evaluation. First, only the operations $+$, $*$, $/$, and
{sqrt} are supported by our instantiation, so we can't use benchmarks with
subtraction or transcendental functions. Second, all constants and variables
must be strictly positive numbers, and the rounding mode must be fixed as round
towards $+\infty$. These limitations follow from the fact that the RP metric
(\Cref{def:rp}) is only well-defined for non-zero values of the same sign. We
leave the exploration of tradeoffs between the choice of metric and the
primitive operations that can be supported by the language to future work.  
Given these limitations, along with the fact that \Lang{} does not currently
support programs with loops, we were able to include 13 of the 129 unique (at
the time of writing) benchmarks from FPBench in our evaluation.

\begin{table}
  \caption{Comparison of \Lang{} to {FPTaylor} and {Gappa}. The {Bound} column
    gives upper bounds on relative error (smaller is better); the bounds for
    {FPTaylor} and {Gappa} assume all variables are in $[0.1, 1000]$.  The
    {Ratio} column gives the ratio of \Lang's relative error bound to the
  tightest (best) bound of the other two tools; values less than 1 indicate that
\Lang{} provides a tighter bound. The {Ops} column gives the number of operations
in each benchmark.  Benchmarks from FPBench are marked with a (*).} 
\label{tab:comp} 
\begin{adjustbox}{max width=1.0\textwidth,center}
  \begin{tabular}{l c c c c c c c c} 
\hline
{Benchmark}
& {Ops}
& \multicolumn{3}{c}{Bound} 
& {Ratio}
& \multicolumn{3}{c}{Timing (s)} 
 \\ 
\cline{3-5}
\cline{7-9}
&
& {\Lang}
& {{FPTaylor}}
& {{Gappa}}
&
& {\Lang}
& {{FPTaylor}}
& {{Gappa}}
\\
\hline 
{hypot*} & 4 & 5.55e-16& {5.17e-16} & \textbf{4.46e-16} 
&1.3  & 0.002 & 3.55 & 0.069
\\ 
{x\_by\_xy*}& 3  &{4.44e-16} & fail & \textbf{2.22e-16} 
& 2 & 0.002 & -  & 0.034 
\\
{one\_by\_sqrtxx} & 3 &{5.55e-16} & 5.09e-13 & \textbf{3.33e-16} 
& 1.7 & 0.002 & 3.34  & 0.047 
\\
{{sqrt\_add*} }  & 5 & 9.99e-16 & {6.66e-16} & \textbf{5.54e-16} 
& 1.5  & 0.003 & 3.28 & 0.092
\\
{test02\_sum8*} & 8 & {1.55e-15} & 9.32e-14 & {1.55e-15} 
& 1 &  0.002 & 14.61 & 0.244 
\\
{nonlin1*}  & 2 & {4.44e-16} & {4.49e-16} & \textbf{2.22e-16} 
& 2 & 0.003 & 3.24 & 0.040 
\\
{test05\_nonlin1*} & 2 & {4.44e-16} & {4.46e-16} & \textbf{2.22e-16} 
& 2 & 0.008 & 3.27 & 0.042 
\\
{verhulst*} & 4 & {8.88e-16} & {7.38e-16} & \textbf{{4.44e-16}} 
& 2 & 0.002 & 3.25 & 0.069 
\\
{predatorPrey*} & 7 & {1.55e-15} & {4.21e-11} & \textbf{8.88e-16}
& 1.7 & 0.002 & 3.28 & 0.114 
\\
{test06\_sums4\_sum1*} & 4 & {6.66e-16} & {6.71e-16} & {6.66e-16} 
& 1 & 0.003 & 3.84 & 0.069
\\
{test06\_sums4\_sum2*} & 4 & {6.66e-16} & {1.78e-14} & 
\textbf{4.44e-16} & 1.5 & 0.002 & 11.02 & 0.055 
\\
{i4*} & 4 & {4.44e-16} & {4.50e-16} & {4.44e-16} 
& 1 & 0.002 & 3.30 & 0.055 
\\
{Horner2} & 4 & {4.44e-16} & 6.49e-11 & {4.44e-16}
& 1  & 0.002 & 11.72  & 0.052 
\\
{{Horner2\_with\_error}} & 4 & {1.55e-15} & 1.61e-10 & 
\textbf{1.11e-15} & 1.4 & 0.002 & 19.56 & 0.119
\\
{Horner5} & 10 & {1.11e-15} & 1.62e-01 & {1.11e-15} 
& 1 & 0.003 & 22.08  & 0.209
\\
{Horner10} & 20 & {2.22e-15} & 1.14e+13 & {2.22e-15} 
& 1 & 0.003 & 40.68  & 0.650
\\
{Horner20} & 40  & {4.44e-15} & 2.53e+43 & {4.44e-15} 
& 1 & 0.003 & 109.42  & 2.246
\\
\hline 
\end{tabular}
\end{adjustbox}
\end{table}

\subsubsection{Small Benchmarks}
The results for benchmarks with fewer than 50 floating-point operations are
given in \Cref{tab:comp}. Eleven of the seventeen benchmarks are taken from the
FPBench benchmarks.  Both {FPTaylor} and {Gappa} require user provided interval
bounds on the input variables in order to compute the relative error; we used
an interval of $[0.1,1000]$ for each of the benchmarks.  We used the default
configuration for {FPTaylor}, and used {Gappa} without providing hints for
interval subdivision.  The floating-point format of each benchmark is binary64,
and the rounding mode is set at round towards $+\infty$; the unit roundoff in
this setting is $2^{-52}$ (approximately {$2.22\text{e-}16$}).  Only
\lstinline{Horner2_with_error} assumes error in the inputs.

\subsubsection{Large Benchmarks}
\Cref{tab:large} shows the results for benchmarks with 100 or more
floating-point operations. Five of the nine benchmarks are taken from
\textsc{Satire} \citep{SATIRE}, an \emph{empirically sound} static analysis tool
that computes absolute error bounds.  Although \textsc{Satire} does not
statically compute relative error bounds for the benchmarks listed in
\Cref{tab:large}, most of these benchmarks have well-known worst case relative
error bounds that we can compare against.  These bounds are given in the {Std.}
column in \Cref{tab:large}; the relevant references are as follows: Horner's
scheme \citep[cf.][p. 95]{Higham:2002:Accuracy}, summation 
\citep[cf.][p.260]{boldo:2023:floatingpoint}, and matrix multiply 
\citep[cf.][p.63]{Higham:2002:Accuracy}. For matrix multiplication, we report the max elementwise
relative error bound produced by \Lang.  When available, the Timing column in
\Cref{tab:large} lists the time reported for \textsc{Satire} to compute
\emph{absolute} error bounds \citep[cf.][Table III]{SATIRE}. 

\subsubsection{Conditional Benchmarks}
\Cref{tab:cond} shows the results for conditional benchmarks. Two of the four
benchmarks are taken from FPBench and the remaining benchmarks are examples
from Dahlquist and Bj\"{o}rck \citep[cf.][p. 119]{Dahlquist}.  We were unable
to compare the performance and computed relative error bounds shown in
\Cref{tab:cond} against other tools. While Daisy, FPTaylor, and Gappa compute
relative error bounds, they don't handle conditionals. And, while PRECiSA can
handle conditionals, it doesn't compute relative error bounds. Only Rosa
computes relative error bounds for floating-point conditionals, but Rosa
doesn't compute bounds for the directed rounding modes. 

\begin{table}
  \caption{
The performance of \Lang{} on benchmarks with 100 or more floating-point operations.  
The {Std.} column gives relative error bounds from the literature.
Benchmarks from  \textsc{Satire} are marked with with an ({a});
the \textsc{Satire} subcolumn gives timings for the computation of \emph{absolute} error bounds as reported in \citep{SATIRE}.
} 
\label{tab:large} 
  \begin{tabular}{l c c c c c} 
\hline
{Benchmark}
& {Ops}
& {Bound (\Lang)} 
& {Bound (Std.)} 
& \multicolumn{2}{c}{Timing (s)} 
\\
\cline{5-6}
&
&
&
& {\Lang}
& {{\textsc{Satire}}}

\\
\hline 
{Horner50}$^{a}$& 100  &{1.11e-14}  & {1.11e-14} & 0.009 & 5
\\
{MatrixMultiply4}& 112  &{1.55e-15}  & {8.88e-16} & 0.003 & -
\\
{Horner75}& 150  &{1.66e-14} & {1.66e-14} & 0.020 & -
\\
{Horner100}& 200  &{2.22e-14} & {2.22e-14} & 0.040 & -
\\
{SerialSum}$^{\text{a}}$& 1023  &{2.27e-13}  &  {2.27e-13}  & 5 & 5407
\\
{{Poly50}}$^{\text{a}}$  & 1325 &{2.94e-13} &  - & 2.120 & 3
\\
{MatrixMultiply16}& 7936  &{6.88e-15}  & {3.55e-15} & 0.040 & -
\\
{MatrixMultiply64}$^{\text{a}}$ & 520192 &{2.82e-14}  & {1.42e-14} & 10 & 65
\\
{MatrixMultiply128}$^{\text{a}}$ & 4177920  &{5.66e-14}  & {2.84e-14} & 1080 & 763
\\
\hline 
\end{tabular}
\end{table} 
\begin{table}
  \caption{
The performance of \Lang{} on conditional benchmarks. Benchmarks 
from FPBench are marked with a (${*}$). 
Benchmarks 
from Dahlquist and Bj\"{o}rck ~\citep[cf.][p. 119]{Dahlquist} are marked with with a ($\text{b}$). 
} \label{tab:cond} 
\centering
  \begin{tabular}{l c c } 
\hline
{Benchmark}
& {Bound} 
& {Timing (ms)} 

\\
\hline 
{PythagoreanSum}$^{\text{b}}$ &{8.88e-16} & 2 
\\
{HammarlingDistance}$^{\text{b}}$  &{1.11e-15} & 2 
\\
{squareRoot3}$^{*}$  &{4.44e-16} & 2 
\\
 {{squareRoot3Invalid}}$^{*}$  &{4.44e-16} & 2
\\
\hline 
\end{tabular}
\end{table} 

\subsubsection{Evaluation Summary}

We draw three main conclusions from our evaluation. 

\paragraph{\emph{Roundoff error analysis via type checking is fast.}}
On small and conditional benchmarks, \Lang{} infers an error bound in the order
of milliseconds. This is at least an order of magnitude faster than either
{Gappa} or {FPTaylor}. On larger benchmarks, \Lang's performance surpasses that
of comparable tools by computing bounds for problems with up to 520k operations
in under a minute.

\paragraph{ \emph{Roundoff error bounds derived via type checking are useful.}}
On most small benchmarks \Lang{} produces a tighter relative error bound than
either FPTaylor or Gappa. On the few benchmarks where FPTaylor computes a
tighter bound, \Lang's results are still well within an order of magnitude. For
benchmarks where rounding errors are composed and magnified, such as
\lstinline{Horner2_with_error}, and on somewhat larger benchmarks like
\lstinline{Horner2}-\lstinline{Horner20}, our type-based approach performs
particularly well. On larger benchmarks that are intractable for the other
tools, \Lang{} produces bounds that are nearly optimal in comparison to those
from the literature. \Lang{} is also able to provide non-trivial relative error
bounds for floating-point conditionals.

\paragraph{ \emph{Roundoff error bounds derived via type checking are strong.}} 
The relative error bounds produced by \Lang{} hold for all positive real inputs,
assuming the absence of overflow and underflow. In comparison, the relative
error bounds derived by  {FPTaylor} and {Gappa} only hold for the user provided
interval bounds on the input variables, which we took to be $[0.1,1000]$.
Increasing this interval range allows {FPTaylor} and {Gappa} to give stronger
bounds, but can also lead to slower analysis.  Furthermore, given that relative
error is poorly behaved for values near zero, some tools are sensitive to the
choice of interval.  We see this in the results for the benchmark
\lstinline{x_by_xy} in \Cref{tab:comp}, where we are tasked with calculating
the roundoff error produced by the expression $x/(x+y)$, where $x$ and $y$ are
binary64 floating-point numbers in the interval $[0.1,1000]$.  For these
parameters, the expression lies in the interval $[5.0\text{e-}05,1.0]$ and the
relative error should still be well defined.  However, {FPTaylor} (used with
its default configuration) fails to provide a bound, and issues a warning due
to the value of the expression being too close to zero. 

\begin{remark}[User specified Input Ranges]
Allowing users to specify input ranges is a feature of many tools used for
floating-point error analysis, including FPTaylor and Gappa.  In some cases, a
useful bound can't be computed for an unbounded range, but can be computed
given a well-chosen bounded range for the inputs.  Input ranges are also
required for computing absolute error bounds.  Extending \Lang{} with bounded
range inputs is left to future work; we note that this feature could be
supported by adding a new type to the language, and by adjusting the types of
primitive operations.
\end{remark}

\section{Related Work}\label{sec:nfuzz:related}

\subsubsection{Abstract Interpretation}
The theory of abstract interpretation offers a generic framework for designing
sound static analysis tools. At the heart of any abstract interpretation
framework is the concept of an \emph{abstract domain}, which provides a
mathematical approximation of the program properties being analyzed. For
floating-point programs, abstract interpretation frameworks use numerical
abstract domains to soundly overapproximate the set of values that program
variables can represent. Common numerical abstract domains for analyzing
floating-point programs include interval arithmetic \citep{Moore:2009:Interval},
affine arithmetic \citep{Figueiredo:2004:Affine}, and convex
polyhedra \citep{Chen:2008:Sound}.   

Many tools based on abstract interpretation aim to derive sound and accurate
bounds for floating-point variables but do not compute bounds or estimates on
the rounding error in a floating-point result. These tools make it possible
to validate numerical behaviors of systems without precisely tracking the
rounding error associated with each floating-point operation, and are described
in works by \citet{Rivera:2024:TVPI}, \citet{mine:2004:relational},
\citet{Chen:2008:Sound,Chen:2009:Polyhedra,Chen:2010:Abstract},
\citet{Bertrand:2009:Apron}, \citet{Chapoutot:2010:Interval}, and
\citet{Chapoutot:2009:Simulink}.  In a somewhat orthogonal research direction,
abstract interpretation frameworks have also been used in the development of
satisfiability decision procedures for constraints over floating-point
arithmetic \citep{Haller:2012:deciding}.

Tools that use abstract domains to derive sound rounding error bounds include
Gappa \citep{Gappa:2010}, Rosa \citep{Rosa2}, Daisy \citep{Daisy:2018},
PRECiSA \citep{PRECISA,PRECISA:2024}, Fluctuat \citep{Fluctuat:2011}, and
Astr{\'e}e \citep{Astree:2005}. 

While abstract interpretation is flexible and can be generally applied to programs with
conditionals and loops, it can significantly overestimate rounding error,
and it is difficult to model the cancellation of errors.
Unlike abstract interpretation, type-based approaches like \Lang{} provide a
mechanism for defining valid programs, and can support features like foreign
function interfaces \citep{Ghica:2014:bounded}. 

\subsubsection{Type Systems for Numerical Computations}
A type system due to \citet{Martel:2018:types} uses dependent types to track
numerical errors. A significant difference between \Lang{} and the type system 
proposed by Martel is error soundness. In Martel's system, the soundness result 
says that a semantic relation capturing the notion of accuracy between a 
floating-point expression and its ideal counterpart is preserved by a reduction 
relation. This is weaker than a standard type soundness guarantee. In particular, 
it is not shown that  well-typed terms satisfy the semantic relation.  In \Lang, 
the central novel property guaranteed by our type system is much stronger:
well-typed programs of monadic type satisfy the error bound indicated by their
type.

\subsubsection{Optimization Techniques for Program Analysis}
To provide more precise bounds, many methods rely on optimization.
Conceptually, these methods bound the roundoff error by representing the error
symbolically as a function of the program inputs and the error variables
introduced during the computation, and then perform global optimization over
all settings of the errors \citep{truong:2014:finding}. 
Since the error expressions are typically complex,
verification methods use approximations to simplify the error expression to
make optimization more tractable, and mostly focus on straight-line programs.
For instance, {Real2Float} \citep{REAL2FLOAT} separates the error expression
into a linear term and a non-linear term; the linear term is bounded using
semidefinite programming, while the non-linear term is bounded using interval
arithmetic. FPTaylor \citep{FPTaylor} was the first tool to use Taylor
approximations of error expressions.  \citet{Abbasi:2023:compose}
describe a modular method for bounding the propagation of errors using Taylor
approximations, and {Rosa} \citep{Rosa1,Rosa2} uses Taylor series to
approximate the propagation of errors in possibly non-linear terms.

In contrast, our type system does not rely on global optimization
and can be instantiated
to different models of floating-point arithmetic with minimal changes. Our
language supports a variety of datatypes and higher-order functions. While our
language does not support recursive types and general recursion, similar
languages support these features \citep{Fuzz,Amorim:2017:metric,DalLago:2022:relational}
and we expect they should be possible in \Lang; however, the precision of the
error bounds for programs using general recursion might be poor. Another
limitation of our method is in typing conditionals: while \Lang{} can only
derive error bounds when the ideal and floating-point executions follow the
same branch, tools that use general-purpose solvers (e.g., {PRECiSA} and
{Rosa}) can produce error bounds for programs where the ideal and
floating-point executions take different branches.

\subsubsection{Verification and Synthesis} 
Formal verification has a long history in the area of numerical computations,
starting with the pioneering work of Harrison
\citep{Harrison:1999:holfloat,Harrison:1997:holfloat1,Harrison:2000:trig}.
Formalized specifications of floating-point arithmetic have been developed in
the Coq \citep{Flocq}, Isabelle \citep{IEEE_Floating_Point-AFP}, and
PVS \citep{MINER199531} proof assistants. These specifications have been used to
develop sound tools for floating-point error analysis that generate proof
certificates, such as {VCFloat} \citep{VCFloat1,Appel:CPP:2024} and
{PRECiSA} \citep{PRECISA,PRECISA:2024}.  They have also been used to  mechanize proofs of
error bounds for specific numerical programs~(e.g.,
\citet{Kellison:2022:ode,boldo2014,Tekriwal:2023:jacobi,Kellison:Arith:2023,Moscato:2019,
tekriwal:2023:phd}).
Work by \citet{ASTREE} has applied abstract interpretation to verify the
absence of floating-point faults in flight-control software, which have caused
real-world accidents.  Finally, recent work uses \emph{program synthesis}:
{Herbie} \citep{Herbie} automatically rewrites numerical programs to reduce
numerical error, while {RLibm} \citep{RLIBM} automatically generates
correctly-rounded math libraries.

\subsubsection{Sensitivity Type Systems}
\Lang{} belongs to a line of work on linear type systems for sensitivity
analysis, starting with \emph{Fuzz} \citep{Fuzz}. We point out a few especially
relevant works. Our syntax and typing rules are inspired by \citet{DalLago:2022:relational}, who
propose a family of \emph{Fuzz}-like languages and define various notions of
operational equivalence; we are inspired by their syntax, but our case
elimination rule ($+$ E) is different: we require $s$ to be strictly positive
when scaling the conclusion. This change is due to a subtle difference in how
sums are treated. 

In \Lang, as in \emph{Fuzz}, the distance between left and right injections is
$\infty$, whereas in the system by \citet{DalLago:2022:relational}, left and right injections are
not related at any distance. Our approach allows non-trivial operations
returning booleans to be typed as infinite sensitivity functions, but the case
rule must be adjusted: to preserve soundness, the conclusion must retain a
dependence on the guard expression, even if the guard is not used in the
branches.

\citet{Amorim2019} added a graded monadic type to \emph{Fuzz} to handle more complex
variants of differential privacy; in their application, the grade does not
interact with the sensitivity language.  
Finally, recent work by
\citet{BunchedFuzz} proposes a variant of \emph{Fuzz} with ``bunched'' (tree-shaped)
contexts, with two methods of combining contexts.  It could be interesting to
develop a bunched version of \Lang---the metrics for $\otimes$ and $\tand$
could be naturally accommodated in the contexts, possibly leading to more
precise error analysis.

\subsubsection{Other Approaches}
The numerical analysis literature has explored other conceptual tools for
static error analysis, such as stochastic error analysis
\citep{doi:10.1137/20M1334796}.  Techniques for \emph{dynamic} error analysis,
which estimate the rounding error at runtime, have also been proposed
\citep{Higham:2002:Accuracy}. 

It would be interesting to consider these techniques from a formal methods
perspective, whether by connecting dynamic error analysis with ideas like
runtime verification, or developing methods to verify the correctness of
dynamic error analysis.

\section{Conclusion}\label{sec:nfuzz:conclusion}
\Lang{} is a functional programming language designed to express quantitative
bounds on forward rounding error.  The rounding error analysis modeled by
\Lang{} is standard: a sensitivity analysis is combined with a local rounding
error analysis to derive a global rounding error bound.  \Lang{} uses a linear
type system and a graded comonad to perform a sensitivity analysis, and uses a
novel a graded monad to track rounding error.  A major benefit of our
type-based approach is \emph{soundness}: \Lang{} programs are guaranteed to
satisfy the error bounds assigned to them by the type system, which are
overapproximations of the true rounding error. Another advantage is
scalability: \Lang{} can infer tight error bounds for significantly larger
programs than existing static analysis tools in a reasonable amount of time.
Moreover, on well known benchmarks, \Lang{} infers error bounds that are
competitive with those produced by existing tools, often with significantly
faster performance. \Lang{} can be extended in various ways, and we conclude
this chapter with a discussion of  promising  directions for future
development. 

\subsubsection{Additional Language Features}
Rounding error bounds are typically parametric in the length of the input
data. For instance, the error bound for Horner's scheme (cf.
\Cref{sec:nfuzz:examples}) is usually expressed in terms of the polynomial's
degree, which corresponds to the length of the vector of polynomial
coefficients.  Currently, \Lang{}, like \emph{Fuzz}, only supports numeric
annotations (grades), but these are insufficient for expressing how error
bounds depend on properties of input data. To address this limitation,
\Citet{Gaboardi:2013:dfuzz} introduced lightweight dependent types---sensitivity 
variables and quantifiers over these variables---to the types of
\emph{Fuzz}, enabling the expression of more general sensitivity properties.
Given this prior work, we expect \Lang{}'s type system could similarly be
extended to support lightweight dependent types.

Combining this extension with a bounded loop construct would further enhance
\Lang{} by enabling the expression of more general error bounds and reducing
the verbosity of programs.  Compared to a bounded loop construct, it is less
obvious that extending \Lang{} to support general recursive functions and
types, even those with  precise sensitivity as described by
\citet{Amorim:2017:metric} for \emph{Fuzz}, would be immediately useful for
potential users. (Although it is clear that it would complicate the metatheory). 

\subsubsection{Probabilistic Rounding}
Probabilistic models of rounding errors have been proposed for both
deterministic computations \citep{doi:10.1137/18M1226312,
Ipsen:2020:probabilistic} and probabilistic computations
\citep{Constantinides:2021:probabilistic}.  \Lang{} could be extended to track
probabilistic rounding errors by incorporating techniques from probabilistic
languages, such as those described by \citet{Amorim2019} or
\citet{Crubille:2015:metric}.  While \citet{kahan:1996:improbability} and
\citet{Ipsen:2020:probabilistic} provide critical assessments of probabilistic
rounding error analyses, \citet{Constantinides:2021:probabilistic} argue that
the probabilistic approach is necessary for analyzing rounding error in
probabilistic computations. 

\subsubsection{Mixed-Precision Computations}
It is possible to represent mixed-precision computations in \Lang{} by adding
additional $\rnd$ constructs to the language, with each construct corresponding
to a different supported precision. The challenge lies in  accurately modeling
the expected behavior when composing rounding operations.  According to the
IEEE standard \citep{IEEE2019}, conversions between formats with the same radix
but wider precision should always be exact. Evaluating the expression
$\letm{x}{(\rnds ~3.0)}{(\rndd ~x)}$, where $\rnds$ rounds to binary32 and
$\rndd$ rounds to binary64, should therefore produce only a single rounding
effect, due to the evaluation of $(\rnds ~3.0)$. However, under the current
monadic sequencing (MLet) rule, each operation introduces its own error,
effectively modeling a scenario where both rounding steps contribute to the
total error. While this is a sound overapproximation, it raises the question of
whether \Lang{} can support a more precise, fine-grained analysis that
distinguishes between scenarios where no additional error is introduced. One
possible approach to achieving this finer-grained analysis is to use
\emph{graded monad transformers}, as described by \citet{ivaskovic:2023:phd},
for combining two analyses of computations based on the same type of effectful
operation.

\subsubsection{Additional Error Measures}
A natural follow up to our work on \Lang{} is to consider 
whether or not other error measures from the literature can be used in
place of relative precision (\Cref{def:rp}).
Error measures that can uniformly represent floating-point error on
both large and small values are the \emph{units in the last place} (ULP) error,
which measures the number of floating-point values between an approximate and
exact value, and its logarithm, \emph{bits of error} \citep{FPBench}:
\begin{align}
  er_{\textsc{ulp}}(x,\tilde{x}) = 
  |\mathbb{F} \cap [\min(x,\tilde{x}),\max(x,\tilde{x})] |
        \quad \text{and} \quad
  er_{bits}(x,\tilde{x}) =  \log_2{er_{\textsc{ulp}}(x,\tilde{x})}. 
  \label{def:ulps_ers}
\end{align}
While static analysis tools that provide sound, worst-case error bounds
for floating-point programs compute relative or absolute error 
bounds (or both), the ULP error and its logarithm are often used in tools 
that optimize either the performance or accuracy of floating-point programs,
like Herbie \citep{Herbie} and \textsc{Stoke} \citep{STOKE-FP}.

Generalizations of the relative precision metric have been proposed by 
\citet{ziv:1982:relative} and 
\citet{Pryce:1985:Matrix,Pryce:1984:Vectors},
for analyzing the error of programs involving vectors and matrices. 
It would be interesting 
to explore whether \Lang{} could accommodate these metrics, though this 
would naturally require extending the type system to support 
types for matrices and vectors; \emph{Fuzz}-like languages with 
matrix types have been described by \citet{Duet} and \citet{BunchedFuzz}. 

\subsubsection{Mechanization}
It would be possible to mechanize several results about \Lang{}
in a proof assistant like Coq. For instance, formalizing
key aspects of the operational semantics, such as type-preservation and 
strong normalization (\Cref{thm:SN}), as well as soundness of the 
type checking algorithm (\Cref{thm:algsound}), would be straightforward 
but valuable exercises. To our knowledge, 
there is currently no mechanization of sensitivity type 
systems like \Lang{} in a proof assistant, making this an interesting area 
for further exploration.

%%%%%%%%%%%%%%%%%%%
\chapter{A Language for Backward Error Analysis} \label{chapter:bean}

This chapter presents \bea{}, a programming language for
\textsc{\scalefont{1.25}{\textbf{b}}}ackward
\textsc{\scalefont{1.25}{\textbf{e}}}rror
\textsc{\scalefont{1.25}{\textbf{an}}}alysis.  \bea{} features a type system
that tracks how backward error flows through programs, and ensures that
well-typed programs have bounded backward error. 

\section{Introduction}\label{sec:bean:introduction}
With \bean{}, our point of departure from other static analysis tools for floating-point
rounding error analysis is our focus on deriving backward error bounds rather
than forward error bounds. To facilitate our description of backward error, we
will use the following notation: floating-point approximations to real-valued
functions, as well as data with perturbations due to floating-point rounding
error, will be denoted by a tilde. For instance, the floating-point
approximation of a real-valued function $f$ will be denoted by $\tilde{f}$, and
data that are intended to represent slight perturbations of $x$ will be denoted
by $\tilde{x}$.

\subsubsection*{Backward Error and Backward Stability}
Given a floating-point result $\tilde{y}$ approximating $y = f(x)$ with $f:
\R^n \rightarrow \R^m$, a forward error analysis directly measures the accuracy
of the floating-point result by bounding the distance between $\tilde{y}$ and
$y$. In contrast, a backward error analysis identifies an input $\tilde{x}$
that would yield the floating-point result when provided as input to $f$; i.e.,
such that $\tilde{y} = f(\tilde{x})$.  The backward error is a measure of the
distance between the input $x$ and the input $\tilde{x}$. 

\begin{figure}
\begin{center}
\begin{tikzpicture}
  \draw[draw=blue!20, thick,fill=blue!20, opacity = 0.6] (0,0) 
  ellipse (1.25cm and 1.5cm);
  \draw[draw=green!20, thick,fill=green!20, opacity = 0.6] (5,0) 
  ellipse (1.25cm and 1.5cm);
  \node[label=below:{Input Space $\R^n$}] (A) at (0,2.2) {};
  \node[label=below:{Output Space $\R^m$}] (B) at (5,2.2) {};
  % Draw and label points inside the carrier set
  \node[circle,fill=black,inner sep=0pt,minimum size=3pt,
  label=above:{$x$}] (A) at (0,0.5) {};
  \node[circle,fill=black,inner sep=0pt,minimum size=3pt,
  label=below:{$\tilde{x}$}] (B) at (0,-0.5) {};
  \node[circle,fill=black,inner sep=0pt,minimum size=3pt] (C) at (5,0) {};
  \node[align=center,anchor=north] (lab) at (5.4,0.75) 
    {$\tilde{{f}}(x)$ \\ = \\ $f(\tilde{x})$};
  \draw (0.1,0.5)  edge[bend left, ->, thick] node[above] 
  {$\tilde{{f}}$} (4.9,0.1);
  \draw (0.1,-0.5) edge[bend right, ->, thick] node[above] {${{f}}$} 
  (4.9,-0.1);
  \draw[dashed,thick,draw = red] node[left] {${\delta x}$} (A) -- (B);
\end{tikzpicture}
\end{center}
\caption{An illustration of backward error.
  The function $\tilde{f}$ represents a floating-point implementation of the
  function $f$. Given the points $\tilde{x} \in \R^n$ and $x \in \F^n \subset
  \R^n$ such that $\tilde{f}(x) = f(\tilde{x})$, the backward error is the
  distance $\delta x$ between $x$ and $\tilde{x}$. }
\label{fig:BEA}
\end{figure}
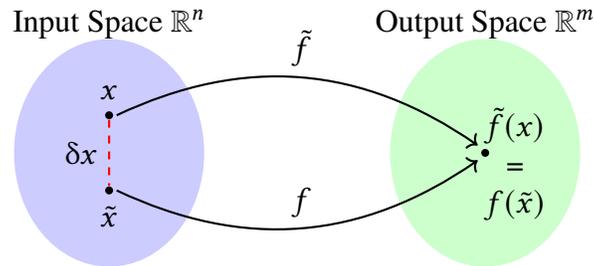

An illustration of the backward error is given in \Cref{fig:BEA}. If the
backward error is small for every possible input, then an implementation is
said to be \emph{backward stable}:

\begin{definition}(Backward Stability)\label{def:bstab}
  A floating-point implementation $\tilde{f} : \F^n \rightarrow \F^m$ of a
  real-valued function ${f} : \R^n \rightarrow \R^m$ is \emph{backward stable}
  if, for every input $x \in \F^n$, there exists an input 
  $\tilde{x} \in \R^n$ such that 
  \begin{equation} 
    f(\tilde{x}) = \tilde{f}(x) \text{ and } d(x,\tilde{x}) \le \alpha u
  \end{equation} 
  where $u$ is the \emph{unit roundoff}---a value that depends on the precision
  of the floating-point format $\F$, $\alpha$ is a small constant, and 
  $d :\R^n \times \R^n \rightarrow \R \cup \{+\infty\}$ provides a measure 
  of distance in $\R^n$.
\end{definition}

In general, a large forward error can have two causes: the conditioning of the
problem being solved or the stability of the program used to solve it. If the
problem being solved is \emph{ill-conditioned}, then it is highly sensitive to
floating-point rounding errors, and can amplify these errors to produce
arbitrarily large changes in the result. Conversely, if the problem is
well-conditioned but the program is \emph{unstable}, then inaccuracies in the
result can be attributed to the way rounding errors accumulate during the
computation. While forward error alone does not distinguish between these two
sources of error, backward error provides a controlled way to separate them. 
The relationship between forward error and backward error is governed by the 
\emph{condition number}, which provides a quantitative measure of the 
conditioning of a problem:
\begin{align}
  \text{forward error $\le$ condition number $\times$ backward error}.
\end{align}
A more precise definition of the condition number is given in \Cref{def:cnum}. 

By automatically deriving sound backward error bounds that indicate the
backward stability of programs, \bea{} addresses a significant gap in current
tools for automated error analysis.  To quote Dianne P. O'Leary
\citep{Oleary:2009:book}: ``Life may toss us some ill-conditioned problems, 
but there is no good reason to settle for an unstable algorithm.''

\subsubsection*{Backward Error Analysis by Example} 
A motivating example illustrating the importance of backward error is the dot
product of two vectors. While the dot product can be computed in a backward
stable way, if the vectors are orthogonal (i.e., when the dot product is zero)
the floating-point result can have arbitrarily large relative forward
error.  This means that, for certain inputs, a forward error analysis can only
provide trivial bounds on the accuracy of a floating-point dot product. In
contrast, a backward error analysis can provide non-trivial bounds describing
the quality of an implementation for all possible inputs. 

To see how a backward error analysis works in practice, suppose we are given
the vectors $x = (x_0,x_1) \in \R^2$ and $y = (y_0,y_1) \in \R^2$ with
floating-point entries. The exact dot product simply computes the sum $s = x_0
\cdot y_0 + x_1 \cdot y_1$, while the floating-point dot product computes
$\tilde{s} = (x_0 \odot_\F y_0) \oplus_\F (x_1 \odot_\F y_1)$, where
$\oplus_\F$ and $\odot_\F$ represent floating-point addition and
multiplication, respectively. A backward error bound for the computed result
$\tilde{s}$ can be derived based on bounds for addition and multiplication.
Following the error analysis proposed by Olver \citep{Olver:1978:rp}, and
assuming no overflow and underflow, floating-point addition and multiplication
behave like their exact arithmetic counterparts, with each input subject to
small perturbations.  Specifically, for any $a_1,a_2,b_1,b_2 \in \R$, we have:
\begin{align}
a_1 \oplus_\F a_2  &=  a_1e^\delta + {a}_2e^\delta \\
&= \tilde{a}_1 + \tilde{a}_2
\label{eq:be_add} \\
b_1 \odot_\F b_2    &= {b}_1e^{\delta'/2} \cdot {b}_2e^{\delta'/2} \\
&= \tilde{b}_1 \cdot \tilde{b}_2 \label{eq:be_mul}
\end{align}
with $|\delta|,|\delta'| \le u/(1-u)$, where $u$ is the unit roundoff. 
For convenience, we use the notation $\varepsilon = u/(1-u)$. 
The basic intuition behind a perturbed input like $\tilde{a}_1 = a_1e^\delta$
in \Cref{eq:be_add} is that $\tilde{a}$ is approximately equal to $a_1 +
a_1\delta$ when the magnitude of $\delta$ is extremely small. 

We can use \Cref{eq:be_add} and \Cref{eq:be_mul} to perform a backward error
analysis for the dot product: we can define the vectors $\tilde{x} =(\tilde{x}_0
, \tilde{x}_1)$ and $\tilde{y} = (\tilde{y}_0,\tilde{y}_1)$ such that their dot
product computed in exact arithmetic is equal to the floating-point result
$\tilde{s}$: 
\begin{align}
\tilde{s} = (x_0 \odot_\F y_0) \oplus_\F (x_1 \odot_\F y_1) 
&= (x_0 e^{\delta_0/2} \cdot y_0 e^{\delta_0/2}) \oplus_\F 
  (x_1 e^{\delta_1/2} \cdot y_1 e^{\delta_1/2}) 
  \nonumber \\
&= (x_0 e^{\delta_0/2} \cdot y_0 e^{\delta_0/2})e^{\delta_2} +
    (x_1 e^{\delta_1/2} \cdot y_1 e^{\delta_1/2}) e^{\delta_2}
  \nonumber \\
&= (x_0 e^{\delta} \cdot y_0 e^{\delta}) +
    (x_1 e^{\delta'} \cdot y_1 e^{\delta'})
  \nonumber \\
&= (\tilde{x}_0  \cdot \tilde{y}_0) + (\tilde{x}_1 \cdot \tilde{y}_1)
  \label{eq:dpbe}
\end{align}
where $|\delta|, |\delta'| \le {3\varepsilon/2}$. Spelling this out, the above
analysis says that the floating-point dot product of the vectors $x$  and $y$
is equal to an exact dot product of the slightly perturbed inputs $\tilde{x}$
and $\tilde{y}$. This means that, by \Cref{def:bstab}, the dot product can be
implemented in a backward stable way, with the backward error of its two 
input vectors each bounded by $3\varepsilon/2$.

A subtle point is that the backward error for multiplication can be described
in a slightly different way, while still maintaining the same backward error
bound given in \Cref{eq:be_mul}. In particular, floating-point multiplication
behaves like multiplication in exact arithmetic with a \emph{single} input
subject to small perturbations: for any $b_1,b_2\in \R$, we have
\begin{align}\label{eq:be_mul2}
b_1 \odot_\F b_2 = {b}_1 \cdot {b}_2e^\delta = {b}_1 \cdot \tilde{b}_2
\end{align} 
with $|\delta| \le \varepsilon$. There are many other ways to assign backward
error to multiplication as long as the exponents sum to $\delta$; in general, a
given program may satisfy a variety of different backward error bounds depending
on how the backward error is allocated between the program inputs.

The cost of using the backward error analysis for multiplication described in
\Cref{eq:be_mul2} instead of \Cref{eq:be_mul} is that all of the rounding error
in the result of a floating-point multiplication is assigned to a single input,
rather than distributing half of the error to each input.  We will see in
\Cref{sec:linearity}, the payoff is that, in some cases, it enables a backward
error analysis of computations that share variables across subexpressions.

\subsection{Backward Error Analysis in \bea{}: Motivating Examples}
In order to reason about backward error as it has been described so far, the
type system of \bea{} combines three ingredients: coeffects, distances, and
linearity. To get a sense of the critical role each of these components plays
in the type system, we first consider the following \bea{} program for
computing the dot product of 2D-vectors \stexttt{x} and \stexttt{y}:

\begin{lstlisting}
// Bean program for the dot product of vectors x and y
DotProd2 x y :=
let (x0, x1) = x in  
let (y0, y1) = y in
let v = mul x0 y0 in  
let w = mul x1 y1 in  
add v w
\end{lstlisting}

\subsubsection{Coeffects}
The type system of \bea{} allows us to prove the following typing judgment: 
\begin{equation}\label{eq:dpbound} 
\emptyset \mid \stexttt{x} :_{3\varepsilon/2}  \R^2, \stexttt{y}:_{3\varepsilon/2}
\R^2 \vdash  \stexttt{DotProd2 x y}: \R
\end{equation}
The \emph{coeffect} annotations ${3\varepsilon/2}$ in the context bindings
$\stexttt{x}:_{3\varepsilon/2} \R^2$ and $\stexttt{y} :_{3\varepsilon/2} \R^2$ express
per-variable relative backward error bounds for \stexttt{DotProd2}.
Thus, the typing judgment for \stexttt{DotProd2} captures the desired
backward error bound in \Cref{eq:dpbe}.

{Coeffect} systems \citep{Ghica:2014:bounded, Petricek:2014:coeffects,
Brunel:2014:coeffects, Tate:2013:effects} have traditionally been used in the
design of programming languages that perform resource management, and provide a
formalism for precisely tracking the usage of variables in programs.  In
\emph{graded coeffect systems} \citep{Gaboardi:2016:combining}, bindings in a
typing context $\Gamma$ are of the form $x :_r \sigma$, where the annotation
$r$ is some quantity controlling how $x$ can be used by the program.  In
\bea{}, these annotations describe the amount of backward error that can be
assigned to the variable. In more detail, a typing judgment $\emptyset \mid x
:_r \sigma \vdash e : \tau$ ensures that the term $e$ has at most $r$ backward
error with respect to the variable $x$. 

In \bea{}, the coeffect system allows us to derive backward error bounds for
larger programs from the known language primitives; the typing rules are used
to track the backward error of increasingly large programs in a compositional
way.  For instance, the typing judgment given in \Cref{eq:dpbound} for the
program \stexttt{DotProd2} is derived using primitive typing rules for
addition and subtraction. These rules capture the backward error bounds
described in \Cref{eq:be_add} and \Cref{eq:be_mul}: 
\vspace{1em}
\begin{center}
\AXC{}
\RightLabel{(Add)}
\UIC{$\emptyset \mid  \Gamma, x:_\varepsilon \R, 
            y:_\varepsilon\R \vdash
   \add x y : \R$}
\bottomAlignProof
\DisplayProof
\hspace{2em}
\AXC{}
\RightLabel{({Mul})}
\UIC{$\emptyset \mid  \Gamma, x:_{\varepsilon/2} \R, 
            y:_{\varepsilon/2}\R \vdash \mul  x  y : \R$}
\bottomAlignProof
\DisplayProof
\vspace{1em}
\end{center}

The following rule similarly captures the backward error bound described in 
\Cref{eq:be_mul2}:  

\vspace{1em}
\begin{center}
\AXC{}
\RightLabel{({DMul})}
\UIC{$ x: \R \mid 
    \Gamma, y:_{\varepsilon}\R \vdash \dmul  x  y : \R$}
\bottomAlignProof
\DisplayProof
\end{center}

\subsubsection{Distances}
In order to derive concrete backward error bounds, we require a notion
of \emph{distance} between points in an input space. To this 
end, each type $\sigma$ in \bea{} is equipped with a {distance} function 
$d_\sigma : \sigma \times \sigma \to \R_{\ge 0} \cup \{+\infty\}$ describing 
how close pairs of values of type $\sigma$ are to one another. 
For instance, for our numeric type $\R$, choosing the \emph{relative precision}
metric (\Cref{def:rp}) proposed by \citet{Olver:1978:rp} for the distance 
function $d_\R$ allows us to prove backward error bounds for a relative notion of 
error. This idea is reminiscent of type systems capturing function sensitivity
\citep{Fuzz,Gaboardi:2013:dfuzz,NUMFUZZ}; however, the \bean{} type system does not 
capture function sensitivity since this concept does not play a central role in 
backward error analysis.

\subsubsection{Linearity}\label{sec:linearity}
The conditions under which composing backward stable programs yields another
backward stable program are poorly understood. Our development of a static
analysis framework for backward error analysis led to the following insight:
the composition of two backward stable programs remains backward stable \emph{as long
as they do not assign backward error to shared variables}. Thus, to ensure 
that our programs satisfy a backward stability guarantee, 
\bea{} features a \emph{linear} typing discipline to control the
duplication of variables. While most coeffect type systems allow using a
variable $x$ in two subexpressions as long as the grades $x :_r \sigma$ and $x
:_s \sigma$ are combined in the overall program, \bea{} requires a stricter
condition: linear variables cannot be duplicated at all.

To understand why a type system for backward error analysis should disallow
unrestricted duplication, consider the floating-point computation
corresponding to the evaluation of the polynomial $h(x) = ax^2 + bx$. The
variable $x$ is used in each of the subexpressions $f(x) = ax^2$ and $g(x) =
bx$. Using the backward error bound given in \Cref{eq:be_mul} for
multiplication, the backward stability of $f$ is guaranteed by the existence of
the perturbed coefficient $\tilde{a} = ae^{\delta_1}$ and the perturbed
variable  $\tilde{x}_f = xe^{\delta_2}$:
\begin{align}
\tilde{f}(x) = 
ae^{\delta_1} \cdot (xe^{\delta_2})^2  = \tilde{a} \cdot \tilde{x}_f^2
\end{align}
Similarly, the backward stability of $g$ is guaranteed by the existence of the
perturbed coefficient $\tilde{b} = be^{\delta_3/2}$ and the variable
$\tilde{x}_g = xe^{\delta_3/2}$ . However, there is no common variable
$\tilde{x}$ that ensures the stability of $f$ and $g$ simultaneously. That is,
there is no input $\tilde{x}$ such that ${f}(\tilde{x}) + g(\tilde{x}) =
\tilde{f}(x) + \tilde{g}(x)$. 

By requiring linearity, \bea{} ensures that we never need to reconcile multiple
backward error requirements for the same variable. However, this restriction can
be quite limiting, and rules out the backward error analysis of some programs
that are backward stable---for instance, the polynomial $h(x) = ax^2 + bx$ above
\emph{is} actually backward stable! To regain flexibility in \bea{}, we note
that there is a special situation when a variable \emph{can} be duplicated
safely: when it doesn't need to be perturbed in order to provide a backward
error guarantee. For our polynomial $h(x)$, we can obtain a backward error
guarantee using \Cref{eq:be_mul2} to assign zero backward error to the variable
$x$ and non-zero backward error to the coefficients $a$ and $b$.  Since $x$ does
not need to be perturbed in order to provide an overall backward error guarantee
for $h(x)$, it can be duplicated without violating backward stability.

To realize this idea in \bea, the type system distinguishes between linear,
restricted-use data and non-linear, reusable data.  Linear variables are those
we can assign backward error to during an analysis, while non-linear variables
are those we do not assign backward error to during an analysis.  Technically,
\bea{} uses a dual context judgment, reminiscent of work on linear/non-linear
logic \citep{DBLP:conf/csl/Benton94}, to track the two kinds of variables.  In
more detail, a typing judgment of the form  
$y: \alpha \mid x :_r \sigma \vdash e: \tau$
ensures that the term $e$ has at most $r$ backward error with respect to the
\emph{linear} variable $x$, and has \emph{no backward error} with respect to
the \emph{non-linear} variable $y$. (Note that the bindings in the nonlinear
context do not carry an index, because no amount of backward error can be
assigned to these variables.) The soundness theorem for \bea{}, which we
introduce in \Cref{sec:bean:sound}, formalizes this result.

\section{Type System}\label{sec:bean:language}

\begin{figure}[tbp]
  \begin{alignat*}{3}
         &\alpha &::=~ &\dnum 
         \mid \alpha \otimes  \alpha 
         \tag*{(discrete types)} \\
         &\sigma &::=~ &\unit \mid \alpha
         \mid \num
         \mid \sigma \otimes  \sigma
         \mid \sigma + \sigma 
         \tag*{(linear types)} \\
	 &\Gamma &::=~ &\emptyset
	 \mid \Gamma, x:_r \sigma
	 \tag*{(linear typing contexts)} \\
	 &\Phi &::=~ &\emptyset 
	 \mid \Phi, z:\alpha
         \tag*{(discrete typing contexts)} \\
         &e,f \ &::=~ &x \mid z
         \mid ()
         \mid ~!e
         \mid  (e, f) 
         \mid \inl e
         \mid \inr e \\
         & & & \hspace{-0.25em} 
         \mid \slet x e f 
         \mid \slet {(x,y)} e f 
         \mid \dlet z e f 
         \mid \dlet {(x,y)} e f \\
         & & & \hspace{-0.25em} 
         \mid \case e {x.f} {x.f} 
         \mid \add e f
         \mid \sub e f 
         \mid \mul e f 
         \mid \dmul e f 
         \mid \fdiv e f
      %   \mid \fle e f 
        \tag*{(expressions)}
  \end{alignat*}
  \caption{Grammar for $\bea{}$ types and terms.}
  \label{fig:grammar}
\end{figure}

\bea{} is a simple first-order programming language, extended with a few
constructs that are unique to a language for backward error analysis.  The
grammar of the language is presented in \Cref{fig:grammar}, and the typing
relation is presented in \Cref{fig:typing_rules}.

\subsection{{Types}} 
We use \emph{linear} and \emph{discrete} types to distinguish between linear,
restricted-use data that can have backward error, and non-linear,
unrestricted-use data that cannot: linear types $\sigma$ are used for linear
data, and {discrete} types $\alpha$ are used for non-linear data.  Both linear
and discrete types include a base numeric type, denoted  by $\num$ and $\dnum$,
respectively. Linear types also include a tensor product $\otimes$, a unit type
$\unit$, and a sum type $ + $.

\subsection{Typing Judgments}
Terms are typed with judgments of the form 
${\Phi, z: \alpha \mid \Gamma, x:_r \sigma \vdash e : \tau}$ 
where $\Gamma$ is a linear typing context and $\sigma$ is a linear type, $\Phi$
is a discrete typing context and $\alpha$ is a discrete type, and $e$ is an
expression.  For linear typing contexts, variable assignments have the form $x
:_r \sigma$, where the grade $r$ is a member of a preordered monoid
$\mathcal{M} = (\mathbb{R}^{\ge0},+,0)$. Typing contexts, both linear and
discrete, are defined inductively as shown in \Cref{fig:grammar}.

Although linear typing contexts cannot be joined together with discrete typing
contexts, linear typing contexts can be joined with other linear typing
contexts as long as their domains are disjoint.  We write $\Gamma, \Delta$ to
denote the disjoint union of the linear contexts $\Gamma$ and $\Delta$.

While most graded coeffect systems support the composition of linear typing
contexts $\Gamma$ and $\Delta$ via a \emph{sum} operation $\Gamma + \Delta$
\citep{DalLago:2022:relational, Gaboardi:2016:combining}, where the grades of
shared variables in the contexts are added together, this operation is not
supported in \bea{}. This is because the sum operation serves as a mechanism
for the restricted duplication of variables, but \bea{}'s strict linearity
requirement does not allow variables to be duplicated. However, \bea{}'s type
system does support a sum operation that adds a given grade $q \in \mathcal{M}$
to the grades in a linear typing context:  
\begin{align*} 
  q + \emptyset &= \emptyset \\
  q + (\Gamma, x:_r \sigma) &= q + \Gamma, x:_{q+r} \sigma.
\end{align*}

In \bea{}, a well-typed expression ${ \Phi \mid x:_r \sigma \vdash e : \tau}$
is a program that has at most $r$ backward error with respect to the
\emph{linear}  variable $x$, and has \emph{no backward error} with respect to
the \emph{discrete} variables in the context $\Phi$.  For more general programs
of the form 
\[
	\Phi \mid x_1:_{r_1} \sigma_1, \dots, x_i:_{r_i} \sigma_i \vdash e : \tau,
\] 
\bea{} guarantees that the program has at most $r_i$ backward error with
respect to each variable $x_i$, and has \emph{no backward error} with respect
to the {discrete} variables in the context $\Phi$. This idea is formally
expressed in our soundness theorem (\Cref{thm:main}).

\subsection{{Expressions}} 
\bea{} expressions include linear variables $x$ and discrete variables $z$, as
well as a unit () value. Linear variables are bound in let-bindings of the form
{$\slet x {e} {f}$}, while discrete variables are bound in let-bindings of the
form {$\dlet z {e} {f}$}. The !-constructor is a syntactic convenience for
declaring that an expression can be duplicated. The pair constructor $(e,f)$
corresponds to a tensor product, and can be composed of expressions of both
linear and discrete type.  Discrete pairs are eliminated by pattern matching
using the construct {$\dlet {(x,y)} {e} {f}$}, whereas linear pairs are
eliminated by pattern matching using the construct $\slet {(x,y)} {e} {f}$.
The injections $\inl e$ and $\inr e$ correspond to a coproduct, and are
eliminated by case analysis using the construct $\case {e} {x.f} {y.f}$. Some
of the primitive arithmetic operations of the language ($\mathbf{add}$,
$\mathbf{mul}$, $\mathbf{dmul}$) were already introduced in
\Cref{sec:bean:introduction}. \bea{} also supports division ($\mathbf{div}$) and
subtraction ($\mathbf{sub}$).

\subsection{{Typing relation}} 
\begin{figure}
%% ROW1
\begin{center}
%% var1
\AXC{ }
\RightLabel{(Var)}
\UIC{$\Phi\mid \Gamma, x:_r \sigma \vdash x : \sigma$}
\bottomAlignProof
\DisplayProof
\hskip 0.5em
%% var2
\AXC{}
\RightLabel{(DVar)}
\UIC{$\Phi, z:\alpha\mid \Gamma \vdash z : \alpha$}
\bottomAlignProof
\DisplayProof
\vskip 1em
%%
%% ROW2
%% cprod intro
\AXC{$\Phi\mid\Gamma \vdash e : \sigma$}
\AXC{$\Phi\mid\Delta \vdash f : \tau$}
\RightLabel{($\otimes $ I)}
\BinaryInfC{$\Phi\mid\Gamma, \Delta \vdash ( e, f ): \sigma \otimes   \tau $}
\bottomAlignProof
\DisplayProof
\hskip 0.5em
%% unit
\AXC{}
\RightLabel{(Unit)}
\UIC{$\Phi\mid \Gamma \vdash () : \mathbf{unit}$}
\bottomAlignProof
\DisplayProof
\vskip 1em
%%
%% ROW3
%%
%% prod elim
\AXC{$\Phi\mid\Gamma \vdash e : \tau_1 \otimes  \tau_2$}
\AXC{$\Phi\mid\Delta, x:_r \tau_1, 
            y:_r \tau_2 \vdash f : \sigma$}
\RightLabel{($\otimes $ E$_\sigma$)}
\BIC{$\Phi\mid r + \Gamma, \Delta \vdash \slet {(x,y)} e f : \sigma$}
\bottomAlignProof
\DisplayProof
\vskip 1em

%%
%% NEW ROW
%%
%% prod elim discrete
\AXC{$\Phi\mid\Gamma \vdash e : \alpha_1 \otimes  \alpha_2$}
\AXC{$\Phi, z_1: \alpha_1, 
            z_2: \alpha_2\mid\Delta \vdash f : \sigma$}
\RightLabel{($\otimes $ E$_\alpha$)}
\BIC{$\Phi\mid \Gamma, \Delta \vdash 
        \dlet {(z_1,z_2)} e f : \sigma$}
\bottomAlignProof
\DisplayProof
\vskip 1em

%%
%% ROW4
%%
% sum elim
\AXC{$\Phi\mid\Gamma \vdash e' : \sigma+\tau$}
\AXC{$\Phi\mid\Delta, x:_q \sigma \vdash e : \rho$}
\AXC{$\Phi\mid\Delta, y:_q \tau \vdash f: \rho$}
\RightLabel{($+$ E)}
\TIC{$\Phi\mid q+\Gamma,\Delta \vdash \mathbf{case} \ e' \ \mathbf{of} \ (x.e \ | \ y.f) : \rho$}
\bottomAlignProof
\DisplayProof
\vskip 1em
%%
%% ROW 5
%% ind sum intro
\AXC{$\Phi\mid\Gamma \vdash e : \sigma$ }
\RightLabel{($+$ $\text{I}_L$)}
\UIC{$\Phi\mid\Gamma \vdash \inl \ e : \sigma + \tau$}
\bottomAlignProof
\DisplayProof
\hskip 0.5em
%% ind sum intro
\AXC{$\Phi\mid\Gamma \vdash e : \tau$ }
\RightLabel{($+$ $\text{I}_R$)}
\UIC{$\Phi\mid\Gamma \vdash \inr \ e : \sigma + \tau$}
\bottomAlignProof
\DisplayProof
\vskip 1em
%%
%%% ROW 6
% let 
\AXC{$\Phi\mid\Gamma \vdash e :  \tau$}
\AXC{$\Phi\mid\Delta, x :_r \tau \vdash f : \sigma$}
\RightLabel{(Let)}
\BIC{$\Phi\mid r + \Gamma,\Delta \vdash \slet x e f : \sigma$}
\bottomAlignProof
\DisplayProof
\vskip 1.3em
%%
%%% ROW 7
% disc
\AXC{$\Phi\mid\Gamma \vdash e :  \num$}
\RightLabel{(Disc)}
\UIC{$\Phi \mid \Gamma  \vdash ~ !e : \dnum$}
\bottomAlignProof
\DisplayProof
\hskip 0.5em
% dlet 
\AXC{$\Phi\mid\Gamma \vdash e :  \alpha$}
\AXC{$\Phi, z: \alpha \mid\Delta \vdash f : \sigma$}
\RightLabel{(DLet)}
\BIC{$\Phi \mid \Gamma,\Delta \vdash \dlet z e f : \sigma$}
\bottomAlignProof
\DisplayProof
\vskip 1.3em
%%
%%% ROW 8
% add, sub
\AXC{}
\RightLabel{(Add, Sub)}
\UIC{$\Phi\mid \Gamma, x:_{\varepsilon + r_1} \num, 
            y:_{\varepsilon + r_2}\num \vdash
   \{\mathbf{add}, 
    \mathbf{sub}\} \ x \ y : \num$}
\bottomAlignProof
\DisplayProof
\vskip 1.3em
%%
%%% ROW 9
% mul 
\AXC{ }
\RightLabel{(Mul)}
\UIC{$\Phi\mid \Gamma, 
    x:_{\varepsilon/2 + r_1} \num, 
    y:_{\varepsilon/2 + r_2}\num \vdash \mul  x  y : \num$}
\bottomAlignProof
\DisplayProof
\vskip 1.3em
%%% ROW 10
%  div 
\AXC{ }
\RightLabel{(Div)}
\UIC{$\Phi\mid
    \Gamma, x:_{\varepsilon/2 + r_1} \num, 
    y:_{\varepsilon/2 + r_2}\num
     \vdash \fdiv x  y : \num + \err$}
\bottomAlignProof
\DisplayProof
\vskip 1.3em
%%% ROW 11
% mul 
\AXC{ }
\RightLabel{(DMul)}
\UIC{$\Phi,z: \dnum\mid 
        \Gamma, 
    x:_{\varepsilon + r}\num \vdash \dmul  z  x : \num$}
\bottomAlignProof
\DisplayProof
\vskip 1em
\end{center}
    \caption{Typing rules  with 
    $q,r,r_1,r_2 \in \mathbb{R}^{\ge0}$ and fixed parameter $\varepsilon \in \mathbb{R}^{>0}$. }
    \label{fig:bean_typing_rules}
\end{figure}

The full type system for \bea{} is given in \Cref{fig:bean_typing_rules}. It is
parametric with respect to the constant $\varepsilon = u/(1-u)$, where $u$
represents the unit roundoff. 

Let us now describe the rules in \Cref{fig:bean_typing_rules}, starting with those
that employ the sum operation between grades and linear typing contexts: the
linear let-binding rule (Let) and the elimination rules for sums ($+$ E) and
linear pairs ($\otimes$ E$_\sigma$). Using the intuition that a grade describes
the backward error bound of a variable with respect to an expression, we see
that whenever we have an expression $e$ that is well-typed in a context
$\Gamma$ and we want to use $e$ in place of a variable that has a backward
error bound of $r$ with respect to another expression, then we must assign $r$
backward error onto the variables in $\Gamma$ using the sum operation $r +
\Gamma$. That is, if an expression $e$ has a backward error bound of
$q$ with respect to a variable $x$ and the expression $f$ has backward error
bound of $r$ with respect to a variable $y$, then $f[e/y]$ will have backward
error bound of $r + q$ with respect to the variable $x$.

The action of the !-constructor is illustrated in the Disc rule, which promotes
an expression of linear numeric type to  discrete numeric type.  The
!-constructor allows an expression to be used without restriction, but
there is a drawback: once an expression is promoted to discrete type it can no
longer be assigned backward error. The discrete let-binding rule (DLet) allows
us to bind variables of discrete type.

Aside from the rules discussed above, the only remaining rules in
\Cref{fig:bean_typing_rules} that are not mostly standard are the rules for
primitive arithmetic operations: addition (Add), subtraction (Sub),
multiplication between two linear variables (Mul), division (Div), and
multiplication between a discrete and non-linear variable (DMul). While these
rules are designed to mimic the relative backward error bounds for
floating-point operations following analyses described in the numerical
analysis literature
\citep{Olver:1978:rp,Higham:2002:Accuracy,Corless:2013:Analysis} and as briefly
introduced in \Cref{sec:bean:introduction}, they also allow weakening, or relaxing, the
backward error guarantee.  Intuitively, if the backward error of an expression
with respect to a variable is \emph{bounded} by $\varepsilon$, then it is also
\emph{bounded} by $\varepsilon + r$ for some grade $r \in \mathcal{M}$. We also
note that division is a partial operation, where the error result indicates a
division by zero.

The following section is devoted to explaining how a novel categorical
semantics, where \bea{} typing judgments are interpreted as morphisms in the
category \Bel{} of \emph{backward error lenses}, supports the language features
we have described so far.

\section{Denotational Semantics}\label{sec:bean:semantics}
Now that we have seen the syntax of \bea, we turn to its semantics. We first
introduce a novel category \Bel{} of \emph{backward error lenses}, where
morphisms are functions that satisfy a backward error guarantee. We show that
this category supports a variety of useful constructions, which we use to
interpret \bea{} programs.  We will assume knowledge of the basic category theory
concepts (e.g., categories and functors)
that are briefly described in \Cref{background:cat}, introducing less well-known
constructions as we go.  

\subsection{\Bel: The Category of Backward Error Lenses}\label{sec:meta_errcat}
The key semantic structure for \bea{} is the category \Bel{} of \emph{backward
error lenses}. 
Each morphism in \Bel{} corresponds to a pair of functions describing the
continuous problem and its approximating function, along with a \emph{backward
map} that serves as a constructive mechanism for witnessing the existence of a
backward error result. We view the category \Bel{} as conceptually similar to
categories of
\emph{lenses} \citep{Hofmann:2011:symmetriclenses,Johnson:2010:lenscat,
riley:2018:categories,Johnson:2012:lenses}.  Lenses, first introduced by
\citet{Foster:2007:trees}, consist of pairs of transformations between a set of
source structures and a set of target structures: a \emph{forward}
transformation produces a target from a source and a \emph{backward}
transformation ``puts back" a modified target onto a source according to some
laws \citep{Fischer:2015:lenses,ko:2017:axiomatic}. More concretely, if $X$ is a
set of source structures and $Y$ is a set of target structures, then a lens is
comprised of a pair of functions known as \emph{get} of type $X \rightarrow Y$
(the forward transformation) and \emph{put} of type $X \times Y \rightarrow Y$
(the backward transformation). The category \textbf{Lens} of lenses then has
sets as objects and lenses as morphisms, and supplies a well defined notion of
the composition of two lenses \citep{riley:2018:categories}. 
 
In contrast to the traditional definition of lenses, \emph{backward error
lenses} consist of a \emph{triple} of transformations:

\begin{definition}[Backward Error Lenses]\label{def:error_lens} 
A \emph{backward error lens} ${{L} : X \rightarrow Y}$ is a triple
${(f,\tilde{f},b)}$ of set-maps between the \emph{generalized 
distance spaces} $(X,d_X)$ and
$(Y,d_Y)$, described by the following data:
\begin{itemize}
\item the forward map $f: X \rightarrow Y$
\item the approximation map $\tilde{f}: X \rightarrow Y$, and
\item the backward map $b: X \times Y \rightarrow X$ defined as
\[
	\forall x \in X. ~b(x,-) \in 
  \{ y\in Y \mid d_Y(\tilde{f}(x),y) \neq \infty \} \rightarrow X
\] 
\end{itemize}
satisfying the properties
\begin{enumerate}[align=left]
\item[\textbf{Property 1.}] 
  $\forall x \in X, y \in Y.~ d_X(x,b(x,y)) \le d_Y(\tilde{f}(x),y)$  
  \label{eq:prop1}
\item[\textbf{Property 2.}] 
  $\forall x \in X, y \in Y.~f(b(x,y))=y$ \label{eq:prop2}
\end{enumerate}
under the assumption that ${d_Y(\tilde{f}(x),y) \neq \infty}$.
\end{definition}

The backward map $b : X \times Y\rightarrow X$ for backward error lenses given
in \Cref{def:error_lens} maps a point $x \in X$ in the input space and a point
$y \in Y$ in the output space \emph{that is at finite distance from $x$ under
the approximation map $\tilde{f}$} (i.e., such that $d_Y(\tilde{f}(x),y) \neq
+\infty$) to a point $\tilde{x} \in X$ in the input space. By restricting the
backward map to points that are at finite distance in the output space under
the approximation map, we can guarantee that the backward map produces a point
in the input space that is at finite distance from the original input. 

Properties 1 and 2 of \Cref{def:error_lens} are closely related to the
\emph{lens laws} described in the literature: for the forward transformation
\emph{get} of type $X \rightarrow Y$ and the backward transformation \emph{put}
of type $X \times Y \rightarrow Y$, every lens must obey the following laws:
\begin{align} 
	\forall x \in X.~put ~ x ~(get ~x) &= x \label{eq:lens1}\\
	\forall x \in X, \forall y \in Y.~get ~ (put ~ x ~ y) &= y \label{eq:lens2}
\end{align}
Clearly, property 2 of \Cref{def:error_lens} and \Cref{eq:lens2} are closely
related.  For backward error lenses, property 2 requires that the backward map
precisely captures the backward error. To see why, consider instantiating
property 2 with a point $x \in X$ and $\tilde{f}(x) \in Y$: the backward map
$b(x,\tilde{f}(x))$ produces a point $\tilde{x} \in X$ and property 2 requires
that the backward error result $f(\tilde{x}) = \tilde{f}(x)$ holds.

Looking closely at property 1, we can see that it corresponds to a generalized
\Cref{eq:lens1}, reframed as an inequality. Where \Cref{eq:lens1} requires
\emph{put} to exactly restore the original point in the source space under
strict conditions on its arguments, property 1 requires that the distance
between the point produced by the backward map and the original point in the
source space is bounded. The bound in property 2, namely the distance
$d_Y(\tilde{f}(x),y)$, serves as an upper bound for the generalized notion of
backward error. To see how, consider that property 2, instantiated on a point
$x \in X$ and the point $\tilde{f}(x) \in Y$, requires the following inequality
to hold: 
\begin{align}\label{eq:prop2_bound}
	d_X(x,b(x,\tilde{f}(x))) \le d_Y(\tilde{f}(x),\tilde{f}(x)).
\end{align}
Observe that if $d_X$ and $d_Y$ were distance functions with zero self
distance, as is usual for standard metric spaces, then \Cref{eq:prop2_bound}
forces the backward error, given as the distance $d_X(x,b(x,\tilde{f}(x)))$, to
zero. 

However, we would also like our semantics to support maps with bounded, but
\emph{non-zero} backward error. It turns out that we can support these more
maps by allowing the distance functions $d_X$ to take on a wide range of
values, ranging over $\R \cup \{ \pm \infty \}$; we call such functions
\emph{generalized distances}. While it is not obvious what a negative distance
represents, intuitively, we merely use this distances as technical devices to
enable compositional reasoning about backward error. For applications, all
backward error guarantees will involve maps to \emph{standard} metric spaces,
i.e., with non-negative distance function satisfying the usual metric axioms.

\begin{definition}[The Category \Bel{} of Backward Error Lenses]
\label{def:lensC} 
The category \Bel{} of \emph{backward error lenses} is the category with the
following data:
\begin{itemize}
  \item Its objects are generalized distance spaces: $(M, d : M \times M \to \R
    \cup \{ \pm \infty \})$, where the distance function has non-positive
    self-distance: $d(x, x) \leq 0$.
\item Its morphisms from $X$ to $Y$ are backward error lenses from $X$ to $Y$:
  triples of maps $(f, \tilde{f}, b)$, satisfying the two properties in
  \Cref{def:error_lens}.
\item The identity morphism on objects $X$ is given by the triple 
$(id_X,id_X,\pi_2)$.
\item The composition 
\[
(f_2,\tilde{f}_2,b_2) \circ (f_1,\tilde{f}_1,b_1)
\]
of error lenses ${(f_1,\tilde{f}_1,b_1) : X \rightarrow Y}$
and ${(f_2,\tilde{f}_2,b_2) : Y \rightarrow Z}$ is the error lens
${(f,\tilde{f},b): X \rightarrow Z}$ defined by
\begin{itemize}
\item the forward map 
  \begin{equation} \label{eq:fcomp}
    f: x \mapsto (f_1;f_2) \ x 
  \end{equation}
\item the approximation map
  \begin{equation}\label{eq:acomp}
    \tilde{f} : x \mapsto (\tilde{f}_1;\tilde{f}_2) \ x 
  \end{equation}
\item the backward map
  \begin{equation} \label{eq:bcomp}
    b: (x,z) \mapsto b_1\bigl(x, b_2(\tilde{f}_1(x), z)\bigr) 
  \end{equation}
\end{itemize}
\end{itemize}
\end{definition}
The composition in \Cref{def:lensC} is only well-defined if the domain of the
backward map is well-defined, and if the error lens properties hold for the
composition; we check these requirements in \Cref{app:check_bel}.

\subsection{Basic Constructions in \Bel}\label{sec:meta:basic}
We can now begin defining the lenses in \Bel{} that are necessary for
interpreting the language features in \bea{}. Following the description of
\bea{} given in \Cref{sec:bean:language}, we give the constructions below for lenses
corresponding to a tensor product, coproducts, and \emph{a graded comonad}
(see \Cref{background:comonad}) for interpreting linear typing contexts.

\subsubsection{Initial and Final Objects}
We start by introducing the initial and final objects of our category.  Let
$0 \in \Bel$ be the empty metric space $(\emptyset, d_\emptyset)$, and $1 \in
\Bel$ be the singleton metric space $(\{\star\}, \underbar{0})$ with a single
element and a constant distance function $d_1({\star,\star})=0$. Then for any
object $X \in \Bel$, there is a unique morphism $0_X : 0 \to X$ where the
forward, approximate, and backward maps are all the empty map, so $0$ is an
\emph{initial object} for \Bel.

Similarly, for every object $X \in \Bel$ is a morphism $!_X : X \to 1$ given by
$f_{!} = \tilde{f}_!:= x \mapsto \star$ and $b_{!} := (x,\star) \mapsto x$. To
check that this is indeed a morphism in \Bel---we must check the two backward
error lens conditions in \Cref{def:error_lens}. The first condition boils down
to checking $d_X(x, x) \leq d_1(\star, \star) = 0$, but this holds since all
objects in \Bel{} have non-positive self distance. The second condition is
clear, since there is only one element in $1$.  Finally, this morphism is
clearly unique, so $1$ is a \emph{terminal object} for \Bel.

\subsubsection{Tensor Product}
Next, we turn to products in \Bel. Like most lens categories, \Bel{} does not
support a Cartesian product \citep{Hofmann:2011:symmetriclenses}. 
In particular, it is not possible to define a
diagonal morphism $\Delta_A : A \to A \times A$, where the space $A \times A$
consists of pairs of elements of $A$. The problem is the second lens condition
in \Cref{def:error_lens}: given an approximate map $\tilde{f} : A \to A \times
A$ and a backward map $b : A \times (A \times A) \to A$, we need to satisfy
\[
  \tilde{f}(b(a_0, (a_1, a_2))) = (a_1, a_2)
\]
for all $(a_0, a_1, a_2) \in A$. But it is not possible to satisfy this
condition when $a_1 \neq a_2$: the backward map can only return 
one of $a_0, a_1$, or
$a_2$. As a consequence, there is not enough information for the approximate map
$\tilde{f}$ to recover $(a_1, a_2)$. More conceptually, this is the technical
realization of the problem described in \Cref{sec:bean:introduction}: if we think of
$a_1$ and $a_2$ as backward error witnesses for two subcomputations that both
use $A$, we may not be able to reconcile these two witnesses into a single
backward error witness.

Although a Cartesian product does not exist, \Bel{} does support a weaker,
monoidal product, which makes it a \emph{symmetric monoidal category}.  
Specifically, given two objects $X$ and $Y$ in \Bel{} we have the object 
$(X \times Y,d_{X\otimes Y})$
where the metric $d_{X\otimes Y}$ takes the componentwise max. 
Additionally, given any two morphisms ${(f,\tilde{f},b) : A \rightarrow X}$ 
and
${(g,\tilde{g},b'): B \rightarrow Y}$, we have the morphism
\[
  {(f,\tilde{f},b) \otimes  (g,\tilde{g},b') : 
  A \otimes  B \rightarrow X \otimes  Y}
\]
defined by:
\begin{itemize}
\item the forward map
  \begin{equation}\label{eq:tensor_lens1}
     (a_1, a_2) \mapsto (f(a_1),g(a_2)) 
   \end{equation}
\item the approximation map
  \begin{equation}\label{eq:tensor_lens2}  
    (a_1,a_2) \mapsto (\tilde{f}(a_1),\tilde{g}(a_2)) 
  \end{equation}
\item the backward map
  \begin{equation} \label{eq:tensor_lens3}
    ((a_1, a_2),(x_1,x_2)) \mapsto (b(a_1,x_1),b'(a_2,x_2)) 
  \end{equation}
\end{itemize}

We check that the tensor product given in
\Cref{eq:tensor_lens1,eq:tensor_lens2,eq:tensor_lens3} is well-defined in
\Cref{app:products}.

\begin{lemma}\label{lem:bifun}
	The tensor product operation on lenses induces a bifunctor on  \Bel.
\end{lemma}
\noindent The proof of \Cref{lem:bifun} requires checking conditions expressing
preservation of composition and identities, and is given in \Cref{app:products}.

The bifunctor $\otimes $ on the category \Bel{} gives rise to a symmetric
monoidal category of error lenses. The unit object $I$ is defined to be the
terminal object $1 = (\{\star\}, {0})$ with a single element and a
constant distance function $d_I({\star,\star})=0$ along with natural
isomorphisms for the associator ($\alpha_{X,Y,Z} : X \otimes  (Y \otimes Z))$,
and we can define the usual left-unitor ($\lambda_X : I \otimes  X \rightarrow
X)$, right-unitor ($\rho_X : X \otimes I \rightarrow X )$, and symmetry
($\gamma_{X,Y} : X \otimes  Y \rightarrow Y \otimes  X$) maps.  These
definitions are provided in \Cref{app:products}.

\subsubsection{Projections}
For any two spaces $X$ and $Y$ with the same self distance, i.e., with $d_X(x,
x) = d_Y(y, y)$ for all $x \in X$ and $y \in Y$, we can define a projection map
$\pi_1 : X \otimes Y \rightarrow X$ via:
\begin{itemize}
\item the forward map 
  \[
    f: (x, y) \mapsto x
  \]
\item the approximation map
  \[
    \tilde{f} : (x, y) \mapsto x
  \]
\item the backward map
  \[
    b: ((x, y),z) \mapsto (z, y)
  \]
\end{itemize}
The projection $\pi_2 : X \otimes Y \rightarrow Y$ is defined similarly.

\subsubsection{Coproducts}
For any two objects $X$ and $Y$ in \Bel{} we have the object $(X + Y,d_{X+Y})$,
where the metric $d_{X+Y}$ is defined as 
\\
\
\begin{equation} 
  d_{X+Y}(x,y) \triangleq   
   \begin{cases} 
      d_X(x_0,y_0) & \text{if } x = inl \ x_0  \text{ and } y = 
      inl \ y_0  \\
      d_Y(x_0,y_0) & \text{if } x = inr \ x_0  \text{ and } y = 
      inr \ y_0 \\
     \infty & \text{otherwise.}
    \end{cases}
\label{eq:prod_met}
\end{equation}
\
We define the morphism for the first projection 
$in_1: X \rightarrow X + Y$ as the triple
\\
\
\begin{align}
  f_{in_1}(x) = \tilde{f}_{in_1}(x) &\triangleq inl \ x \label{eq:inlf} \\
  b_{in_1}(x,z) &\triangleq 
      \begin{cases}  
        x_0 & \text{if } z = inl \ x_0 \\
        x & \text{otherwise.}
      \end{cases} \label{eq:inlb}
\end{align}
\
We check that the first projection is well-defined in \Cref{app:coproducts}.
The morphism $in_2 : Y \rightarrow X + Y $ for the second projection can be
defined similarly. 

Now, given any two morphisms 
\begin{align}
g : X \rightarrow C &\triangleq (f_g,\tilde{f}_g,b_g) \\
h : Y \rightarrow C &\triangleq (f_h,\tilde{f}_h,b_h)
\end{align}
we define the unique \emph{copairing} morphism ${[g,h] : X + Y \rightarrow C}$ 
such that $[g,h] \circ in_1 = g$ and $[g,h]  \circ in_2= h$ by the following triple:
\\
\
\begin{align}
  f_{[g,h]}(z) &\triangleq       
      \begin{cases}  
        f_g(x) & \text{if } z = inl \ x \\
        f_h(y) & \text{if } z = inr \ y
      \end{cases}  \\
  \tilde{f}_{[g,h]}(z) &\triangleq 
      \begin{cases}  
        \tilde{f}_g(x) & \text{if } z = inl \ x \\
        \tilde{f}_h(y) & \text{if } z = inr \ y
      \end{cases} \\
  b_{[g,h]}(z,c) &\triangleq 
      \begin{cases}  
        inl \ (b_g(x,c)) & \text{if } z = inl \ x \\
        inr \ (b_h(y,c)) & \text{if } z = inr \ y.
      \end{cases} 
\end{align}
\
We check that the morphism ${[g,h] : X + Y \rightarrow C}$ 
is well-defined in \Cref{app:coproducts}.

\subsubsection{A Graded Comonad}
Next, we turn to the key construction in \Bel{} that enables our semantics to
capture morphisms with non-zero backward error. The rough idea is to use a
graded comonad to shift the distance by a numeric constant; this change then
introduces slack into the lens conditions in \Cref{def:error_lens} to support
backward error.

More precisely, we construct a comonad graded by the real numbers.  Let the
pre-ordered monoid $\mathcal{R}$ be the non-negative real numbers ${R^{\ge 0}}$
with the usual order and addition.  We define a graded comonad on \Bel{} by the
family  of functors 
\begin{align*}
\{D_r :
  \Bel\rightarrow \Bel \ | \ r \in \mathcal{R}\}
\end{align*}
as follows.
\begin{itemize}
\item The object-map $D_r: \mathbf{Bel} \rightarrow \mathbf{Bel}$ takes
  $(X,d_X)$ to $(X,d_X - r)$, where $\pm \infty - r$ is defined to be equal to
  $\pm \infty$.
\item The arrow-map $D_r$ takes an error lens 
$(f,\tilde{f},b) : A \rightarrow X$ to an error lens 
\begin{align}
	(D_rf,D_r\tilde{f},D_rb) : D_rA \rightarrow D_rX
\end{align}
 where
\begin{align}\label{eq:monad1}
	(D_rg)x \triangleq g(x). 
\end{align}
\item The \emph{counit} map $\varepsilon_X: D_0 X \rightarrow X$ is 
the identity lens.
\item The \emph{comultiplication} map 
$\delta_{q,r,X}: D_{q+r}X \rightarrow D_q(D_rX)$ is the identity lens.
\item The \emph{2-monoidality} map 
$m_{r,X,Y} : D_rX \otimes D_rY \xrightarrow{\sim} D_r(X \otimes Y)$ is the identity
lens. 
\item The map
$m_{q\le r,X} : D_rX \rightarrow D_qX$ is the identity lens.
\end{itemize}
Unlike similar graded comonads considered in the literature on coeffect
systems, our graded monad does not support a graded contraction map: there is
no lens morphism $c_{r, s, A} : D_{r + s} A \to D_r A \otimes D_s A$. This is
for the same reason that our category does not support diagonal maps: it is not
possible to satisfy the second lens condition in \Cref{def:error_lens}. Thus,
we have a graded comonad, rather than a graded exponential comonad
\citep{Brunel:2014:coeffects}.

\subsubsection{Discrete Objects}
While there is no morphism $A \to A \otimes A$ in general, graded or not, there
is a special class of objects where we do have a diagonal map: the
\emph{discrete} spaces.

\begin{definition}[Discrete space]
  We say a generalized distance space $(X, d_X)$ is \emph{discrete} if its
  distance function satisfies $d_X(x_1, x_2) = \infty$ for all $x_1 \neq x_2$.
\end{definition}
We write \Del{} for category of the discrete spaces and backward error lenses;
this forms a full subcategory of \Bel.  Discrete objects are closed under the
monoidal product in \Bel, and the unit object $I$ is discrete.

There are two other key facts about discrete objects. First, it is possible to
define a diagonal lens.
\begin{lemma}[Discrete diagonal] \label{lem:discr-diag}
  For any discrete object $X \in \Del$ there is a lens morphism $t_X : X
  \to X \otimes X$ defined via
  \begin{align*}
    f_t = \tilde{f}_t &\triangleq x \mapsto (x,x) \\
    b_t &\triangleq (x,(x_1,x_2)) \mapsto x.
  \end{align*}
\end{lemma}
The key reason this is a lens morphism is that according to the lens
requirements in \Cref{def:error_lens}, we only need to establish the lens
properties for $(x_1, x_2)$ at \emph{finite} (i.e., not equal to $+\infty$)
distance from $f_t(x) = (x, x)$ under the distance on $X \otimes X$. Since
this is a discrete space, we only need to consider pairs $(x_1, x_2)$ that are
\emph{equal} to $(x, x)$; thus, the lens conditions are obvious. More
conceptually, we can think of a discrete object as a space that can't have any
backward error pushed onto it. Thus, the backward error witnesses $x_1$ and
$x_2$ are always equal to the input, and can always be reconciled.

Second, the graded comonad $D_r$ restricts to a graded comonad on $\Del$. In
particular, if $X$ is a discrete object, then $D_r X$ is also a discrete object.

\subsection{Interpreting \bea{}}
With the basic structure of \Bel{} in place, we can now interpret the types and
typing judgments of \bea{} as objects in \Bel{} and morphisms in \Bel{},
respectively.
Given a type $\tau$, we define a metric space $\denot{\tau}$ with the rules
\begin{align*}
	\denot{\dnum} &\triangleq (\mathbb{R},d_\alpha) \\
	\denot{\num}  &\triangleq (\R,d_\R) \\
	\denot{\sigma \otimes  \tau} &\triangleq 
	\denot{\sigma} \otimes \denot{\tau} \\
	\denot{\sigma + \tau} &\triangleq \denot{\sigma} + \denot{\tau} \\
  \denot{\textbf{unit}} &\triangleq (\{\star\}, \underbar{$0$}) 
\end{align*}
where the distance function $d_\alpha$ is the discrete metric on $\R$
where self-distance is zero and the distance 
between two distinct point is $+\infty$, and 
$d_\R$ is the relative precision metric (\Cref{def:rp}). 
By definition, if $d_\R$
is a standard distance function, then $\denot{\tau}$ is a standard metric space.   

The interpretation of typing contexts is defined inductively as follows:
\begin{align*}
  \denot{\emptyset \mid \emptyset}  &\triangleq I \\
  \denot{\emptyset \mid  \Gamma, x:_r\sigma}  &\triangleq 
  \denot{\Gamma} \otimes D_r\denot{\sigma} \\
  \denot{\Phi, z : \alpha \mid \emptyset}  &\triangleq 
  \denot{\Phi} \otimes \denot{\alpha} \\
  \denot{\Phi, z : \alpha \mid  \Gamma, x:_r\sigma}  &\triangleq 
  \denot{\Phi} \otimes \denot{\alpha} \otimes \denot{\Gamma} \otimes
  D_r\denot{\sigma}
\end{align*}
where the graded comonad $D_r$ is used to interpret the linear variable
assignment $x:_r \sigma$, and $I$ is the monoidal unit $(\{ \star \},
\underbar{$-\infty$})$.

Given the above interpretations of types and typing environments, we can
interpret \bea{} programs in \Bel:

\begin{definition}(Interpretation of \bea{} Terms.)\label{def:interpL}
	We can interpret each well-typed term $\Phi \mid \Gamma \vdash e : \tau$ as an
	error lens $\denot{e} : \denot{\Phi} \otimes \denot{\Gamma} \rightarrow
	\denot{\tau}$ in \Bel, by structural induction on the typing derivation. 
\end{definition}
\noindent The details of each construction for \Cref{def:interpL} can be found
in \Cref{sec:app_interp_bea}. We provide the cases for the (Let), (Add), and (Mul) 
rules here as a demonstration.
\begin{description}
% LET
\item[Case (Let).] Given the maps
\begin{align*}
h_1 &= \denot{\Phi\mid \Gamma \vdash e : \tau} : 
  \denot{\Phi} \otimes \denot{\Gamma} \rightarrow {\denot{\tau}}\\
h_2 &= \denot{\Phi\mid \Delta,x :_r \tau \vdash f : \sigma} : 
  \denot{\Phi} \otimes \denot{\Delta} \otimes D_r\denot{\tau} \rightarrow 
  {\denot{\sigma}}
\end{align*}
we need to define a map 
$\denot{\Phi\mid r + \Gamma,\Delta \vdash \slet x e f  : \sigma } $.

We first define a map $h : D_r \denot{\Phi} \otimes D_r \denot{\Gamma} \otimes
\denot{\Delta} \to \denot{\sigma}$ as the following composition:
\[  h_2 \circ  
  (D_r(h_1) \otimes (\varepsilon_{\denot{\Phi}} \circ m_{0 \leq r, \denot{\Phi}})\otimes id_{\denot{\Delta}}) \circ 
  (m_{r,\denot{\Phi},\denot{\Gamma}} \otimes id_{D_r \denot{\Phi}} \otimes id_{\denot{\Delta}}) \circ
  (t_{D_r \denot{\Phi}} \otimes id_{D_r\denot{\Gamma} \otimes \denot{\Delta}}) .
\] 
Here, the map $\varepsilon$ is the counit of the graded comonad.

Since $\denot{\sigma}$ is a metric space, its distance is bounded below by $0$.
Since $\denot{\Phi}$ is a discrete space, we observe that forward, approximate,
and backward maps in $h$ are also a lens morphism between objects:
\[
  h : \denot{\Phi} \otimes D_r \denot{\Gamma} \otimes \denot{\Delta} \to \denot{\sigma}
\]
This is the desired map to interpret let-binding.
% ADD
\item[Case (Add).] Suppose the contexts $\Phi$ and $\Gamma$ have total length
  $i$. We define the map
  \[
    \denot{\Phi\mid \Gamma, x :_{\varepsilon + q} \num, y :_{\varepsilon + r} \num \vdash
    \add x y : \num }
  \]
as the composition
\[ 
  \pi_i
  \circ
  (id_{\denot{\Phi}}
  \otimes (\overline{\varepsilon_{\denot{\sigma_j}}}
  \circ \overline{m_{0 \leq q_j, \denot{\sigma_j}}})
  \otimes id_{\denot{\num}})
  \circ
  (id_{\denot{\Phi} \otimes \denot{\Gamma}}
  \otimes {L}_{add})
  \circ
  (id_{\denot{\Phi} \otimes \denot{\Gamma}}
  \otimes m_{\varepsilon \le\varepsilon + q, \denot{\num}}
  \otimes m_{\varepsilon \le \varepsilon + r, \denot{\num}}) ,
\]
where the map $\overline{\varepsilon_{\denot{\sigma_j}}}$ applies the counit map
$\varepsilon_X : D_0 X \to X$ to each object in the context $\denot{\Gamma}$,
and the map $\overline{m_{0 \leq q_j, \denot{\sigma_j}}}$ applies the map $m_{0
\leq q, X} : D_q X \to D_0 X$ to each binding $\denot{x :_q \sigma}$ in the
context $\denot{\Gamma}$.

The lens ${L}_{add} : D_\varepsilon(\R) \otimes D_\varepsilon(\R) \rightarrow \R$ 
is given by the triple
\begin{align*}
  f_{add}(x_1,x_2) &\triangleq x_1 + x_2 \\
  \tilde{f}_{add}(x_1,x_2)&\triangleq  (x_1 + x_2)e^\delta; \quad |\delta|
  \le \varepsilon \\
  b_{add}((x_1,x_2),x_3) &\triangleq  \left(\frac{x_3x_1}{x_1+x_2},
  \frac{x_3x_2}{x_1+x_2}\right)
\end{align*}
where $\varepsilon = u/(1-u)$ and $u$ is the unit roundoff.

We now show ${L}_{add} : D_\varepsilon(\R) \otimes D_\varepsilon(\R) \rightarrow
\R$ is well-defined: it is clear that $L_{add}$ satisfies property 2 of a
backward error lens, and so we are left with checking property 1:
\noindent Assuming, for any $x_1,x_2,x_3 \in \R$,
\begin{equation}
  d_{R}\left(\tilde{f}_{add}(x_1,x_2),x_3\right) \neq \infty,
  \label{eq:add_asum}
\end{equation}
we are required to show
\begin{align*}
	d_{\R \otimes \R}\left((x_1,x_2),b_{add}((x_1,x_2),x_3)\right) - \varepsilon
	&\le d_{\R}\left(\tilde{f}_{add}(x_1,x_2),x_3\right) 
  = d_{\R}\left((x_1+x_2)e^\delta,x_3\right).
\end{align*}

Note that \Cref{eq:add_asum} implies $d_R((x_1+x_2)e^\delta, x_3) \neq \infty$:
by \Cref{eq:olver}, we have that $(x_1+x_2)$ and $x_3$ are either both zero, or
are both non-zero and of the same sign.
We can assume, without loss of generality, 
\begin{equation}
  d_\R\left(x_2,\frac{x_3x_2}{x_1+x_2}\right)  \le 
  d_\R\left(x_1,\frac{x_3x_1}{x_1+x_2}\right).
\end{equation}
Under this assumption, we have
\begin{equation}
  d_{\R \otimes  \R}\left((x_1,x_2),b_{add}((x_1,x_2),x_3)\right) =
  d_\R\left(x_1,\frac{x_3x_1}{x_1+x_2}\right)
\end{equation}
and we are then required to show
\begin{align}
  d_\R\left(x_1,\frac{x_3x_1}{x_1+x_2}\right) &\le
  d_{\R}\left((x_1+x_2)e^\delta,x_3\right) + \varepsilon. \label{eq:add_req}
\end{align}

Using the distance function given in \Cref{eq:olver}, the inequality in
\Cref{eq:add_req} becomes
\begin{align}
  \left|{\ln\left(\frac{x_1+x_2}{x_3}\right)}\right| \le
  \left|{\ln\left(\frac{x_1+x_2}{x_3}\right)} + \delta\right| + \varepsilon,
  \label{eq:p1_add_end}
\end{align}
which holds under the assumptions of $|\delta| \le \varepsilon$ and $0 <
\varepsilon$.
Set $\alpha = \left|{\ln\left((x_1+x_2)/{x_3}\right)}\right|$,
and assume, without loss of generality, that $\alpha < 0$. If
$\alpha + \delta < 0$ then $|\alpha| = -\alpha$ and
$|\alpha + \delta| = -(\alpha + \delta)$; the inequality in \Cref{eq:p1_add_end}
reduces to $\delta \le \varepsilon$, which follows by assumption.
Otherwise, if $0 \le \alpha + \delta$, then $-\alpha \le \delta \le \varepsilon$
and it suffices to show that $\varepsilon \le \alpha + \delta + \varepsilon$.
% MUL
\item[Case (Mul).]
  We proceed the same as the case for (Add), with slightly different indices.
  We define a lens 
  $\mathcal{L}_{mul} : D_{\varepsilon/2}(\R) \otimes 
            D_{\varepsilon/2}(\R) \rightarrow R$ 
  given by the triple
\begin{align*}
    f_{mul}(x_1,x_2) &\triangleq x_1  x_2 \\
    \tilde{f}_{mul}(x_1,x_2)&\triangleq  x_1 x_2e^{\delta}; \quad |\delta| 
  \le \varepsilon \\
    b_{mul}((x_1,x_2),x_3) &\triangleq  \left(x_1\sqrt{\frac{x_3}{x_1x_2}},
  x_2\sqrt{\frac{x_3}{x_1x_2}}\right).
\end{align*}

We check that $\mathcal{L}_{mul} : D_{\varepsilon/2}(R) \otimes
D_{\varepsilon/2}(R) \rightarrow  R$ is well-defined.

For any $x_1,x_2,x_3 \in R$ such that 
\begin{equation} 
d_{R}\left(\tilde{f}_{mul}(x_1,x_2),x_3\right) \neq \infty.
\label{eq:mul_asum}
\end{equation}
holds, we need to check the that $\mathcal{L}_{mul}$ satisfies the properties 
of an error lens. We again take the distance function $d_R$ as the metric 
given in \Cref{eq:olver}, so \Cref{eq:mul_asum} implies that $(x_1x_2)$ 
and $x_3$ are either both zero or are both non-zero and of the same sign; 
this guarantees that the backward map (containing square roots) is indeed 
well defined. 
\begin{enumerate}[align=left]
\item[Property 1.] We are required to show
\begin{align*} \label{eq:p1_mul}
d_{R \otimes  R}\left((x_1,x_2),b_{mul}((x_1,x_2),x_3)\right) - \varepsilon/2&\le 
     d_{ R}\left(\tilde{f}_{mul}(x_1,x_2),x_3\right) \\
    &\le d_{R}\left(x_1 x_2e^{\delta},x_3\right) 
\end{align*}
Unfolding the definition of the distance function (\Cref{eq:olver}), we have
\begin{align*}
d_{R \otimes  R}\left((x_1,x_2),b_{mul}((x_1,x_2),x_3)\right)  
&= d_R\left(x_1,x_1\sqrt{\frac{x_3}{x_1x_2}}\right) \\
&= d_R\left(x_2,x_2\sqrt{\frac{x_3}{x_1x_2}}\right) \\
&= \frac{1}{2}\left|\ln \left( \frac{x_1x_2}{x_3}\right)\right|,
\end{align*}
and so we are required to show 
\begin{equation}  \label{eq:p1_mul_end}
\frac{1}{2}\left|\ln \left( \frac{x_1x_2}{x_3}\right)\right| \le
\left|\ln \left( \frac{x_1x_2}{x_3}\right) + \delta \right| + 
  \frac{1}{2}\varepsilon
\end{equation}
which holds under the assumptions of $|\delta| \le \varepsilon$ and $0 < \varepsilon$. 
Setting $\alpha = \ln(x_1x_2/x_3)$, assume, without loss of generality,
that $\alpha < 0$. If $\alpha + \delta < 0 $ then $\alpha < - \delta$ and 
it suffices to show that $-\frac{1}{2} \delta \le - \delta + \frac{1}{2} \varepsilon$, 
which follows by assumption. Otherwise, if $0 \le \alpha + \delta $ then 
$-\alpha \le \delta$ and it suffices to show that 
$\frac{1}{2} \delta \le \alpha + \delta + \frac{1}{2} \varepsilon$,
which follows by assumption.

\item[Property 2.] 
\begin{align*}
f_{mul}\left(b_{mul}((x_1,x_2),x_3)\right) = 
     f_{mul}\left(x_1\sqrt{\frac{x_3}{x_1x_2}},x_2\sqrt{\frac{x_3}{x_1x_2}}\right)
    = x_3.
\end{align*} 
\end{enumerate}
\end{description}
\section{Backward Error Soundness}\label{sec:bean:sound}
Recall the intuition behind the guarantee for \bea{}'s type system: a
well-typed term of the form 
\[ 
  \Phi \mid x_1:_{r_1} \sigma_1, \dots, x_i:_{r_i} \sigma_i \vdash e : \tau
\]
is a program that has at most $r_i$ backward error with respect to each
variable $x_i$, and has {no backward error} with respect to the {discrete}
variables in the context $\Phi$. By interpreting \bea{} programs as morphisms
in \Bel{}, we are able to precisely describe how every \bea{} program captures
the complex interaction between an ideal problem, its approximate program, and
a map that constructs the backward error between them.  Using these
constructions \Bel{}, we can clearly see a path towards a \emph{backward error
soundness theorem}: given the interpretation 
\[
  \denot{z_1: \alpha_1, \dots, z_j: \alpha_j \mid x_1:_{r_1} \sigma_1, 
  \dots, x_i:_{r_i} \sigma_i \vdash e : \tau} = 
  (f,\tilde{f},b) : \denot{\alpha_1} \otimes
  \dots \otimes  D_{r_i}\denot{\sigma_i} \rightarrow \denot{\tau}
\]
if $\tilde{f}[u_1/z_1]\dots[v_1/x_1]$ evaluates to a value $w$ for the
well-typed substitutions $({u})_{1 \le n \le j}$ and $({{v}})_{1 \le n \le i}$,
then we can guarantee our desired backward error result if we can witness the
existence of a well-typed substitution $(\tilde{v})_{1 \le n \le i}$ such that
${f}[u_1/z_1]\dots[\tilde{v}_1/x_1]$ also evaluates to the value $w$, and
$d_{\sigma_i}(v_i,\tilde{v}_i) \le r_i$; our backward map $b$ can be used to
construct the required witness. 

Formalizing the above result requires explicit access to each transformation in
the backward error lens individually. We achieve this by defining an
intermediate language, which we call \LangS, where programs denote morphisms in
$\Set$.  We then define an ideal and approximate operational semantics for
\LangS, and relate these semantics to the backward error lens semantics of
$\bea{}$ via the $\Set$ semantics of \LangS. As we will see, the semantic
constructions for $\bea{}$ can be transformed in a straightforward way to
semantic constructions for \LangS{} using the forgetful functors $U_{id} : \Bel
\rightarrow \Set$ and $U_{ap} : \Bel \rightarrow \Set$; these functors
associate each metric space with its underlying set, and associate each
backward error lens with its underlying ideal (resp., approximate) function on
sets. The actions of the forgetful functors $U_{id}$ and $U_{ap}$ on objects
are both denoted by $U$ for simplicity.

\subsection{\LangS{}: A Language for Projecting \bea{} into \Set \ }

\begin{figure}
%% ROW1
\begin{center}
%% var
\AXC{}
\RightLabel{(Var)}
\UIC{$\Gamma, x: \sigma, \Delta \vdash x : \sigma$}
\bottomAlignProof
\DisplayProof
\hskip 0.5em
%% unit
\AXC{}
\RightLabel{(Unit)}
\UIC{$\Gamma \vdash (): \mathbf{unit}$}
\bottomAlignProof
\DisplayProof
\hskip 0.5em
%% const
\AXC{$k \in R$}
\RightLabel{(Const)}
\UIC{$\Gamma \vdash k : \textbf{num}$}
\bottomAlignProof
\DisplayProof
\hskip 0.5em
\vskip 1em
%%
%% ROW2
%% prod intro
\AXC{$\Phi,\Gamma \vdash e : \sigma$}
\AXC{$\Phi,\Delta \vdash f : \tau$}
\RightLabel{($\otimes $ I)}
\BinaryInfC{$\Phi,\Gamma, \Delta \vdash ( e, f ): \sigma \otimes   \tau $}
\bottomAlignProof
\DisplayProof
\vskip 1em
%% prod elim
\AXC{$\Phi,\Gamma \vdash e : \tau_1 \otimes  \tau_2$}
\AXC{$\Phi,\Delta, x : \tau_1, y : \tau_2 \vdash f : \sigma $}
\RightLabel{($\otimes $ E)}
\BIC{$\Phi,\Gamma, \Delta  \vdash \slet {(x,y)} e f : \sigma$}
\bottomAlignProof
\DisplayProof
\vskip 1em
%%
%% ROW3
% sum elim
\AXC{$\Phi,\Gamma \vdash e' : \sigma+\tau$}
\AXC{$\Phi,\Delta, x: \sigma \vdash e : \rho$}
\AXC{$\Phi,\Delta, y: \tau \vdash f: \rho$}
\RightLabel{($+$ E)}
\TIC{$\Phi,\Gamma,\Delta \vdash \mathbf{case} \ e' 
    \ \mathbf{of} \ (\inl x.e \ | \ \inr y.f) : \rho$}
\bottomAlignProof
\DisplayProof
\vskip 1em
%%
%% ROW 5
%% ind sum intro
\AXC{$\Phi,\Gamma \vdash e : \sigma$ }
\RightLabel{($+$ $\text{I}_L$)}
\UIC{$\Phi,\Gamma \vdash \inl \ e : \sigma + \tau$}
\bottomAlignProof
\DisplayProof
\hskip 0.5em
%% ind sum intro
\AXC{$\Phi,\Gamma \vdash e : \tau$ }
\RightLabel{($+$ $\text{I}_R$)}
\UIC{$\Phi,\Gamma \vdash \inr \ e : \sigma + \tau$}
\bottomAlignProof
\DisplayProof
\vskip 1em
%%
%%% ROW 6
% let 
\AXC{$\Phi,\Gamma \vdash e :  \tau$}
\AXC{$\Phi,\Delta, x : \tau \vdash f : \sigma$}
\RightLabel{(Let)}
\BIC{$\Phi,\Gamma,\Delta \vdash \slet x e f : \sigma$}
\bottomAlignProof
\DisplayProof
\vskip 1em
%%
%%% ROW 7
% add, sub
\AXC{$\Phi,\Gamma \vdash e : \num$}
\AXC{$\Phi,\Delta \vdash f : \num $}
\AXC{$ \mathbf{Op} \in \{\mathbf{add}, \mathbf{sub}, 
    \mathbf{mul},  \mathbf{dmul} \}$}
\RightLabel{(Op)}
\TIC{$\Phi,\Gamma, \Delta \vdash \mathbf{Op} \ e \ f : \num$}
\bottomAlignProof
\DisplayProof
\vskip 1em
%%
%%% ROW 8
%  div 
\AXC{$\Phi,\Gamma \vdash e : \num$}
\AXC{$\Phi,\Delta \vdash f : \num $}
\RightLabel{(Div)}
\BIC{$\Phi,\Gamma, \Delta \vdash \fdiv e f : \num + \err$}
\bottomAlignProof
\DisplayProof

\end{center}
    \caption{Full typing rules for \LangS. }
    \label{fig:typing_rules_2_full}
\end{figure}

\subsubsection{A Type System for \LangS{}}
The type system of \LangS{} corresponds closely to \bea's.  Terms are typed
with judgments of the form $\Phi, \Gamma \vdash e : \tau$, where the typing
context $\Gamma$ corresponds to the linear typing contexts of \bea{} with all
of the grade information erased, and the typing context $\Phi$ corresponds to
the discrete typing contexts of \bea{}; we will denote the erasure of grade
information from a linear typing environment $\Delta$ as $\Delta^\circ$. Under
the erasure of grade information from a linear context $\Delta$, the disjoint
union of the contexts $\Phi, \Delta^\circ$ is well-defined.

In contrast to \bea{}, types in \LangS{} are not categorized as linear and
discrete:
\begin{align}
  \sigma \ ::=~ \num \mid \unit \mid \sigma \otimes \sigma \mid 
                \sigma + \sigma \tag*{(\LangS{} types)}
\end{align}

The grammar of terms in \LangS{} is mostly unchanged from the grammar of
\bea{}, except that \LangS{} extends $\bea{}$ to include primitive constants
drawn from a signature $R$:
\begin{align*}
  e,f \ ::=~ \dots \mid k \in R \tag*{(\LangS{} terms)}
\end{align*}
The typing relation of \LangS{} is entirely standard for a first-order simply
typed language; the
full set of rules is given in \Cref{fig:typing_rules_2_full}. 
The close correspondence between derivations in
$\bea{}$ and derivations in \LangS{} is summarized in the following lemma.  

\begin{lemma} \label{lem:derive_ls} Let 
  $\Phi \mid \Gamma \vdash e : \tau$ be a well-typed term in $\bea{}$.  Then
  there is a derivation of $\Phi,\Gamma^\circ \vdash e : \tau$ in \LangS.
\end{lemma}
\noindent The proof of \Cref{lem:derive_ls} is given in 
\Cref{sec:app_lem_sound}.

The proof of backward error soundness requires that 
\LangS{} satisfies the basic properties of weakening and substitution: 

\begin{lemma}[Weakening] 
  Let $\Gamma \vdash e : \tau$ be a well-typed \LangS{} term. Then for any
  typing environment $\Delta$ disjoint with $\Gamma$, there is a derivation of
  $\Gamma, \Delta \vdash e : \tau$.  \label{lem:weakening}
\end{lemma}

In the following theorem statement,
we write $e[v/x]$ for the capture avoiding substitution of the value $v$ for
all free occurrences of $x$ in $e$. Given a typing environment $x_1 : \tau_1,
\dots, x_i : \tau_i = \Gamma$, we denote the simultaneous substitution of a
vector of values $v_1, \dots, v_i = \bar{v}$ for the variables in $\Gamma$ as
$e[\bar{v}/dom(\Gamma)]$. Additionally, 
for a vector $\gamma_1,\dots, \gamma_i = \bar{\gamma}$ of well-typed closed values and 
a typing environment $x_1 : \tau_1, \dots, x_i : \tau_i = \Gamma$ (note the  
assumption that $\gamma$ and $\Gamma$ have the same length) we write 
$\bar{\gamma} \vDash \Gamma$ to denote the following 
\begin{align} \label{eq:subst_sat}
\bar{\gamma} \vDash \Gamma \triangleq
  \forall x_i \in dom(\Gamma). ~ \emptyset \vdash \gamma_i : \Gamma(x_i).
\end{align}

\begin{theorem}[Substitution] \label{thm:subst}
  Let $\Gamma \vdash e : \tau$ be a well-typed \LangS{} term. Then for any
  well-typed substitution $\bar{\gamma} \vDash \Gamma$ of closed values, there
  is a derivation $\emptyset \vdash e[\bar{\gamma}/dom(\Gamma)] : \tau.$
\end{theorem}
\noindent Most cases for substitution are routine; we provide the details of the
proof in \Cref{sec:app_lem_sound}.

\subsubsection{An Operational Semantics for \LangS{}} 
\begin{figure}
%% ROW1
\begin{center}
\AXC{}
\UIC{$() \Downarrow ()$}
\bottomAlignProof
\DisplayProof
\hskip 0.5em
\AXC{$e \Downarrow u$}
\AXC{$f \Downarrow v$}
\BIC{$(e,f) \Downarrow (u,v)$}
\bottomAlignProof
\DisplayProof
\hskip 0.5em
\AXC{$e \Downarrow (u,v)$}
\AXC{$f[u/x][v/y] \Downarrow w$}
\BIC{$\slet {(x,y)} e f \Downarrow w$}
\bottomAlignProof
\DisplayProof
\hskip 0.5em
\vskip 1em
%% 
%% ROW2
%%
%% 
\AXC{}
\UIC{$k \in R \Downarrow k \in R $}
\bottomAlignProof
\DisplayProof
\hskip 0.5em
\AXC{$e \Downarrow v$ }
\UIC{$\inl e \Downarrow \inl v$}
\bottomAlignProof
\DisplayProof
\hskip 0.5em
\AXC{$e \Downarrow v$ }
\UIC{$\inr e \Downarrow \inr v$}
\bottomAlignProof
\DisplayProof
\vskip 1em
%%
%% ROW3
% 
\AXC{$e \Downarrow u$ }
\AXC{$f[u/x] \Downarrow v$ }
\BIC{$\slet x e f \Downarrow v$}
\bottomAlignProof
\DisplayProof
\vskip 1em
%%
%% ROW3
% 
\AXC{$e \Downarrow \inl v$}
\AXC{$e_1[v/x] \Downarrow w$}
\BIC{$\mathbf{case} \ e \ \mathbf{of} \ (x.e_1 \ |\ y.e_2) \Downarrow w$}
\bottomAlignProof
\DisplayProof
\hskip 0.5em
\AXC{$e \Downarrow \inr v$}
\AXC{$e_2[v/y] \Downarrow w$}
\BIC{$\mathbf{case} \ e \ \mathbf{of} \ (x.e_1 \ | \ y.e_2) \Downarrow w$}
\bottomAlignProof
\DisplayProof
\vskip 1em
%%
%% ROW 4
%%
\AXC{$e_1 \stepid k_1$}
\AXC{$e_2 \stepid k_2$}
\AXC{$\mathbf{Op} \in \{\mathbf{Add}, \mathbf{Sub}, \mathbf{Mul}, \mathbf{Div}, \mathbf{LE}\}$}
\TIC{$\mathbf{Op} \ {e_1} \ {e_2} \stepid f_{op} \ ({k_1}, {k_2})$}
\bottomAlignProof
\DisplayProof
\vskip 1em
%%
%% ROW 5
%%
\AXC{$e_1 \stepap k_1$}
\AXC{$e_2 \stepap k_2$}
\AXC{$\mathbf{Op} \in \{\mathbf{Add}, \mathbf{Sub}, \mathbf{Mul}, \mathbf{Div}, \mathbf{LE}\}$}
\TIC{$\mathbf{Op} \ {e_1} \ {e_2} \stepap \tilde{f}_{op} \ ({k_1}, {k_2})$}
\bottomAlignProof
\DisplayProof
\vskip 1em

\end{center}
    \caption{Evaluation rules for \LangS. 
    A generic step relation ($\Downarrow$) is used when the rule is identical for 
    both the ideal ($\stepid$) and approximate ($\stepap$) step relations.}
    \label{fig:op_semantics_full}
\end{figure}

Intuitively, an ideal problem and its approximating program can behave
differently given the same input. Following this intuition, we allow programs
in \LangS{} to be executed under an ideal or approximate big-step operational
semantics. The full set of evaluation rules is given in 
\Cref{fig:op_semantics_full}. We write $e \stepid v$ (resp., $e \stepap v$) to
denote that a term $e$ evaluates to value $v$ under the ideal (resp.,
approximate) semantics.  Values, the subset of terms that are allowed as
results of evaluation, are defined as follows.
  \begin{alignat*}{1}
         &\text{Values } v ~::=~ 
         ()
         \mid k \in R
         \mid  (v, v)
         \mid \inl v
         \mid \inr v
  \end{alignat*}

An important feature of \LangS{} is that it is deterministic and strongly
normalizing:

\begin{theorem}[Strong Normalization] \label{thm:normalizing}
  If $\emptyset \vdash e : \tau$, then the well-typed closed values $\emptyset
  \vdash v, v' : \tau$ exist such that $e \stepid v$ and $e \stepap v'$. 
\end{theorem}

In our main result of backward error soundness, we will relate the ideal and
approximate operational semantics given above to the backward error lens
semantics of \bea{} via an interpretation of programs in \LangS{} as morphisms
in the category \Set.

\subsection{Interpreting \LangS}
Our main backward error soundness theorem requires that we have explicit access
to each transformation in a backward error lens. We achieve this by lifting the
close syntactic correspondence between \LangS{} and \bea{} to a close semantic
correspondence using the forgetful functors $U_{id} : \Bel \rightarrow \Set$
and $U_{ap} : \Bel \rightarrow \Set$ to interpret \LangS{} programs in \Set. 

We start with the interpretation of \LangS{} types, defined as follows 
\begin{align*}
    \pdenot{\num} &\triangleq U\denot{\num} 
        = U\denot{\dnum}  \\
    \pdenot{\textbf{unit}} &\triangleq U\denot{(\{\star\},\underline{0})} \\
    \pdenot{\sigma \otimes \tau}  &\triangleq 
      U\denot{\sigma} \times  U\denot{\tau}  \\
   \denot{\sigma + \tau} &\triangleq U\denot{\sigma} + U\denot{\tau} 
\end{align*}  

Given the above interpretation of types, the interpretation $\pdenot{\Gamma}$
of a \LangS{} typing context $\Gamma$ is then defined as 
\begin{align*}
  \pdenot{\emptyset} &\triangleq U\denot{(\{\star\},\underline{0})} \\
  \pdenot{\Gamma,x:\sigma} &\triangleq 
  \pdenot{\Gamma} \otimes \pdenot{\sigma}
\end{align*}  

Now, using the above definitions for the interpretations of \LangS{} types and
contexts, we can use the interpretation of \bea{} (\Cref{def:interpL}) terms
along with the functors $U_{id}$ and $U_{ap}$ to define the interpretation of
\LangS{} programs as morphisms in \Set:

\begin{definition}(Interpretation of \LangS{} terms.) \label{def:interpS}
  Each typing derivation $\Gamma \vdash e : \tau$ in \LangS{} yields the set
  maps $\pdenot{e}_{id} : \pdenot{\Gamma} \rightarrow \pdenot{\tau}$ and
  $\pdenot{e}_{ap} : \pdenot{\Gamma} \rightarrow \pdenot{\tau}$, by structural
  induction on the \LangS{} typing derivation $\Gamma \vdash e : \tau$.
\end{definition}
\noindent We give the detailed constructions for \Cref{def:interpS} in
\Cref{sec:app_interp_LS}. 

Given \Cref{def:interpS}, we can now show that \LangS{} is semantically sound
and compuationally adequate: a \LangS{} program computes to a value if and only
if their interpretations in \Set{} are equal. Because \LangS{} has an ideal and
approximate operational semantics as well as an ideal and approximate
denotational semantics, we have two version of the standard theorems for
soundness and adequacy:

\begin{theorem}[Soundness of $\pdenot{-}$]\label{thm:soundid}
  Let $\Gamma \vdash e : \tau$ be a well-typed \LangS{} term.  
  Then for any
  well-typed substitution of closed values $\bar{\gamma} \vDash \Gamma$, if
  $e[\bar{\gamma}/dom(\Gamma)]\stepid v$ for some value $v$, then
  ${\pdenot{\Gamma \vdash e : \tau}_{id}\pdenot{\bar{\gamma}}_{id} =
  \pdenot{v}_{id}}$ (and similarly for $\stepap$ and $\pdenotap{-}$).
\end{theorem} 
\noindent The proof of \Cref{thm:soundid} is given in \Cref{sec:app_lem_sound}.

\begin{theorem}[Adequacy of $\pdenot{-}$]\label{thm:adequacy}
  Let $\Gamma \vdash e : \tau$ be a well-typed \LangS{} term. Then for any
  well-typed substitution of closed values $\bar{\gamma} \vDash \Gamma$, if $
  \pdenot{\Gamma \vdash e : \tau}_{id}\pdenot{\bar{\gamma}}_{id} =
  \pdenot{v}_{id}$ for some value $v$, then $e[\bar{\gamma}/dom(\Gamma)]\stepid
  v$ (and similarly for $\stepap$ and $\pdenotap{-}$).
\end{theorem}

\noindent Details of the proof of \Cref{thm:adequacy} can be found in
\Cref{sec:app_lem_sound}.

Our main error backward error soundness theorem requires one final piece of
information: we must know that the functors $U_{id}$ and $U_{ap}$ project
directly from interpretations of \bea{} programs in \Bel{} (\Cref{def:interpL})
to interpretations of \LangS{} programs in \Set{} (\Cref{def:interpS}):  
\begin{lemma}[Pairing]\label{lem:pairing}
  Let $\Phi \mid \Gamma \vdash e : \sigma$ be a \bea{} program. Then we have 
\[
  U_{id}\denot{\Phi \mid \Gamma \vdash e : \sigma} = 
  \pdenotid{\Phi,\Gamma^\circ \vdash e : \sigma} 
  \quad \text{and} \quad 
  U_{ap}\denot{\Phi \mid \Gamma \vdash e : \sigma} = 
  \pdenotap{\Phi,\Gamma^\circ \vdash e : \sigma}.
\]
\end{lemma}
\noindent A proof of \Cref{lem:pairing} follows by induction on the structure
of the \bea{} derivation $\Phi\mid \Gamma \vdash e : \sigma$; details of the
proof can be found in \Cref{sec:app_lem_sound}.

\begin{theorem}[Backward Error Soundness]\label{thm:main}
  Let \[
  \Phi\mid  x_1:_{r_1}\sigma_1,\cdots,x_n:_{r_n}\sigma_n = \Gamma \vdash e
  : \sigma
  \] be a well-typed \bea{} term. Then for any well-typed substitutions
  $\bar{p} \vDash \Phi$ and $\bar{k} \vDash \Gamma^\circ$, if
  \[e[\bar{p}/dom(\Phi)][\bar{k}/dom(\Gamma)] \stepap v\] for some value $v$,
  then the well-typed substitution $\bar{l} \vDash \Gamma^\circ$ exists such
  that \[e[\bar{p}/dom(\Phi)][\bar{l}/dom(\Gamma)] \stepid v,\] and
  $d_{\denot{\sigma_i}}({k}_i,{l}_i) \le r_i$ for each $k_i \in \bar{k}$ and
  $l_i \in \bar{l}$.
\end{theorem}

\begin{proof} 
We sketch the proof here; details are provided in \Cref{sec:app_soundness}.
The key idea is to use the backward map $b :  (\denot{\Phi} \otimes
\denot{\Gamma} \otimes \denot{\sigma}) \rightarrow (\denot{\Phi} \otimes
\denot{\Gamma})$ to construct the well-typed substitutions $\bar{s}$ and
$\bar{l}$ such that $(\bar{s},\bar{l}) = b((\bar{p},\bar{k}),v)$.  From the
second property of backward error lenses we then have
\[
f\pdenotid{(\bar{s},\bar{l})} = f\pdenotid{b((\bar{p},\bar{k}),v)} = v.
\]
We can use this result along with  pairing (\Cref{lem:pairing}) and adequacy
(\Cref{thm:adequacy}) to show 
\[ 
e[\bar{s}/dom({\Phi})][\bar{l}/dom(\Gamma)]\stepid v.
\]
 By soundness (\Cref{thm:soundid}) we can then derive a backward
error result: 
\[ 
\tilde{f}\pdenotap{(\bar{p},\bar{k})} = f\pdenotid{(\bar{s},\bar{l})}.
\]

Two things remain to be shown. First, we must show the values of discrete type
carry no backward error, i.e., $\bar{s} = \bar{p}$.  Second, we  must show the
values of linear type have bounded backward error. Both follow from the first
property of error lenses: from the inequality 
\[ 
	\max\left(d_{\denot{\Phi}}(\bar{p},\bar{s}),
	d_{\denot{\Gamma}}({\bar{k},\bar{l}})\right) \le d_{\denot{\sigma}}\left( v,
	v\right) \neq \infty
\] 
we can conclude $\bar{s}=\bar{p}$ and $d_{\denot{\sigma_i}}({k}_i,{l}_i) \le
r_i$ for each $k_i \in \bar{k}$ and $l_i \in \bar{l}$.
\end{proof}

\section{Example \bea{} Programs}\label{sec:bean:examples}
We will present a range of case studies demonstrating how algorithms with
well-known backward error bounds from the literature can be implemented in
\bea{}.  We begin by comparing two implementations of polynomial evaluation, a
naive evaluation and Horner's scheme. Next, we write several programs which
compose to perform generalized matrix-vector multiplication. Finally, we write a
triangular linear solver.

To improve the readability of our examples, we adopt several conventions.
First, matrices are assumed to be stored in row-major order.  Second, following
the convention used in the grammar for \bea{} in \Cref{sec:bean:language}, we use
\stexttt{x} and \stexttt{y} for linear variables and \stexttt{z} for discrete variables. Finally, for
types, we denote both discrete and linear numeric types by $\R$, and use a
shorthand for type assignments of vectors and matrices. For instance: $\R^2
\equiv (\R \otimes \R)$ and $\R^{3 \times 2} \equiv (\R \otimes \R) \otimes (\R
\otimes \R) \otimes (\R \otimes \R)$.

Since \bea{} is a simple first-order language and currently does not support
higher-order functions or variable-length tuples, programs can become verbose.
To reduce code repetition, we use basic user-defined abbreviations in our
examples. 

\subsection*{Polynomial Evaluation} 
To illustrate how \bea{} can provide a fine-grained backward error analysis for
numerical algorithms, we begin with simple programs for polynomial evaluation.
The first program, \stexttt{PolyVal}, evaluates a polynomial by naively
multiplying each coefficient by the variable multiple times and then summing
the resulting terms. The second program, \stexttt{Horner}, applies Horner's
method, which iteratively adds the next coefficient and then multiplies the sum
by the variable \cite[p.94]{Higham:2002:Accuracy}. We consider here \bea{} 
implementations of these algorithms for a second-order polynomial; in 
\Cref{sec:bean:implementation}, we describe a prototype implementation of \bea{} 
and evaluate the backward error bounds it infers for higher-degree polynomials.

Given a tuple $\stexttt{a} : \R^3$ of coefficients and a discrete variable
$\stexttt{z} : \R$, the \bea{} programs for evaluating a second-order
polynomial $p(z) = a_0 + a_1z + a_2z^2$ using naive polynomial evaluation and
Horner's method are shown below. 

\begin{center}
\begin{minipage}[t]{.5\textwidth}
\begin{lstlisting}
PolyVal a z :=
let (a0, a') = a  in 
let (a1, a2) = a' in
let y1  = dmul z a1  in 
let y2' = dmul z a2  in
let y2  = dmul z y2' in
let x   = add a0 y1  in
add x y2
\end{lstlisting}
\end{minipage}
\begin{minipage}[t]{.5\textwidth}
\begin{lstlisting}
Horner a z :=
let (a0, a') = a   in 
let (a1, a2) = a'  in
let y1 = dmul z a2 in
let y2 = add a1 x  in
let y3 = dmul z y2 in
add a0 y3
\end{lstlisting}
\end{minipage}
\end{center}
\noindent
Recall from \Cref{sec:bean:introduction} that the $\textbf{dmul}$ operation assigns
backward error onto its second argument; in the programs above, the operation
indicates that backward error should not be assigned to the discrete variable
\stexttt{z}. Using \bea{}'s type system, the following typing judgments are
valid:
\[ 
  \stexttt{z} : \R \mid \stexttt{a} :_{3\varepsilon} \R^3 \vdash \stexttt{PolyVal a z}: \R 
  \qquad  \qquad
  \stexttt{z} : \R \mid \stexttt{a} :_{4\varepsilon} \R^3 \vdash \stexttt{Horner a z} : \R
\] 
From these judgments, \emph{backward error soundness} (\Cref{thm:main})
guarantees that \stexttt{PolyVal} has backward error of at most ${3\varepsilon}$
with respect to each element in the tuple \stexttt{a}, while \stexttt{Horner}
has backward error of at most ${4\varepsilon}$ with respect to each element in
the tuple \stexttt{a}.

Surprisingly, though Horner's scheme is considered more numerically stable as
it minimizes the number of floating-point operations, we find it has
potentially greater backward error with respect to the vector of coefficients.
A closer look at each coefficient individually, however, reveals more
information about the two implementations. By adjusting the implementations to
take each coefficient as a separate input, we can derive the backward error
bounds for each coefficient individually. Now, \bea{}'s type system derives the
following valid judgments:
\[ 
  \stexttt{z} : \R \mid \stexttt{a0} :_{2\varepsilon}\R, \stexttt{a1}:_{3\varepsilon}\R, 
  \stexttt{a2}:_{3\varepsilon}\R \vdash
   \stexttt{PolyVal'}: \R
  \qquad
  \stexttt{z} : \R \mid \stexttt{a0} :_{\varepsilon}\R, \stexttt{a1}:_{3\varepsilon}\R, 
  \stexttt{a2}:_{4\varepsilon}\R \vdash
  \stexttt{Horner'}: \R
\] 
We see that Horner's scheme assigns more backward error onto the coefficients
of higher-order terms than lower-order terms, while naive polynomial evaluation 
assigns the same error onto all but the lowest-order coefficient.  In this way, 
\bea{} can be used to investigate the numerical stability of different polynomial 
evaluation schemes by providing a fine-grained error analysis.

\subsection*{Matrix-Vector Multiplication}
A key feature of \bea{}'s type and effect system is its ability 
to precisely track backward error across increasingly large programs.  Here, we
demonstrate this process with several programs that gradually build up to a
scaled matrix-vector multiplication. 

Given a matrix  $M \in \R^{m \times n}$, vectors $v \in \R^{n}$ and $u \in
\R^m$, and constants $a,b \in \R$, a scaled matrix-vector operation computes $a
\cdot (M \cdot v) + b \cdot u$. Since \bea{} does not currently support
variable-length tuples, we present the details of a \bea{} implementation for a
$2 \times 2$ matrix. 

We first define the program \stexttt{SVecAdd}, which computes a scalar-vector
product using \stexttt{ScaleVec} and then adds the result to another vector.
Given a discrete variable $\stexttt{a} : \R$, along with linear variables
$\stexttt{x}: \R^2$ and $\stexttt{y}: \R^2$, we implement these programs as
follows:
\begin{center}
\begin{minipage}[t]{.45\textwidth}
\begin{lstlisting}
ScaleVec a x :=
let (x0, x1) = x  in 
let u = dmul a x0 in 
let v = dmul a x1 in 
(u, v)
\end{lstlisting}
\end{minipage}
\begin{minipage}[t]{.45\textwidth}
\begin{lstlisting}
SVecAdd a x y :=
let (x0, x1) = ScaleVec a x in 
let (y0, y1) = y  in 
let u = add x0 y0 in
let v = add x1 y1 in
(u, v)
\end{lstlisting}
\end{minipage}
\end{center}
\noindent
These programs have the following valid typing judgments:
\[ 
  \stexttt{a}: \R \mid \stexttt{x} :_{\varepsilon}
  \R^2 \vdash \stexttt{ScaleVec a x} : \R
  \qquad 
  \stexttt{a}:  \R \mid \stexttt{x} :_{2\varepsilon} \R^2,
  \stexttt{y}:_{\varepsilon} \R^2 \vdash \texttt{\small{SVecAdd a x y}} : \R
\]
In the typing judgment for \texttt{\small{SVecAdd}}, we observe that the linear
variable \stexttt{x} has a backward error bound of $2 \varepsilon$, while the
linear variable \stexttt{y} has backward error bound of only $\varepsilon$.
This difference arises because  \stexttt{x} accumulates $\varepsilon$ backward
error from \stexttt{ScaleVec} and an additional $\varepsilon$ backward error
from the vector addition with the linear variable \stexttt{y}.

Now, given discrete variables $\stexttt{a} : \R$
and $\stexttt{b} : \R$, and $\stexttt{v}: \R^{2}$,
along with the linear variables 
$\stexttt{M}: \R^{2 \times 2}$ 
and $\stexttt{u}: \R^{2}$, 
we can compute a matrix-vector product of 
\stexttt{M} and \stexttt{v}
with  
\stexttt{MatVecMul},
and use the result in the scaled matrix-vector 
product, \stexttt{SMatVecMul}:

\begin{center}
\begin{minipage}[t]{.45\textwidth}
\begin{lstlisting}
MatVecMul M v :=
let (m0, m1) = M in 
let u0 = InnerProduct m0 v in 
let u1 = InnerProduct m1 v in 
(u0, u1)
\end{lstlisting}
\end{minipage}
\begin{minipage}[t]{.45\textwidth}
\begin{lstlisting}
SMatVecMul M v u a b :=
let x = MatVecMul M v in
let y = ScaleVec b u  in
SVecAdd a x y 
\end{lstlisting}
\end{minipage}
\end{center}
For \stexttt{MatVecMul}, we rely on a program \stexttt{InnerProduct},
which computes the dot product of two $2 \times 2$ vectors. Notably,
\stexttt{InnerProduct} differs from the \stexttt{DotProd2} program 
described in \Cref{sec:bean:introduction} because it assigns backward error only onto 
the first vector. The type of this program is:
\[
  \stexttt{v}: \R^2 \mid \stexttt{u} :_{2\varepsilon}
  \R^2 \vdash \stexttt{InnerProduct u v}: \R
\]
The \bea{} programs \stexttt{MatVecMul} and 
\stexttt{SMatVecMul} have the following valid typing judgments: 
\begin{align*}
  \stexttt{v} :  \R^2 \mid \stexttt{M} :_{2\varepsilon}  \R^{2\times 2} &\vdash
  \stexttt{MatVecMul M v} : \R^2
  \\
  \stexttt{a}:\R, \stexttt{b}:\R, \stexttt{v}:\R^{2} 
  \mid \stexttt{M} :_{4\varepsilon} \R^{2\times 2},
  \stexttt{u}:_{2\varepsilon}\R^{2} &\vdash \stexttt{SMatVecMul M v u a b}:
  \R^2
\end{align*}

By error soundness, these judgments say that the computation
\stexttt{SMatVecMul} produces at most $2 \varepsilon$ backward error with
respect to the vector \stexttt{u} and at most $4 \varepsilon$ backward error
with respect to the matrix \stexttt{M}.  The backward error bound for
\stexttt{M} can be understood as follows: the computation \stexttt{MatVecMul M
v} assigns at most $2 \varepsilon$ backward error to \stexttt{M}, and the
computation \stexttt{SVecAdd a x y} assigns an additional $2 \varepsilon$
backward error to \stexttt{M}, resulting in a backward error bound of $4
\varepsilon$.  Similarly, the backward error bound of $2 \varepsilon$ for the
variable \stexttt{u} arises from the computation \stexttt{ScaleVec b u}, which
assigns at most $\varepsilon$ backward error to \stexttt{u}, and
\stexttt{SVecAdd a x y}, which assigns at most an additional $\varepsilon$
backward error to \stexttt{u}, leading to a total backward error bound of $2
\varepsilon$.  In \Cref{sec:evaluation}, we will see that the backward error
bounds for matrix-vector multiplication derived by \bea{} match the worst-case
theoretical backward error bounds given in the literature.

Overall, these examples highlight the compositional nature of \bea's analysis: 
like all type systems, the type of a \bea{} program is derived from the types of its
subprograms. While the numerical analysis literature is unclear on whether (and
when) backward error analysis can be performed compositionally (e.g.,
~\citep{Bornemann:2007:backwardcompose}), \bea{} demonstrates that this is in
fact possible.

\subsection*{Triangular Linear Solver}
One of the benefits of integrating error analysis with a type system is the
ability to weave common programming language features, such as conditionals
(if-statements) and error-trapping, into the analysis. We demonstrate these
features in our final, and most complex example: a linear solver for triangular
matrices. Given a lower triangular matrix $A\in\R^{2\times 2}$ and a vector
$b\in\R^2$, the linear solver should compute return a vector $x$ satisfying $Ax
= b$ if there is a unique solution.

We comment briefly on the program \stexttt{LinSolve}, shown below.  The matrix
and vector are given as inputs $\stexttt{((a00, a01), (a10, a11))}:\R^{2\times
2}$ where \stexttt{a01} is assumed to be $0$, and $\stexttt{(b0, b1)}:\R^2$.
The program either returns the solution $x$ as a vector, or returns error if the
linear system does not have a unique solution.  The $\mathbf{div}$ operator has
return type $\R + \err$, where $\err$ represents division by zero. Ensuing
computations can check if the division succeeded using $\mathbf{case}$
expressions. This example also uses the $!$-constructor to convert a
linear variable into a discrete one; this is required since the later entries in
the vector $x$ depend on---i.e., require duplicating---earlier entries in the
vector.
\begin{center}
\begin{lstlisting}
    LinSolve ((a00, a01), (a10, a11)) (b0, b1) :=
    let x0_or_err = div b0 a00 in // solve for x0 = b0 / a00
    case x0_or_err of 
      inl (x0) => // if div succeeded
        dlet d_x0 = !x0 in // make x0 discrete for reuse
        let s0 = dmul d_x0 a10 in // s0 = x0 * a10
        let s1 = sub b1 s0 in // s1 = b1 - x0 * a10
        let x1_or_err = div s1 a11 in // solve for x1 = (b1 - x0 * a10) / a11
        case x1_or_err of 
          inl (x1) => inl (d_x0, x1) // return (x0, x1)
        | inr (err) => inr err // division by 0
    | inr (err) => inr err // division by 0
\end{lstlisting}
\end{center}
\noindent 
The type of \stexttt{LinSolve} is
$
  A:_{5\varepsilon/2}\R^{2\times 2}, b:_{3\varepsilon/2}\R^2
  \vdash \stexttt{LinSolve A b} : \R^2 + \err.
$
Hence, \stexttt{LinSolve} has a guaranteed backward error bound of at most
$5\varepsilon/2$ with respect to the matrix \stexttt{M} and at most
$3\varepsilon/2$ with respect to the vector \stexttt{b}. If either of the
division operations fail, the program returns $\err$. This example demonstrates
how various features in \bea{} combine to establish backward error guarantees
for programs involving control flow and duplication, via careful control of how
to assign and accumulate backward error through the program.
\section{Implementation}\label{sec:bean:implementation}
\subsection{{Implementation}}\label{sec:algorithm}
We implemented a type checking and coeffect inference algorithm for \bea{} in
OCaml. It is based on the sensitivity inference algorithm introduced by
\citet{Amorim:2014:typecheck}, which is used in implementations of
\emph{Fuzz}-like languages \cite{Gaboardi:2013:dfuzz, NUMFUZZ}.  Given a \bea{} program
without any error bound annotations in the context, the type checker ensures
the program is well-formed, outputs its type, and infers the tightest possible
backward error bound on each input variable.  Using the type checker, users can
write large \bea{} programs and automatically infer backward error with respect
to each variable.

More precisely, let $\Gamma^\bullet$ denote a context \emph{skeleton}, a linear
typing context with no coeffect annotations. If $\Gamma$ is a linear context,
let $\overline{\Gamma}$ denote its skeleton. Next, we say $\Gamma_1$ is a
\emph{subcontext} of $\Gamma_2$, $\Gamma_1\sqsubseteq\Gamma_2$, if
$\dom{\Gamma_1}\subseteq\dom{\Gamma_2}$ and for all $x:_r\sigma\in\Gamma_1$, we
have $x:_q\sigma\in\Gamma_2$ where $r \leq q$. In other words, $x$ has a
tighter backward error bound in the subcontext. Now, we can say the input to
the type checking algorithm is a typing context skeleton $\Phi\mid
\Gamma^\bullet$ and a \bea{} program, $e$. The output is the type of the
program $\sigma$ and a linear context $\Gamma$ such that $\Phi\mid\Gamma\vdash
e:\sigma$ and $\overline{\Gamma}\sqsubseteq\Gamma^\bullet$.  Calls to the
algorithm are written as $\Phi\mid \Gamma^\bullet;e\Rightarrow \Gamma;\sigma$.
The algorithm uses a recursive, bottom-up approach to build the final context.

For example, to type the \bea{} program $(e,f)$, where $e$ and $f$ are
themselves programs, we use the algorithm rule \[
  \AXC{$\Phi\mid\Gamma^\bullet;e\Rightarrow \Gamma_1;\sigma$} \AXC{$\Phi\mid
  \Gamma^\bullet;f\Rightarrow\Gamma_2;\tau$}
  \AXC{$\dom\Gamma_1\cap\dom\Gamma_2=\emptyset$} \RightLabel{($\tensor$ I)}
  \TrinaryInfC{$\Phi\mid\Gamma^\bullet;(e,f)\Rightarrow\Gamma_1,\Gamma_2;\sigma\otimes\tau$}
  \bottomAlignProof \DisplayProof \] In practice, this means recursively
calling the algorithm on $e$ and $f$ then combining their outputted contexts.
The output contexts discard unused variables from the input skeletons; thus,
the requirement $\dom\Gamma_1\cap\dom\Gamma_2=\emptyset$ ensures the strict
linearity requirement is met. 

The type checking algorithm is sound and complete, meaning that it agrees
exactly with \bea{}'s typing rules. Precisely: \begin{theorem}[Algorithmic
  Soundness]\label{thm:algo_sound} If $\Phi\mid\Gamma^\bullet;e\Rightarrow
  \Gamma;\sigma$, then $\overline{\Gamma}\sqsubseteq\Gamma^\bullet$ and the
derivation $\Phi\mid\Gamma\vdash e:\sigma$ exists.  \end{theorem}
\begin{theorem}[Algorithmic Completeness]\label{thm:algo_complete} If
  $\Phi\mid\Gamma\vdash e:\sigma$ is a valid derivation in \bea{}, then there
exists a context $\Delta\sqsubseteq\Gamma$ such that
$\Phi\mid\overline{\Gamma}; e\Rightarrow \Delta;\sigma$.  \end{theorem} The
full algorithm and proofs of its correctness are given in \Cref{app:algorithm}.
The \bea{} implementation is parametrized only by unit roundoff, which is
dependent on the floating-point format and rounding mode and is fixed for a
given analysis. 

\subsection{{Evaluation}}\label{sec:evaluation}
\begin{table}
\caption{ 
    A comparison of \bea{} to \citet{Fu:2015:BEA} on polynomial approximations
    of $\sin$ and $\cos$. The \bea{} implementation matches the 
    programs evaluated by \citet{Fu:2015:BEA} for the given range of input values.
} 
\label{tab:benchFu} 
\centering
\begin{tabular}{l l c c c c} 
\hline
{Benchmark} & {Range} & \multicolumn{2}{c}{Backward Bound} 
& \multicolumn{2}{c}{Timing (ms)}  \\
\cline{3-6} 
{} & {} & {{\bea}}
& {{\citet{Fu:2015:BEA}}}
& {{\bea}}
& {{\citet{Fu:2015:BEA}}} \\
\hline 
\stexttt{cos} & [0.0001, 0.01] & {1.33e-15} & {5.43e-09} & 1 & 1310 \\
\arrayrulecolor{gray}\hline
\stexttt{sin} & [0.0001, 0.01] & {1.44e-15} & {1.10e-16} & 1 & 1280
\\ \arrayrulecolor{black}\hline
\end{tabular}
\end{table} 

\begin{table}
\caption{ 
    The performance of \bea{} benchmarks with known backward error bounds
    from the literature. The Input Size column gives 
    the length of the input vector or dimensions of the input matrix; the 
    Ops column gives the total number of floating-point operations.  
    The Backward Bound column reports the bounds
    inferred by \bea{} and well as the standard bounds (Std.) from
    the literature. The Timing column reports the time in seconds for 
    \bea{} to infer the backward error bound. 
} 
\label{tab:benchStd} 
\centering
\begin{tabular}{l l l c c r} 
\hline
{Benchmark}
& {Input Size}
& {Ops}
& \multicolumn{2}{c}{Backward Bound} 
& {Timing (s)} 
\\ \cline{4-5}
& & & {{\bea}} & {{Std.}} & 
\\ \hline 
\multirow{4}{1em}{{DotProd}} 
& 20 & {39} & {2.22e-15} & {2.22e-15} & 0.004
\\ \arrayrulecolor{gray}\cline{2-6}
& 50 & {99} & {5.55e-15} & {5.55e-15} & 0.04
\\ \cline{2-6}
& 100 & 199 & {1.11e-14} & {1.11e-14} & 0.3 
\\ \cline{2-6}
& 500 & 999 & {5.55e-14} & {5.55e-14} & 30
\\ \arrayrulecolor{black}\hline
\multirow{4}{1em}{{Horner}} 
& 20 & {40} & {4.44e-15} & {4.44e-15} & 0.002
\\ \arrayrulecolor{gray}\cline{2-6}
& 50 & {100} & {1.11e-14} & {1.11e-14} & 0.02
\\ \cline{2-6}
& 100 & 200 & {2.22e-14} & {2.22e-14} & 0.1
\\ \cline{2-6}
& 500 & 1000 & {1.11e-13} & {1.11e-13} & 10
\\ \arrayrulecolor{black}\hline
\multirow{4}{1em}{{PolyVal}} 
& 10 & {65} & {1.22e-15} &{1.22e-15} & 0.004
\\ \arrayrulecolor{gray}\cline{2-6}
& 20 & {230} & {2.33e-15} & {2.33e-15} & 0.06
\\ \cline{2-6}
& 50 & 1325 & {5.66e-15} & {5.66e-15} & 5
\\ \cline{2-6}
& 100 & 5150 & {1.12e-14} & {1.12e-14} & 200
\\ \arrayrulecolor{black}\hline
\multirow{4}{1em}{{MatVecMul}} 
& 5 $\times$ 5 & 45 & {5.55e-16} & {5.55e-16} & 0.003
\\ \arrayrulecolor{gray}\cline{2-6}
& 10 $\times$ 10 & 190 & {1.11e-15} & {1.11e-15} & 0.1
\\ \cline{2-6}
& 20 $\times$ 20 & 780 & {2.22e-15} & {2.22e-15} & 6
\\ \cline{2-6}
& 50 $\times$ 50 & 4950 & {5.55e-15} & {5.55e-15} & 1000
\\ \arrayrulecolor{black}\hline
\multirow{4}{1em}{{Sum}} 
& 50 & {49} & {5.44e-15} & {5.44e-15} & 0.008
\\ \arrayrulecolor{gray}\cline{2-6}
& 100 & {100} & {1.10e-14} & {1.10e-14} & 0.04
\\ \cline{2-6}
& 500 & 499 & {5.54e-14} & {5.54e-14} & 4
\\ \cline{2-6}
& 1000 & 999 & {1.11e-13} & {1.11e-13} & 30
\\ \arrayrulecolor{black}\hline
\end{tabular}
\end{table} 

\begin{table}
\caption{ 
A comparison of forward bounds derived from \bea{}'s
backward error bounds to those of NumFuzz and Gappa.
For Gappa, we assume all variables are in the interval $[0.1,1000]$.} 
\label{tab:benchFz} 
\centering
\begin{tabular}{l c c c c c c} 
\hline
{Benchmark}
& {Input Size}
& {Ops}
& \multicolumn{3}{c}{Forward Bound} 
\\ \cline{4-6}
& & & {{\bea}} & {NumFuzz } & Gappa \\
\hline 
Sum
& 500 & 499 & {1.11e-13} & {1.11e-13} & {1.11e-13} \\
\arrayrulecolor{gray}\hline
DotProd
& 500 & 999 & {1.11e-13} & {1.11e-13} & {1.11e-13} \\
\hline
Horner
& 500 & 1000 & {2.22e-13} & {2.22e-13} & {2.22e-13} \\
\hline
PolyVal
& 100 & {5150} & {2.24e-14} & {2.24e-14} & {2.24e-14} 
\\ \arrayrulecolor{black}\hline
\end{tabular}
\end{table}

In this section, we report results from an empirical evaluation of our \bea{}
implementation, focusing primarily on the quality of the inferred bounds. Since
\bea{} is the first tool to statically derive \emph{sound} backward error
bounds, a direct comparison with existing tools is challenging. We therefore
evaluate the inferred bounds using three complementary methods.

First, we compare our results to those from a dynamic analysis tool for
automated backward error analysis introduced by \citet{Fu:2015:BEA}. To our
knowledge, the results reported by \citet{Fu:2015:BEA} provide the only
automatically derived quantitative bounds on backward error available for
comparison; these results serve as a useful baseline for assessing the tightness
of the bounds inferred by \bea{}. However, the experimental results reported by
\citet{Fu:2015:BEA} are limited to transcendental functions, while \bea{} is
designed to handle larger programs oriented towards linear algebra primitives.
Therefore, we also include an evaluation against theoretical worst-case backward
error bounds described in the literature. This allows us to benchmark \bea{}'s
bounds in relation to established theoretical limits, providing a measure of how
closely \bea{}'s inferred bounds approach these worst-case values. Finally, we
evaluate the quality of the backward error bounds derived by \bea{} using
forward error as a proxy. Specifically, using known values of the relative
componentwise condition number (\Cref{def:cnum}), we compute forward error
bounds from our backward error bounds. This approach enables a comparison to
existing tools focused on forward error analysis. We compare our derived forward
error bounds to those produced by two tools that soundly and automatically bound
relative forward error: NumFuzz \cite{NUMFUZZ} and Gappa \cite{Gappa:2010}.
Both tools are capable of scaling to larger benchmarks involving over 100
floating-point operations, making them suitable tools for comparison with
\bea{}. All of our experiments were performed on a MacBook Pro with an Apple M3
processor and 16 GB of memory.

\subsubsection{Comparison to Dynamic Analysis}
The results for the comparison of \bea{} to the optimization based tool for
automated backward error analysis due to \citet{Fu:2015:BEA} is given in
\Cref{tab:benchFu}. The benchmarks are polynomial approximations of $\sin$
and $\cos$ implemented using Taylor series expansions following the GNU C
Library (glibc) version 2.21 implementations.  Our \bea{} implementations match
the benchmarks from \Cref{tab:benchFu} on the input range $[0.0001,0.01]$.
Specifically, the Taylor series expansions implemented in \bea{} only match the
glibc implementations for inputs in this range.  Since the glibc
implementations analyzed by \citet{Fu:2015:BEA} use double-precision and
round-to-nearest, we instantiated \bea{} with a unit roundoff of $u=2^{-53}$.
Although we include timing information for reference, the implementation
described by \citeauthor{Fu:2015:BEA} is neither publicly available nor
maintained, preventing direct runtime comparisons; thus, all values are taken
from Table 6 of \citet{Fu:2015:BEA}. 

\subsubsection{Evaluation Against Theoretical Worst-Case Bounds}
\Cref{tab:benchStd} presents results for several benchmark problems with known
backward error bounds from the literature. Each benchmark was run on inputs of
increasing size (given in Input Size), with the total number of
floating-point operations listed in the Ops column. The Std. column provides
the worst-case theoretical backward error bound reported in the literature
assuming double-precision and round-to-nearest; the relevant references are 
\cite[p.63, p.94, p.82]{Higham:2002:Accuracy}. 
For simplicity, the \bea{} programs are written with a single
linear variable, while the remaining inputs are treated as discrete variables.
The maximum elementwise backward bound is computed with respect to the linear
input. The \bea{} programs emulate the following analyses for input size $N$: 
\begin{itemize}
\item \stexttt{DotProd} computes the dot product
of two vectors in $\R^N$, assigning backward error to a single vector.
\item \stexttt{Horner}  evaluates an $N$-degree polynomial using Horner's
scheme, assigning backward error onto the vector of coefficients.
\item \stexttt{PolyVal}  naively evaluates an $N$-degree polynomial, 
assigning backward error onto the vector of coefficients.
\item \stexttt{MatVecMul} computes the product of a matrix in $\R^{N\times
N}$ and a vector in $\R^N$, assigning backward error onto the matrix.
\item \stexttt{Sum} sums the elements of a vector
in $\R^N$, assigning backward error onto the vector.
\end{itemize}
Since we report the backward error bounds from the literature under the
assumption of double-precision and round-to-nearest, we instantiated \bea{}
with a unit roundoff of $u=2^{-53}$.

\subsubsection{Using Forward Error as a Proxy}
We can compare the quality of the backward error bounds derived by
\bea{} to existing tools using forward error as a proxy. Specifically, by using
known values of the relative componentwise condition number $\kappa_{rel}$, we
can compute relative forward error bounds from relative backward error bounds
\cite[Definition 2.12]{hohmann:2003:numerical}: 

\begin{definition}[Relative Componentwise Condition Number]
\label{def:cnum}
The \emph{relative componentwise condition number} of a scalar function $f:
\R^n \rightarrow \R$  is the smallest number $\kappa_{rel} \ge 0$ such that,
for all $x\in\R^n$, 
\begin{align}\label{eq:krel}
d_{\R}(f(x),\tilde{f}(x)) \le \kappa_{rel} \max_{i} d_\R(x_i,\tilde{x}_i)
\end{align}
\end{definition}
\noindent where $\tilde{f}$ is the approximating program and $\tilde{x}$ is 
the perturbed input witnessing $f(\tilde{x})=\tilde{f}(x)$.
In \Cref{eq:krel}, $d_{\R}(f(x),\tilde{f}(x))$ is the \emph{relative forward
error}, and $\max_{i} d_\R(x_i,\tilde{x}_i)$ is the maximum \emph{relative
backward error}. Thus, for problems where the relative condition number is
known, we can compute relative forward error bounds from the relative backward
error bounds inferred by \bea{}. 

\Cref{tab:benchFz} presents the results for several benchmark problems with
$\kappa_{rel}=1$. For these problems, according to \Cref{eq:krel}, the maximum
relative backward error serves as an upper bound on the relative forward error.
As an example, the problem of summing $n$ values $(a_i)_{1\le i\le n}$ has a
relative condition number  $\kappa_{rel} = \sum_{i=1}^n |a_i|/ |\sum_{i=1}^n
a_i|$ \cite{MullerBook}, which clearly reduces to $\kappa_{rel}= 1$ when all
$a_i >0$. In fact, for each of the benchmarks listed in \Cref{tab:benchFz},
$\kappa_{rel}= 1$ is only guaranteed for strictly positive inputs. This 
assumption is already required for NumFuzz in order to guarantee the 
soundness of its forward error bounds. To enforce this in Gappa, we used an 
interval of $[0.1, 1000]$ for each input. Since NumFuzz assumes 
double-precision and round towards positive infinity, we instantiated \bea{} 
and Gappa with a unit roundoff of $u=2^{-52}$.

\subsubsection{Evaluation Summary} 
The main conclusions from our evaluation results are as follows.
\emph{\textbf{\bea{}'s backward error bounds are useful}}: In all of our
experiments, \bea{} produced competitive error bounds. Compared to the
backward error bounds reported by \citet{Fu:2015:BEA} for their dynamic
backward error analysis tool, \bea{} was able to derive \emph{sound} backward
error bounds that were close to or better than those produced by the dynamic
tool.  Furthermore, \bea{}'s sound bounds precisely match the worst-case
theoretical backward error bounds from the literature, demonstrating that our
approach guarantees soundness without being overly conservative.  Finally,
when using forward error as a proxy to assess the quality of \bea{}'s backward
error bounds, we find that \bea{}'s bounds again precisely match the bounds
produced by NumFuzz and Gappa.  \emph{\textbf{\bea{} performs well on large
programs}}: In our comparison to worst-case theoretical error bounds, we find
that \bea{} takes under a minute to infer backward error bounds on benchmarks
with fewer than 1000 floating-point operations. Overall, \bea{}'s performance
scales linearly with the number of floating-point operations in a benchmark.

\section{Related Work}\label{sec:bean:related}

\subsubsection{Automated Backward Error Analysis}
Existing automated methods for backward error analysis are based on automatic 
differentiation and optimization techniques. Unlike \bean{}, existing methods
do not provide a soundness guarantee and are based on heuristics.
Miller's algorithm \citep{Miller:78:BEA} first appeared in a 
FORTRAN package and used automatic differentiation to compute  
partial derivatives of function outputs with respect to function inputs as a 
proxy for backward error. The algorithm was later augmented to handle a 
broader range of program features (loops and conditional expressions)
in a MATLAB implementation \citep{Gati:2012:BEA}. 

The first optimization based tool for automated backward 
error analysis was introduced by \citet{Fu:2015:BEA}. The key idea of the 
approach is to separate the analysis into a \emph{local} error analysis and a 
\emph{global} error analysis. Given a program and specific
inputs, the local error simulates the ideal, continuous problem by lifting the 
program to a higher-precision version. Then, a generic minimizer is used to derive a 
backward error function that associates the input and output of the original 
program to an input of the higher-precision program that hits the same output.  
The global error analysis uses the backward error function as a black-box
function to heuristically estimate the maximal backward error for a range of
inputs by Markov Chain Monte Carlo techniques. 

\bea{} is similar to the optimization technique due to \citet{Fu:2015:BEA} 
by virtue of the direct construction of the backward function: in order to 
perform a backward error analysis, both \bea{} and the optimization technique 
require an ideal function, an approximating function, and an explicit backward 
function. However, unlike \bea{}, the existing optimization method must perform 
a sometimes costly analysis to construct the ideal and backward functions for 
every program. In \bea{}, it is not necessary to construct these functions for
typechecking since they are built into our semantic model.

\subsubsection{Residual Based Methods}
When a backward error bound does exist, it is possible to compute \emph{a
posteriori} estimates of the error dynamically using \emph{residual-based
methods} as described by \citet{Corless:2013:Analysis}.  These methods require
constructing a \emph{defining function}, which depends on both the input and
the output of the program---similar to the backward maps in \bea{}.  Whereas
residual-based methods require the manual construction of a defining function,
in \bea{}, the backward maps of complex programs are composed from their
individual components, enabling more automated reasoning. Moreover,
residual-based estimates are not \emph{bounds} and, unlike \bea{}, do not
provide a guaranteed sound overapproximation of the true error. In situations
where soundness is not required, residual-based methods can be manually
incorporated into programs.  Unfortunately, in some instances, computing these
estimates can be more computationally expensive than solving the original
problem. 

\subsubsection{Type Systems and Formal Methods}
A diverse set of tools for reasoning about forward rounding error bounds have 
been proposed in the formal methods literature; these tools are discussed in
\Cref{sec:nfuzz:related}. In comparison, for backward error analysis, the
formal methods literature is sparse. The LAProof library due to
\citet{Kellison:Arith:2023} provides formal proofs of backward error bounds for
basic linear algebra subprograms in the Coq proof assistant, and gives an
example of how these proofs can be used to verify real C-programs. 
These proofs are parametric in both the floating-point format and the size of
the underlying data structures.  In \bea{}, as with other type-based
approaches, we trade some of the expressivity offered by proof assistant-based
methods for a more lightweight and potentially more automated system. While
less expressive, valid typing derivations in \bea{} correspond to formal
proofs that a given program satisfies the backward error bound 
the type system assigns it. This guarantee is rigorously established by our
backward error soundness theorem, \Cref{thm:main}.

The only other type-based approach to rounding error analysis is \Lang{},
which, like \bean{}, also uses a linear type system and coeffects. However,
\Lang{} is specifically designed for forward error analysis, whereas \bean{}
focuses on backward error analysis.  While the syntactic similarities between
\Lang{} and \bean{} may suggest that \bean{} is simply a derivative of \Lang{}
modified for backward error analysis, this is not the case. We designed \bean{}
by first developing the category \Bel{} of backward error lenses, and then
developing the language described in \Cref{sec:bean:language} to fit this
category.  Indeed, the semantics of the two systems are entirely different:
\begin{itemize}
\item The primary semantic novelty in \Lang{} is the neighborhood monad, which
  tracks forward error but cannot be adapted for backward error. \bea{} does not
  use the neighborhood monad, or any monad at all.
\item While both \Lang{} and \bea{} use a graded comonad, their interpretations
  are different and are used for different purposes. The graded
  comonad in \Lang{} scales the metric to track function sensitivity, while
  the graded comonad in \bea{} shifts (translates) the metric to track backward error.
\item Similar to other \emph{Fuzz}-like languages, \Lang{} interprets programs
  in the category of metric spaces, which lacks the necessary structure for
  reasoning about backward error. To address this, we introduced the novel
  category of backward error lenses, offering a completely new semantic
  foundation that distinguishes \bean{} from all languages in the \emph{Fuzz} family.
\end{itemize}

\subsubsection{Linear Type Systems and Coeffects}
After \Lang{}, which we have already discussed, \bea{} is 
most closely related to coeffect-based type systems, like those described
by \citet{Petricek:2014:coeffects}; references can be found in the brief 
description of coeffects given in \Cref{background:coeffects}.
As discussed before, a notable difference between \bean{}'s type system 
and other coeffect systems is that our model does
not support contraction, and so our type system enforces strict linearity for
variables with coeffect annotations, with a separate context for variables 
that can be reused. Our dual context approach is similar to the Linear/Non-Linear
(LNL) calculus described by \citet{DBLP:conf/csl/Benton94}.

\subsubsection{Lenses and Bidirectional Programming Languages}
Our semantic model is inspired by work on lenses, proposed by
\citet{Foster:2007:trees}, as a tool to address the view-update problem in
databases. Basically, a lens is a pair of a forward transformation \emph{get} and a 
backward transformation \emph{put} which are used to synchronize  
related data. In general, lenses satisfy  
several \emph{lens laws}, which can be framed as equations that specify the 
relationship between the lens transformations and the data they operate on; 
the equations defining \emph{well-behaved lenses} are given in \Cref{eq:lens1}
and \Cref{eq:lens2}. These equations correspond closely to the properties of
backward error lenses (\Cref{def:error_lens}). The concept of a lens has been
rediscovered multiple times in different contexts, ranging from categorical
proof theory and G\"odel's Dialectica translation \citep{dialectica} to more
recent work on open games \citep{Ghani:2018:gamelenses}, and supervised
learning \citep{DBLP:conf/lics/FongST19}; the interested reader can see
\citet{LensBlog} for a good summary. While the formal similarity between our
backward error lenses and existing work on lenses is undeniable, we are not
aware of any existing notion of lens that includes ours.

\section{Conclusion}\label{sec:bean:conclusion}
\bea{} is a typed first-order programming language that guarantees backward
error bounds. Its type system is based on the combination of three elements: a
notion of distances for types, a coeffect system for tracking backward error,
and a linear type system for controlling how backward error can flow though
programs. Although the backward error analysis modeled by \bean{} is 
more general than the standard approach, we can capture the standard definition 
as a special case, as shown by our main theorem of backward error soundness
(\Cref{thm:main}). A major benefit of our proposed approach is that it is 
structured around the idea of composition: when backward error bounds exist,
the backward error bounds of complex programs are composed from the backward 
error bounds of their subprograms. The linear type system of bean correctly 
rejects programs that do not have bounded backward error, and is also flexible 
enough to capture the backward error analysis of well-known algorithms from 
the literature. \bea{} is the first demonstration of a static analysis framework for
reasoning about backward error, and can be extended in various ways. We conclude
this chapter with a discussion of promising directions for future development. 

\subsubsection{Implementation}
The linear fragment of \bea{} resembles a first-order version of \Lang{},
and fully automated type checking and grade inference algorithms
developed for the \emph{Fuzz} family of languages, such as those 
described by \citet{NUMFUZZ},\citet{DAntoni:2013:sens}, and 
\citet{Amorim:2014:typecheck} can be adapted for \bea{}.
The main difference between the inference algorithms designed for 
\emph{Fuzz}-like languages and \bean{} is strict linearity and dual 
contexts, which introduces a 
small but manageable complication to existing implementations.
Similar to the inference algorithm described in \Cref{sec:nfuzz:implementation}, 
the general idea is that given three inputs---a \bean{} term $e$, 
a linear context $x_1:\sigma_1,\dots,x_i:{\sigma_i}=\Gamma^{\circ}$ with all 
grade annotations erased, and a discrete context $\Phi$---the algorithm infers 
a type $\tau$ of $e$ and the backward error bounds $\delta_i$ such that 
$\Phi \mid x_1:_{\delta_1} \sigma_1,\dots,x_i :_{\delta_i}\sigma_i \vdash e : \tau$. 
While using \bea{} in practice would require users to understand a linear type 
system and select operations based on whether a variable is linear or discrete
(i.e., the choice of using \linec{dmul} or \linec{mul}), this is
standard practice in languages with linear types. From a numerical
standpoint, it is also often known which variables will have backward
error assigned to them during the analysis, and the primary concern is 
computing a backward error bound for compositions of programs; \bea{} 
is well-suited to this task.

\subsubsection{Additional Language Features}
\bean{} does not support higher-order functions, limiting code reuse.
Technically, we do not know if the \Bel{} category supports linear
exponentials, which would be needed to interpret function types. While most
lens categories do not support higher-order functions, there are some notable
situations where the lens category is symmetric monoidal closed
(e.g.,~\citet{dialectica}).  Connecting our work to these lens categories could
suggest how to support higher-order functions in our framework.

%%%%%
\chapter{Conclusion}

This chapter concludes our investigation of designing languages 
that unify the tasks of writing numerical programs and 
reasoning about their accuracy.  

In \Cref{chapter:nfuzz}, we demonstrated that it is possible to 
design a language for forward error analysis with \Lang{}, a
higher-order functional programming language with a linear type system that 
can express quantitative bounds on forward error. \Lang{} combines a 
sensitivity type system with a novel graded monad to track the forward relative 
error of programs. We proposed a metric semantics for \Lang{} and 
proved a central soundness theorem using this semantics, which 
guarantees that well-typed programs with graded monadic type adhere to 
the relative error bound specified by their type. A prototype implementation of 
\Lang{} illustrates that type-based approaches to rounding error analysis
can provide strong theoretical guarantees while being useful practice, 
narrowing the gap between rigorous analysis and realistic numerical computations.  

In \Cref{chapter:bean}, we presented \bean{}, a programming language for
\textsc{\scalefont{1.25}{\textbf{b}}}ackward
\textsc{\scalefont{1.25}{\textbf{e}}}rror
\textsc{\scalefont{1.25}{\textbf{an}}}alysis. As the first static analysis tool 
for backward error analysis, \bean{} demonstrates the feasibility of using 
static analysis techniques to automatically infer backward 
error bounds bounds and ensure the backward stability of programs. A key
insight from our development of \bean{} is that the composition of two backward 
stable programs remains backward stable \emph{provided they do not assign
backward error to shared variables}. Thus, to ensure 
that \bean{} programs satisfy a backward stability guarantee, 
\bean{} employs a \emph{strict} linear type system that restricts the
duplication of variables. 

We established the soundness of \bean{} by developing the \emph{category of
backward error lenses}---\Bel{}---and interpreting \bean{} programs 
as morphisms in this category. This semantic foundation 
distinguishes \bean{} from \emph{Fuzz}-like languages, including \Lang{},
which interpret programs in the category of metric spaces. Using \Bel{}, 
we formulated and proved a soundness theorem for \bean, which guarantees 
that well-typed programs have bounded backward error. A prototype implementation 
of \bean{} demonstrates that our approach can be used to infer accurate 
per-variable backward error bounds.

Topics for further work on \bean{} and \Lang{} are plentiful, and are discussed in 
\Cref{sec:nfuzz:conclusion} and \Cref{sec:bean:conclusion}.

\bibliographystyle{ACM-Reference-Format}
\bibliography{main}

\appendix
\chapter{Appendix for \Lang{}}\label{app:nfuzz}

\section{Termination (Strong Normalization)}\label{app:nfuzz:termination}
This appendix gives the proof of \Cref{thm:SN}, which verifies that \Lang{} 
is strongly normalizing using a logical relations argument. In addition to 
\Cref{lem:SN_aux1} and \Cref{thm:subsumption}, the proof relies on the following 
lemma. 
\begin{lemma}\label{lem:SN_aux2}
Let $\Gamma \vdash e : \monad{u}\tau$ be a well-typed term, and let $\vec{v}_1,
\vec{v}_2$ be well-typed substitutions of closed values such that $\vec{v}_1,
\vec{v}_2 \vDash \Gamma$. If $e[\vec{v}_1/dom(\Gamma)]\in
\mathcal{VR}^m_{\monad{u}\tau}$ for some $m\in\mathbb{N}$ and
$e[\vec{v}_2/dom(\Gamma)]\in \mathcal{R}_{\monad{u}\tau}$, then
$e[\vec{v}_2/dom(\Gamma)]\in \mathcal{VR}^m_{\monad{u}\tau}$.
\end{lemma}

\begin{namedtheorem}[\Cref{thm:SN}]
If $\emptyset \vdash e : \tau$ then there exists $v \in CV(\tau)$ such that $e
\mapsto^* v$.  
\end{namedtheorem}

\begin{proof} 
We first prove the following stronger statement. Let $\Gamma \triangleq x_1
:_{s_1} \sigma_1, \dots, x_i :_{s_i} \sigma_i$ be a typing environment and let
$\vec{w}$ denote the values $\vec{w} \triangleq w_1, \dots, w_i$. If {$\Gamma
\vdash e : \tau$} and {$w_i \in \mathcal{R}_{\Gamma(x_i)}$} for every {$x_i \in
dom(\Gamma)$}, then {$e[\vec{w}/dom(\Gamma)] \in \mathcal{R}_\tau$}. The proof
follows by induction on the derivation $\Gamma \vdash e : \tau$. We consider
the monadic cases, as the non-monadic cases are standard. The base cases
(Const), (Ret), and (Rnd) follow by definition, and (MSub) follows by
\Cref{thm:subsumption}. The case for (MLet) requires some detail. The rule is
\begin{center} 
\vskip 1em \noLine 
\AXC{(MLet)} \kernHyps{-2em} \UIC{}
\noLine \UIC{$\Gamma \vdash v : \monad{r} \sigma$}
\AXC{$\Theta, x:_{s} \sigma \vdash
f: \monad{q} \tau$} 
\BIC{$s \cdot \Gamma + \Theta \vdash \letm{x}{v}{f}: \monad{s\cdot r+q}\tau$} 
\DisplayProof 
\vskip 1em 
\end{center} 
and so we are required to show ($\letm{x}{v}{f}) \in \mathcal{R}_{\monad{s\cdot
r+q}\tau}$.  We proceed by cases on $v$. 

\begin{description}
\item[\textbf{Subcase:} $v \equiv \rnd \ k$ with $k \in R$.]
Let $\Delta =
s \cdot \Gamma + \Phi$.
By \Cref{thm:substitution}, 
\[
(\letm{x}{\rnd \ k}{f})[\vec{w}/dom(\Delta)] \in CV(\monad{s\cdot r + q} \tau)
\]
and it remains to be shown that 
\[ 
(\letm{x}{\rnd \ k}{f})[\vec{w}/dom(\Delta)] \in 
	\mathcal{VR}_{\monad{s\cdot r+q}\tau}.
\]
By definition of the logical relation, 
$\rnd \ k \in \mathcal{VR}^0_{\monad{s}\sigma}$.  
  Given $\vec{w} \in \mathcal{R}_{\Delta(x_i)}$ for
every $x_i \in \Delta$, we have the vector of values $\vec{w}'$ such that $w'_i \in
\mathcal{R}_{\Phi(y_i)}$ for every $y_i \in dom(\Phi)$. 
By the induction hypothesis and \Cref{lem:SN_aux1} we have 
\[ 
f[\vec{w}'/dom(\Phi)][k/x] \in \mathcal{VR}^n_{\monad{q}\tau}
\]  
for some $n \in \mathbb{N}$. If $n = 0$, then the conclusion follows trivially.
Otherwise, $n > 0$, and we need to show
that, given some $t \in \mathcal{VR}_\num$, 
\[ 
f[\vec{w}'/dom(\Phi)][t/x] \in \mathcal{VR}^n_{\monad{q}\tau}
\]
This follows by the induction hypothesis and \Cref{lem:SN_aux2}.
\item[\textbf{Subcase:} $v \equiv \ret~ v'$ with $v' \in \mathcal{VR}_{\sigma}$.]
The conclusion follows by applying reasoning identical to that of the previous 
subcase.
\item[\textbf{Subcase:} $v \equiv \letm{y}{\rnd~ k}{g}$.] 
From the reduction rules we have
\[
	(\letm{x}{v}{f})[\vec{w}/dom(\Delta)] \mapsto 
	(\letm{y}{\rnd \ k}{(\letm{x}{g}{f})})[\vec{w}/dom(\Delta)]
\] with $y \notin FV(f)$. 
By \Cref{thm:substitution} we have 
\[
	\emptyset \vdash (\letm{x}{(\letm{y}{\rnd~ k}{g})}{f})[\vec{w}/dom(\Delta)]: 
	\monad{s \cdot r + q} \tau
\]
and, by \Cref{lem:SN_aux1}, it is therefore sufficient to show  
\[
	(\letm{y}{\rnd \ k}{(\letm{x}{g}{f})})[\vec{w}/dom(\Delta)] 
	\in \mathcal{R}_{\monad{s\cdot r+q}}\tau,
\]
which follows by applying reasoning identical to that of the previous 
two subcases.
\end{description}
\end{proof}

\begin{lemma} \label{thm:SN1}
If $\emptyset \vdash e : \tau$ then $e \in \mathcal{R}_\tau$.
\end{lemma}

The following lemma follows directly from the definition of the 
reducibility predicate.
\begin{lemma} 
If $e \in \mathcal{R}_\tau$ then there exists a $v \in CV(\tau)$ such that $e
\mapsto^* v$. \label{thm:SN2} 
\end{lemma}
The proof of termination (\Cref{thm:SN}) then follows from \Cref{thm:SN2} and
\Cref{thm:SN1}.

\section{The Neighborhood Monad}\label{app:nfuzz:monad}
This appendix provides the details 
verifying that the neighborhood monad
defined in \Cref{sec:nfuzz:monad}
forms an $\mathcal{R}$-strong graded monad on $\Met$.  

\begin{namedtheorem}[\Cref{lem:monad-nat}]
  Let $q, r \in \mathcal{R}$. For any metric space $A$, the maps $(q \leq r)_A$,
  $\eta_A$, and $\mu_{q, r, A}$ are non-expansive maps and natural in $A$.
\end{namedtheorem}

\begin{proof}
  Non-expansiveness and naturality for the subeffecting maps $(q \leq r)_A : T_q
  A \to T_r A$ and the unit maps $\eta_A : A \to T_0 A$ are straightforward. We
  describe the checks for the multiplication map $\mu_{q, r, A}$. 
  
  First, we check that the multiplication map has the claimed domain and
  codomain.  Note that $d_A(x, x') = d_{T_r A} ( (x, y), (x', y') ) \leq q$
  because of the definition of $T_q$, and $d_A(x', y') \leq r$ because of the
  definition of $T_r$, so via the triangle inequality we have $d_A(x, y') \leq r
  + q$ as claimed.

  Second, we check non-expansiveness. Let $((x, y),(x', y'))$ and $((w, z),
  (w', z'))$ be two elements of $T_q (T_r A)$. Then:
  \begin{align}
    d_{T_{r + q} A}( \mu( (x, y), (x', y') ) , \mu( (w, z), (w', z') ) )
    &= d_{T_{r + q} A}( (x, y'), (w, z') ) \tag{def. $\mu$} \\
    &= d_A(x, w) \tag{def. $d_{T_{r + q} A}$} \\
    &= d_{T_r A} ( (x, y), (w, z) ) \tag{def. $d_{T_r} A$} \\
    &= d_{T_q (T_r A)} ( ((x, y), (x', y')), ((w, z), (w', z')) ) \tag{def. $d_{T_q (T_r A)}$}
  \end{align}
  
  Finally, we can check naturality. Let $f : A \to B$ be any non-expansive map.
  By unfolding definitions, it is straightforward to see that $\mu_{q, r, B}
  \circ T_q (T_r f) = T_{r + q} f \circ \mu_{q, r, A}$.
\end{proof}

\begin{namedtheorem}[\Cref{lemma:monad_strong}]
  The neighborhood monad (\Cref{def:nhd-monad}) together with 
  the \emph{tensorial strength maps} 
  $st_{r, A, B} : A \otimes T_r B \to T_r (A \otimes B)$ 
  defined as
  \begin{align*}
    st_{r, A, B}(a, (b, b')) &\triangleq ((a, b), (a, b'))
  \end{align*}
  for every $r \in \mathcal{R}$ form a 
  $\mathcal{R}$-strong graded monad on $\Met$.
\end{namedtheorem}

\begin{proof}
We must verify the non-expansiveness (\Cref{def:non_expansive}) and 
naturality (\Cref{def:nat_trans}) of the 
tensorial strength map. We check the non-expansiveness
with respect to the tensor product $\otimes$,
although it is also easy to show 
non-expansiveness
with respect to the product $\&$:
For $(a, (b, b'))$ and $(c, (d, d'))$ in $A
\otimes T_r B$, we have:
\begin{align*}
  d_{T_r (A \otimes B)} (st (a, (b, b')), st (c, (d, d')))
  &= d_{T_r (A \otimes B)} (((a, b), (a, b')), ((c, d), (c, d')))
  \tag{def.  $st$} \\
  &= d_{A \otimes B} ((a, b), (c, d))
  \tag{def. $d_{T_r (A \otimes B)}$}
  \\
  &= d_A (a, c) + d_B (b, d)
  \tag{def. $d_{A \otimes B}$} \\
  &= d_A (a, c) + d_{T_r B} ((b, b'), (d, d'))
  \tag{def. $d_{T_r B}$} \\
  &= d_{A \otimes T_r B}  ((a, (b, b')), (c, (d, d')))
  \tag{def. $d_{A \otimes T_r B}$}
\end{align*}

Finally, checking the naturality of the strength 
follows directly by unfolding 
definitions.
%Furthermore, it is possible to show that these maps satisfy the commutative
%diagrams needed to make $T_r$ a graded strong
%monad~\citep{DBLP:conf/popl/Katsumata14,DBLP:conf/fossacs/FujiiKM16}, though we
%will not need this fact for our development. 
\end{proof}

\begin{namedtheorem}[\Cref{lem:distr}]
  Let $s \in \mathcal{S}$ and $r \in \mathcal{R}$ be grades, and let $A$ be a
  metric space. Then identity map on the carrier set $|A| \times |A|$ is a
  non-expansive map
  \[
    \lambda_{s, r, A} : D_s (T_r A) \to T_{s \cdot r} (D_s A)
  \]
  Moreover, these maps are natural in $A$.
\end{namedtheorem}

\begin{proof}
  We first check the domain and codomain. Let $x, y \in A$ be such that $(x, y)$
  is in the domain $D_s(T_r A)$ of the map. Thus $(x, y)$ must also be in $T_r
  A$, and satisfy $d_A(x, y) \leq r$ by definition of $T_r$.  To show that this
  element is also in the range, we need to show that $d_{D_s A} (x, y) \leq s
  \cdot r$, but this holds by definition of $D_s$.
  We can also check that this map is non-expansive:
  \begin{align*}
    d_{T_{s \cdot r} (D_s A)} ((x, y), (x', y'))
  &\triangleq d_{D_s A} (x, x')
  \tag{def. $T_{s \cdot r}$} \\
  &\triangleq s \cdot d_A(x, x')
  \tag{def. $D_s$} \\
  &\triangleq s \cdot d_{T_r A}((x, y), (x', y'))
  \tag{def. $T_r$} \\
  &\triangleq d_{D_s(T_r A)}((x, y), (x', y'))
  \tag{def. $D_s$}
  \end{align*}
  Since $\lambda_{s, r, A}$ is the identity map on the underlying set $|A|
  \times |A|$, it is evidently natural in $A$.
\end{proof}

\section{Interpreting \texorpdfstring{\Lang}{} Terms}\label{sec:app_interp_nfuzz}
This appendix provides the constructions of the interpretation of
\Lang{} terms for \Cref{def:interp-prog} that were not included
in \Cref{sec:nfuzz:semantics}.  

To reduce
notation, we elide the  
the unitors $\lambda_A : I \otimes A \to A$
and $\rho_A: A \otimes I \to I$; the associators $\alpha_{A, B, C}: (A \otimes
B) \otimes C \to A \otimes (B \otimes C)$; and the symmetries 
$\sigma_{A, B}: A \otimes B \to B \otimes A$. 
\begin{description}
% UNIT
\item[(\textbf{Unit}).]
Define $\denot{\Gamma \vdash (): \unit}$ as the map that sends all points in 
$\Gamma$ to $\star \in \denot{\unit}$.
% VAR
\item[(\textbf{Var}).] We define $\denot{\Gamma \vdash x : \tau}$ to be the map
that maps $\denot{\Gamma}$ to the $x$-th component $\denot{\tau}$. All other
components are mapped to $I$ and then removed with the unitor.
% ABS
\item[(\textbf{Abs}).] Let $f = \denot{\Gamma, x :_1 \sigma \vdash e :
\tau} : \denot{\Gamma} \otimes D_1 \denot{\sigma} \to \denot{\tau}$. 
Define 
\[
\denot{\Gamma \vdash \lambda x. e : \sigma \lin \tau} \triangleq \lambda(f)
\]
The map $\lambda(f) : \denot{\Gamma} \to (\denot{\sigma} \lin \denot{\tau})$
is guaranteed by the closed symmetric monoidal structure of $\Met$ 
(\Cref{thm:smcc}). The equality $D_1 \denot{\sigma} = \denot{\sigma}$ 
follows by definition of the comonad.
% APP
\item[(\textbf{App}).] Let $f = \denot{\Gamma \vdash v : \sigma
\lin \tau}$ and $g = \denot{\Theta \vdash w : \sigma}$. Define
\[
\denot{\Gamma + \Theta \vdash v w : \tau} \triangleq
c_{\denot{\Gamma}, \denot{\Theta}} ; (f \otimes g) ; ev
: \denot{\Gamma + \Theta} \to \denot{\tau}
\]
The map $ev: (A \lin B) \otimes A \to B$
is guaranteed by the closed symmetric monoidal structure of $\Met$ 
(\Cref{thm:smcc}). 
% PROD INTRO
\item[($\&$ \textbf{I}).] Let 
$f = \denot{\Gamma \vdash v: \sigma}$ and 
$ g = \denot{\Gamma \vdash w: \tau}$. Then, define: 
\[
\denot{\Gamma \vdash \langle v,w\rangle : \sigma \tand \tau} \triangleq 
\langle f, g \rangle
\]
% PROD ELIM
\item[($\&$ \textbf{E}).] Let 
$f = \denot{\Gamma \vdash v: \tau_1 \tand t_2}$ 
be the denotation of the premise. Define 
\[
\denot{\Gamma \vdash \pi_i ~ v : \tau_i} \triangleq  f;\pi_i
\]
% TENS INTRO
\item[($\otimes$ \textbf{I}).] Let 
$f = \denot{\Gamma \vdash v : \sigma}$ and 
$g = \denot{\Theta \vdash w : \tau}$ be the denotations of the
premises. Then, define: 
\[
\denot{\Gamma + \Theta \vdash (v,w): \sigma \otimes \tau} \triangleq
c_{\denot{\Gamma}, \denot{\Theta}} ; (f \otimes g)
:\denot{\Gamma + \Theta} \to \denot{\sigma \otimes \tau}
\]
% TENS ELIM
\item[($\otimes$ \textbf{E}).] Let the denotations for the premises be:
\begin{align*}
f &= \denot{\Gamma \vdash v : \sigma
\otimes \tau} : \denot{\Gamma} \to \denot{\sigma} \otimes \denot{\tau} \\
g &= \denot{\Theta, x :_s \sigma, y :_s \tau \vdash e : \rho} :
\denot{\Theta} \otimes D_s \denot{\sigma} \otimes D_s \denot{\tau} \to
\denot{\rho} 
\end{align*}
Then, define:  
\[
\denot{s \cdot \Gamma + \Theta \vdash \lett{(x,y)}{v}{e}: \rho} \triangleq
c_{\denot{s \cdot \Gamma}, \denot{\Theta}}; h ; g 
\]
The map $h: \denot{s \cdot \Gamma} \otimes \denot{\Theta}
\to \denot{\Theta} \otimes D_s \denot{\sigma} \otimes D_s \denot{\tau}$ 
is constructed as follows. Applying the functor $D_s$ to $f$ and 
pre-composing with $m$ yields
\[
m;D_s f: \denot{s \cdot \Gamma} \to D_s (\denot{\sigma} \otimes \denot{\tau})
\]
Since the map $m$ is the identity, we can post-compose by its
inverse to get:
\[
m;D_s f ; m^{-1}: \denot{s \cdot \Gamma} \to D_s (\denot{\sigma}) 
\otimes D_s (\denot{\tau})
\]
Composing in parallel with $id_{\denot{\Theta}}$, we get:
\[
h = id_{\denot{\Theta}} \otimes (m ; D_s f ; m^{-1}): 
\denot{\Theta} \otimes \denot{s \cdot \Gamma}
\to \denot{\Theta} \otimes D_s \denot{\sigma} \otimes D_s \denot{\tau}
\]
% SUM INTRO L 
\item[($+$~\textbf{I$_L$}).] Let $f = \denot{\Gamma \vdash v: \sigma}$ 
be the denotation of the premise. Then, define: 
\[
\denot{\Gamma \vdash \inl v: \sigma + \tau} \triangleq f ; \iota_1
\]
where $\iota_1$ is the first injection into the coproduct.
% SUM INTRO R
\item[($+$~\textbf{I$_R$}).] Let $f = \denot{\Gamma \vdash v: \tau }$ 
be the denotation of the premise. Then, define: 
\[
\denot{\Gamma \vdash \inr v: \sigma + \tau} \triangleq f ; \iota_r
\]
where $\iota_2$ is the second injection into the coproduct.
% SUM ELIM
\item[($+$~\textbf{E}).] This particular case requires as few additional 
facts about the structures in our category. First, when $s > 0$, there 
is an isomorphism 
\[
\text{dist}^D_{s} : D_s (A + B) \cong D_s A + D_s B
\]  
Second, there is a map 
\[
\text{dist}_{A,B,C} : A \times (B + C) \to A \times B + A \times C
\]
that pushes the first component into the disjoint union. 
This map is non-expansive, and in fact
$\Met$ is a distributive category.

Now, let $f = \denot{\Gamma \vdash v : \sigma + \tau}$ and 
$g_i = \denot{\Theta, x_i :_s \sigma \vdash e_i : \rho}$ for $i = 1, 2$
and $s > 0$.  
Then, define 
\[
\denot{s \cdot \Gamma + \Theta \vdash \case{v}{x.e}{y.f} : \rho}
\triangleq c_{\denot{s \cdot \Gamma}, \denot{\Theta}};h;[g_1, g_2]
\]
The map $h$ is constructed as follows. 
Since $s > 0$, $\text{dist}_{s}$ is an isomorphism. Using
the functor $D_{s}$ on $f$, composing in parallel with
$id_{\denot{\Theta}}$ and distributing, we have:
\[
h = (id_{\denot{\Theta}} \otimes (m; D_{s} f ; 
\text{dist}^D_{s})) ; \text{dist}_{\denot{\Theta},\denot{\sigma},\denot{\tau}}
: \denot{\Theta} \otimes \denot{s\cdot \Gamma}
\to \denot{\Theta} \otimes D_{s} \denot{\sigma}
+ \denot{\Theta} \otimes D_{s} \denot{\tau}
\]
By post-composing with the pairing map $[g_1, g_2]$ from the coproduct, and
pre-composing with $c_{\denot{s \cdot \Gamma}, \denot{\Theta}}$ (and symmetry maps),
we get a map $\denot{s \cdot \Gamma + \Theta} \to \denot{\rho}$ as desired.

% BOX INTRO
\item[($!$~\textbf{I}).] Let $f = \denot{\Gamma \vdash v : \sigma}$ 
be the denotation of the premise. Then, define:  
\[ 
\denot{s \cdot \Gamma \vdash \boxx{v} : \bang{s} : \sigma} \triangleq m ; D_s f
\]
% BOX ELIM
\item[($!$~\textbf{E}).] Let $f = \denot{\Gamma \vdash v : \bang{s}
\sigma}$ and $g = \denot{\Theta, x :_{m \cdot n} \sigma \vdash e : \tau}$.
Then, define:
\[ \denot{t \cdot \Gamma + \Theta \vdash \lett{\boxx{x}}{v}{e} : \tau}
\triangleq m ; D_m f ; \delta_{m, n, \denot{\sigma}}^{-1}: 
\denot{m \cdot \Gamma} \to D_{m \cdot n} \denot{\sigma}
\]
Here we use the fact that $\delta_{m, n, \denot{\sigma}}$ is an
isomorphism in our model. By composing in parallel with
$id_{\denot{\Theta}}$, we can then post-compose by $g$. Pre-composing with
$c_{\denot{s \cdot \Gamma}, \denot{\Theta}}$ gives a map $\denot{s
\cdot \Gamma + \Theta} \to \denot{\tau}$, as desired.
% LET
\item[(\textbf{Let}).] 
Let $f = \Gamma \vdash e : \tau$ and let $g = \denot{\Theta, x:_s \tau \vdash f
:\sigma}$. Then, define:
\[
\denot{s \cdot \Gamma + \Theta \vdash \lett{x}{e}{f}: \sigma} \triangleq
c_{\denot{s \cdot \Gamma},\denot{\Theta}};h;g
\] 
Similar to the other elimination cases, the map $h: \denot{s \cdot \Gamma}
\otimes \denot{\Theta} \to \denot{\Theta} \otimes D_s\denot{\tau}$ is
constructed as follows. Applying the functor $D_s$ to $f$, pre-composing
with $m$, and composing in parallel with $\id_{\denot{\Theta}}$ yields
\[
  h = (m;D_s f) \otimes \id_{\denot{\Theta}}: \denot{\Theta} \otimes \denot{s
  \cdot \Gamma} \to D_s (\denot{tau}) \otimes \denot{\Theta} 
\]
\end{description}

\section{Denotational Semantics}\label{app:nfuzz:semantics}
This appendix provides 
basic lemmas about the denotational semantics: weakening 
(\Cref{lem:weak-sem}), subsumption (\Cref{lem:subsump}), and  
substitution (\Cref{lem:subst-sem}). It also includes 
a computational soundness lemma (\Cref{lem:pres-sem}) showing that our 
metric interpretation of \Lang{} terms respect the operational 
semantics given in \Cref{fig:eval_rules}; these are the semantics defined
prior to the refinement into a ideal and floating-point step relations. 

\begin{lemma}[Weakening] \label{lem:weak-sem}
  Let $\Gamma, \Gamma' \vdash e : \tau$ be a well-typed term. Then for any
  context, there is a derivation of $\Gamma, \Delta, \Gamma' \vdash e : \tau$
  with semantics $\denot{\Gamma, \Gamma' \vdash e : \tau} \circ \pi$, where $\pi
  : \denot{\Gamma, \Delta, \Gamma'} \to \denot{\Gamma, \Gamma'}$ projects the
  components in $\Gamma$ and $\Gamma'$.
\end{lemma}

\begin{proof}
  By induction on the typing derivation of $\Gamma, \Gamma' \vdash e: \tau$.
\end{proof}

\begin{lemma}[Subsumption] \label{lem:subsump}
  Let $\Gamma \vdash e : M_r \tau$ be a well-typed program of monadic type,
  where the typing derivation concludes with the subsumption rule. Then either
  $e$ is of the form $\ret v$ or $\rnd~k$, or there is a
  derivation of $\Gamma \vdash e : M_r \tau$ with the same semantics that
  does not conclude with the subsumption rule.
\end{lemma}

\begin{proof}
  By straightforward induction on the typing derivation, using the fact that
  subsumption is transitive, and the semantics of the subsumption rule leaves
  the semantics of the premise unchanged since the subsumption map $(r \leq
  s)_A$ is the identity function.
\end{proof}

\begin{lemma}[Substitution] \label{lem:subst-sem}
Let $\Gamma, \Delta, \Gamma' \vdash e : \tau$ be a well-typed term, and let
$\vec{v} : \Delta$ be a well-typed substitution of closed values, i.e., we
have derivations $\emptyset \vdash v_x : \Delta(x)$. Then there is a derivation of
\[
\Gamma, \Gamma' \vdash e[\vec{v}/dom(\Delta)] : \tau
\]
with semantics
$
\denot{\Gamma, \Gamma' \vdash e[\vec{v}/dom(\Delta)] : \tau}
= (id_{\denot{\Gamma}} \otimes \denot{\emptyset \vdash \vec{v} : \Delta} 
\otimes id_{\denot{\Gamma'}})
; \denot{\Gamma, \Delta, \Gamma' \vdash e : \tau} .
$
\end{lemma}

\begin{proof}
By induction on the typing derivation of $\Gamma, \Delta, \Gamma' \vdash e :
\tau$. The base cases Unit and Const are obvious. The other base case Var
follows by unfolding the definition of the semantics. Most of the rest of the
cases follow from the substitution lemma for \emph{Fuzz}~\citep[Lemma
3.3]{Amorim:2017:metric}. We show the cases for \textbf{Rnd},
\textbf{Ret}, and \textbf{MLet}, which differ from \emph{Fuzz}. We omit the
bookkeeping morphisms.

\begin{description}
\item[Case \textbf{Rnd}.] Given a derivation $f = \denot{\Gamma, \Delta,
\Gamma' \vdash w : \num}$, by induction, there is a derivation
$\denot{\Gamma, \Gamma' \vdash w[\vec{v}/\Delta] : M_{\rnderr} \num} =
(id_{\denot{\Gamma}} \otimes \denot{\emptyset \vdash \vec{v} : \Delta} \otimes
id_{\denot{\Gamma'}}) ; f$. By applying rule \textbf{Round} and by
definition of the semantics of this rule, we have a derivation
\[
\denot{\Gamma, \Gamma' \vdash (\rnd w)[\vec{v}/\Delta] : M_{\rnderr} \num}
= (id_{\denot{\Gamma}} \otimes \denot{\emptyset \vdash \vec{v} : \Delta} \otimes id_{\denot{\Gamma'}}) ; f ; \langle id, \rho \rangle .
\]
We are done since $\denot{\Gamma, \Delta, \Gamma' \vdash \rnd~w :
M_{\rnderr} \num} = f ; \langle id, \rho \rangle$.
\item[Case \textbf{Ret}.] Same as previous, using the unit of the monad
$\eta_{\denot{\tau}}$ in place of $\langle id, \rho \rangle$.
\item[Case \textbf{MLet}.] Suppose that $\Gamma = \Gamma_1, \Delta_1,
\Gamma_2$ and $\Theta = \Theta_1, \Delta_2, \Theta_2$ such that $\Delta =
s \cdot \Delta_1 + \Delta_2$. By combining \Cref{lem:ctx-contr} and
\Cref{lem:ctx-scale}, there is a natural transformation $\sigma : \denot{s
\cdot \Gamma + \Theta} \to \denot{\Theta} \otimes D_s \denot{\Gamma}$.

Let $g_1 = \denot{\Gamma_1, \Delta_1, \Gamma_2 \vdash w : M_r \sigma}$ and
$g_2 = \denot{\Theta_11, \Delta_2, \Theta_2, x :_s \sigma \vdash f : M_{r'} \tau}$. By
induction, we have:
\begin{align*}
\tilde{g_1} &= \denot{\Gamma_1, \Gamma_2 \vdash w[\vec{v}/\Delta] : M_r \sigma}
= (id_{\denot{\Gamma_1}} \otimes \denot{\emptyset \vdash \vec{v} : \Delta}
\otimes id_{\denot{\Gamma_2}}) ; g_1 \\
\tilde{g_2} &= \denot{\Theta_1, \Theta_2, x :_s \sigma \vdash f[\vec{v}/\Delta] : M_{r'} \tau}
= (id_{\denot{\Theta_1}} \otimes \denot{\emptyset \vdash \vec{v} : \Delta}
\otimes id_{\denot{\Theta_2, x :_s \sigma}}) ; g_2
\end{align*}
Thus we have a derivation of the judgment $s \cdot (\Gamma_1, \Gamma_2) +
(\Theta_1, \Theta_2) \vdash \letm{x}{w}{f}[\vec{v} / \Delta] : M_{s
\cdot r + r'} \tau$, and by the definition of the semantics of
\textbf{MLet}, its semantics is:
\[
split
;
(id_{\denot{\Theta_1, \Theta_2}} \otimes (D_s \tilde{g_1} ; \lambda_{s, r, \denot{\sigma}}))
;
st_{\denot{\Theta_1, \Theta_2}, \denot{\sigma}}
;
T_{s \cdot r} \tilde{g_2}
;
\mu_{s \cdot r, r', \denot{\tau}}
\]
From here, we can conclude by showing that the first morphisms in
$\tilde{g_1}$ and $\tilde{g_2}$ can be pulled out to the front. For
instance,
\[
D_s \tilde{g_1} = (id_{\denot{s \cdot \Gamma_1}}
\otimes D_s \denot{\emptyset \vdash \vec{v} : \Delta}
\otimes id_{\denot{s \cdot \Gamma_2}})
; D_s g_1
\]
by functoriality. By naturality of $split$, the first morphism can be
pulled out in front of $split$.

Similarly, for $\tilde{g_2}$, we have:
\begin{multline*}
st_{\denot{\Theta_1, \Theta_2}, \denot{\sigma}} ;
T_{s \cdot r} (id_{\denot{\Theta_1}} \otimes \denot{\emptyset \vdash \vec{v} : \Delta} \otimes id_{\denot{\Theta_2, x :_s \sigma}})
\\
= (id_{\denot{\Theta_1}} \otimes \denot{\emptyset \vdash \vec{v} : \Delta}
\otimes id_{\denot{\Theta_2}} \otimes id_{T_{s \cdot r} D_s \denot{\sigma}})
;
st_{\denot{\Theta_1, \Delta_2, \Theta_2}, \denot{\sigma}}
\end{multline*}
by naturality of strength. By naturality, we can pull the first
morphism out in front of $split$.
\end{description}
\end{proof}

\begin{lemma}[Computational Soundness (Metric Semantics)]\label{lem:pres-sem}
Let $\emptyset \vdash e : \tau$ be a well-typed closed term, and suppose $e \mapsto e'$.
Then there is a derivation of $\emptyset \vdash e' : \tau$, and the semantics of both
derivations are equal: \[\denot{\vdash e : \tau} = \denot{\vdash e' : \tau}.\]
\end{lemma}

\begin{proof}
By case analysis on the step rule, using the fact that $e$ is well-typed. For
the beta-reduction steps for programs of non-monadic type, preservation
follows by the soundness theorem; these cases are exactly the same as in \emph{\emph{\emph{Fuzz}}}
\citep{Amorim:2017:metric}. %

The two step rules for programs of monadic type are new. It is possible to
show soundness by appealing to properties of the graded monad $T_r$, but
we can also show soundness more concretely by unfolding definitions and
considering the underlying maps. 
\begin{description}
\item[MLet $q$.] Suppose that $e = \letm{x}{\ret v}{f}$ is a well-typed program with type $M_{s \cdot r + q}
\tau$. Since subsumption is admissible~(\Cref{lem:subsump}), we may assume
that the last rule is \textbf{MLet} and we have derivations $\emptyset \vdash
\ret v : M_r \sigma$ and $x :_s \sigma \vdash f : M_{q}
\tau$. By definition, the semantics of $\emptyset \vdash e : M_{s \cdot r +
q} \tau$ is given by the composition:
\[\begin{tikzcd}
I & {D_s I} && {D_s T_r \sigma} & {T_{s \cdot r} D_s \sigma} & {T_{s \cdot r} T_q \tau} & {T_{s \cdot r + q} \tau}
\arrow["{\lambda_{s, r, \sigma}}", from=1-4, to=1-5]
\arrow[from=1-1, to=1-2]
\arrow["{T_{s \cdot r} f}", from=1-5, to=1-6]
\arrow["{\mu_{s \cdot r, q, \tau}}", from=1-6, to=1-7]
\arrow["{D_s (v; \eta_\sigma; (0 \leq r)_\sigma)}", from=1-2, to=1-4]
\end{tikzcd}\]
By substitution~(\Cref{lem:subst-sem}), we have a derivation of $\cdot
\vdash f[v/x] : M_{q} \tau$. By applying the subsumption rule, we have a
derivation of $\emptyset \vdash f[v/x] : M_{s \cdot r + q} \tau$ with
semantics:
\[\begin{tikzcd}
I & {D_s I} & {D_s \sigma} & {T_q \tau} & {T_{s \cdot r + q} \tau}
\arrow["{\lambda_{q, s \cdot r + q, \tau}}", from=1-4, to=1-5]
\arrow["f", from=1-3, to=1-4]
\arrow["{D_s v}", from=1-2, to=1-3]
\arrow[from=1-1, to=1-2]
\end{tikzcd}\]
Noting that the underlying maps of $s$ and $\mu$ are the identity function,
both compositions have the same underlying maps, and hence are equal
morphisms.
\item[MLet Assoc.] Suppose that $e =
\letm{y}{ \letm{x}{\rnd~k}{f}}{g} $ is a
well-typed program with type $M_{s \cdot r + q} \tau$. Since
subsumption is admissible~(\Cref{lem:subsump}), we have derivations:
\[
\vdash \rnd~k : M_{r_1} \num
\qquad\qquad
x :_t \num \vdash f : M_{r_2} \sigma
\qquad\qquad
y :_s \sigma \vdash g : M_{q} \tau
\]
such that $t \cdot r_1 + r_2 = r$. By applying
\textbf{MLet} on the latter two derivations, we have:
\[
x :_{s \cdot t} \num \vdash \letm(f, y. g) : M_{s \cdot r_2 + q} \tau
\]
And by applying \textbf{MLet} again, we have:
\[
\vdash \letm(\rnd~k, x.  \letm(f, y. g))
: M_{s \cdot t \cdot r_1 + s \cdot r_2 + q} \tau
\]
This type is precisely $M_{s \cdot r + q} \tau$. The semantics of $e$ and
$e'$ have the same underlying maps, and hence are equal morphisms.
\qedhere
\end{description}
\end{proof}

\chapter{Appendix for \texorpdfstring{\bea}{}}\label{app:bean}

\section{The Category of Backward Error Lenses}\label{app:check_bel}
This appendix verifies that the composition in the  category \Bel{} of
backward error lenses (\Cref{def:lensC}) is well-defined.
We first restate the definition: the composition 
\[
  {{(f_2,\tilde{f}_2,b_2) \circ (f_1,\tilde{f}_1,b_1)}}
\]
of error lenses ${(f_1,\tilde{f}_1,b_1) : X \rightarrow Y}$
and ${(f_2,\tilde{f}_2,b_2) : Y \rightarrow Z}$ is the error lens
${(f,\tilde{f},b): X \rightarrow Z}$ defined by
\begin{itemize}
\item the forward map 
  \begin{equation} 
    f: x \mapsto (f_1;f_2) \ x (\tag*{\Cref{eq:fcomp}}) 
  \end{equation}
\item the approximation map
  \begin{equation}
    \tilde{f} : x \mapsto (\tilde{f}_1;\tilde{f}_2) \ x   
    (\tag*{\Cref{eq:acomp}})
  \end{equation}
\item the backward map
  \begin{equation} 
    b: (x,z) \mapsto b_1\bigl(x, b_2(\tilde{f}_1(x), z)\bigr) 
    (\tag*{\Cref{eq:bcomp}})
  \end{equation}
\end{itemize}

\begin{figure}
\[\begin{tikzcd}
	{X\times Y \times Z} && {X \times Y} \\
	\\
	{X \times Z} && X
	\arrow["{\langle id_X, \tilde{f}_{1} \rangle ~ \times ~ id_Z}", 
    from=3-1, to=1-1]
	\arrow["{id_X \times  ~b_2}", from=1-1, to=1-3]
	\arrow["{b}"', from=3-1, to=3-3]
	\arrow["{b_{1}}", from=1-3, to=3-3]
\end{tikzcd}\]
\caption{The backward map $b$ for the composition 
${{(f_2,\tilde{f}_2,b_2) \circ (f_1,\tilde{f}_1,b_1)}}$.}
\label{fig:compose}
\end{figure}
\noindent The diagram for the backward map for the composition of error lenses
is given in \Cref{fig:compose}.  

Let $L_1 = (f_1,\tilde{f}_1,b_1)$ and let $L_2 = (f_2,\tilde{f}_2,b_2)$.  We
first check the domain: for all $x \in X$ and $z \in Z$, and assuming
$d_Z\left(\tilde{f}(x),z\right) \neq \infty$, we must show  
\[
  {d_Y\left(\tilde{f}_1(x),b_2(\tilde{f}_1(x), z)\right) \neq \infty}.
\] 
This follows from Property 1 for $L_2$ and the assumption:
\begin{align}
  d_Y\left(\tilde{f}_1(x),b_2(\tilde{f}_1(x), z)\right) &\le 
  d_Z\left(\tilde{f}_2(\tilde{f}_1(x)),z\right) \\ 
  &= d_Z\left( \tilde{f}(x),z\right) 
  \neq \infty
\end{align}

Now, given that 
${d_Y\left(\tilde{f}_1(x),b_2(\tilde{f}_1(x), z)\right) \neq \infty}$ 
holds for all $x \in X$ and $z \in Z$ under the assumption of 
$d_Z\left(\tilde{f}(x),z\right) \neq \infty$, we can freely use 
Properties 1 and 2 of the lens $L_1$ to show that the 
lens properties hold for the composition:
\begin{enumerate}[align=left]
\item[{Property 1.}] 
  \begin{align*}
  d_X(x,b(x,z)) &= 
      d_X\left(x,b_1\left(x, b_2(\tilde{f}_1(x), z)\right)\right) 
      &&  \text{\cref{eq:bcomp}} \\
  &\le d_Y\left(\tilde{f}_1(x),b_2(\tilde{f}_1(x), z)\right) 
      && \text{Property 1 for  $L_1$} \\
  &\le d_Z\left(\tilde{f}_2(\tilde{f}_1(x)),z\right) 
      && \text{Property 1 for  $L_2$} \\
  &= d_Z(\tilde{f}(x),z) 
      && \text{\cref{eq:acomp}}
  \end{align*}
\item[{Property 2.}] 
\begin{align*}
f(b(x,z)) &= 
  f_2\left(f_1\left(b_1\left(x, b_2(\tilde{f}_1(x), z)\right) 
  \right)\right) && \text{\cref{eq:bcomp} \&  \cref{eq:fcomp}} \\
&= 
  f_2 \left( b_2(\tilde{f}_1(x), z)\right) &&  
  \text{Property 2 for  $L_1$} \\
&=z &&   \text{Property 2 for  $L_2$}
\end{align*}
\end{enumerate}

\section{Basic Constructions in {\Bel{}}}\label{app:check_bel_cons}
This appendix verifies that the basic constructions in \Bel{} from
\Cref{{sec:meta:basic}} are well-defined.
\subsection{Tensor Product}\label{app:products}
The tensor product given in
\Cref{eq:tensor_lens1,eq:tensor_lens2,eq:tensor_lens3} is only well-defined if
the domain of the backward map is well-defined, and if the error lens
properties hold. We check these properties below, and restate the definition of
the tensor product lens here for convenience: 

Given any two morphisms ${(f,\tilde{f},b) : A \rightarrow X}$ and
${(g,\tilde{g},b'): B \rightarrow Y}$, we have the morphism
\[
  {(f,\tilde{f},b) \otimes  (g,\tilde{g},b') : 
  A \otimes  B \rightarrow X \otimes  Y}
 \]
defined by
\begin{itemize}
\item the forward map
  \begin{equation*}
     (a_1, a_2) \mapsto (f(a_1),g(a_2))  
    \tag*{\Cref{eq:tensor_lens1}}
   \end{equation*}
\item the approximation map
  \begin{equation*}
    (a_1,a_2) \mapsto (\tilde{f}(a_1),\tilde{g}(a_2)) 
    \tag*{\Cref{eq:tensor_lens2}}
  \end{equation*}
\item the backward map
  \begin{equation*} 
    ((a_1, a_2),(x_1,x_2)) \mapsto (b(a_1,x_1),b'(a_2,x_2)) 
    \tag*{\Cref{eq:tensor_lens3}}
  \end{equation*}
\end{itemize}

We first check the domain: for all ${(a_1,a_2) \in A \otimes  B}$ and 
${(x_1,x_2) \in X \otimes  Y}$, we assume 
\begin{equation}\label{eq:prod_assum}
  d_{X\otimes Y}(\tilde{f}_{ \otimes }(a_1,a_2),(x_1,x_2)) 
  \neq \infty
\end{equation} 
and we are required to show  
\begin{equation}\label{eq:domten}
  d_{X}(\tilde{f}(a_1),x_1) \neq \infty \text{ and } 
  d_{Y}(\tilde{g}(a_2),x_2) \neq \infty
\end{equation}
which follows directly by assumption.

Given that \Cref{eq:domten} holds for all ${(a_1,a_2) \in A \otimes  B}$ 
and ${(x_1,x_2) \in X \otimes  Y}$ under the assumption given in 
\Cref{eq:prod_assum}, we can freely use Properties 1 and  2 of the lenses 
$(f,\tilde{f},b)$ and $(g,\tilde{g},b')$ to show that the lens properties 
hold for the product:
\begin{enumerate}[align=left]
  \item[Property 1.] 
  \begin{align*}
  d_{A\otimes B}\bigl((a_1,a_2) ,b_{ \otimes } 
    ((a_1,a_2),(x_1,x_2))\bigr)  
  &= \max\bigl(d_A(a_1,b(a_1,x_1)),d_B(a_2,b'(a_2,x_2))\bigr) 
  \tag*{\cref{eq:tensor_lens3}} \\
  &\le \max\bigl(d_X(\tilde{f}(a_1),x_1),d_Y(\tilde{g}(a_2),x_2)\bigr) 
   \tag*{(Property 1 of $(f,\tilde{f},b)$ \& $(g,\tilde{g},b')$)} 
  \end{align*}
  \item[{Property 2.}] As above, the property follows directly from 
    Property 2 of the component $(f,\tilde{f},b)$ and $(g,\tilde{g},b')$.
\end{enumerate}

\subsubsection{Tensor product as bifunctor}

\begin{namedtheorem}[\Cref{lem:bifun}]
  The tensor product operation on lenses induces a bifunctor on \Bel.
\end{namedtheorem}

\begin{proof}
The functoriality of the triple given in
\Cref{eq:tensor_lens1,eq:tensor_lens2,eq:tensor_lens3} follows by checking
conditions expressing preservation of composition and identities.
Specifically, for any error lenses 
$h : A \rightarrow B$
$h': A' \rightarrow B'$
$g : B \rightarrow C $ and 
$g': B' \rightarrow C'$ 
we must show  
\[
 (g \otimes g') \circ (h \otimes h') = (g\circ h) \otimes  (g' \circ h')
\]
We check the backward map:

Given any $(a_1,a_2) \in A \otimes A'$ and $(c_1,c_2) \in C \otimes C'$ we have 
\begin{align}
  b_{(g \otimes g') \circ (h \otimes h')}((a_1,a_2),(c_1,c_2)) &=
  b_{h \otimes h'}\bigl((a_1,a_2),b_{g \otimes g'}
  \bigl(\tilde{f}_{h \otimes h'}(a_1,a_2),(c_1,c_2)\bigr)\bigr)\\
  &= (b_h(a_1,b_g(\tilde{f}_h(a_1),c_1)),b_{h'}(a_2,b_{g'}(\tilde{f}_{h'}(a_2),c_2)))\\
  &= b_{(g\circ h) \otimes  (g' \circ h')}((a_1,a_2),(c_1,c_2))
\end{align}

Moreover, for any objects $X$ and $Y$ in $\Bel$, the identity lenses 
$id_X$ and $id_Y$ clearly satisfy 
\[ 
id_X \otimes  id_Y = id_{X \otimes  Y}
\]
\end{proof}

\subsubsection{{Associator}} 
We define the associator 
$\alpha_{X,Y,Z} : X \otimes  (Y  \otimes  Z) 
\rightarrow (X \otimes  Y) \otimes  Z$ as the following triple:
\begin{align}
  f_{\alpha}(x,(y,z)) &\triangleq ((x,y),z)  \\
  \tilde{f}_{\alpha}(x,(y,z)) &\triangleq ((x,y),z) 
      \label{eq:aassoc} \\
  b_{\alpha}((x,(y,z)),((a,b),c)) &\triangleq (a,(b,c)) 
      \label{eq:bassoc}.
\end{align}

It is straightforward to check that $\alpha_{X,Y,Z}$ is an error lens
satisfying Properties 1 and 2. To check that the associator is an isomorphism,
we are required to show the existence of the lens
\[
  \alpha' :  (X \otimes Y) \otimes Z \rightarrow 
  X \otimes  (Y \otimes Z) 
\] 
satisfying 
\[
  \alpha' \circ \alpha_{X,Y,Z} = id_{X\otimes (Y\otimes Z)}
\]
and 
\[
  \alpha_{X,Y,Z} \circ \alpha' = id_{(X\otimes Y)\otimes Z}
\]
where $id$ is the identity lens (see \Cref{def:lensC}).  Defining the forward
and approximation maps for $\alpha'$ is straightforward; for the forward map we
have
\begin{align*} 
  f_{\alpha'}((x,y),z) &\triangleq (x,(y,z))
\end{align*}
and the approximation map is defined identically. For the backward map we have
\begin{align*} 
  b_{\alpha'}\bigl(((x,y),z),(a,(b,c))\bigr) \triangleq ((a,b),c)
\end{align*}
It is straightforward to check that $\alpha'$ satisfies Properties 1 and 2 of 
an error lens. 

The naturality of the associator follows by checking that the following diagram
commutes. 
\[\begin{tikzcd}
	{X \otimes  (Y \otimes  Z)} && {X \otimes  (Y \otimes  Z)} \\ \\
	{(X \otimes  Y) \otimes  Z} && {(X \otimes  Y) \otimes  Z}
	\arrow["{\alpha_{X,Y,Z}}"', from=1-1, to=3-1]
	\arrow["{g_X \otimes   (g_Y  \otimes   g_Z)}", from=1-1, to=1-3]
	\arrow["{\alpha_{X,Y,Z}}", from=1-3, to=3-3]
	\arrow["{(g_X  \otimes   g_Y)  \otimes   g_Z}"', from=3-1, to=3-3]
\end{tikzcd}\]
That is, we check that 
\[
  ((g_X \otimes  g_Y) \otimes   g_Z) \circ \alpha_{X,Y,Z} = 
  \alpha_{X,Y,Z} \circ (g_X \otimes (g_Y \otimes  g_Z))
\] 
for the error lenses 
\begin{align*}
  g_X&: X \rightarrow X \triangleq (f_X,\tilde{f}_X,b_X) \\ 
  g_Y&: Y \rightarrow Y \triangleq (f_Y,\tilde{f}_Y,b_Y) \\
  g_Z&: Z \rightarrow Z \triangleq (f_Z,\tilde{f}_Z,b_Z)
\end{align*}
This follows from the definitions of lens composition (\Cref{def:error_lens})
and the tensor product on lenses
(\cref{eq:tensor_lens1,eq:tensor_lens2,eq:tensor_lens3}). 
We detail here the case of the backward map. 

Using the notation $b_g$ (resp. $\tilde{f}_g$) to refer to both of the 
backward maps (resp. approximation maps) of the tensor product lenses of 
the lenses $g_X$, $g_Y$, and $g_Z$, we are required to show that 
\begin{equation}
b_\alpha\left(xyz,b_{g}\left(\tilde{f}_{\alpha}(xyz),x'y'z'\right)\right) 
  = b_{g}\left(xyz,b_\alpha\left(\tilde{f}_g(xyz),x'y'z'\right)\right) 
\label{eq:verif_assoc1}
\end{equation}
for any $xyz \in X \ \otimes \ (Y\otimes  Z)$ and 
$x'y'z' \in (X\otimes  Y)\otimes  Z$:
\begin{align*}
b_\alpha\left(xyz,b_{g}\left(\tilde{f}_{\alpha}(xyz),x'y'z'\right)\right) 
  &= b_{g}\left(xyz,b_\alpha\left(\tilde{f}_g(xyz),x'y'z'\right)\right) \\
b_\alpha\left(xyz,b_{g}\left(\tilde{f}_{\alpha}(xyz),x'y'z'\right)\right)   
  &= b_{g}\left(xyz,(x',(y',z'))\right) & \text{by \Cref{eq:bassoc}}\\
b_{\alpha}\left(xyz,b_{g}\left(((x,y),z),x'y'z'\right)\right) 
  &= b_{g}\left(xyz,(x',(y',z'))\right) & \text{by \Cref{eq:aassoc}}\\
b_{\alpha}\left(xyz,((b_X(x,x'),b_Y(y,y')),b_Z(z,z'))\right) 
  &= (b_X(x,x'),(b_Y(y,y'),b_Z(z,z'))) & \text{by \Cref{eq:tensor_lens3}}\\
(b_X(x,x'),(b_Y(y,y'),b_Z(z,z')))
  &= (b_X(x,x'),(b_Y(y,y'),b_Z(z,z'))) & \text{by \Cref{eq:bassoc}}\\
\end{align*}

\subsubsection{{Unitors}} We define the left-unitor 
$\lambda_X: I \otimes  X \rightarrow X$ as
\begin{align*}
  f_{\lambda}(\star,x) &\triangleq x \\
  \tilde{f}_{\lambda}(\star,x) &\triangleq x \\
  b_{\lambda}((\star,x),x') &\triangleq (\star,x')
\end{align*}

The right-unitor is similarly defined. 

The fact that $d_I(\star,\star) = -\infty$ is essential in order for 
$\lambda_X$ to satisfy the first property of an error lens: 
\begin{align*}
d_{I\otimes X}(x,b_\lambda((\star,x),x')) 
  &\le d_X(\tilde{f}_\lambda(\star,x),x')\\
\max{(-\infty,d_X(x,x'))} &\le d_X(x,x')
\end{align*}

Checking the naturality of $\lambda_X$ amounts to checking that 
the following diagram commutes for all error lenses 
$g : X \rightarrow Y$. 
\[\begin{tikzcd}
	{I \otimes  X} && {I \otimes  Y} \\ \\
	{X} && {Y}
	\arrow["{\lambda_X}"', from=1-1, to=3-1]
	\arrow["{ id_I \otimes  g }", from=1-1, to=1-3]
	\arrow["{\lambda_Y}", from=1-3, to=3-3]
	\arrow["{g}"', from=3-1, to=3-3]
\end{tikzcd}\]

\subsubsection{{Symmetry}}
We define the symmetry map $\gamma_{X,Y} : X \otimes   Y \rightarrow Y \otimes
X$ as the following triple:
\begin{align*}
  f_{\gamma}(x,y) &\triangleq (y,x) \\
  \tilde{f}_{\gamma}(x,y) &\triangleq (y,x) \\
  b_{\gamma}((x,y),(y',x')) &\triangleq (x',y')
\end{align*} 

It is straightforward to check that  $\gamma_{X,Y}$ is an error 
lens. Checking the naturality of  $\gamma_{X,Y}$ amounts to 
checking that the following diagram commutes for any error 
lenses $g_1 : X \rightarrow Y$ and $g_2 : Y \rightarrow X$.

\[\begin{tikzcd}
	{X \otimes  Y} && {Y \otimes  X} \\ \\
	{Y \otimes  X} && {X \otimes  Y}
	\arrow["{\gamma_{X,Y}}"', from=1-1, to=3-1]
	\arrow["{ g_1 \otimes  g_2 }", from=1-1, to=1-3]
	\arrow["{\gamma_{Y,X}}", from=1-3, to=3-3]
	\arrow["{g_2 \otimes  g_1}"', from=3-1, to=3-3]
\end{tikzcd}\]

\subsection{Coproducts}\label{app:coproducts}
\begin{enumerate}[align=left]
\item[Property 1] For any $x \in X$ and $z \in X+Y$, 
$d_{X}(x,b_{in_1}(x,z)) \le d_{X+Y}(\tilde{f}_{in_1}(x),z)$ 
  supposing 
\[{d_{X+Y}(\tilde{f}_{in_1}(x),z) = {d_{X+Y}(inl \ x,z) 
  \neq \infty}}.\] 

From ${d_{X+Y}(inl \ x,z) \neq \infty}$ and \Cref{eq:prod_met}, we know  
$z = inl \ x_0$ for some $x_0 \in X$, and so we must show 
\[
d_{X}(x,x_0) \le d_{X+Y}(inl \ x, inl \ x_0)
\]
which follows from \Cref{eq:prod_met} and reflexivity.
\item[Property 2] For any $x \in X$ and $z \in X+Y$, 
  ${f_{in_1}(b_{in_1}(x,z)) = z}$ supposing 
\[
d_{X+Y}(\tilde{f}_{in_1}(x),z) = {d_{X+Y}(inl \ x,z) \neq \infty}
\]

From ${d_{X+Y}(inl \ x,z) \neq \infty}$ and \Cref{eq:prod_met}, we know 
$z = inl \ x_0$ for some $x_0 \in X$, and so we must show  
\[
f_{in_1}(x_0) = inl \ x_0
\]
which follows from \Cref{eq:inlf}.
\end{enumerate}

\begin{enumerate}[align=left]
\item[Property 1] For all $z \in X  + Y$ and $c \in C$, 
$d_{X+Y}\left(z,b_{[g,h]}(z,c)\right) \le 
  d_{C}\left(\tilde{f}_{[g,h]}(z),c\right)$ 
supposing 
\[d_C\left(\tilde{f}_{[g,h]}(z),c\right) \neq \infty.\]

This follows directly given that $g$ and $h$ are error lenses:

If $z=inl \ x$ for some $x \in X$ then 
$d_C\left(\tilde{f}_{g}(x), c\right) \neq \infty$ and we use Property 1 
for $g$ to satisfy the desired conclusion:  
${d_X(x, (b_g(x,c))) \le d_C(\tilde{f}_g(x),c)}$. Otherwise, $z=inr \ y$ 
for some $y \in Y$ then $d_C\left(\tilde{f}_{h}(y), c\right) \neq \infty$ 
and we use Property 1 for $h$.

\item[Property 2] For all $z  \in  X+ Y$ and $c \in C$, 
$f_{[g,h]}\left(b_{[g,h]}(z,c)\right) = c$ \\ 
supposing $d_C \left(\tilde{f}_{[g,h]}(z),c\right) \neq \infty$.
If $z=inl \ x$ for some $x \in X$ then 
$d_C\left(\tilde{f}_{g}(x), c\right) \neq \infty$ 
and we use Property 2 for $g$. Otherwise, $z=inr \ y$ for some $y \in Y$ 
then $d_C\left(\tilde{f}_{h}(y), c\right) \neq \infty$ and we use 
Property 2 for $h$.
\end{enumerate}

To show that $[g,h] \circ in_1 = g$ (resp. $[g,h]  \circ in_2= h$), 
we observe that the following diagrams, by definition, commute.

\begin{figure}[!htb]
\minipage{0.32\textwidth}
\[\begin{tikzcd}
	{X + Y} && C \\
	\\
	X
	\arrow["{f_{in_1}}", from=3-1, to=1-1]
	\arrow["{f_{[g,h]}}", from=1-1, to=1-3]
	\arrow["{f_g}"', from=3-1, to=1-3]
\end{tikzcd}\]
\endminipage\hfill
\minipage{0.32\textwidth}
\[\begin{tikzcd}
	{X + Y} && C \\
	\\
	X
	\arrow["{\tilde{f}_{in_1}}", from=3-1, to=1-1]
	\arrow["{\tilde{f}_{[g,h]}}", from=1-1, to=1-3]
	\arrow["{\tilde{f}_g}"', from=3-1, to=1-3]
\end{tikzcd}\]
\endminipage\hfill
\minipage{0.32\textwidth}%
\[\begin{tikzcd}
	{(X+Y)\times C} && {X + Y} \\
	\\
	{X \times C} && X
	\arrow["{\tilde{f}_{in_1}\times id_C}", from=3-1, to=1-1]
	\arrow["{b_{[g,h]}}", from=1-1, to=1-3]
	\arrow["{b_g}"', from=3-1, to=3-3]
	\arrow["{b_{in_1}}", from=1-3, to=3-3]
\end{tikzcd}\]
\endminipage
\end{figure}

\subsubsection{Uniqueness of the copairing}
We check the uniqueness of copairing by showing that for any two morphisms 
$g_1 : X \rightarrow C$ and $g_2 : Y \rightarrow C$, if 
$h \circ in_1 = g_1$ and $h \circ in_2 = g_2$ for any 
$h : X + Y \rightarrow C$, then $h = [g_1,g_2]$. 

We detail the cases for the forward and backward map; the case for the 
approximation map is identical to that of the forward map. 

\begin{enumerate}[align=left]
\item[\textbf{forward map}] We are required to show that 
$f_h(z) = f_{[g_1,g_2]}(z)$ for any ${z \in X + Y}$ assuming that 
$f_{in_1};f_{h} = f_{g_1}$ and $f_{in_2};f_{h} = f_{g_2}$. 
The desired conclusion follows by cases on $z$; i.e., $z = inl \ x$ 
for some $x \in X$ or $z = inr \ y$ for some $y \in Y$.
\item[\textbf{backward map}] We are required to show that 
$b_{h}(z,c) = b_{[g_1,g_2]}(z,c)$ for any $z \in X+Y$ and $c \in C$. 
Unfolding definitions in the assumptions $b_{h \circ in_1} = b_{g_1}$ 
and $b_{h \circ in_2} = b_{g_2}$, we have that 
$b_{in_1}(x,b_h\left( \tilde{f}_{in_1}(x),c_1 \right)) = b_{g_1}(x,c_1)$ 
for any $x \in X$ and $c_1 \in C$ and 
$b_{in_2}(y,b_h\left( \tilde{f}_{in_2}(y),c_2 \right)) = b_{g_2}(y,c_2)$ 
for any $y \in Y$ and $c_2 \in C$. We proceed by cases on $z$. 

If $z = inl \ x$ for some $x \in X$ then we are required to show that 
\[b_{h}(inl \ x,c) = inl \ (b_{g_1}(x,c)).\]  By definition of lens 
composition, we have that 
\[d_{X+Y}\left( inl \ x, b_h(inl \ x, c)\right) \neq \infty,\] so 
$b_h(inl \ x, c) = inl \ x_0$ for some $x_0 \in X$. By assumption, we 
then have that $b_{g_1}(x,c) = b_{in_1}(x,inl \ x_0) = x_0$, from which 
the desired conclusion follows. 

The case of $z = inr \ y$ for some $y \in Y$ is identical.
\end{enumerate}

\section{Interpreting \texorpdfstring{\bea}{} Terms}\label{sec:app_interp_bea}
In this section of the appendix, we detail the constructions for 
interpreting \bea{} terms (\Cref{def:interpL}). 

Applications of the symmetry map $s_{X,Y} : X \times Y \rightarrow Y \times X$
and 2-monoidality $m_{r,A,B} : D_R(A \otimes B) \xrightarrow{\sim} D_r(A \otimes B)$
are often elided for succinctness. Recall the discrete diagonal $t_X : X \rightarrow X
\times X$ (\Cref{lem:discr-diag}), which will be used frequently in the 
following constructions. 
\begin{description}
% VAR LINEAR
\item[Case (Var).] Suppose that $\Gamma = x_{0} :_{q_0} \sigma_0, \dots, x_{i-1}
  :_{q_{i-1}} \sigma_{i-1}$. Define the map   
$\denot{\Phi\mid \Gamma, x:_r \sigma \vdash x:\sigma}$ 
as the composition
\[
  \pi_i
  \circ (\varepsilon_{\denot{\sigma_0}}
    \otimes \cdots \otimes
    \varepsilon_{\denot{\sigma_{i - 1}}}
    \otimes \varepsilon_{\denot{\sigma}})
  \circ (m_{0\le r, \denot{\sigma_0}}
    \otimes \cdots \otimes
    m_{0\le r, \denot{\sigma_{i - 1}}}
    \otimes m_{0\le r, \denot{\sigma}}) ,
\]
where the lens $\pi_i$ is the $i$th projection. Note that all types $\sigma$ and
$\sigma_j$ are interpreted as metric spaces, i.e., satisfying reflexivity.

% VAR DISCRETE
\item[Case (DVar).] Define the map   
$\denot{\Phi,z : \alpha \mid\Gamma \vdash x:\alpha}$ 
as the $i$th projection lens $\pi_i$, assuming
$\Phi = z_{0} : \alpha_0, \dots, z_{i-1} : \alpha_{i-1}$. Note that all discrete
types $\alpha_j$ and $\alpha$ are interpreted as discrete metric spaces, i.e.,
with self-distance zero.
% UNIT
\item[Case (Unit).] Define the map 
$
\denot{\Phi\mid \Gamma \vdash ( ) : \unit }
$
as the lens $\mathcal{L}_{unit}$ from a tuple 
$\bar{x} \in \denot{\Phi\mid \Gamma}$ to the singleton 
of the carrier in $I = (\{\star\},\underline{0})$ , defined as 
\begin{align*}
  f_{unit}(\bar{x}) &\triangleq \star\\
  \tilde{f}_{unit}(\bar{x}) &\triangleq \star \\ 
  b_{unit}(\bar{x},\star) &\triangleq \bar{x}.
\end{align*}
We verify that the triple $\mathcal{L}_{unit}$ is an error lens.  
\begin{enumerate}[align=left]
\item[Property 1.] For any $\bar{x} \in X_1 \otimes  \cdots \otimes  X_i$ we must show    
\begin{align*}
d_{X_1 \otimes  \cdots \otimes  X_i}\left(\bar{x},b_{c}(\bar{x},\star)\right) &\le
     d_I\left(\tilde{f}_{unit}(\bar{x}),\star\right) \\
\max(d_{X_1}(x_1,x_1),\cdots,d_{X_i}(x_i,x_i)) &\le 0,
\end{align*}
which holds under the assumption that all types are interpreted as metric spaces 
with negative self distance. 
\item[Property 2.] For any $\bar{x} \in X_1 \otimes  \cdots \otimes  X_i$ we have
\[f_{unit}\left(b_{unit}(\bar{x},\star) \right) = f_{unit}(\bar{x}) = \star.\]
\end{enumerate}
% TENS INTRO
\item[Case ($\otimes $ I).] 
Given the maps  
\begin{align*}
  h_1 &= \denot{\Phi\mid \Gamma \vdash e : \sigma}: \denot{\Phi\mid \Gamma}
  \rightarrow \denot{\sigma}\\ 
  h_2 &=\denot{\Phi\mid \Delta \vdash f : \tau} : \denot{\Phi\mid \Delta} 
  \rightarrow \denot{\tau}
\end{align*}
define the map 
$\denot{\Phi\mid \Gamma,\Delta \vdash (e,f) : \sigma \otimes \tau}$
as the composition
\[ 
  (h_1 \otimes h_2) \circ (t_{\denot{\Phi}} \otimes id_{\denot{\Gamma,\Delta}}),
\]
where the map $t_{\denot{\Phi}}: \denot{\Phi} 
  \rightarrow \denot{\Phi} \otimes \denot{\Phi}$ is 
  the diagonal lens on discrete metric spaces (\Cref{lem:discr-diag}).
% TENS ELIM LINEAR
\item[Case ($\otimes $ E$^\sigma$).] Given the maps
\begin{align}
h_1 &= \denot{\Phi\mid \Gamma \vdash e : \tau_1 \otimes \tau_2} : 
  \denot{\Phi} \otimes \denot{\Gamma} \rightarrow {\denot{\tau_1 \otimes \tau_2}} \\
h_2 &= \denot{\Phi\mid \Delta, x :_r \tau_1, y :_r \tau_2 \vdash 
  f : \sigma} : \denot{\Phi} \otimes \denot{\Delta} \otimes 
  D_r\denot{\tau_1} \otimes D_r\denot{\tau_2} \rightarrow {\denot{\sigma}}
\end{align}
we must define a 
$\denot{\Phi\mid r + \Gamma, \Delta \vdash \slet {(x,y)} e f : \sigma}$. 
We first define a map 
\[h: D_r\denot{\Phi} \otimes D_r\denot{\Gamma} \otimes 
  \denot{\Delta} \rightarrow \denot{\sigma}\]
as the composition
\[
  h_2 \circ (
  (m^{-1}_{r,\denot{\tau_1},\denot{\tau_2}} \circ 
  D_r(h_1)) \otimes (\varepsilon_{\denot{\Phi}} \circ m_{0\le r, \denot{\Phi}})
  \otimes id_{\denot{\Delta}})
 \circ (t_{D_r\denot{\Phi}} \otimes id_{D_r\denot{\Gamma} \otimes  \denot{\Delta}}).
\]
Now, observing that $\denot{\Phi}$ is a discrete space, the set maps from the lens $h$ 
define the desired lens:
\[\denot{\Phi} \otimes D_r\denot{\Gamma} \otimes 
  \denot{\Delta} \rightarrow \denot{\sigma}\]
% TENS ELIM DISCRETE
\item[Case ($\otimes $ E$^\alpha$).] Given the maps
\begin{align}
h_1 &= \denot{\Phi\mid \Gamma \vdash e : \alpha_1 \otimes \alpha_2} : 
  \denot{\Phi} \otimes \denot{\Gamma} \rightarrow 
  {\denot{\alpha_1 \otimes \alpha_2}} \\
h_2 &= \denot{\Phi, x : \alpha_1, y : \alpha_2; \Delta \vdash 
  f : \sigma} : \denot{\Phi} \otimes \denot{\alpha_1} \otimes \denot{\alpha_2} 
    \otimes \denot{\Delta} \rightarrow {\denot{\sigma}}
\end{align}
define the map 
$\denot{\Phi\mid \Gamma, \Delta \vdash \slet {(x,y)} e f : \sigma}$
as the composition
\[
h_2 \circ (h_1 \otimes id_{\denot{\Phi} \otimes  \denot{\Delta}})
 \circ (t_{\denot{\Phi}} \otimes id_{\denot{\Gamma} \otimes  \denot{\Delta}}).
\]
% SUM ELIM
\item[Case ($+$ E).] Given the maps
\begin{align*}
h_1 &= \denot{\Phi\mid  \Gamma \vdash e' : \sigma + \tau} : 
  \denot{\Phi} \otimes \denot{\Gamma} \rightarrow 
  \denot{\sigma + \tau}\\
h_2 &= \denot{\Phi\mid \Delta,x :_q \sigma \vdash e : \rho} : 
  \denot{\Phi} \otimes \denot{\Delta} \otimes D_q\denot{\sigma} \rightarrow 
  \denot{\rho}\\
h_3 &= \denot{\Phi\mid \Delta, y :_q \tau \vdash f : \rho} : 
  \denot{\Phi} \otimes \denot{\Delta} \otimes D_q\denot{\tau}\rightarrow 
  \denot{\rho} 
\end{align*}
we require a lens
$\denot{\Phi\mid q+\Gamma,\Delta \vdash \case {e'} {x.e} {y.f}: \rho}$. 
We first define a lens 
\[h: D_q\denot{\Phi} \otimes D_q\denot{\Gamma} \otimes
  \denot{\Delta} \rightarrow \denot{\sigma}\]
as the composition
\[
  [h_2,h_3] \circ \Theta \circ
  ((\eta \circ D_q(h_1)) \otimes (\varepsilon_{\denot{\Phi}} 
  \circ m_{0\le r, \denot{\Phi}}) \otimes id_{\denot{\Delta}}) \circ 
  (t_{D_q{\denot{\Phi}}} \otimes id_{D_q\denot{\Gamma} \otimes  \denot{\Delta}})
\]
Now, observing that $\denot{\Phi}$ is a discrete space, the set maps from the lens $h$ 
define the desired lens. Above, the map 
$\eta : D_q \denot{\sigma + \tau} \rightarrow D_q\denot{\sigma} + D_q\denot{\tau}$
is the identity lens, and the map 
$
\Theta_{X,Y,Z} : X \otimes (Y + Z) \rightarrow (X \otimes Y) + (X \otimes Z)
$ 
is given by the triple 
\\
\
\begin{align*}
    f_{\Theta}(x,w) &\triangleq 
        \begin{cases}
          inl \ (x,y) & \text{if } w = inl \ y \\ 
          inr \ (x,z) & \text{if } w = inr \ z
        \end{cases}\\
    \tilde{f}_{\Theta}(x,w) &\triangleq  f_{\Theta}(x,w) \\
    b_{\Theta}((x,w),u) &\triangleq  
        \begin{cases}
          (\pi_1 a, inl \ (\pi_2 a)) & 
                \text{if } u = inl \ a \\ 
          (\pi_1 a, inr \ (\pi_2 a)) & 
                \text{if } u = inr \ a  \\
          (x,w) & \text{ otherwise.}     
        \end{cases} 
\end{align*}
\
We check that the triple 
\[
\Theta_{X,Y,Z} : X \otimes  (Y + Z) \rightarrow (X \otimes  Y) + (X \otimes  Z) 
     \triangleq (f_\Theta,\tilde{f}_{\Theta},b_\Theta)
\] is well-defined.  
\begin{enumerate}[align=left]
\item[Property 1.] For any $x \in X$, $w \in Y+Z$, and
     $u \in (X \otimes  Y) + (X \otimes  Z)$ we are required to show    
\begin{align}
d_{X \otimes  (Y+Z)}\left((x,w),b_{\Theta}((x,w),u)\right) \le
     d_{(X\otimes Y) + (X \otimes Z)}\left(\tilde{f}_{\Theta}(x,w),u\right)
\label{eq:prod_req1}
\end{align}
supposing 
\begin{equation}
d_{(X\otimes Y) + (X \otimes Z)}\left(\tilde{f}_{\Theta}(x,w),u\right) \neq \infty.
\label{eq:prop1_theta}
\end{equation}
From \Cref{eq:prop1_theta}, and by unfolding definitions, we have
\begin{enumerate}
\item 
if $w = inl \ y$ for some $y \in Y$, then $u = inl \ (x_1,y_1)$ 
     for some $(x_1,y_1) \in X \otimes  Y$
\item 
if $w = inr \ z$ for some $z \in Z$, then $u = inr \ (x_1,z_1)$ 
     for some $(x_1,z_1) \in X \otimes  Z$.
\end{enumerate}
In both cases, the \Cref{eq:prod_req1} is an equality.
\item[Property 2.] 
For any $x \in X$, $w \in Y+Z$, and
     $u \in (X \otimes  Y) + (X \otimes  Z)$ we are required to show 
\begin{equation}
{f_{\Theta}(b_{\Theta}((x,w),u)) = u} \label{eq:prop2_theta}
\end{equation} 
supposing \Cref{eq:prop1_theta} holds.

We consider the cases when $u = inl \ (x_1,y_1)$ for some 
$(x_1,y_1) \in X \otimes  Y$ and when 
$u = inr \  (x_1,z_1)$ for some $(x_1,z_1) \in X \otimes  Z$ 
as we did for Property 1

In the first case, we have 
\begin{align*}
f_{\Theta}(b_{\Theta}((x,w),u)) &= f_{\Theta}(x_1, inl \ y_1) \\
&= inl \ (x_1,y_1).
\end{align*}
In the second case we have 
\begin{align*}
f_{\Theta}(b_{\Theta}((x,w),u))&= f_{\Theta}(x_1, inr \ z_1) \\
&= inr \ (x_1,z_1).
\end{align*}
\end{enumerate}
% SUM INTRO L
\item[Case ($+ \ \text{I}_{L,R}$).] Given the maps
\begin{align*}
h_l &= \denot{\Phi\mid \Gamma \vdash e : \sigma} : 
  \denot{\Phi} \otimes \denot{\Gamma} \rightarrow \denot{\sigma} \\
h_r &= \denot{\Phi\mid \Gamma \vdash e : \sigma} : 
  \denot{\Phi} \otimes \denot{\Gamma} \rightarrow \denot{\tau}
\end{align*}
define the maps 
\begin{align*}
\denot{\Phi\mid \Gamma \vdash \inl e  : \sigma + \tau} &\triangleq in_1\circ h_l \\
\denot{\Phi\mid \Gamma \vdash \inr e  : \sigma + \tau} &\triangleq in_2\circ h_r .
\end{align*}
% LET
\item[Case (Let).] See \Cref{sec:bean:semantics}.
% DISCRETE INTRO
\item[Case (Disc).] Given the lens 
$\denot{\Phi \mid \Gamma \vdash e : \num}$ from the premise, 
we can define the map $\denot{\Phi \mid \Gamma \vdash e : \dnum}$
directly by verifying the lens conditions. 
% DISCRETE LET
\item[Case (DLet).] 
Given the maps
\begin{align*}
h_1 &= \denot{\Phi\mid \Gamma \vdash e : \alpha} : 
  \denot{\Phi} \otimes \denot{\Gamma} \rightarrow {\denot{\alpha}}\\
h_2 &= \denot{\Phi,x :\alpha \mid \Delta \vdash f : \sigma} : 
  \denot{\Phi} \otimes \denot{\alpha} \otimes \denot{\Delta} \rightarrow 
  {\denot{\sigma}}
\end{align*}
define the map 
$\denot{\Phi \mid \Gamma  \vdash \slet x e f  : \sigma } $ as the composition
\[
h_2 \circ (h_1 \otimes id_{\denot{\Phi} \otimes \denot{\Delta}} )\circ 
(t_{\denot{\Phi}} \otimes id_{\denot{\Gamma} \otimes \denot{\Delta}})
\]
% ADD
\item[Case (Add).] See \Cref{sec:bean:semantics}.
%%%%%%%%%%%%%%%%
% SUB

\item[Case (Sub).] We proceed the same as the case for (Add). We define a lens
  $\mathcal{L}_{sub} : D_\varepsilon(R) \otimes D_\varepsilon(R) \rightarrow R$
  given by the triple
\begin{align*}
    f_{sub}(x_1,x_2) &\triangleq x_1 - x_2 \\
    \tilde{f}_{sub}(x_1,x_2)&\triangleq  (x_1 - x_2)e^\delta; \quad |\delta| 
        \le \varepsilon \\
    b_{sub}((x_1,x_2),x_3) &\triangleq  \left(\frac{x_3x_1}{x_1-x_2},
      \frac{x_3x_2}{x_1-x_2}\right).
\end{align*}

We check that $\mathcal{L}_{sub} : D_\varepsilon(R) \otimes D_\varepsilon(R)
\rightarrow R$ is well-defined. 

For any $x_1,x_2,x_3 \in R$ such that 
\begin{equation}
d_{R}\left(\tilde{f}_{sub}(x_1,x_2),x_3\right) \neq \infty.
\label{eq:sub_asum}
\end{equation}
holds, we need to check that $\mathcal{L}_{sub}$ satisfies the properties of an 
error lens. We take the distance function $d_R$ as the metric given in 
\Cref{eq:olver}, so \Cref{eq:sub_asum} implies that $(x_1-x_2)$ and $x_3$ 
are either both zero or are both non-zero and of the same sign. 

\begin{enumerate}[align=left]
\item[Property 1.] We are required to show that 
\begin{align*}
d_{R \otimes  R}\left((x_1,x_2),b_{sub}((x_1,x_2),x_3)\right) - \varepsilon &\le 
     d_{ R}\left(\tilde{f}_{sub}(x_1,x_2),x_3\right) \\
    &\le d_{R}\left((x_1-x_2)e^\delta,x_3\right).
\end{align*}
Without loss of generality, we consider the case when 
\[
d_{R \otimes  R}\left((x_1,x_2),b_{sub}((x_1,x_2),x_3)\right) =
    d_R\left(x_1,\frac{x_3x_1}{x_1-x_2}\right); 
\]
that is, 
\[
d_R\left(x_2,\frac{x_3x_2}{x_1-x_2}\right)  \le 
  d_R\left(x_1,\frac{x_3x_1}{x_1-x_2}\right).
\]
Unfolding the definition of the distance function given in \Cref{eq:olver}, 
we are required to show 
\begin{equation}
\left|{\ln\left(\frac{x_1-x_2}{x_3}\right)} \right| \le 
\left|{\ln\left(\frac{x_1-x_2}{x_3}\right)} + \delta\right| + \varepsilon.
\end{equation}
which holds under the assumptions of $|\delta| \le \varepsilon$ and 
$0 < \varepsilon$; the proof is identical to that given for the case of the 
Add rule. 
\item[Property 2.] 
\begin{align*}
f_{sub}\left(b_{sub}((x_1,x_2),x_3)\right) = 
     f_{sub}\left(\frac{x_3x_1}{x_1-x_2},\frac{x_3x_2}{x_1-x_2}\right) 
    = x_3.
\end{align*} 
\end{enumerate}
%%%%%%%
% MUL
%%%
\item[Case (Mul).] See \Cref{sec:bean:semantics}.
%%%
% DIV
%%%
\item[Case (Div).]
  We proceed the same as the case for (Add), with slightly different indices.
  We define a lens $\mathcal{L}_{div} : D_{\varepsilon/2} (R) \otimes
  D_{\varepsilon/2}(R) \rightarrow (R + \diamond)$ given by the triple
\\
\
\begin{align*}
    f_{div}(x_1,x_2) &\triangleq
      \begin{cases}
        {x_1}/{x_2} & \text{ if } x_2 \neq 0 \\
        \diamond    & \text{ otherwise }
      \end{cases}  \\
    \tilde{f}_{div}(x_1,x_2)&\triangleq        
        \begin{cases}
        {x_1}e^\delta/{x_2} & \text{ if } x_2 \neq 0 ; \quad |\delta| \le \varepsilon \\
        \diamond    & \text{ otherwise }
      \end{cases}  \\
    b_{div}((x_1,x_2),x) &\triangleq 
        \begin{cases}
         \left(\sqrt{x_1x_2x_3},\sqrt{{x_1x_2}/x_3}\right) & \text{ if } x = inl \ x_3 \\
         (x_1,x_2)    & \text{ otherwise }
      \end{cases}.
\end{align*}
\

We check that $\mathcal{L}_{div} : D_{\varepsilon/2} (R) \otimes
D_{\varepsilon/2}(R ) \rightarrow (R + \diamond)$ is well-defined.

For any $x_1,x_2 \in R$ and $x \in R + \diamond$ such that  
\begin{equation}\label{eq:div_asum}
d_{R+ \diamond}\left(\tilde{f}_{div}(x_1,x_2),x\right) \neq \infty.
\end{equation}
holds, we are required to show that $\mathcal{L}_{div}$ satisfies the 
properties of an error lens. From \Cref{eq:div_asum} and again assuming the 
distance function is given by \Cref{eq:olver}, we know $x = inl \ x_3$ for some 
$x_3 \in R$, $x_2 \neq 0$, and $x_1/x_2$ and $x_3$ are either both zero or both 
non-zero and of the same sign; this guarantees that the backward map 
(containing square roots) is indeed well defined. 
\begin{enumerate}[align=left]
\item[Property 1.] We need to show 
\begin{align}\label{eq:p1_div}
d_{R \otimes  R}\left((x_1,x_2),b_{div}((x_1,x_2),x)\right) -\varepsilon/2 &\le 
     d_{R + \diamond}\left(\tilde{f}_{div}(x_1,x_2),x\right) \nonumber \\
    &\le d_{R}\left(\frac{x_1}{x_2}e^\delta,x_3\right).
\end{align}
Unfolding the definition of the distance function (\Cref{eq:olver}), we have
\begin{align*}
d_{R \otimes  R}\left((x_1,x_2),b_{div}((x_1,x_2),x_3)\right)  
&= d_R\left(x_1,x_1\sqrt{{x_1}{x_2x_3}}\right) \\
&= d_R\left(x_2,x_2\sqrt{{x_1x_2}/{x_3}}\right) \\
&= \frac{1}{2}\left|\ln \left( \frac{x_1}{x_2x_3}\right)\right|,
\end{align*}
and so we are required to show 
\begin{align*}
\frac{1}{2}\left|\ln \left( \frac{x_1}{x_2x_3}\right)\right|\le
\left|\ln \left( \frac{x_1}{x_2x_3}\right) + \delta \right| + 
  \frac{1}{2}\varepsilon,
\end{align*}
which holds under the assumptions of $|\delta| \le \varepsilon$ and $0 < \varepsilon$; the proof is identical to that given for the case of the Mul rule. 
\item[Property 2.] 
\begin{align*}
f_{div}\left(b_{div}((x_1,x_2),x_3)\right) = 
     f_{div}\left(\sqrt{x_1x_2x_3},\sqrt{{x_1x_2}/x_3}\right) 
    = x_3.
\end{align*} 
\end{enumerate}
%%%
% DMUL
%%%
\item[Case (DMul).]
  We proceed similarly as for (Add). We define a
  lens $\mathcal{L}_{dmul} : (R^\alpha \otimes  D_{\varepsilon}R) \rightarrow R$
  given by the triple
\begin{align*}
    f_{dmul}(x_1,x_2) &\triangleq x_1  x_2 \\
    \tilde{f}_{dmul}(x_1,x_2)&\triangleq  x_1 x_2e^{\delta}; \quad |\delta| 
  \le \varepsilon \\
    b_{dmul}((x_1,x_2),x_3) &\triangleq (x_1,x_3/x_1).
\end{align*}
%%%

We check that 
$\mathcal{L}_{dmul} : R^\alpha \otimes  D_{\varepsilon}R \rightarrow  R$ is well-defined.

For any $x_1,x_2,x_3 \in R$ such that 
\begin{equation} 
d_{R}\left(\tilde{f}_{dmul}(x_1,x_2),x_3\right) \neq \infty.
\label{eq:dmul_asum}
\end{equation}
holds, we need to check the that $\mathcal{L}_{dmul}$ satisfies the properties 
of an error lens. We again take the distance function $d_R$ as the metric 
given in \Cref{eq:olver}, so \Cref{eq:dmul_asum} implies that $(x_1x_2)$ 
and $x_3$ are either both zero or are both non-zero and of the same sign; 
this guarantees that the backward map (containing square roots) is indeed 
well defined. 
\begin{enumerate}[align=left]
\item[Property 1.] We are required to show
\begin{align*} \label{eq:dp1_mul}
d_{R^\alpha \otimes  D_\varepsilon R}\left((x_1,x_2),b_{dmul}((x_1,x_2),x_3)\right)  &\le 
     d_R\left(\tilde{f}_{dmul}(x_1,x_2),x_3\right) \\
    &\le d_{R}\left(x_1 x_2e^{\delta},x_3\right) 
\end{align*}
Unfolding the definition of the distance function (\Cref{eq:olver}), we have
\begin{align*}
d_{R^\alpha \otimes  D_\varepsilon R}
  \left((x_1,x_2),b_{dmul}((x_1,x_2),x_3)\right)  
&= \max\left(d_\alpha(x_1,x_1),d_R(x_2,x_3/x_1) -\varepsilon \right) \\
&= \left|\ln \left( \frac{x_1x_2}{x_3}\right)\right| -\varepsilon,
\end{align*}
and so we are required to show 
\begin{equation}  \label{eq:dp1_mul_end}
\left|\ln \left( \frac{x_1x_2}{x_3}\right)\right| \le
\left|\ln \left( \frac{x_1x_2}{x_3}\right) + \delta \right| + 
  \varepsilon
\end{equation}
which holds under the assumptions of $|\delta| \le \varepsilon$ and $0 < \varepsilon$;
the proof is identical to that given in the Add rule. 

\item[Property 2.] 
\begin{align*}
f_{dmul}\left(b_{dmul}((x_1,x_2),x_3)\right) = 
     f_{dmul}\left(x_1,x_3/x_1\right)
    = x_3.
\end{align*} 
\end{enumerate}

\end{description}

\section{Interpreting \texorpdfstring{\LangS}{} Terms}\label{sec:app_interp_LS}
This appendix provides the detailed constructions of the interpretation of
\LangS{} terms for \Cref{def:interpS}.  The interpretation of terms is defined
over the typing derivations for \LangS{} given in
\Cref{fig:typing_rules_2_full}.  For each case, the ideal interpretation
$\pdenotid{-}$  is constructed explicitly, but the construction for
$\pdenotap{-}$ is nearly identical, requiring only that the forgetful functor
$\Uap$ is used in place of $\Uid$.  

Applications of the symmetry map $s_{X,Y} : X \times Y \rightarrow Y \times X$
are elided for succinctness. The diagonal map $d_X : X \rightarrow X \times X$
on $\Set$ is used frequently and is not elided.
\begin{description}
% VAR 
\item[\textbf{Case (Var).}] 
Define the maps 
$\pdenot{\Phi,\Gamma,x :\sigma, \Delta \vdash x : \sigma}_{id}$ and 
$\pdenot{\Phi,\Gamma,x :\sigma, \Delta \vdash x : \sigma}_{ap}$ 
in $\Set$ as the appropriate projection $\pi_i$.
% Unit 
\item[\textbf{Case (Unit).}] 
Define the set maps 
$\pdenot{\Phi,\Gamma \vdash () : \unit}_{id}$ and 
$\pdenot{\Phi,\Gamma \vdash () : \unit}_{ap}$ as the constant function 
returning the value $\star$. 
% CONST 
\item[\textbf{Case (Const).}] 
Define the maps 
$\pdenot{\Phi,\Gamma \vdash k : \num}_{id}$  and 
$\pdenot{\Phi,\Gamma \vdash k : \num}_{ap}$ 
in $\Set$ as the constant function taking points in 
$\pdenot{\Phi,\Gamma}$ to the value $k \in R$. 
% PROD INTRO
\item[\textbf{Case ($\otimes $ I).}] 
Given the maps 
\begin{align*}
h_1 &= \pdenotid{\Phi,\Gamma \vdash e : \sigma}: 
  \pdenot{\Phi} \times \pdenot{\Gamma} 
  \rightarrow {\pdenot{\sigma}}\\
h_2 &= \pdenotid{\Phi,\Delta \vdash f : \tau}: 
  \pdenot{\Phi} \times\pdenot{\Delta} 
  \rightarrow {\pdenot{\tau}}
\end{align*}
in $\Set$, define the map $\pdenotid{\Phi,\Gamma, \Delta \vdash (e,f) : \sigma
    \otimes  \tau}$ as 
\[
  (d_{\pdenot{\Phi}},id_{\pdenot{\Gamma}},id_{\pdenot{\Delta}});(h_1, h_2)
\]
% PROD ELIM
\item[\textbf{Case ($\otimes $ E).}] Given the maps 
\begin{align*}
h_1 &= \pdenot{\Phi,\Gamma \vdash e : \tau_1 \otimes  \tau_2}_{id}: 
  \pdenot{\Phi} \times \pdenot{\Gamma} 
  \rightarrow {\pdenot{\tau_1} \times  \pdenot{\tau_2}}\\
h_2 &= \pdenot{\Phi,\Delta, x : \tau_1, y: \tau_2 \vdash f : \sigma}_{id} : 
  \pdenot{\Phi} \times \pdenot{\Delta} \times \pdenot{\tau_1} \times
  \pdenot{\tau_2} \rightarrow {\pdenot{\sigma}}
\end{align*}
in $\Set$, define $\pdenotap{\Phi,\Gamma, \Delta \vdash \slet {(x,y)} e f :
    \sigma}$ as 
\[
  (d_{\pdenot{\Phi}},id_{\pdenot{\Gamma}},id_{\pdenot{\Delta}}); (h_1,
  id_{\pdenot{\Phi} \times \pdenot{\Delta}});h_2
\]
% SUM ELIM
\item[Case ($+$ E).] 
Given the maps 
\begin{align*}
h_1 &= \pdenotid{\Phi,\Gamma \vdash e' : \sigma + \tau}: 
  \pdenot{\Phi} \times \pdenot{\Gamma} 
  \rightarrow {\pdenot{\sigma + \tau}}\\
h_2 &= \pdenotid{\Phi,\Delta, x : \sigma \vdash e : \rho} : 
  \pdenot{\Phi} \times\pdenot{\Delta} \times \pdenot{\sigma} 
  \rightarrow {\pdenot{\rho}}\\
h_3 &= \pdenotid{\Phi,\Delta, y : \tau \vdash f : \rho} : 
  \pdenot{\Phi} \times \pdenot{\Delta} \times \pdenot{\tau} 
  \rightarrow {\pdenot{\rho}}
\end{align*}
in $\Set$, define $\pdenotid{\Phi,\Gamma, \Delta \vdash \case {e'} {x.e} {y.f}
: \rho}$ as 
\[
  (d_{\pdenot{\Phi}},
  id_{\pdenot{\Gamma} \times \pdenot{\Delta}}); 
  (h_1, id_{\pdenot{\Phi} \times \pdenot{\Delta}});
  \Theta^S_{\pdenot{\Phi} \times \pdenot{\Delta},\pdenot{\sigma},
  \pdenot{\tau}};[h_2,h_3]
\]
where $\Theta^S_{X,Y,Z}$ is a map in $\Set$:
\[
  \Theta_{X,Y,Z} : X \times  (Y + Z) \rightarrow 
  (X \times  Y) + (X \times  Z)
\] 
% SUM INTRO
\item[\textbf{Case ($+$ I$_L$).}] 
Given the map 
\[
  h = \pdenotid{\Phi,\Gamma \vdash e : \sigma}:
  \pdenot{\Phi} \times \pdenot{\Gamma} 
  \rightarrow {\pdenot{\sigma}}
\]
in $\Set$, define the map
\[
\pdenotid{\Phi,\Gamma \vdash \inl e: \sigma + \tau} 
\] 
as the composition 
\[
  h;in_1
\]
% SUM INTRO
\item[\textbf{Case ($+$ $I_R$).}] 
Given the map 
\[
  h = \pdenotid{\Phi,\Gamma \vdash e : \sigma}: \pdenot{\Phi} \times
  \pdenot{\Gamma} \rightarrow {\pdenot{\sigma}}
\]
in $\Set$, define the map
\[
  \pdenotid{\Phi,\Gamma \vdash \inr e: \sigma + \tau} 
\] 
as the composition 
\[
  h;in_r
\]
% Let 
\item[\textbf{Case (Let).}] 
Given the maps 
\begin{align*}
h_1 &= \pdenotid{\Phi,\Gamma \vdash e : \sigma} 
  : \pdenot{\Phi} \times \pdenot{\Gamma} \rightarrow \pdenot{\sigma} \\
h_2 &= \pdenotid{\Phi,\Delta, x: \sigma \vdash f : \tau}
  : \pdenot{\Phi} \times \pdenot{\Delta} \times \pdenot{\sigma} 
      \rightarrow \pdenot{\tau} 
\end{align*}
in $\Set$, define the map 
\[
  \pdenotid{\Phi,\Gamma, \Delta \vdash \slet x e f : \tau}
\]
as the composition 
\[
  (t_{\pdenot{\Phi}},
  id_{\pdenot{\Gamma} \times \pdenot{\Delta}}); 
  (h_1, id_{\pdenot{\Phi} \times \pdenot{\Delta}}); h_2
\]
% OP
\item[\textbf{Case (Op).}] 
Given the maps
\begin{align*}
h_1 &= \pdenotid{\Phi,\Gamma \vdash e : \num}: 
  \pdenot{\Phi} \times \pdenot{\Gamma} \rightarrow {\pdenot{\num}}\\
h_2 &= \pdenotid{\Phi,\Delta \vdash f : \num} :
  \pdenot{\Phi} \times \pdenot{\Delta} \rightarrow {\pdenot{\num}}
\end{align*}
in $\Set$, define the map
\[
  \pdenotid{\Phi,\Gamma,\Delta \vdash \op ~e ~f : \num} 
\]
as the composition 
\[
  d_{\pdenot{\Phi}};(h_1, h_2); U_{id} \mathcal{L}_{op} =
  d_{\pdenot{\Phi}};(h_1, h_2); f_{op}
\]
for $\op \in \{\mathbf{add},\mathbf{sub},\mathbf{mul},\mathbf{dmul}\}$.
% Div
\item[\textbf{Case (Div).}] Given the maps
\begin{align*}
h_1 &= \pdenotid{\Phi,\Gamma \vdash e : \num}: 
  \pdenot{\Phi} \times \pdenot{\Gamma} \rightarrow {\pdenot{\num}}\\
h_2 &= \pdenotid{\Phi,\Delta \vdash f : \num} :
  \pdenot{\Phi} \times \pdenot{\Delta} \rightarrow {\pdenot{\num}}
\end{align*}
in $\Set$, define the map
\[
  \pdenotid{\Phi,\Gamma,\Delta \vdash \textbf{div} ~ e ~f : \num + \err}
\]
as the composition
\[
  d_{\pdenot{\Phi}};(h_1, h_2); U_{id} \mathcal{L}_{div} =
  d_{\pdenot{\Phi}};(h_1, h_2); f_{div}
\]
\end{description}

\section{Details for Soundness}\label{sec:app_lem_sound}
This appendix provides details for the proofs in \Cref{sec:bean:sound}, which
presents the main backward error soundness theorem for \bea{}. A detailed 
proof of the main theorem is provided in \Cref{sec:app_soundness}.
In this appendix, we provide the details for auxiliary results. 
The full typing relation for \LangS{} is defined by the rules in
\Cref{fig:typing_rules_2_full}, and the operational semantics for \LangS{} are
given in \Cref{fig:op_semantics_full}.
\begin{namedtheorem}[\Cref{lem:derive_ls}]
Let $\Phi \mid \Gamma \vdash e : \tau$ be a well-typed term in \bea{}. 
Then there is a derivation of $\Phi,\Gamma^\circ \vdash e : \tau$ in \LangS. 
\end{namedtheorem}
\begin{proof}
The proof of \Cref{lem:derive_ls} follows by induction on the \bea{} derivation
$\Phi \mid \Gamma \vdash e : \tau$.  Most cases are immediate by application of
the corresponding \LangS{} rule.  The rules for primitive operations require
application of the \LangS{} (Var) rule.  We demonstrate the derivation for the
case of the (Add) rule:
\begin{description}
\item[\textbf{Case (Add).}]
Given a \bea{} derivation of 
\[
  \Phi \mid \Gamma, x :_{\varepsilon + r_1} \num,
    y:_{\varepsilon + r_2}\num \vdash \add x y : \num
\] 
we are required to show a \LangS{} derivation of 
\[
  \Phi,\Gamma, x:\num, y:\num \vdash \add x y : \num
\]
which follows by application of the Var rule for \LangS : \\
  \begin{center}
  \AXC{}
  \LeftLabel{\footnotesize (Var)}
  \UIC{$\Phi,\Gamma, x: \num \vdash x :  \num$}
  \AXC{}
  \LeftLabel{\footnotesize (Var)}
  \UIC{$y : \num \vdash y : \num$}
  \LeftLabel{\footnotesize (Add)}
  \BIC{ $\Phi,\Gamma, x : \num, y:\num \vdash \add x y : \num$ }
  \bottomAlignProof
  \DisplayProof
  \end{center}
\end{description}
\end{proof}

\begin{namedtheorem}[\Cref{lem:pairing}]
Let $\Phi \mid \Gamma \vdash e : \sigma$ be a \bea{} program. Then we have 
\[
\Uid\denot{\Phi \mid \Gamma \vdash e : \sigma} = 
    \pdenotid{\Phi,\Gamma^\circ \vdash e : \sigma} 
\quad \text{and} \quad 
\Uap\denot{\Phi \mid \Gamma \vdash e : \sigma} = 
    \pdenotap{\Phi,\Gamma^\circ \vdash e : \sigma}.
\]
\end{namedtheorem}
\begin{proof}
The proof of \Cref{lem:pairing} follows by induction on the structure of the 
\bea{} derivation $\Phi\mid \Gamma \vdash e : \sigma$. We detail here the 
cases of pairing for the ideal semantics. 
\[
  \Uid\denot{\Phi\mid \Gamma \vdash e : \sigma} = 
  \pdenotid{\Phi,\Gamma \vdash e : \sigma}
\] 
\begin{description}
% VAR
\item[\textbf{Case (Var).}] 
\begin{align*}
  \Uid\denot{\Phi\mid \Gamma, x:_r \sigma \vdash x : \sigma} &= 
  \Uid(\varepsilon_{\denot{\sigma}} \circ m_{0\le r, \denot{\sigma}} \circ
  \pi_i) \tag*{(\Cref{def:interpL})} \\
  &= \pi_i \tag*{(Definition of $\Uid$)} \\
  &= \pdenotid{\Phi,\Gamma^\circ,x:\sigma \vdash x : \sigma }
                                                \tag*{(\Cref{def:interpS})} 
\end{align*}
% DVAR
\item[\textbf{Case (DVar).}] 
\begin{align*}
\Uid\denot{\Phi,z: \alpha\mid\Gamma \vdash z : \sigma} &= \pi_i \tag*{(\Cref{def:interpL})} \\
&= \pdenotid{\Phi,x:\sigma,\Gamma^\circ \vdash x : \sigma } \tag*{(\Cref{def:interpS})}
\end{align*}
% UNIT
\item[\textbf{Case (Unit).}] 
\begin{align*}
  \Uid\denot{\Phi\mid \Gamma \vdash () : \unit} &= f_{unit}
  \tag*{(\Cref{def:interpL})} \\ 
  &= \pdenotid{\Phi\mid \Gamma^\circ \vdash () : \unit} 
                            \tag*{(\Cref{def:interpS})}
\end{align*}
% TENS INTRO
\item[\textbf{Case ($\otimes $ I).}] 
From the induction hypothesis we have
\begin{align*}
  \Uid{(h_1)} &= \pdenotid{\Phi,\Gamma \vdash e : \sigma} \\
  \Uid{(h_2)} &= \pdenotid{\Phi,\Delta \vdash f : \tau}
\end{align*}
We conclude as follows:
\begin{align*}
\Uid\denot{\Phi \mid \Gamma, \Delta \vdash (e,f) : \sigma \otimes  \tau} &=
  \Uid(t_{\denot{\Phi}} \otimes  id_{\denot{\Gamma \otimes \Delta}}); \Uid(h_1 \otimes
  h_2) \tag*{(\Cref{def:interpL})} \\ 
  &= (t_{\pdenot{\Phi}},id_{\pdenot{\Gamma} \times \pdenot{\Delta}});
  (\Uid(h_1), \Uid(h_2)) \tag*{(Definition of $\Uid$)} \\
  &= \pdenotid{\Phi,\Gamma, \Delta \vdash (e,f) : \sigma \otimes  \tau} 
  \tag*{(IH \& \Cref{def:interpS})}
\end{align*}
% TENS ELIM
\item[\textbf{Case ($\otimes $ E$_\sigma$).}] 
From the induction hypothesis we have 
\begin{align*}
  \Uid{(h_1)} &= \pdenotid{\Phi,\Gamma^\circ \vdash e : \tau_1 \otimes  \tau_2} \\
  \Uid{(h_2)} &= \pdenotid{\Phi,\Delta^\circ,x : \tau_1, y:\tau_2 
  \vdash f : \sigma}
\end{align*}
We conclude with the following: 
\begin{align*}
  &\Uid\denot{\Phi\mid  r+ \Gamma, \Delta \vdash \slet {(x,y)} e f  :\sigma} 
  \nonumber \\ 
  &= \Uid\left(
  h_2 \circ (
  (m^{-1}_{r,\denot{\tau_1},\denot{\tau_2}} \circ 
  D_r(h_1)) \otimes (\varepsilon_{\denot{\Phi}} \circ m_{0\le r, \denot{\Phi}})
  \otimes id_{\denot{\Delta}})
  \circ (t_{D_r\denot{\Phi}} \otimes id_{D_r\denot{\Gamma} \otimes  \denot{\Delta}})
  \right) \tag*{(\Cref{def:interpL})} \\ 
  &= (t{\pdenot{\Phi}},id_{\pdenot{\Gamma} \times  \pdenot{\Delta}});
  (\Uid(h_1), id_{\pdenot{\Phi}}, id_{\pdenot{\Delta}}); \Uid(h_2) 
  \tag*{(Definition of $\Uid$)}\\
  &= (t{\pdenot{\Phi}},id_{\pdenot{\Gamma} \times  \pdenot{\Delta}});
  (\Uid (h_1), id_{\pdenot{\Phi} \times \pdenot{\Delta}}); \Uid(h_2)
  \tag*{(Definition of $\Uid$)} \\
  &= \pdenotid{\Phi,\Gamma^\circ,\Delta^\circ \vdash \slet {(x,y)} e f: \sigma} 
  \tag*{(IH \& \Cref{def:interpS})}
\end{align*} 
% TENS ELIM
\item[\textbf{Case ($\otimes $ E$_\alpha$)}] 
From the induction hypothesis we have 
\begin{align*}
  \Uid{(h_1)} &= \pdenotid{\Phi,\Gamma^\circ \vdash e : \tau_1 \otimes  \tau_2} \\
  \Uid{(h_2)} &= \pdenotid{\Phi,\Delta^\circ,x : \tau_1, y:\tau_2  \vdash f : \sigma}
\end{align*}
We conclude with the following: 
\begin{align*}
  \denot{\Phi\mid  \Gamma, \Delta \vdash \slet {(x,y)} e f  :\sigma}
  &= h_2 \circ (h_1 \otimes  id_{\denot{\phi} \otimes  \denot{\Delta}})
  \circ (t_{\denot{\Phi}} \otimes  id_{\denot{\Gamma} \otimes  \denot{\Delta}})
  \tag*{(\Cref{def:interpL})}\\ 
  &= (\Uid (h_1), id_{\pdenot{\phi} \times \pdenot{\Delta}}); \Uid(h_2) 
  \tag*{(Definition of $\Uid$)} \\
  &= (\Uid (h_1), id_{\pdenot{\Delta}}); \Uid(h_2)
  \tag*{(Definition of $\Uid$)} \\
  &= \pdenotid{\Phi,\Gamma^\circ,\Delta^\circ \vdash 
  \slet {(x,y)} e f: \sigma} \tag*{(IH \&\Cref{def:interpS})}
\end{align*} 
% LET
\item[\textbf{Case (Let).}]
From the induction hypothesis we have
\begin{align*}
  \Uid{(h_1)} &= 
  \pdenotid{\Phi,\Gamma^\circ \vdash e : \tau} \\
  \Uid{(h_2)} &= 
  \pdenotid{\Phi,\Delta^\circ, x : \tau \vdash f : \sigma}.
\end{align*}

We conclude with the following: 
\begin{align*}
  &\Uid \denot{\Phi\mid r + \Gamma, \Delta \vdash \slet x e f: \sigma} \\
  &=\Uid \bigr( 
  h_2 \circ (D_r(h_1) \otimes (\varepsilon_{\denot{\Phi}} \circ m_{0 \leq r, 
  \denot{\Phi}})\otimes id_{\denot{\Delta}}) \\ &\qquad \qquad \qquad \circ 
  (m_{r,\denot{\Phi},\denot{\Gamma}} \otimes id_{D_r \denot{\Phi}} \otimes 
  id_{\denot{\Delta}}) \circ (t_{D_r \denot{\Phi}} \otimes id_{D_r\denot{\Gamma} 
  \otimes \denot{\Delta}}) \bigl)\\
  \hfill \tag*{(\Cref{def:interpL})} \\
  &=(t_{\pdenot{\Phi}}, id_{\pdenot{\Gamma} \times \pdenot{\Delta}});
  (\Uid(h_1),id_{\pdenot{\Phi}},id_{\pdenot{\Delta}}); \Uid(h_2) 
  \tag*{(Definition of $\Uid$)} \\
  &=\pdenot{\Phi^\circ, (r+\Gamma)^\circ, \Delta^\circ \vdash 
  \slet x e f: \sigma} \tag*{(IH \& \Cref{def:interpS})}
\end{align*}
% SUM ELIM
\item[\textbf{Case ($+$ E).}]
From the induction hypothesis, we have 
\begin{align*}
  \Uid{(h_1)}
  &= \pdenotid{\Phi,\Gamma^\circ \vdash e' : \sigma + \tau}\\
  \Uid{(h_2)} 
  &= \pdenotid{\Phi,\Delta^\circ, x :_q \sigma \vdash e: \rho}\\
  \Uid{(h_3)} 
  &= \pdenotid{\Phi,\Delta^\circ, y:_q \tau \vdash f: \rho}
\end{align*}
We conclude with the following: 
\begin{align*}
&\Uid{\denot{\Phi\mid q+\Gamma, \Delta \vdash \case {e'} {x.e} {y.f}: \sigma}} \\
&= \Uid\bigl(  
  [h_2,h_3] \circ \Theta \circ
  ((\eta \circ D_q(h_1)) \otimes
  (\varepsilon_{\denot{\Phi}} \circ m_{0\le r, \denot{\Phi}}) 
  \otimes id_{\denot{\Delta}}) \circ 
  (t_{D_q{\denot{\Phi}}} \otimes id_{D_q\denot{\Gamma} \otimes  \denot{\Delta}})
        \bigr) \\
\hfill \tag*{(\Cref{def:interpL})} \\
&= (t{\pdenot{\Phi}}, id_{\pdenot{\Gamma} \times \pdenot{\Delta}});
(\Uid(h_1), id_{\pdenot{\Phi}},id_{ \pdenot{\Delta}});
\Uid(\Theta) ;[\Uid(h_2),\Uid(h_3)] \\
\hfill \tag*{(Definition of $\Uid$)} \\
&= \pdenot{\Phi,\Gamma^\circ, \Delta^\circ \vdash 
\case {e'} {x.e} {y.f}: \sigma} \tag*{(IH \& \Cref{def:interpS})}
\end{align*}
% SUM INTRO
\item[\textbf{Case (+ I).}]  
From the induction hypothesis, we have
\[\Uid(h) = \pdenot{\Phi,\Gamma^\circ \vdash e : \sigma}\]
\begin{align}
  \Uid(\denot{\Phi\mid \Gamma \vdash \inl e: \sigma + \tau}) 
  &= \Uid\left(in_1\circ h \right)\tag*{(\Cref{def:interpL})} \\
  &= \Uid(h);\Uid(in_1) \tag*{(Definition of $\Uid$)} \\
  &= \pdenotid{\Phi, \Gamma \vdash \inl e : \sigma + \tau}  
  \tag*{(IH \& \Cref{def:interpS})}
\end{align}
% ADD
\item[\textbf{Case (Add).}] 
From \Cref{def:interpL} we have 
\begin{align*}
&\denot{\Phi\mid \Gamma,x:_\varepsilon\num,
    y:_\varepsilon\num \vdash \add x y : \num}\\
& \qquad \qquad =
  \pi_i \circ \dots \circ
  (id_{\denot{\Phi} \otimes \denot{\Gamma}}
  \otimes {L}_{add})
  \circ
  (id_{\denot{\Phi} \otimes \denot{\Gamma}}
  \otimes m_{\varepsilon \le\varepsilon + q, \denot{\num}}
  \otimes m_{\varepsilon \le \varepsilon + r, \denot{\num}}),
\end{align*}
We conclude as follows:
\begin{align*}
  \Uid{\denot{\Phi\mid \Gamma,x:_{\varepsilon}\num,y:_{\varepsilon}\num 
  \vdash \add x y : \num}} 
  &= \pi_i; (id_{\pdenot{\Phi} \otimes \pdenot{\Gamma}}, f_{add}) 
  \tag*{(Definition of $\Uid$)} \\
  &=\pdenotid{\Phi,\Gamma^\circ,x:\num,y:\num 
    \vdash \add x y : \num}  \\ \hfill \tag*{(\Cref{def:interpS})}
\end{align*}
\end{description}
The cases for the remaining arithmetic operations are nearly identical 
to the case for \textbf{Add}.
\end{proof}

\subsection*{Substitution}
\begin{namedtheorem}[\Cref{{thm:subst}}]
Let $\Gamma \vdash e : \tau$ be a well-typed \LangS{} term. Then for any
well-typed substitution $\bar{\gamma} \vDash \Gamma$ of closed values, there
is a derivation $\emptyset \vdash e[\bar{\gamma}/dom(\Gamma)] : \tau.$
\end{namedtheorem}

\begin{proof}
By induction on the structure of the derivation $\Gamma \vdash e : \tau$. The
cases for (Var), (Unit), (Const), and (+ I) are trivial; \LangS{} is a simple
first-order language and the remaining cases are routine.
\begin{description}
% TENS INTRO
\item[\textbf{Case ($\otimes$ I).}]
We have a well-typed substitution of closed values $\bar{\gamma} \vDash
dom(\Phi,\Gamma,\Delta)$ and it is straightforward to show that the induction
hypothesis yields the premises needed for applying the typing rule ($\otimes$
I). The desired conclusion then follows from the definition of substitution. 
% TENS ELIM
\item[\textbf{Case ($\otimes$ E).}] 
We are required to show 
\[
\emptyset \vdash (\slet{(x,y)}{e}{f})[\bar{\gamma}/dom(\Phi,\Gamma,\Delta)] : \sigma
\]
given the well-typed substitution of closed values $\bar{\gamma} \vDash
dom(\Phi,\Gamma,\Delta)$.  From $\bar{\gamma}$ we derive a substitution
$\bar{\gamma}' \vDash \Phi,\Gamma$, and from the induction hypothesis on the
left premise we have $\emptyset \vdash e[\bar{\gamma}'/dom(\Phi,\Gamma)] :
\tau_1 \otimes \tau_2$; by inversion on this hypothesis, we derive a
substitution which allows us to use the induction hypothesis for the right
premise. This provides the premises needed to apply the typing rule ($\otimes$
E). The desired conclusion then follows from the definition of substitution. 
% SUM ELIM
\item[\textbf{Case (+ E).}] 
We are required to show 
\[
\emptyset \vdash (\case{e'}{x.e}{y.e})
  [\bar{\gamma}/dom(\Phi,\Gamma,\Delta)] : \rho
\]
given the well-typed substitution of closed values $\bar{\gamma} \vDash
dom(\Phi,\Gamma,\Delta)$. From $\bar{\gamma}$ we derive a substitution
$\bar{\gamma}' \vDash \Phi,\Gamma$, and from the induction hypothesis on the
left premise we have $\emptyset \vdash e'[\bar{\gamma}'/dom(\Phi,\Gamma)] :
\sigma + \tau$; we first apply inversion to this hypothesis and then reason by
cases to derive a substitution which allows us to use the induction hypothesis
for the right premise. This provides the premises needed to apply the typing
rule (+ E). The desired conclusion then follows from the definition of
substitution.  
% LET
\item[\textbf{Case (Let).}] 
We are required to show 
\[
\emptyset \vdash (\slet{x}{e}{f})
  [\bar{\gamma}/dom(\Phi,\Gamma,\Delta)] : \sigma
\]
given a well-typed substitution of closed values $\bar{\gamma} \vdash
dom(\Phi,\Gamma,\Delta)$.  From $\bar{\gamma}$ we derive a substitution
$\bar{\gamma}' \vDash \Phi,\Gamma$, and from the induction hypothesis on the
left premise we have $\emptyset \vdash e[\bar{\gamma}'/dom(\Phi,\Gamma)]$; by
inversion on this hypothesis, we derive a substitution which allows us to
use the induction hypothesis for the right premise. This provides the premises
needed to apply the typing rule (Let).  The desired conclusion then follows
from the definition of substitution. 
% OP
\item[\textbf{Case (Op).}] 
We have a well-typed substitution of closed values $\bar{\gamma} \vDash
dom(\Phi,\Gamma,\Delta)$ and it is straightforward to show that the induction
hypothesis yields the premises needed for applying the typing rule (Op). 
The desired conclusion then follows from the definition of substitution. 
% Div
\item[\textbf{Case (Div).}] 
We have a well-typed substitution of closed values $\bar{\gamma} \vDash
dom(\Phi,\Gamma,\Delta)$ and it is straightforward to show that the induction
hypothesis yields the premises needed for applying the typing rule (Div). 
The desired conclusion then follows from the definition of substitution. 
\end{description}
\end{proof}

\subsection*{Soundness of $\pdenotid{-}$}
In this section of the appendix, we prove the soundness of our denotational
semantics.  Namely, we show that our interpretation of \LangS{}
(\Cref{def:interpS}) respects the operational semantics given in
\Cref{fig:op_semantics_full}. 

Applications of the symmetry map $s_{X,Y} : X \times Y \rightarrow Y \times X$
are elided for succinctness. Recall the diagonal map $d_X : X \rightarrow X
\times X$ on $\Set$, which is used frequently in the interpretation of
\LangS{}. 

\begin{namedtheorem}[\Cref{thm:soundid}]
Let $\Gamma \vdash e : \tau$ be a well-typed \LangS{} term.  Then for any
well-typed substitution of closed values $\bar{\gamma} \vDash \Gamma$, if
$e[\bar{\gamma}/dom(\Gamma)]\stepid v$ for some value $v$, then
${\pdenot{\Gamma \vdash e : \tau}_{id}\pdenot{\bar{\gamma}}_{id} =
\pdenot{v}_{id}}$ (and similarly for $\stepap$ and $\pdenotap{-}$).
\end{namedtheorem}

\begin{proof}
By induction on the structure of the \LangS{} derivations $\Gamma \vdash e :
\tau$. The cases for (Var), (Unit), (Const), and (+ I) are trivial. In each 
case we apply inversion on the step relation to obtain the premise for the 
induction hypothesis. 
\begin{description}
% TENS INTRO
\item[\textbf{Case ($\otimes $ I).}] 
We are required to show
\[
\pdenotid{\Phi, \Gamma, \Delta \vdash (e,f) : \sigma \otimes  \tau}
\pdenotid{\bar{\gamma}} = \pdenotid{(u,v)}
\] 
for some well-typed closed substitution  
$\bar{\gamma} \vDash \Phi, \Gamma, \Delta$ 
and value $(u,v)$ such that  
\[
(e,f)[\bar{\gamma}'/dom(\Phi,\Gamma,\Delta)] \stepid (u,v)
\]
From $\bar{\gamma}$ we derive the substitutions $\bar{\gamma}' \vDash \Phi$,
$\bar{\gamma}_1 \vDash \Gamma$, and $\bar{\gamma}_2 \vDash \Delta$.
By inversion on the step relation we then have 
\begin{align*}
e[\bar{\gamma}',\bar{\gamma}_1/dom(\Phi,\Gamma)] &\stepid u \\
f[\bar{\gamma}',\bar{\gamma}_2/dom(\Phi,\Delta)] &\stepid v
\end{align*}
We conclude as follows: 
\begin{align*}
  \pdenotid{\Phi, \Gamma, \Delta \vdash (e,f) : \sigma \otimes  \tau}
  \pdenotid{\bar{\gamma}} 
  &= \bigl((d_{\pdenot{\Phi}},id_{\pdenot{\Gamma}},id_{\pdenot{\Delta}});
  (\pdenotid{\Phi,\Gamma \vdash e : \sigma}, 
  \pdenotid{\Phi,\Delta \vdash f : \tau})\bigr) \pdenotid{\bar{\gamma}} \\
  \hfill \tag*{(\Cref{def:interpS})}\\
  &= \bigl(\pdenotid{\Phi,\Gamma \vdash e : \sigma}, 
  \pdenotid{\Phi,\Delta \vdash f : \tau}\bigr) 
  \bigl(\pdenot{\bar{\gamma}'},\pdenot{\bar{\gamma}_1},
  \pdenot{\bar{\gamma}'},\pdenot{\bar{\gamma}_2}\bigr) \\
  \hfill \tag*{(Definition of $d_{\pdenot{\Phi}}$)} \\
  &= (\pdenotid{u},\pdenotid{v}) \tag*{(IH)}
\end{align*}
% TENS ELIM
\item[\textbf{Case ($\otimes $ E).}] 
We are required to show 
\[
\pdenotid{\Phi,\Gamma,\Delta \vdash \slet {(x,y)} e f: \sigma}
\pdenotid{\bar{\gamma}} = \pdenotid{w}
\] 
for some well-typed closed substitution $\bar{\gamma} \vDash
\Phi,\Gamma,\Delta$ and value $w$ such that 
\[
(\slet {(x,y)} e f)[\bar{\gamma}/dom(\Phi,\Gamma,\Delta)] \stepid w
\]
From $\bar{\gamma}$ we derive the substitutions
$\bar{\gamma}' \vDash \Phi$, $\bar{\gamma}_1 \vDash \Gamma$,
and $\bar{\gamma}_2 \vDash \Delta$. By inversion on the step 
relation we then have 
\begin{align*}
e[\bar{\gamma}',\bar{\gamma}_1/dom(\Phi,\Gamma)] &\stepid (u,v) \\
f[\bar{\gamma}',\bar{\gamma}_2/dom(\Phi,\Delta)][u/x][v/y] &\stepid w
\end{align*}
We conclude as follows: 
\begin{align*}
  &\pdenotid{\Phi,\Gamma,\Delta \vdash \slet {(x,y)} e f: \sigma}
  \pdenotid{\bar{\gamma}} \\
  &= \bigl(  (d_{\pdenot{\Phi}},id_{\pdenot{\Gamma}},id_{\pdenot{\Delta}}); 
  (h_1, id_{\pdenot{\Phi} \times \pdenot{\Delta}});h_2\bigr) 
  \pdenotid{\bar{\gamma}} \tag*{(\Cref{def:interpS})}\\
  &= \bigl((h_1, id_{\pdenot{\Phi} \times \pdenot{\Delta}});h_2\bigr) 
  \bigl(\pdenot{\bar{\gamma}'},\pdenot{\bar{\gamma}_1},
  \pdenot{\bar{\gamma}'},\pdenot{\bar{\gamma}_2}\bigr) 
  \tag*{(Definition of $d_{\pdenot{\Phi}}$)} \\
  &= \bigl(\pdenotid{\Phi,\Delta,x:\tau_1,y:\tau_2 \vdash f : \sigma}\bigr) 
  \bigl(\pdenot{\bar{\gamma}'},\pdenot{\bar{\gamma}_2},\pdenot{u}, 
  \pdenot{v}\bigr) \tag*{(IH)} \\
  &= \pdenotid{w} \tag*{(IH)}
\end{align*}
% SUM ELIM
\item[\textbf{Case ($+$ E).}] 
We are required to show 
\[
\pdenotid{\Phi,\Gamma, \Delta \vdash \case {e'} {x.e} {y.f} : \rho}
\pdenotid{\bar{\gamma}} = \pdenotid{w}
\]
for some well-typed closed substitution $\bar{\gamma} \vDash
\Phi,\Gamma,\Delta$ and value $w$ such that 
\[
(\case {e'} {x.e} {y.f})[\bar{\gamma}/dom(\Phi,\Gamma,\Delta)] \stepid w
\]

We consider the case when $e' = \inl e_1$ for some $e_1 : \sigma$. 
From $\bar{\gamma}$ we derive the substitutions
$\bar{\gamma}' \vDash \Phi$, $\bar{\gamma}_1 \vDash \Gamma$,
and $\bar{\gamma}_2 \vDash \Delta$. By inversion on the step 
relation we then have 
\begin{align*}
e'[\bar{\gamma}',\bar{\gamma}_1/dom(\Phi,\Gamma)] &\stepid \inl v \\
e[\bar{\gamma}',\bar{\gamma}_2/dom(\Phi,\Delta)][v/x] &\stepid w 
\end{align*}
We conclude as follows: 
\begin{align*}
  &\pdenotid{\Phi,\Gamma, \Delta \vdash \case {e'} {x.e} {y.f} : \rho}
  \pdenotid{\bar{\gamma}} \\
  &= \bigl((d_{\pdenot{\Phi}},
  id_{\pdenot{\Gamma} \times \pdenot{\Delta}}); 
  (h_1, id_{\pdenot{\Phi} \times \pdenot{\Delta}});
  \Theta^S_{\pdenot{\Phi} \times \pdenot{\Delta},\pdenot{\sigma},
  \pdenot{\tau}};[h_2,h_3]\bigr) 
  \pdenotid{\bar{\gamma}} \tag*{(\Cref{def:interpS})}\\
  &= \bigl((h_1, id_{\pdenot{\Phi} \times \pdenot{\Delta}});
  \Theta^S_{\pdenot{\Phi} \times \pdenot{\Delta},\pdenot{\sigma},
  \pdenot{\tau}};[h_2,h_3]\bigr) 
  \bigl(\pdenot{\bar{\gamma}'},\pdenot{\bar{\gamma}_1},
  \pdenot{\bar{\gamma}'},\pdenot{\bar{\gamma}_2}\bigr) 
  \tag*{(Definition of $d_{\pdenotid{\Phi}}$)} \\
  &= \bigl(\Theta^S_{\pdenot{\Phi} \times \pdenot{\Delta},\pdenot{\sigma},
  \pdenot{\tau}};[h_2,h_3]\bigr) 
  \bigl(\pdenot{\bar{\gamma}'},\pdenot{\bar{\gamma}_2},
  \pdenot{\inl v}\bigr) \tag*{(IH)} \\
  &= \pdenotid{w} \tag*{(IH)}
\end{align*}
% LET 
\item[\textbf{Case (Let).}]
We are required to show 
\[
  \pdenotid{\Gamma, \Delta \vdash \slet x e f : \tau}
  \pdenotid{\bar{\gamma}} = 
  \pdenotid{v}
\]

for some well-typed closed substitution $\bar{\gamma} \vDash
\Phi,\Gamma,\Delta$ and value $w$ such that 
\[
(\slet {x} e f)[\bar{\gamma}/dom(\Phi,\Gamma,\Delta)] \stepid v
\]
From $\bar{\gamma}$ we derive the substitutions
$\bar{\gamma}' \vDash \Phi$, $\bar{\gamma}_1 \vDash \Gamma$,
and $\bar{\gamma}_2 \vDash \Delta$. By inversion on the step 
relation we then have 
\begin{align*}
e[\bar{\gamma}',\bar{\gamma}_1/dom(\Phi,\Gamma)] &\stepid u \\
f[\bar{\gamma}',\bar{\gamma}_2/dom(\Phi,\Delta)][u/x] &\stepid v
\end{align*}
We conclude as follows: 
\begin{align*}
  &\pdenotid{\Phi,\Gamma,\Delta \vdash \slet {x}{e}{f}: \sigma}
  \pdenotid{\bar{\gamma}} \\
  &= \bigl(  (d_{\pdenot{\Phi}},id_{\pdenot{\Gamma}},id_{\pdenot{\Delta}}); 
  (h_1, id_{\pdenot{\Phi} \times \pdenot{\Delta}});h_2\bigr) 
  \pdenotid{\bar{\gamma}} \tag*{(\Cref{def:interpS})}\\
  &= \bigl((h_1, id_{\pdenot{\Phi} \times \pdenot{\Delta}});h_2\bigr) 
  \bigl(\pdenot{\bar{\gamma}'},\pdenot{\bar{\gamma}_1},
  \pdenot{\bar{\gamma}'},\pdenot{\bar{\gamma}_2}\bigr) 
  \tag*{(Definition of $d_{\pdenot{\Phi}}$)} \\
  &= \bigl(\pdenotid{\Phi,\Delta,x:\tau_1 \vdash f : \sigma}\bigr) 
  \bigl(\pdenot{\bar{\gamma}'},\pdenot{\bar{\gamma}_2},\pdenot{u} \bigr) 
  \tag*{(IH)} \\
  &= \pdenotid{v} \tag*{(IH)}
\end{align*}
% OP 
\item[\textbf{Case (Op).}] 
We are required to show 
\[
\pdenotid{\Phi, \Gamma, \Delta \vdash \textbf{Op} ~e ~f : \num }
  \pdenotid{\bar{\gamma}} = \pdenotid{f_{op}(k_1,k_2)}
\]
for some well-typed closed substitution $\gamma \vDash \Phi, \Gamma, \Delta$
and value $f_{op}(k_1,k_2)$ such that 
\[
(\Phi, \Gamma, \Delta \vdash \textbf{Op} ~e ~f)
  [\bar{\gamma}/dom(\Phi,\Gamma,\Delta)]\stepid ~ f_{op}(k_1,k_2)
\]
From $\bar{\gamma}$ we derive the substitutions
$\bar{\gamma}' \vDash \Phi$, $\bar{\gamma}_1 \vDash \Gamma$,
and $\bar{\gamma}_2 \vDash \Delta$. By inversion on the step 
relation we then have 
\begin{align*}
e[\bar{\gamma}',\bar{\gamma}_1/dom(\Phi,\Gamma)] &\stepid k_1 \\
f[\bar{\gamma}',\bar{\gamma}_2/dom(\Phi,\Delta)][u/x] &\stepid k_2
\end{align*}
We conclude as follows: 
\begin{align*}
  &\pdenotid{\Phi,\Gamma,\Delta \vdash \textbf{Op} ~ e ~f: \num}
  \pdenotid{\bar{\gamma}} \\
  &= \bigl(d_{\pdenot{\Phi}};(h_1, h_2); f_{op}\bigr) 
  \pdenotid{\bar{\gamma}} \tag*{(\Cref{def:interpS})}\\
  &= \bigl((h_1, h_2);f_{op}\bigr) 
  \bigl(\pdenot{\bar{\gamma}'},\pdenot{\bar{\gamma}_1},
  \pdenot{\bar{\gamma}'},\pdenot{\bar{\gamma}_2}\bigr) 
  \tag*{(Definition of $d_{\pdenot{\Phi}}$)} \\
  &= \pdenotid{f_{op}(k_1,k_2)} \tag*{(IH)}
\end{align*}
% OP 
\item[\textbf{Case (Div).}] Identical to the proof for (Op).
\end{description}
\end{proof}

\subsection*{Adequacy of $\pdenotid{-}$}
In this section of the appendix, we prove the computational adequacy of our
denotational semantics.  Namely, we show that if two \LangS{} terms are equal
under our denotational semantics (\Cref{def:interpS}), then they will evaluate
to the same value under our operational semantics given in
\Cref{fig:op_semantics_full}. 

\begin{namedtheorem}[\Cref{thm:adequacy}]
Let $\Gamma \vdash e : \tau$ be a well-typed \LangS{} term. Then for any
well-typed substitution of closed values $\bar{\gamma} \vDash \Gamma$, if
$\pdenot{\Gamma \vdash e : \tau}_{id}\pdenot{\bar{\gamma}}_{id} =
\pdenot{v}_{id}$ for some value $v$, then $e[\bar{\gamma}/dom(\Gamma)]\stepid
v$ (and similarly for $\stepap$ and $\pdenotap{-}$).
\end{namedtheorem}

\begin{proof}
The proof follows directly by cases on $e$. Many cases are immediate and the
remaining cases, given that \LangS{} is deterministic, follow by substitution
(\Cref{thm:subst}) and normalization (\Cref{thm:normalizing}).  We show two
representative cases. 
\begin{description}
% VAR
\item[Case.] 
Given $\Gamma, x : \sigma, \Delta \vdash x : \sigma$ and $\pdenotid{\Gamma, x :
\sigma, \Delta \vdash x : \sigma} \pdenotid{\bar{\gamma}} = \pdenotid{v}$ for
some value $v$ and some well-typed substitution $\bar{\gamma} \vDash \Gamma,
x:\sigma, \Delta$ we are required to show 
\[
x[\bar{\gamma}/dom(\Gamma, x : \sigma, \Delta)] \stepid v
\] 
which follows by substitution (\Cref{thm:subst}) and normalization
(\Cref{thm:normalizing}). 
% TENS ELIM
\item[Case.] 
Given $\Gamma, \Delta \vdash \slet x e f :  \tau$ and $\pdenotid{\Gamma, \Delta
\vdash \slet x e f : \tau} \pdenotid{\bar{\gamma}} = \pdenotid{w}$ for some
value $w$ and some well-typed derivation $\bar{\gamma} \vDash \Gamma, \Delta$
we are required to show 
\[
(\slet x e f)[\bar{\gamma}/dom(\Gamma,\Delta)] \stepid w
\] 
which follows by substitution (\Cref{thm:subst}) and normalization
(\Cref{thm:normalizing}).  
\end{description}
\end{proof}

\section{Proof of Backward Error Soundness}\label{sec:app_soundness}
This appendix provides a detailed proof of the main backward error soundness
theorem for \bea{} (\Cref{thm:main}).

\begin{namedtheorem}[\Cref{thm:main}]
Let $ \Phi\mid  x_1:_{r_1}\sigma_1,\cdots,x_n:_{r_n}\sigma_n = 
    \Gamma \vdash e : \sigma$ be a well-typed \bea{} term. Then for 
any well-typed substitutions 
$\bar{p} \vDash \Phi$ and $\bar{k} \vDash \Gamma^\circ$, if 
\[e[\bar{p}/dom(\Phi)][\bar{k}/dom(\Gamma)] \stepap v\] 
for some value $v$, then the well-typed substitution $\bar{l} \vDash \Gamma^\circ$ 
exists such that \[e[\bar{p}/dom(\Phi)][\bar{l}/dom(\Gamma)] \stepid v,\] and 
$d_{\denot{\sigma_i}}({k}_i,{l}_i) \le r_i$ for each $k_i \in \bar{k}$ 
and $l_i \in \bar{l}$.
\end{namedtheorem}

\begin{proof}
From the lens semantics (\Cref{def:interpL}) of \bea{} we have the triple 
\[ \denot{\Phi\mid \Gamma \vdash e : \sigma} = (f,\tilde{f},b) : 
   \denot{\Phi} \otimes  \denot{\Gamma} \rightarrow \denot{\sigma}. \]
Then, using the backward map $b$, we can define the tuple of vectors 
of values 
$\label{eq:bmapdef}
	(\bar{s},\bar{l}) \triangleq b((\bar{p},\bar{k}),v)
$
such that $\bar{s} \vDash \Phi$ and $\bar{l} \vDash \Gamma$. 

From the second property of backward error lenses we then have  
$$f\pdenotid{(\bar{s},\bar{l})} = f\pdenotid{b((\bar{p},\bar{k}),v)} 
  = v.$$

We can now show a backward error result, i.e.,
$\tilde{f}\pdenotap{(\bar{p},\bar{k})} = f\pdenotid{(\bar{s},\bar{l})}$:
\begin{align}
\pdenotap{\Phi,\Gamma \vdash e : \sigma} \pdenotap{(\bar{p},\bar{k})}
	&= \Uap{\denot{\Phi \mid \Gamma \vdash e : \sigma}} 
		\tag*{(\Cref{lem:pairing})}\pdenotap{(\bar{p},\bar{k})} \\
	&=\tilde{f}\pdenotap{(\bar{p},\bar{k})}\tag*{(\Cref{def:interpL})}\\
	&= v \tag*{(\Cref{thm:soundid})} \\ 
	&= f\pdenotid{(\bar{s},\bar{l})} \tag*{} 
\end{align}

From the first property of error lenses we have 
$d_{\denot{\Phi}\otimes \denot{\Gamma}}\left( (\bar{p},\bar{k}), 
    b((\bar{p},\bar{k}),v)\right) 
  \le d_{\denot{\sigma}}\left( \tilde{f}(\bar{p},\bar{k}), v\right)
$
so long as  
\begin{align}
d_{\denot{\sigma}}\left( \tilde{f}(\bar{p},\bar{k}), v\right) 
 =d_{\denot{\sigma}}\left( v, v\right) \neq \infty 
  \label{eq:error_assum}.
\end{align}
If the base numeric type is interpreted as a metric
space with a standard distance function, then 
$d_{\denot{\sigma}}\left( v, v\right) \neq \infty$ for any type 
$\sigma$, and so \Cref{eq:error_assum} is satisfied.

Unfolding definitions, and using the fact that 
$\tilde{f}\pdenotap{(\bar{p},\bar{k})} = v$ 
from above, we have 
\begin{align}
\max\left(d_{\denot{\Phi}}(\bar{p},\bar{s}), 
    d_{\denot{\Gamma}}({\bar{k},\bar{l}})\right) &\le 
    d_{\denot{\sigma}}\left( v, v\right) \label{eq:error_ineq}
\end{align}
From \Cref{eq:error_ineq} we can conclude two things. First,
using the definition of the distance function on discrete metric 
spaces, we can conclude 
$\bar{p}=\bar{s}$: the discrete variables carry no backward error. 
Second, for linear variables we can derive the required 
backward error bound:
\begin{align*} 
    \max\left( d_{\denot{\sigma_1}}(k_1,l_1)-r_1,\dots,
     d_{\denot{\sigma_n}}(k_n,l_n)-r_n\right) &\le 0.
\end{align*}
\end{proof}

\section{Type Checking Algorithm}\label{app:algorithm}
\begin{figure}
\begin{center}
%% ROW 1
%% var
\AXC{ }
\RightLabel{(Var)}
\UIC{$\Phi\mid\Gamma^\bullet,x:\sigma;x\Rightarrow  \{x:_0\sigma\};\sigma$}
\bottomAlignProof
\DisplayProof
\hskip 0.5em
%% dvar
\AXC{}
\RightLabel{(DVar)}
\UIC{$\Phi,z:\alpha\mid\Gamma^\bullet;z\Rightarrow \emptyset; \alpha$}
\bottomAlignProof
\DisplayProof
\vskip 1em
%% ROW 2
%% prod intro
\AXC{$\Phi\mid\Gamma^\bullet;e\Rightarrow \Gamma_1; \sigma$}
\AXC{$\Phi\mid\Gamma^\bullet;f\Rightarrow \Gamma_2; \tau$}
\AXC{$\dom\Gamma_1\cap\dom\Gamma_2=\emptyset$}
\RightLabel{($\otimes $ I)}
\TrinaryInfC{$\Phi\mid\Gamma^\bullet; (e,f)\Rightarrow \Gamma_1,\Gamma_2; \sigma\otimes\tau$}
\bottomAlignProof
\DisplayProof
\vskip 1em
%% ROW 3
%% unit
\AXC{}
\RightLabel{(Unit)}
\UIC{$\Phi\mid\Gamma^\bullet;()\Rightarrow \emptyset; \mathbf{unit}$}
\bottomAlignProof
\DisplayProof
\vskip 1em
%% ROW 4
%% prod elim
\AXC{$\Phi\mid\Gamma^\bullet;e\Rightarrow \Gamma_1;\tau_1\otimes\tau_2$}
\AXC{$\Phi\mid\Gamma^\bullet,x:\tau_1, y:\tau_2; f\Rightarrow \Gamma_2 ; \sigma$}
\AXC{$\dom\Gamma_1\cap\dom\Gamma_2=\emptyset$}
\RightLabel{($\otimes $ E$_\sigma$)}
\TrinaryInfC{$\Phi\mid \Gamma^\bullet; \slet {(x,y)} e f \Rightarrow 
    (r+\Gamma_1),\Gamma_2\setminus \{x,y\}; \sigma$}
\bottomAlignProof
\DisplayProof
\vskip 0.5em
where $x,y\not\in\Gamma^\bullet$ and $r=\max\{r_1,r_2\}$ if at least one of
$x:_{r_1}\tau_1,y:_{r_2}\tau_2\in\Gamma_2$ (else $r=0$)
\vskip 1em
%% ROW 5
%% prod elim discrete
\AXC{$\Phi\mid\Gamma^\bullet;e\Rightarrow  \Gamma_1;\alpha_1\otimes\alpha_2$}
\AXC{$\Phi,z_1:\alpha_1,z_2:\alpha_2\mid\Gamma^\bullet;f\Rightarrow \Gamma_2;\sigma$}
\AXC{$\dom\Gamma_1\cap\dom\Gamma_2=\emptyset$}
\RightLabel{($\otimes $ E$_\alpha$)}
\TrinaryInfC{$\Phi\mid \Gamma^\bullet; 
        \dlet {(z_1,z_2)} e f \Rightarrow \Gamma_1,\Gamma_2;\sigma$}
\bottomAlignProof
\DisplayProof
\vskip 0.5em
where $z_1,z_2\not\in \Phi$
\vskip 1em
%% ROW 6
% case
\AXC{$\Phi\mid\Gamma^\bullet;e'\Rightarrow\Gamma_1;\sigma+\tau$}
\AXC{$\Phi\mid \Gamma^\bullet,x:\sigma;e\Rightarrow \Gamma_2;\rho$}
\AXC{$\Phi\mid\Gamma^\bullet,y:\tau;f\Rightarrow \Gamma_3;\rho$}
\AXC{$\dom\Gamma_1\cap\dom\Gamma_2$}
\noLine
\UnaryInfC{$=\dom\Gamma_1\cap\dom\Gamma_3$}
\noLine
\UnaryInfC{$=\emptyset$}
\RightLabel{($+$ E)}
\QuaternaryInfC{$\Phi\mid\Gamma^\bullet;\mathbf{case} \ e' \ \mathbf{of} \ (x.e \ | \ y.f) 
    \Rightarrow (q+\Gamma_1),\max\{\Gamma_2\setminus \{x\}, \Gamma_3\setminus\{y\}\};\rho$}
\bottomAlignProof
\DisplayProof
\vskip 0.5em
where $x,y\not\in \Gamma^\bullet$ and $q=\max\{q_1,q_2\}$ if at least one of $x:_{q_1}\sigma\in\Gamma_2$ 
or $y:_{q_2}\tau\in\Gamma_3$ (else $q=0$)
\vskip 1em
%% ROW 7
%% ind sum intro
\AXC{$\Phi\mid \Gamma^\bullet;e\Rightarrow \Gamma;\sigma$}
\RightLabel{($+$ $\text{I}_L$)}
\UIC{$\Phi\mid\Gamma^\bullet; \inl_\tau \ e \Rightarrow\Gamma;\sigma + \tau$}
\bottomAlignProof
\DisplayProof
\hskip 0.5em
%% ind sum intro
\AXC{$\Phi\mid \Gamma^\bullet;e\Rightarrow\Gamma;\tau$}
\RightLabel{($+$ $\text{I}_R$)}
\UIC{$\Phi\mid\Gamma^\bullet; \inr_\sigma \ e \Rightarrow \sigma + \tau$}
\bottomAlignProof
\DisplayProof
\vskip 1em
%% ROW 8
% let 
\AXC{$\Phi\mid\Gamma^\bullet;e\Rightarrow\Gamma_1;\tau$}
\AXC{$\Phi\mid\Gamma^\bullet,x:\tau;f\Rightarrow\Gamma_2;\sigma$}
\AXC{$\dom\Gamma_1\cap\dom\Gamma_2=\emptyset$}
\RightLabel{(Let)}
\TrinaryInfC{$\Phi\mid \Gamma^\bullet; \slet x e f
 \Rightarrow (r+\Gamma_1),\Gamma_2\setminus\{x\}; \sigma$}
\bottomAlignProof
\DisplayProof
\vskip 0.5em
where $x\not\in\Gamma^\bullet$ and $x:_r\sigma\in\Gamma_2$ (else $r=0$)
\vskip 1em
%% ROW 9
%% disc
\AXC{$\Phi\mid\Gamma^\bullet;e\Rightarrow\Gamma;\num$}
\RightLabel{(Disc)}
\UIC{$\Phi\mid\Gamma^\bullet;!e\Rightarrow \Gamma;\dnum$}
\bottomAlignProof
\DisplayProof
\vskip 1em
%% ROW 10
%% dlet
\AXC{$\Phi\mid \Gamma^\bullet;e\Rightarrow\Gamma_1;\dnum$}
\AXC{$\Phi,z:\dnum\mid \Gamma^\bullet;f\Rightarrow\Gamma_2;\sigma$}
\AXC{$\dom\Gamma_1\cap\dom\Gamma_2=\emptyset$}
\RightLabel{(DLet)}
\TIC{$\Phi\mid\Gamma^\bullet; \dlet{z}{e}{f}\Rightarrow \Gamma_1,\Gamma_2;\sigma$}
\bottomAlignProof
\DisplayProof
\vskip 0.5em
where $z\not\in\Phi$
\vskip 1em
%% ROW 11
% add, sub
\AXC{}
\RightLabel{(Add, Sub)}
\UIC{$\Phi\mid \Gamma^\bullet, x: \num, y:\num;
    \{\mathbf{add}, \mathbf{sub}\} \ x \ y \Rightarrow 
    \{x:_{\varepsilon}\num,y:_{\varepsilon}\num\};\num$}
\bottomAlignProof
\DisplayProof
\vskip 1em
%%% ROW 12
% mul 
\AXC{ }
\RightLabel{(Mul)}
\UIC{$\Phi\mid \Gamma^\bullet, 
    x:\num, 
    y:\num; \mul  x  y \Rightarrow 
    \{x:_{\varepsilon/2}\num, y:_{\varepsilon/2}\num\};\num $}
\bottomAlignProof
\DisplayProof
\vskip 1em
%% ROW 13
% div 
\AXC{ }
\RightLabel{(Div)}
\UIC{$\Phi\mid \Gamma^\bullet, x: \num, y:\num;
    \{\mathbf{add}, \mathbf{sub}\} \ x \ y \Rightarrow 
    \{x:_{\varepsilon}\num,y:_{\varepsilon}\num\};\num + \err$}
\bottomAlignProof
\DisplayProof
\vskip 1em
%%% ROW 14
% dmul 
\AXC{ }
\RightLabel{(DMul)}
\UIC{$\Phi,z: \dnum\mid 
        \Gamma^\bullet, 
    x:\num; \dmul  z  x \Rightarrow \{x:_\varepsilon \num\}; \num$}
\bottomAlignProof
\DisplayProof

\end{center}
    \caption{Type checking algorithm for \bea{}.}
    \label{fig:algorithm}
\end{figure}

This appendix, authored by Laura Zielinski, presents the type-checking algorithm for \bea{} 
as described in \Cref{sec:algorithm}, along with proofs of its soundness and completeness.
First, we give the full type checking algorithm in \Cref{fig:algorithm}.
Recall that algorithm calls are written as $\Phi\mid\Gamma^\bullet;e\Rightarrow \Gamma;\sigma$
where $\Gamma^\bullet$ is a linear context skeleton, $e$ is a \bea{} program,
$\Gamma$ is, intuitively, the \emph{minimal} linear context required to type $e$ such that
$\overline{\Gamma}\sqsubseteq\Gamma^\bullet$, and $\sigma$ is the type of $e$. 
Note that we only require $\Phi$ to contain the discrete variables used in the program 
and we do nothing more; thus, it is not returned by the algorithm.
We do require that discrete and linear contexts are always disjoint, and we will denote
linear variables by $x$ and $y$ and discrete variables by $z$. 
Finally, we define the \emph{max} of two linear contexts, $\max\{\Gamma,\Delta\}$, 
to have domain $\dom\Gamma\cup\dom\Delta$ and, if $x:_q\sigma\in\Gamma$ and $x:_r\sigma\in\Delta$, 
then $x:_{\max\{q,r\}}\sigma\in\max\{\Gamma,\Delta\}$. 

Before we give proofs of \Cref{thm:algo_sound} and \Cref{thm:algo_complete},
we must prove two lemmas about type system and algorithm weakening. 
Intuitively, type system weakening says that if we can derive the type of a program from a context $\Gamma$, 
then we can also derive the same program from a larger context $\Delta$ which subsumes $\Gamma$.
\begin{lemma}[Type System Weakening]\label{lem:type_weak}
  If $\Phi\mid\Gamma\vdash e:\sigma$ and $\Gamma\sqsubseteq\Delta$, then $\Phi\mid\Delta\vdash e:\sigma$.
\end{lemma}
\begin{proof}
  Suppose $\Phi\mid\Gamma\vdash e:\sigma$. We proceed by induction on the final typing rule applied
  and show some representative cases below. 
  \begin{description}
    \item[\textbf{Case (Var).}] Suppose the last rule applied was 
    \[
      \Phi\mid\Gamma,x:_r\sigma\vdash x:\sigma.
    \]
    Let $\Delta$ be a context such that $(\Gamma,x:_r\sigma)\sqsubseteq\Delta$. Thus, $x:_q\sigma\in\Delta$ 
    where $r\leq q$. By the same rule, $\Phi\mid\Delta\vdash x:\sigma$. 

    \item[\textbf{Case ($\otimes $ I).}] Suppose the last rule applied was
    \[
      \Phi\mid\Gamma,\Delta\vdash (e,f):\sigma\otimes\tau
    \]
    and thus, we also have that 
    \[
      \Phi\mid\Gamma\vdash e:\sigma\text{ and }\Phi\mid\Delta\vdash f:\tau.
    \]
    Let $\Lambda$ be a context such that $(\Gamma,\Delta)\sqsubseteq\Lambda$.
    As $\Gamma$ and $\Delta$ are disjoint, we can split $\Lambda$ into the contexts 
    $\Gamma_1$ and $\Delta_1$ such that $\Gamma\sqsubseteq\Gamma_1$ and $\Delta\sqsubseteq\Delta_1$.
    By out inductive hypothesis, it follows that 
    \[
      \Phi\mid\Gamma_1\vdash e:\sigma\text{ and }\Phi\mid\Delta_1\vdash f:\tau.
    \]
    By the same rule, we conclude that 
    \[
      \Phi\mid\Gamma_1,\Delta_1\vdash (e,f):\sigma\otimes\tau.
    \]

    \item[\textbf{Case ($\otimes $ E$_\sigma$).}] Suppose the last rule applied 
    was 
    \[
      \Phi\mid r+\Gamma,\Delta\vdash\slet{(x, y)}e f:\sigma.
    \]
    Let $\Lambda$ be a context such that $(r+\Gamma,\Delta)\sqsubseteq\Lambda$ and 
    $x,y\not\in\dom\Lambda$. As before, split $\Lambda$ into contexts $\Gamma_1$ and $\Delta_1$
    such that $(r+\Gamma)\sqsubseteq\Gamma_1$ and $\Delta\sqsubseteq\Delta_1$ but where 
    $\dom\Gamma=\dom\Gamma_1$. Now, for each $x\in\dom\Gamma_1$, we have that $x:_q\sigma\in\Gamma_1$
    where $r\leq q$. Therefore, we can define the context $-r+\Gamma_1$ which subtracts $r$ from the 
    error bound of every variable in $\Gamma_1$, and hence $\Gamma\sqsubseteq (-r+\Gamma_1)$. 
    Finally, use our inductive hypothesis to get that
    \[
      \Phi\mid(-r+\Gamma_1)\vdash e:\tau_1\otimes\tau_2\text{ and }
      \Phi\mid\Delta_1,x:_r\tau_1,y:_r\tau_2\vdash f:\sigma
    \]
    and we can apply the same rule to get our conclusion.

    \item[\textbf{Case (Add).}] Suppose the last rule applied was 
    \[
      \Phi\mid\Gamma,x:_{\varepsilon+r_1}\num,y:_{\varepsilon+r_2}\num\vdash\add{x}{y}:\num
    \]
    Let $\Delta$ be a context such that
    $(\Gamma,x:_{\varepsilon+r_1}\num,y:_{\varepsilon+r_2}\num)\sqsubseteq\Delta$.
    Hence, $x:_{q_1}\num,y:_{q_2}\num\in\Delta$ where $\varepsilon+r_1\leq q_1$ and $\varepsilon+r_2\leq q_2$.
    Rewrite $q_1=\varepsilon + (q_1-\varepsilon)$ and $q_2=\varepsilon+(q_2-\varepsilon)$ and apply the same rule.
  \end{description}
\end{proof}
Similarly, algorithm weakening says that if we pass a context skeleton $\Gamma^\bullet$ into the algorithm and it
infers context $\Gamma$, then if we pass in a larger skeleton $\Delta^\bullet$, the algorithm will 
still infer context $\Gamma$. (Here, we extend the notion of subcontexts to context skeletons, where
$\Gamma^\bullet \sqsubseteq\Delta^\bullet$ if $\Gamma^\bullet\subseteq\Delta^\bullet$.)
This is because the algorithm discards unused variables from the context.
\begin{lemma}[Type Checking Algorithm Weakening]\label{lem:algo_weak}
  If $\Phi\mid\Gamma^\bullet;e\Rightarrow\Gamma;\sigma$ and $\Gamma^\bullet\sqsubseteq\Delta^\bullet$, 
  then $\Phi\mid\Delta^\bullet;e\Rightarrow \Gamma;\sigma$.
\end{lemma}
\begin{proof}
  Suppose that $\Phi\mid\Gamma^\bullet;e\Rightarrow\Gamma;\sigma$. 
  We proceed by induction on the final algorithmic step applied and 
  show some representative cases below.
  \begin{description}
    \item[\textbf{Case (Var).}] Suppose the last step applied was 
    \[
      \Phi\mid\Gamma^\bullet,x:\sigma;x\Rightarrow\{x:_0\sigma\};\sigma.
    \]
    Let $\Delta^\bullet$ be a context skeleton such that 
    $(\Gamma^\bullet,x:\sigma)\sqsubseteq\Delta^\bullet$. Thus, $x:\sigma\in\Delta^\bullet$
    so we can apply the same rule.

    \item[\textbf{Case ($\otimes $ E$_\sigma$).}] Suppose the last step applied 
    was 
    \[
      \Phi\mid\Gamma^\bullet;\slet{(x, y)} e f\Rightarrow 
      (r+\Gamma_1),\Gamma_2\setminus\{x,y\};\sigma.
    \]
    Let $\Delta^\bullet$ be a context skeleton such that $\Gamma^\bullet\sqsubseteq\Delta^\bullet$
    and $x,y\not\in\Delta^\bullet$. By induction, we have that 
    \[
      \Phi\mid\Delta^\bullet;e\Rightarrow\Gamma_1;\tau_1\otimes\tau_2\text{ and }
      \Phi\mid\Delta^\bullet,x:\tau_1,y:\tau_2;f\Rightarrow\Gamma_2;\sigma
    \]
    and $\dom\Gamma_1\cap\dom\Gamma_2=\emptyset$. By the same rule, we conclude that
    \[
      \Phi\mid\Delta^\bullet;\slet{(x,y)}{e}{f}\Rightarrow (r+\Gamma_1),\Gamma_2\setminus\{x,y\};\sigma.
    \]
  \end{description}
\end{proof}
Finally, we give proofs of algorithmic soundness and completeness. 
Soundness states that if the algorithm returns a linear context $\Gamma$, then we can use $\Gamma$
to derive the program using \bea{}'s type system.
\begin{namedtheorem}[\Cref{thm:algo_sound}]
  If $\Phi\mid\Gamma^\bullet;e\Rightarrow \Gamma;\sigma$, then 
  $\overline{\Gamma}\sqsubseteq\Gamma^\bullet$ and the derivation 
  $\Phi\mid\Gamma\vdash e:\sigma$ exists.
\end{namedtheorem}
\begin{proof}
  Suppose that $\Phi\mid\Gamma^\bullet;e\Rightarrow\Gamma;\sigma$. 
  We proceed by induction on the final algorithmic step applied and show some 
  representative cases below. We use the fact that if $\Gamma,\Delta$ are disjoint,
  then $\overline{\Gamma,\Delta}=\overline\Gamma,\overline\Delta$. 
  \begin{description}
    \item[\textbf{Case (Var).}] Suppose the last step applied was 
    \[
      \Phi\mid\Gamma^\bullet,x:\sigma;x\Rightarrow\{x:_0\sigma\};\sigma.
    \]
    By the typing rule (Var$_\sigma$), we have that $\Phi\mid\{x:_0\sigma\}\vdash x:\sigma$.
    Moreover,
    $\overline{\{x:_0\sigma\}}\sqsubseteq(\Gamma^\bullet,x:\sigma)$.

    \item[\textbf{Case ($\otimes $ I).}] Suppose the last step applied was
    \[
      \Phi\mid\Gamma^\bullet;(e,f)\Rightarrow\Gamma_1,\Gamma_2;\sigma\otimes\tau
    \]
    where 
    \[
      \Phi\mid\Gamma^\bullet;e\Rightarrow\Gamma_1;\sigma\text{ and }
      \Phi\mid\Gamma^\bullet;f\Rightarrow\Gamma_2;\sigma
    \]
    and $\dom\Gamma_1\cap\dom\Gamma_2=\emptyset$. By our inductive hypothesis, 
    we have that $\Phi\mid\Gamma_1\vdash e:\sigma$ and $\Phi\mid\Gamma_2\vdash f:\tau$.
    Therefore, we can apply the typing rule ($\otimes$ I) to get that 
    \[
      \Phi\mid\Gamma_1,\Gamma_2\vdash (e,f):\sigma\otimes\tau.
    \]
    Finally, as $\overline{\Gamma_1}\sqsubseteq\Gamma^\bullet$ and 
    $\overline{\Gamma_2}\sqsubseteq\Gamma^\bullet$,
    we have that $\overline{\Gamma_1,\Gamma_2}\sqsubseteq\Gamma^\bullet$.

    \item[\textbf{Case ($\otimes $ E$_\sigma$).}] Suppose the last step applied 
    was 
    \[
      \Phi\mid\Gamma^\bullet;\slet{(x, y)} e f\Rightarrow 
      (r+\Gamma_1),\Gamma_2\setminus\{x,y\};\sigma
    \]
    By induction, we have that
    \[
      \Phi\mid\Gamma_1\vdash e:\tau_1\otimes\tau_2\text{ and }\Phi\mid\Gamma_2\vdash f:\sigma,
    \]  
    where $x,y$ may be in $\dom\Gamma_2$.
    Let $\Delta=\Gamma_2\setminus\{x,y\}$. 
    Since $r$ is defined to be the maximum of the bounds on $x,y$ if they exist in 
    $\Gamma_2$, we have that $\Gamma_2\sqsubseteq(\Delta,x:_r\tau_1,y:_r\tau_2)$. 
    From \Cref{lem:type_weak}, it follows that
    \[
      \Phi\mid\Delta,x:_r\tau_1,y:_r\tau_2\vdash f:\sigma.
    \]
    Thus, we can apply the typing rule ($\otimes$ E$_\sigma$) to conclude that 
    \[
      \Phi\mid r+\Gamma_1,\Delta\vdash \slet{(x, y)}e f:\sigma.
    \]
    Finally, as $\overline{\Gamma_1}\sqsubseteq\Gamma^\bullet$ and 
    $\overline{\Gamma_2}\sqsubseteq(\Gamma^\bullet,x:\tau_1,y:\tau_2)$, we have that
    \[
      \overline{r+\Gamma_1,\Delta}=\overline{\Gamma_1,\Gamma_2\setminus\{x,y\}}
      =\overline{\Gamma_1},\overline{\Gamma_2\setminus\{x,y\}}\sqsubseteq\Gamma^\bullet.
    \]

    \item[\textbf{Case (+ E).}] Suppose the last step applied was 
    \[
      \Phi\mid\Gamma^\bullet;\mathbf{case} \ e' \ \mathbf{of} \ (x.e \ | \ y.f) 
    \Rightarrow (q+\Gamma_1),\max\{\Gamma_2\setminus \{x\}, \Gamma_3\setminus\{y\}\};\rho.
    \]
    By induction, we have that 
    \[
      \Phi\mid\Gamma_1\vdash e':\sigma+\tau\text{ and }
      \Phi\mid\Gamma_2\vdash e:\rho\text{ and }\Phi\mid\Gamma_3\vdash f:\rho.
    \]
    Let $\Delta=\max\{\Gamma_2\setminus\{x\}, \Gamma_3\setminus\{y\}\}$, and we
    still have that $\dom\Gamma_1\cap\dom\Delta=\emptyset$. By \Cref{lem:type_weak}, 
    it follows that 
    \[
      \Phi\mid\Delta,x:_q\sigma\vdash e:\rho\text{ and }\Phi\mid\Delta,y:_q\tau\vdash f:\rho
    \]
    by weakening the bounds on $x$ and $y$ to $q$. Thus, we can apply typing rule (+ E) to 
    conclude that 
    \[
      \Phi\mid q+\Gamma_1,\Delta\vdash \case{e'}{x.e}{y.f}:\rho.
    \]
    Moreover, as $\overline{\Gamma_1}\sqsubseteq\Gamma^\bullet$ and 
    $\overline{\Gamma_2}\sqsubseteq(\Gamma^\bullet,x:\sigma)$ and 
    $\overline{\Gamma_3}\sqsubseteq(\Gamma^\bullet,y:\tau)$, we have that 
    \[
      \overline{q+\Gamma_1,\Delta}=\overline{\Gamma_1},
      \overline{\max\{\Gamma_2\setminus\{x\},\Gamma_3\setminus\{y\}\}}\sqsubseteq\Gamma^\bullet.
    \]

    \item[\textbf{Case (Add).}] Suppose the last step applied was 
    \[
      \Phi\mid \Gamma^\bullet, x: \num, y:\num;
      \mathbf{add} \ x \ y \Rightarrow 
      \{x:_{\varepsilon}\num,y:_{\varepsilon}\num\};\num.
    \]
    By the typing rule (Add) we have that 
    \[
      \Phi\mid\{x:_\varepsilon\num,y:_\varepsilon\num\}\vdash\mathbf{add}~x~y:\num.
    \]
  \end{description}
\end{proof}
Conversely, completeness says that if from $\Gamma$ we can derive the type of a program $e$, then
inputting $\overline{\Gamma}$ and $e$ into the algorithm will yield a valid output. 
\begin{namedtheorem}[\Cref{thm:algo_complete}]
  If $\Phi\mid\Gamma\vdash e:\sigma$ is a valid derivation in \bea{}, then there
  exists a context $\Delta\sqsubseteq\Gamma$ such that 
  $\Phi\mid\overline{\Gamma}; e\Rightarrow \Delta;\sigma$.
\end{namedtheorem}
\begin{proof}
  Suppose that $\Phi\mid\Gamma\vdash e:\sigma$. We proceed by induction on the final typing rule applied
  and show some representative cases below.
  \begin{description}
    \item[\textbf{Case (Var$_\sigma$).}] Suppose the last rule applied was 
    \[
      \Phi\mid\Gamma,x:_r\sigma\vdash x:\sigma.
    \]
    By algorithm step (Var), we have that
    \[
      \Phi\mid\overline{\Gamma},x:\sigma;x\Rightarrow\{x:_0\sigma\};\sigma
    \]
    and $\{x:_0\sigma\}\sqsubseteq(\Gamma,x:_r\sigma)$ as $0\leq r$. 

    \item[\textbf{Case ($\otimes $ I).}] Suppose the last rule applied was
    \[
      \Phi\mid\Gamma,\Delta\vdash (e,f):\sigma\otimes\tau.
    \]
    From this, we deduce that $\dom\Gamma\cap\dom\Delta=\emptyset$. By induction, there exist
    $\Gamma_1\sqsubseteq\Gamma$ and $\Delta_1\sqsubseteq\Delta$ such that 
    \[
      \Phi\mid\overline\Gamma;e\Rightarrow\Gamma_1;\sigma\text{ and }
      \Phi\mid\overline\Delta;f\Rightarrow\Delta_1;\tau.
    \]
    Moreover, $\dom\Gamma_1\cap\dom\Delta_1=\emptyset$ as well. 
    By \Cref{lem:algo_weak}, we also have that 
    \[
      \Phi\mid\overline{\Gamma,\Delta};e\Rightarrow\Gamma_1;\sigma\text{ and }
      \Phi\mid\overline{\Gamma,\Delta};f\Rightarrow\Delta_1;\tau.
    \]
    Thus, we can apply algorithm step
    ($\otimes$ I) to conclude
    \[
      \Phi\mid\overline{\Gamma,\Delta};(e,f)\Rightarrow \Gamma_1,\Delta_1;\sigma\otimes\tau
    \] 
    where we know $(\Gamma_1,\Delta_1)\sqsubseteq(\Gamma,\Delta)$. 

    \item[\textbf{Case ($\otimes $ E$_\sigma$).}] Suppose the last rule applied 
    was 
    \[
      \Phi\mid r+\Gamma,\Delta\vdash\slet{(x, y)}e f:\sigma.
    \]
    By induction, there exist $\Gamma_1\sqsubseteq\Gamma$ and 
    $\Delta_1\sqsubseteq(\Delta,x:_r\tau_1,y:_r\tau_2)$ such that 
    \[
      \Phi\mid \overline{\Gamma};e\Rightarrow \Gamma_1;\tau_1\otimes\tau_2\text{ and }
      \Phi\mid \overline{\Delta},x:\tau_1,y:\tau_2;f\Rightarrow\Delta_1;\sigma.
    \]
    If $x:_{r_1}\tau_1, y:_{r_2}\tau_2\in\Delta_1$, let $r'=\max\{r_1,r_2\}$. As 
    $\Delta_1\sqsubseteq(\Delta,x:_r\tau_1,y:_r\tau_2)$, we know $r'\leq r$. 

    Using \Cref{lem:algo_weak}, we can apply algorithm step ($\otimes$ E$_\sigma$) of
    \[
      \Phi\mid\overline{\Gamma,\Delta};\slet{(x, y)}e f \Rightarrow (r'+\Gamma_1),\Delta_1\setminus\{x,y\};\sigma.
    \]
    Moreover, $(r'+\Gamma_1)\sqsubseteq(r+\Gamma)$ and $(\Delta_1\setminus\{x,y\})\sqsubseteq\Delta$.

    \item[\textbf{Case (+ E).}] Suppose the last rule applied was 
    \[
      \Phi\mid q+\Gamma,\Delta\vdash \case{e'}{x.e}{y.f}:\rho.
    \]
    By induction, there exist $\Gamma_1\sqsubseteq\Gamma$, $\Delta_1\sqsubseteq(\Delta,x:_q\sigma)$, and 
    $\Delta_2\sqsubseteq(\Delta,y:_q\tau)$ such that 
    \[
      \Phi\mid\overline{\Gamma};e'\Rightarrow \Gamma_1;\sigma+\tau\text{ and }
      \Phi\mid\overline{\Delta},x:\sigma;e\Rightarrow \Delta_1;\rho\text{ and }
      \Phi\mid\overline{\Delta}, y:\tau;f\Rightarrow\Delta_2;\rho.
    \]
    If $x:_{q_1}\sigma\in\Delta_1$ and $y:_{q_2}\tau\in\Delta_2$, let $q'=\max\{q_1,q_2\}$,
    and we know $q'\leq q$. Using \Cref{lem:algo_weak}, we can apply algorithm step (+ E) of
    \[  
      \Phi\mid\overline{\Gamma,\Delta};\case{e'}{x.e}{y.f}\Rightarrow (q'+\Gamma_1),
      \max\{\Delta_1\setminus\{x\},\Delta_2\setminus\{y\}\};\rho.
    \]
    Furthermore, we know $(q'+\Gamma_1)\sqsubseteq(q+\Gamma)$ and 
    $\max\{\Delta_1\setminus\{x\},\Delta_2\setminus\{y\}\}\sqsubseteq \Delta$.

    \item[\textbf{Case (Add).}] Suppose the last rule applied was 
    \[
      \Phi\mid\Gamma,x:_{\varepsilon+r_1}\num,y:_{\varepsilon+r_2}\num\vdash\add{x}{y}:\num
    \]
    where $r_1,r_2\geq 0$. We can apply algorithm step (Add) of
    \[
      \Phi\mid\overline{\Gamma},x:\num,y:\num;\add{x}{y}\Rightarrow \{x:_\varepsilon\num,y:_\varepsilon\num\};\num
    \]
    and we have $\{x:_\varepsilon\num,y:_\varepsilon\num\}\sqsubseteq
    (\Gamma,x:_{\varepsilon+r_1}\num,y:_{\varepsilon+r_2}\num)$.
  \end{description}
\end{proof}

\end{document}